%% file: Main.tex
\RequirePackage[l2tabu, orthodox]{nag}

\documentclass[12pt,phd,a4paper,oneside]{ucl_thesis}

\input{MainPackages}

\input{LinksAndMetadata}


\title{Scaling Distributed Ledgers and Privacy-Preserving Applications}
\author{Alberto Sonnino}
\department{Department of Computer Science}

\usepackage{url}
\usepackage{xspace}
\usepackage{xcolor}
\usepackage{color, colortbl}
\usepackage{cleveref}
\usepackage{amsfonts,amssymb}
\usepackage{amsthm,amsmath}
\usepackage{graphicx}
\usepackage{pifont}
\usepackage{array}
\usepackage{booktabs}
\usepackage{tikz}
\usepackage{bm}
\usepackage{todonotes}
\usepackage{epigraph}
\usepackage{algpseudocode}
\usepackage{subcaption}
\usepackage{arydshln}
\usepackage{makecell}
\usepackage{tabularx}
\usepackage[usestackEOL]{stackengine}
\usepackage{authblk}
\usepackage{microtype}
\usepackage{siunitx}
\usepackage{bussproofs}
\usepackage{booktabs}
\usepackage{tabu}
\usepackage{epigraph}
\usepackage[font=footnotesize]{caption}
\usepackage{listings}
\usepackage{wrapfig}
\usepackage{mdframed}
\usepackage{enumitem}

\def\cameraReady{} % set to true
\let\cameraReady\undefined % set to false

\input{macros}

\begin{document}

\nobibliography*
% This is a dumb trick that works with the bibentry package to let
% you put bibliography entries whereever you like.
% I used this to put references to papers a chapter's work was 
% published in at the end of that chapter.
% For more information, see: http://stefaanlippens.net/bibentry

% If you haven't finished making your full BibTex file yet, you
% might find this useful -- it'll just replace all your
% citations with little superscript notes.
% Uncomment to use.
%\renewcommand{\cite}[1]{\emph{\textsuperscript{[#1]}}}

% At last, content! Remember filenames are case-sensitive and 
% *must not* include spaces.

\include{Preamble}

\include{chapters/introduction/introduction}
\include{chapters/literature-review/literature-review}
\include{chapters/chainspace/chainspace}
\include{chapters/byzcuit/byzcuit}

\include{chapters/fastpay/fastpay}
\include{chapters/coconut/coconut}
\include{chapters/conclusion/conclusion}

% This line manually adds the Bibliography to the table of contents.
% The fact that \include is the last thing before this ensures that it
% is on a clear page, and adding it like this means that it doesn't
% get a chapter or appendix number.
\addcontentsline{toc}{chapter}{Bibliography}

% Actually generates your bibliography.
\bibliography{references}

% All done. \o/
\end{document}

%% file: MainPackages.tex
\usepackage{blindtext}

\usepackage{emptypage}

\usepackage{graphicx}

\usepackage{float}

\usepackage{amsmath}

\usepackage{gensymb}
\usepackage{textcomp}
\usepackage{setspace}
\usepackage{multirow}

\usepackage{bibentry} 

\usepackage[format=hang,font=small,labelfont=bf]{caption}

\usepackage{etoolbox}

%% file: LinksAndMetadata.tex
\usepackage{bibentry}
\makeatletter\let\saved@bibitem\@bibitem\makeatother
\usepackage[pdftex,hidelinks]{hyperref}
\makeatletter\let\@bibitem\saved@bibitem\makeatother
\makeatletter
\AtBeginDocument{
    \hypersetup{
        pdfsubject={Thesis Subject},
        pdfkeywords={Thesis Keywords},
        pdfauthor={Author},
        pdftitle={Title},
    }
}
\makeatother

%% file: macros.tex
\newcommand\alberto[1]{\todo[color=yellow,inline]{\textbf{Alberto:} #1}}
\newcommand\mustafa[1]{\todo[color=green,inline]{\textbf{Mustafa:} #1}}
\newcommand\bano[1]{\todo[color=orange,inline]{\textbf{Bano:} #1}}
\newcommand\george[1]{\todo[color=brown,inline]{\textbf{George:} #1}}

\ifdefined\cameraReady
\renewcommand\alberto[1]{}
\renewcommand\mustafa[1]{}
\renewcommand\bano[1]{}
\renewcommand\george[1]{}
\fi
\newcommand\byzcuitreplay{\texttt{byzcuit}\xspace}
\newcommand\byzcuitbaseline{\texttt{byzcuit-baseline}\xspace}

\newcommand\chainspace{Chainspace\xspace}
\newcommand\ethereum{Ethereum\xspace}
\newcommand\hyperledger{Hyperledger\xspace}
\newcommand\omniledger{Omniledger\xspace}
\newcommand\rapidchain{RapidChain\xspace}
\newcommand\elgamal{El-Gamal\xspace}
\newcommand\coconut{Coconut\xspace}

\newcommand\sbac{S-BAC\xspace}
\newcommand\bft{BFT\xspace}
\newcommand\atomix{Atomix\xspace}
\newcommand\cscoin{CSCoin\xspace}
\newcommand\rscoin{RSCoin\xspace}

\newcommand\bftsmart{\textsc{bft-SMaRt}\xspace}

\newcommand{\byzcuit}{Byzcuit\xspace}
\newcommand\bitcoin{Bitcoin\xspace}
\newcommand\bitcoinng{Bitcoin-NG\xspace}

\newcommand\byzcoin{ByzCoin\xspace}
\newcommand\elastico{Elastico\xspace}
\newcommand\algorand{Algorand\xspace}
\newcommand\bigchaindb{BigchainDB\xspace}
\newcommand\hyperledgerfabric{Hyperledger Fabric\xspace}
\newcommand\pbft{PBFT\xspace} %bibtex: pbft
\newcommand\bftlong{byzantine fault-tolerant\xspace}

\newcommand\solidus{Solidus\xspace} %bibtex: solidus
\newcommand\randhound{RandHerd\xspace} %bibtex: randhound
\newcommand\ouroboros{Ouroboros\xspace} %bibtex: ouroboros
\newcommand\ouroborospraos{Ouroboros Praos\xspace} %bibtex: ouroborospraos
\newcommand\snowwhite{Snow-White\xspace} %bibtex: snowwhite
\newcommand\avalanche{Avalanche\xspace} %bibtex: avalanche
\newcommand\tendermint{Tendermint\xspace} %bibtex: tendermint
\newcommand\stellar{Stellar\xspace} %bibtex: mazieres2015stellar
\newcommand\fastpay{FastPay\xspace} %bibtex: mazieres2015stellar

\newcommand\sysmain{Primary\xspace}

\newcommand\blockchain{blockchain\xspace}
\newcommand\blockchains{blockchains\xspace}
\newcommand\algorithm[1]{\textsf{#1}}
\newcommand\hashtopoint{H^*\xspace}

\newcommand\preacceptt{\textsf{pre-accept}($T$)\xspace}
\newcommand\preabortt{\textsf{pre-abort}($T$)\xspace}
\newcommand\preaccepttt{\textsf{pre-accept}($T'$)\xspace}

\newcommand\preacceptts{\textsf{pre-accept}($T,s_T$)\xspace}
\newcommand\preabortts{\textsf{pre-abort}($T,s_T$)\xspace}

\newcommand\acceptt{\textsf{accept}($T$)\xspace}
\newcommand\abortt{\textsf{abort}($T$)\xspace}

\newcommand\aborttt{\textsf{abort}($T'$)\xspace}

\newcommand\acceptts{\textsf{accept}($T,s_T$)\xspace}
\newcommand\abortts{\textsf{abort}($T,s_T$)\xspace}

\newcommand\myrow[1]{row\xspace\textsf{#1}\xspace}
\newcommand\mycolumn[1]{column\xspace\textsf{#1}\xspace}
\newcommand\shardled{shard-led\xspace}
\newcommand\Shardled{Shard-led\xspace}
\newcommand\clientled{client-led\xspace}
\newcommand\Clientled{Client-led\xspace}

\newcommand\locked{\emph{locked}\xspace}
\newcommand\attacker{attacker\xspace}

\newcommand\prerecorded{prerecorded\xspace}
\newcommand\prerecords{prerecords\xspace}

\newcommand\activeObj{`active'\xspace}
\newcommand\inactiveObj{`inactive'\xspace}
\newcommand\smr{SMR\xspace}
\newcommand\pow{PoW\xspace}
\newcommand\pox{PoX\xspace}

\newcommand\pos{proof-of-Stake\xspace}

\newcommand\dos{DoS\xspace}

\newcommand\sybil{sybil\xspace}
\newcommand\poxStake{proof-of-stake\xspace}

\newcommand\vrf{VRF\xspace}
\newcommand\preimage{pre-image\xspace}
\newcommand\stakeholders{stakeholders\xspace}

\newcommand{\keyword}[1]{\normalfont \texttt{#1}}
\newcommand{\accounts}{\keyword{accounts}}
\newcommand{\balance}{\keyword{balance}}
\newcommand{\totalbalance}{\keyword{total\_balance}}
\newcommand{\val}{\keyword{value}}
\newcommand{\sequencenumber}{\keyword{sequence}}
\newcommand{\lasttransactionindex}{\keyword{last\_transaction}}
\newcommand{\transactionindex}{\keyword{transaction\_index}}
\newcommand{\nextsequencenumber}{\keyword{next\_sequence}}
\newcommand{\redeemlog}{\keyword{redeemed}}
\newcommand{\synchronizationlog}{\keyword{synchronized}}
\newcommand{\funding}{\keyword{funding}}

\newcommand{\confirmedlog}{\keyword{confirmed}}
\newcommand{\pendingconfirmation}{\keyword{pending}}
\newcommand{\amount}{\keyword{amount}}
\newcommand{\sender}{\keyword{sender}}
\newcommand{\recipient}{\keyword{recipient}}
\newcommand{\received}{\keyword{received}}
\newcommand{\transactions}{\keyword{transactions}}
\newcommand{\fastpayaddresses}{\keyword{FastPay}}
\newcommand{\sysmainaddresses}{\keyword{Primary}}
\newcommand{\transfer}{O}
\newcommand{\cert}{C}
\newcommand{\sync}{S}

\newcommand{\authority}{\alpha}

\newcommand{\para}[1]{\vspace{2mm}\noindent\textbf{#1}.\xspace}
\newcommand\mypara[1]{\vspace{0.05in} \noindent \textbf{#1.}}

\newcommand\definition[2]{\ding{118}\xspace \textsf{#1}\xspace$\bm{\rightarrow}$\xspace(#2):\xspace}

\newcommand{\vs}{vs.\@\xspace}
\newcommand{\etc}{etc.\@\xspace}

\newcommand{\etal}{\textit{et al.}\@\xspace}
\newcommand{\eg}{\textit{e.g.}\@\xspace}
\newcommand{\ie}{\textit{i.e.}\@\xspace}
\newcommand{\via}{\textit{via}\@\xspace}

\newtheorem{SecAssumption}{Security Assumption}
\newtheorem{theorem}{Security Theorem}

\newtheorem{lemma}{Lemma}
\newtheorem{proposition}{Proposition}

\def\first{({i})\xspace}
\def\second{({ii})\xspace}
\def\third{({iii})\xspace}
\def\fourth{({iv})\xspace}
\def\fifth{({v})\xspace}
\def\sixth{({vi})\xspace}

\definecolor{verylightgray}{gray}{0.9}

\newcolumntype{L}{l<{\hspace{1cm}}}
\newcolumntype{C}{c<{\hspace{1cm}}}
\newcolumntype{D}{c<{\hspace{0.3cm}}}

\newcommand{\yes}{\ding{51}}
\newcommand{\no}{\ding{55}}
\DeclareRobustCommand\pie[1]{
\tikz[every node/.style={inner sep=0,outer sep=0, scale=1.5}]{
\node[minimum size=1.5ex] at (0,-1.5ex) {}; 
 \draw[fill=white] (0,-1.5ex) circle (0.75ex); \draw[fill=black] (0.75ex,-1.5ex) arc (0:#1:0.75ex); 
}
}
\def\L{\pie{0}} % Low
\def\M{\pie{-180}} % Medium
\def\H{\pie{360}} % High

\newcommand\type{\mathsf{type}}
\newcommand\types{\mathsf{types}}
\newcommand\id{\mathsf{id}}
\newcommand\procedures{\mathsf{proc}}
\newcommand\csobject{o}
\newcommand\contract{c}
\newcommand\csprocedure{p}

\newcommand\pinputs{\vec{w}}
\newcommand\preferences{\vec{r}}
\newcommand\poutputs{\vec{x}}

\newcommand\lparams{\mathsf{lpar}}
\newcommand\lreturns{\mathsf{lret}}
\newcommand\sparams{\mathsf{spar}}
\newcommand\sreturns{\mathsf{sret}}
\newcommand\dependencies{\mathsf{dep}}

\newcommand\faulty{f}
\newcommand\checker{v}
\newcommand\transaction{T\xspace}
\newcommand\shard{\phi}

\newcommand\bi{\begin{itemize}}
\newcommand\ei{\end{itemize}}
\newcommand\ben{\begin{enumerate}}
\newcommand\een{\end{enumerate}}

\lstdefinelanguage{alg}{
morekeywords={fn, let, mut, ensure, match, bail, if, None, Some, return, Result, Ok},
morecomment=[l]///,
}

%% file: Preamble.tex
\maketitle
\makedeclaration

% =========
\begin{abstract} % 300 word limit: check https://wordcounter.net/
This thesis proposes techniques aiming to make blockchain technologies and smart contract platforms practical by improving their scalability, latency, and privacy.
This thesis starts by presenting the design and implementation of \chainspace, a distributed ledger that supports user defined smart contracts and execute user-supplied transactions on their objects. The correct execution of smart contract transactions is publicly verifiable. \chainspace is scalable by sharding state; it is secure against subsets of nodes trying to compromise its integrity or availability properties through Byzantine Fault Tolerance (BFT).
This thesis also introduces a family of replay attacks against sharded distributed ledgers targeting cross-shard consensus protocols; they allow an attacker, with network access only, to double-spend resources with minimal efforts. We then build \byzcuit, a new cross-shard consensus protocol that is immune to those attacks and that is tailored to run at the heart of \chainspace. 
Next, we propose \fastpay, a high-integrity settlement system for pre-funded payments that can be used as a financial side-infrastructure for \chainspace to support low-latency retail payments. This settlement system is based on Byzantine Consistent Broadcast as its core primitive, foregoing the expenses of full atomic commit channels (consensus). The resulting system has extremely low-latency for both confirmation and payment finality. 
Finally, this thesis proposes \coconut, a selective disclosure credential scheme supporting distributed threshold issuance, public and private attributes, re-randomization, and multiple unlinkable selective attribute revelations. It ensures authenticity and availability even when a subset of credential issuing authorities are malicious or offline, and natively integrates with \chainspace to enable a number of scalable privacy-preserving applications.
\end{abstract}

% =========
\chapter*{Impact Statement}
The work in this thesis can inform the design of new and existing projects that implement distributed ledgers, smart contract platforms or applications, in order to increase their scalability, security, and privacy.
\chainspace (\Cref{chainspace}) and \coconut (\Cref{coconut}) are used as part of DECODE (DEcentralized Citizen-owned Data Ecosystems)~\cite{rojo2018}, a European project with a digital democracy pilot in Barcelona implementing a decentralized petitions platform. \chainspace was also commercialized in a company (\texttt{chainspace.io}) co-founded by the author of this thesis, and the team was acquired by Facebook. 
\coconut was also integrated into Cosmos SDK~\cite{coconut-cosmos} and commercialized in a company (Nym~\cite{nym}) aiming to provide an open-ended anonymous overlay network that disguise patterns in internet traffic; it uses the scheme described in \Cref{coconut} as anonymous authentication credentials to enable privacy-enhanced data transfer and decentralized identity.
\byzcuit and the work done on replay attacks against sharded distributed ledgers (\Cref{byzcuit}) have profound impact on the security of recently proposed systems such as \omniledger~\cite{omniledger} and \rapidchain~\cite{rapidchain}; these systems were presented at top security conferences and form the basis of numerous start-ups and open-source projects such as Harmony~\cite{harmony}.
Finally, the content of this thesis has been presented on multiple occasions, both at academic venues and industry conferences; some of its chapters have been published at top-tier security conferences. It produced multiple tools and open-source software, and is freely available online.

% =========
\begin{acknowledgements}
This work would not have been possible without my primary supervisor George Danezis, who helped me throughout the past few years of research, and to whom I wish to express my profound gratitude.
Special thanks to my secondary supervisors Jens Groth and Ioannis Psaras for their unwavering support and generous encouragement, and to my close collaborators, Mustafa Al-Bassam and Shehar Bano, who have been the source of many fruitful discussions.
I have been privileged to have had the opportunity to work with many brilliant and helpful people around the world. Specifically, I thank all my co-authors (in alphabetic order) Christos Andrikos, Sarah Azouvi, Lejla Batina, Mathieu Baudet, Vitalik Buterin, Avery Ching, Lukasz Chmielewski, Andrey Chursin, Dave Hrycyszyn, Ismail Khoffi, Micha\l{} Kr\'{o}l,  Liran Lerman, Zekun Li, Dahlia Malkhi, Vasilis Mavroudis, Sarah Meiklejohn, Patrick McCorry, Kostas Papagiannopoulos, Dmitri Perelman, Guilherme Perin, Giorgos Rassias, Etienne Rivi\`ere, Argyrios Tasiopoulos, Lixia Zhang, and Zhiyi Zhang.
I would also like to thank Ramsey Khoury, my other co-founder at \texttt{chainspace.io} not mentioned above, and all my other former colleagues at \texttt{chainspace.io}: Penny Andrews, Andy Bennett, Stuart Chinery, J\'er\'emy Letang, and Lola Oyelayo-Pearson. I also thank everyone at Facebook Novi, specifically David Marcus, James Everingham, Christian Catalini, Kevin Weil, Ben Maurer, Morgan Beller, and Evan Cheng for providing a supportive working environment.
Finally, I would like to thank the European Commission for funding my research with a PhD scholarship, and all my friends and family for their continuous support and encouragement.
\end{acknowledgements}

\setcounter{tocdepth}{2} 
% Setting this higher means you get contents entries for
%  more minor section headers.

\tableofcontents
\listoffigures
\listoftables

%% file: chapters/introduction/introduction.tex
% =========
% Introduction
% =========
\chapter{Introduction} \label{introduction}
\setlength\epigraphwidth{.82\textwidth}
\epigraph{The root problem with conventional currency is all the trust that's required to make it work. The central bank must be trusted not to debase the currency, but the history of fiat currencies is full of breaches of that trust. Banks must be trusted to hold our money and transfer it electronically, but they lend it out in waves of credit bubbles with barely a fraction in reserve. We have to trust them with our privacy, trust them not to let identity thieves drain our accounts.}{\textit{Satoshi Nakamoto}}

% =========
\section{Problem Statement}
Blockchains lie at the foundation of \bitcoin and other cryptocurrencies, which have a total global market capital of over \$450B as of November 2020~\cite{marketcap}. The blockchain is a decentralized, replicated, immutable and tamper-evident log: data on the blockchain cannot be deleted, and anyone can read data from the blockchain and verify its correctness.  
The blockchain is maintained by a set of authorities (called \emph{nodes}) that form a distributed network. An important implication of this architecture is \emph{disintermediation}: multiple untrusted or semi-trusted parties can \emph{directly} and \emph{transparently} interact with each other without the presence of a trusted intermediary.  This makes blockchains immediately relevant to banks and financial institutions which incur huge middleman costs in settlements and other back office operations.  A number of big players are actively exploring the feasibility of blockchains, including the Bank of England~\cite{BoE}, the Bank of America~\cite{BoA} and the IMF~\cite{IMF}. 
In addition to the financial industry, blockchains have been employed in a diverse array of use cases, ranging from voting~\cite{follow-my-vote}, through data storage~\cite{filecoin}, to the sharing economy~\cite{pastrami, asterisk}. Despite their useful properties and applications, adoption of blockchains is nowhere near as ubiquitous as their traditional counterparts due to their performance limitations. These properties are deeply related to the \emph{consensus} protocol---the core component of the blockchain.   

Consensus protocols are defined by two key properties. The first is related to performance, and requires that requests from correct clients are eventually processed (\emph{liveness}).  The second property is related to security, and states that if an honest node accepts (or rejects) a value then all other honest nodes make the same decision (\emph{safety/consistency}). 
A plethora of consensus protocols exist that offer different trade-offs between the two key properties of consensus liveness and consistency. The distributed systems community has extensively studied consensus for over two decades, and developed robust and practical protocols that can tolerate faulty and malicious nodes~\cite{pbft,lamport1998part}.  
However, these protocols were designed for closed groups, and cannot be readily adapted to blockchains. 

\bitcoin's fundamental innovation was to enable consensus among nodes forming a peer-to-peer network~\cite{bitcoin}.  This was achieved via a leader election based on proof-of-work (\pow): all nodes attempt to find the solution to a hash puzzle and the node that wins adds the next block to the blockchain.   Due to its probabilistic leader election process combined with performance fluctuations in decentralized networks, \bitcoin offers only weak consistency: different nodes might end up having different views of the blockchain leading to \emph{forks}.
Additionally, \bitcoin suffers from poor performance and its \pow consumes a huge amount of energy~\cite{iceland}.  
\bitcoin's underlying blockchain technology suffers from scalability issues: with a current block size of 1MB and 10 minute inter-block interval, throughput is capped at about 7 transactions per second, and a client that creates a transaction has to wait for about 10 minutes to confirm. In contrast, mainstream payment processing companies like Visa confirm transactions within a few seconds, and have high throughput of over 24,000 transactions per second~\cite{visa-performance}. Re-parametrization of \bitcoin---such as \bitcoinng~\cite{bitcoinng}---can improve this to a limited extent up to 27 transactions per second and 12 second latency, respectively~\cite{croman2016scaling}. More significant improvement requires a fundamental redesign of the blockchain paradigm. 
This has led to an array of proposals for new systems and new consensus protocols~\cite{sok-consensus}.  

Unlike application specific blockchain technologies, such as \bitcoin for a currency, or certificate transparency~\cite{laurie2014certificate} for certificate verification, smart contract platforms like \ethereum~\cite{ethereum} introduce a design that offers extensibility allowing nodes to execute user-defined programs on transactions. Traditionally, users submit transactions to the blockchain, which are sequenced through consensus, executed by the nodes, and permanently stored on the blockchain; all transactions and smart contracts are public to allow anyone to verify the correctness of the blockchain.
This paradigm makes it hard to build privacy-preserving applications on top of blockchains as neither the smart contracts logic or the transactions can contain secret values. The restriction of not being able to include secret values inside transactions and smart contracts is extremely limiting. Even a simple smart contract as `digitally sign a document' is problematic since the signing key must be kept secret, and therefore cannot be part of the transaction or smart contract. 

% =========
\section{Overview} \label{sec:introduction:overview}

\begin{figure}[t]
\centering
\includegraphics[width=.7\textwidth]{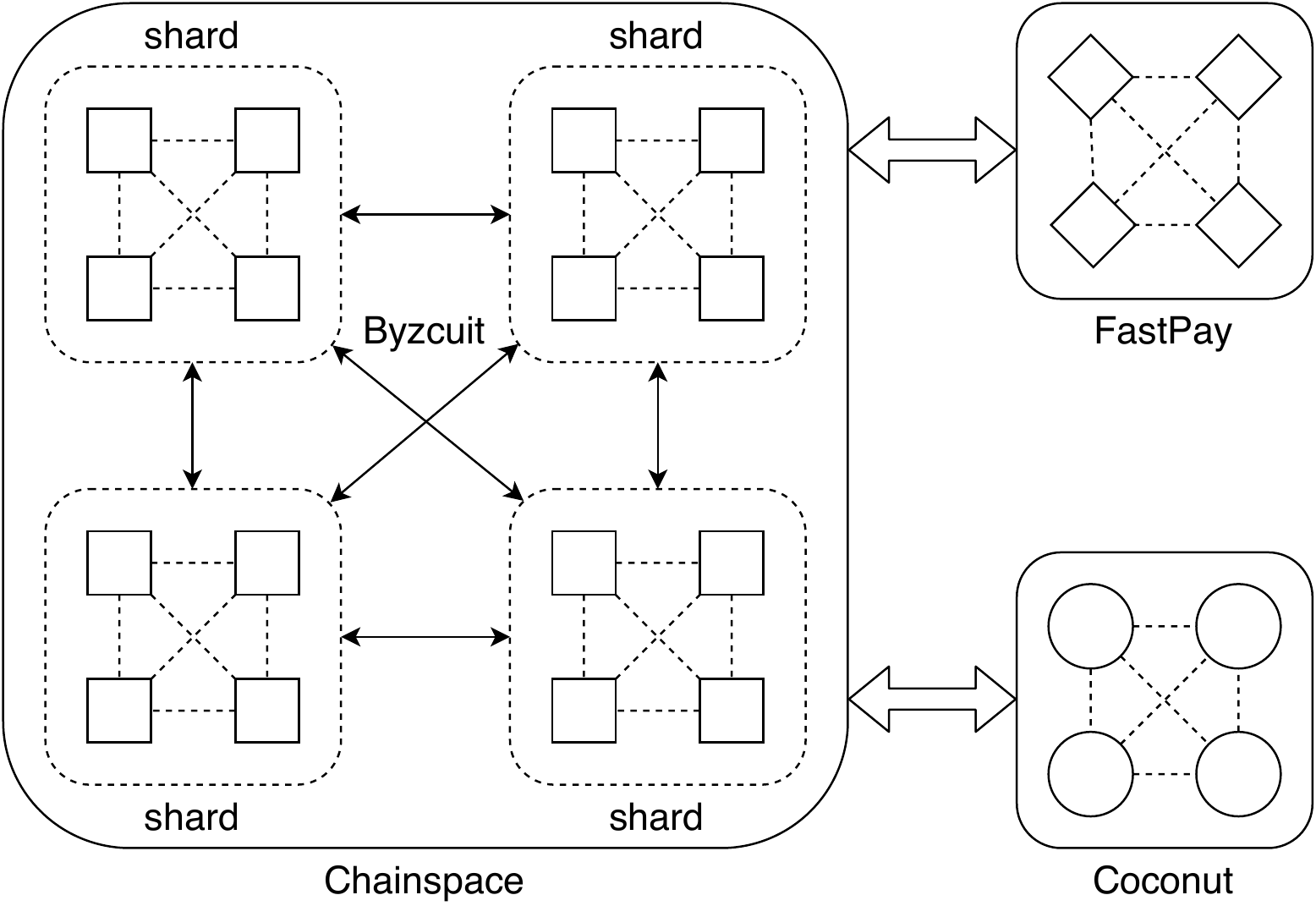}
\caption[Global overview.]{Overview of the system presented in this thesis. \chainspace is the main component and provides a scalable backbone with integrated privacy support, running the \byzcuit consensus protocol at its core; \fastpay is a distributed settlement system to support low-latency retail payments; and \coconut is a selective disclosure credential scheme that integrates with blockchains by distributing the issuance phase across a set of authorities.}
\label{fig:introduction:overview}
\end{figure}

This thesis proposes technologies to make blockchains practical. It specifically aims to overcome the following limitations: \first poor scalability, \second high latency, and \third difficulty to operate on secret values (privacy). \Cref{fig:introduction:overview} presents a high level overview of the system presented in this thesis, that is built from the following main ingredients: \chainspace, \byzcuit, \fastpay, and \coconut.

\para{\chainspace, a scalable backbone with integrated privacy support}
The main component of \Cref{fig:introduction:overview} is \chainspace that constitutes the backbone of the system and allows it to scale to accommodate high throughput. \chainspace is a sharded distributed ledger; sharding is one of the main approaches to address blockchain scalability issues. The key idea is to create groups (called \emph{shards}) of nodes that handle only a subset of all transactions and system state, relying on classical Byzantine Fault Tolerance (BFT) protocols for reaching consensus amongst each shard. Sharded systems achieve great throughput and linear scalability because: \first non-conflicting transactions can be processed in parallel by multiple shards; and \second the system can scale up by adding new shards. That is, the throughput of the system is theoretically unbounded as it can be arbitrarily increased by adding new shards.
However, this separation of transaction handling across shards is not perfectly `clean'---a transaction might rely on data managed by multiple shards, requiring an additional step of \emph{cross-shard consensus} across the concerned shards. To that purpose \chainspace runs \byzcuit at its core, a new cross-shard consensus protocol.

Orthogonally to its sharded design, \chainspace is a smart contract platform supporting privacy-preserving applications natively; it revisits the execution of smart contracts on blockchains, and proposes a system where transactions are executed by the client, and the smart contract only verifies the correctness of the execution.

\para{\fastpay, a low-latency payment system}
Sharded blockchains can be a solid backbone for financial systems, but their latency bottleneck remains their underlying BFT consensus protocol which makes them unpractical for retail payment at physical points of sale.
To overcome this issue, we present a side-infrastructure called \fastpay. \fastpay is a distributed settlement system for pre-funded payments that can be used as a financial side-infrastructure to support low-latency retail payments of a primary system such as \chainspace.
This side-infrastructure complements scalable (high-throughput) blockchains by achieving extremely low latency by foregoing the expenses of consensus, making the system applicable to point of sale payments.
   
\para{\coconut, privacy-preserving credentials for smart contract applications}
\coconut is a novel selective disclosure credential scheme that natively integrates with blockchains by distributing the issuance phase across a set of authorities. \coconut allows to issue privacy-preserving digital identities to users without relying on a trusted third party. We then leverage those credentials along with the new privacy-preserving execution model of \chainspace to build a number of decentralized and scalable privacy-preserving applications as \chainspace smart contracts.

% =========
\section{Dissertation Organization and Contributions} \label{sec:introduction:contributions}
This dissertation is organized as follows. \Cref{literature-review} lays the foundations on which the core chapters rely.
% by setting out a comprehensive analysis of hybrid consensus protocols, that are compositions or variations of classical consensus protocols. It also presents backgrounds on cross-shard atomic commit protocols, smart contract platforms, and settlement systems. 
\Cref{chainspace} presents \chainspace, a novel sharded smart contract platform that supports user defined smart contract, and that scales by sharding state and the execution of transactions through a new cross-shard atomic commit protocol. \Cref{byzcuit} presents the first replay attacks on cross-shard consensus in sharded blockchains, describes the issues that lead to these vulnerabilities, and presents \byzcuit, a novel cross-shard consensus protocol that is immune to those attacks. \Cref{coconut} presents \coconut, a credential system allowing permissioned and semi-permissioned blockchains to issue credentials through smart contracts.

This dissertation has resulted in the following publications (in chronological order, relevant chapter indicated in bold):
\begin{mdframed}
\bibentry{chainspace}
\end{mdframed}
\textbf{(\Cref{chainspace})} \chainspace is a decentralized infrastructure, known as a distributed ledger, that supports user defined smart contracts and executes user-supplied transactions on their objects. The correct execution of smart contract transactions is verifiable by all. The system is scalable, by sharding state and the execution of transactions. \chainspace is secure against subsets of nodes trying to compromise its integrity or availability properties through Byzantine Fault Tolerance (BFT), and extremely high-auditability, non-repudiation and `blockchain' techniques. Even when BFT fails, auditing mechanisms are in place to trace malicious participants. We present the design, rationale, and details of \chainspace; we argue about its scaling and other features; we illustrate a number of privacy-friendly smart contracts for smart metering, polling and banking and measure their performance. 

\vskip 1em
\begin{mdframed}
\bibentry{coconut}
\end{mdframed}
\textbf{(\Cref{coconut})} \coconut is a novel selective disclosure credential scheme supporting distributed threshold issuance, public and private attributes, re-randomization, and multiple unlinkable selective attribute revelations. \coconut integrates with \blockchains to ensure confidentiality, authenticity and availability even when a subset of credential issuing authorities are malicious or offline. We implement and evaluate a generic \coconut smart contract library for \chainspace and \ethereum; and present three applications related to anonymous payments, electronic petitions, and distribution of proxies for censorship resistance.
\coconut uses short and computationally efficient credentials, and our evaluation shows that all \coconut cryptographic primitives can be executed in less than 10 milliseconds on a commodity laptop.
%take just a few milliseconds on average, with verification taking 10 milliseconds. %the longest time (10 milliseconds). 

\vskip 1em
\begin{mdframed}
\bibentry{sok-consensus}
\end{mdframed}
\textbf{(\Cref{literature-review})} The core technical component of blockchains is \emph{consensus}: how to reach agreement among a distributed network of nodes.  
A plethora of blockchain consensus protocols have been proposed---ranging from new designs, to novel modifications and extensions of consensus protocols from the classical distributed systems literature.  The inherent complexity of consensus protocols and their rapid and dramatic evolution makes it hard to contextualize the design landscape.
We address this challenge by conducting a systematization of knowledge of blockchain consensus protocols. After first discussing key themes in classical consensus protocols, we describe: \first protocols based on proof-of-work; \second proof-of-X protocols that replace proof-of-work with more energy-efficient alternatives; and \third hybrid protocols that are compositions or variations of classical consensus protocols.  
This survey is guided by a systematization framework we develop, to highlight the various building blocks of blockchain consensus design, along with a discussion on their security and performance properties.
We identify research gaps and insights for the community to consider in future research endeavors. 

\vskip 1em
\begin{mdframed}
\bibentry{byzcuit}
\end{mdframed}
\textbf{(\Cref{byzcuit})} We present a family of replay attacks against various sharded distributed ledgers targeting cross-shard consensus protocols. They allow an attacker, with network access only, to double-spend or lock resources with minimal effort. The attacker can act independently without colluding with any nodes, and succeed even if all nodes are honest; most of the attacks can also exhibit themselves as faults under periods of asynchrony. These attacks are effective against both \shardled and \clientled cross-shard consensus approaches. We present \byzcuit---a new cross-shard consensus protocol that is immune to those attacks. We implement a prototype of \byzcuit and evaluate it on a real cloud-based testbed, showing that our defenses impact performance minimally, and overall performance surpasses previous works. 

\vskip 1em
\begin{mdframed}
\bibentry{fastpay}
\end{mdframed}
\textbf{(\Cref{fastpay})} \fastpay allows a set of distributed authorities, some of which are Byzantine, to maintain a high-integrity and availability settlement system for pre-funded payments. It can be used to settle payments in a native unit of value (crypto-currency), or as a financial side-infrastructure to support retail payments in fiat currencies. \fastpay is based on Byzantine Consistent Broadcast as its core primitive, foregoing the expenses of full atomic commit channels (consensus). The resulting system has low-latency for both confirmation and payment finality. Remarkably, each authority can be sharded across many machines to allow unbounded horizontal scalability. Our experiments demonstrate intra-continental confirmation latency of less than 100ms, making \fastpay applicable to point of sale payments. In laboratory environments, we achieve a throughput of over 80,000 transactions per second with 20 authorities---surpassing the requirements of current retail card payment networks, while significantly increasing their robustness. 

% =========
\section{Additional Work} \label{sec:introduction:additional-work}
Below are other publications outside scope did while doing a PhD, listed for inclusion in chronological order:
\begin{enumerate}[leftmargin=*]
\item \bibentry{airtnt}
\item \bibentry{proof-of-prestige}
\item \bibentry{sybilquorum}
\item \bibentry{asterisk}
\item \bibentry{fmpc}
\item \bibentry{andrikos2019location}
\item \bibentry{elpasso}
\item \bibentry{pastrami}
\item \bibentry{twins}
\item \bibentry{fraudproofs}
\end{enumerate}

% =========
\section{Work Done in Collaboration} \label{sec:introduction:collaboration}
A large part of this work has been conducted in collaboration with other researchers. All the coauthors contributed to the development of the work listed in \Cref{sec:introduction:contributions}. 

The data model (\Cref{sec:chainspace:data-model}) and application interface (\Cref{sec:chainspace:interface}) of \chainspace is a joint work with all the co-authors of the paper; and I led the design (\Cref{sec:chainspace:applications}), implementation and evaluation (\Cref{sec:chainspace:evaluation}) of the \chainspace's system and application smart contracts. 
I led the design of the replay attacks against cross-shard consensus protocols presented in Sections~\ref{sec:byzcuit:shard-led-consensus} and~\ref{sec:byzcuit:client-led-consensus}, including the techniques to elicit messages to replay (\Cref{sec:byzcuit:prerecoring}); I also led the design of \byzcuit (\Cref{sec:byzcuit:byzcuit}), and the security proofs showing that \byzcuit is immune to those attacks (\Cref{sec:byzcuit:security}). 
The design of \fastpay (\Cref{sec:fastpay:design}) is a joint work with all the co-authors of the paper, and I performed its evaluation (\Cref{sec:fastpay:evaluation}). 
I led the cryptographic construction of \coconut (\Cref{sec:coconut:construction}), as well as its implementation (\Cref{sec:coconut:implementation}) and benchmark (\Cref{sec:coconut:evaluation}); I also led the design and implementation of its smart contract library (\Cref{sec:coconut:smart_contract_library}) and privacy-preserving smart contract applications (\Cref{sec:coconut:applications}) on \chainspace.

This work would not have been possible without my co-authors.
Specifically, 
%Mustafa Al-Bassam mainly contributed to the work on \sbac, the cross-shard consensus protocol at the heart of \chainspace (\Cref{sec:chainspace:design}); Shehar Bano implemented it, and Mustafa Al-Bassam led its evaluation (\Cref{sec:chainspace:evaluation}).
Mustafa Al-Bassam led the implementation of the \chainspace smart contract framework (\Cref{sec:chainspace:evaluation}) and the \coconut smart contract library on \ethereum (\Cref{sec:coconut:tumbler}). 
Mustafa Al-Bassam and Shehar Bano led the implementation and evaluation of \byzcuit (\Cref{sec:byzcuit:implementation}).
%; and Shehar Bano lead the implementation of \chainspace (\Cref{sec:chainspace:evaluation}).
Mathieu Baudet led the security proofs (\Cref{sec:fastpay:security}) and the implementation (\Cref{sec:fastpay:implementation}) of \fastpay; and Sarah Meiklejohn wrote the security proofs of the \coconut threshold credentials scheme (\Cref{sec:coconut:security_proofs}).

%% file: chapters/literature-review/literature-review.tex
\chapter{Background and Related Work} \label{literature-review}
This chapter introduces the terminology and assumptions used throughout this thesis, and provides context to explain the reasons that led to the design choices of \chainspace (\Cref{chainspace}) and \byzcuit (\Cref{byzcuit}). It also provides background on the algorithm at the core of \fastpay (\Cref{fastpay}), and on the cryptographic building blocks used by \coconut (\Cref{coconut}). More niche related works are presented in the appropriate chapters.

\input{chapters/literature-review/sections/terminology.tex}
\input{chapters/literature-review/sections/consensus.tex}

\input{chapters/literature-review/sections/xshard.tex}

\input{chapters/literature-review/sections/sybil.tex}
\input{chapters/literature-review/sections/crypto.tex}
\input{chapters/literature-review/sections/data-structures.tex}
\input{chapters/literature-review/sections/conclusion.tex}

%% file: chapters/literature-review/sections/terminology.tex
% =========
% Terminology
% =========
\section{Terminology and Assumptions} \label{sec:literature-review:terminology}
This section presents basic concepts, terminology and assumptions related to consensus and blockchains. We refer readers interested in detailed, formal consensus definitions to the work by Garay and Kiayias~\cite{Garay:2018}.

\para{Consensus} The consensus protocol enables a distributed network of nodes to agree on the total order of some input values~\cite{cachinBook}. In the blockchain context, consensus helps reach agreement on whether \emph{transactions} should be accepted or rejected, and in which order.
A transaction specifies some transformation on the blockchain state. If a transaction passes validity and verification checks (\emph{transaction validation}), it is included in a candidate \emph{block} (a set of transactions) to be added to the blockchain.

\para{Permissioned \vs permissionless blockchains}
In \emph{permissioned blockchains}~\cite{hyperledger, libra}, identities of all the nodes that run consensus are known (trusted or semi-trusted), and their admission is controlled by a single entity or federation. In \emph{permissionless blockchains}~\cite{bitcoin}, anyone can run a node and join the network. Permissioned blockchains sometimes imply limited write access; in this work, we only refer to its meaning within the context of consensus, as defined earlier.

\para{Consistency} 
The fact that a network of $n$ nodes reaches consensus on a proposed value; it can be either strong~\cite{pbft, chainspace, omniledger, libraBFT} or weak~\cite{bitcoin, ethereum, zerocoin, zcash}. In strong consistency, the shared state across honest nodes does not diverge.
In weak consistency, the shared state across nodes might diverge temporarily leading to \emph{forks}, and additional mechanisms are needed for reconciling forks. 
This is related to \emph{eventual consistency}---\ie the blockchain becomes consistent eventually. \emph{Finality} refers to the guarantee that a block will be permanently added to the blockchain.

%\para{Validation} 
%A number of extensions to consensus protocols include a validation step, that ensures the transactions accepted are valid---however the validation rules must be deterministic and uniform across all nodes, and does not afford nodes any discretion about what constitutes a valid message.

\para{Properties} 
We consider \emph{liveness} and \emph{safety} as enumerated by Cachin~\etal~\cite{Cachin:2017}. For liveness, \emph{validity} ensures that if a node broadcasts a message, eventually this message will be ordered within the consensus, and \emph{agreement} ensures that if a message is delivered to one honest
node, it will eventually be delivered to all honest nodes. For safety, \emph{integrity} guarantees that only broadcast messages are
delivered, and they are delivered only once, and \emph{total order} ensures that all honest nodes extract the same order for all delivered messages.

\para{Synchrony assumptions} 
Networks may be \emph{synchronous} or \emph{asynchronous}, or offer \emph{eventual synchrony}~\cite{dwork1988consensus}. In a \emph{synchronous} network the delays messages may suffer can be bound by some time $\Delta$. On the other hand, in \emph{asynchronous} networks messages may be delayed arbitrarily, and there exists no reliable bound $\Delta$ for their delay. Networks with partial synchrony (or eventual synchrony, or semi-synchronous networks) assume that the network will become and remain synchronous at Global Stabilization Time (GST) despite potentially a long period of asynchrony.

\para{Network propagation}
Consensus protocols make certain assumptions about how messages propagate across nodes within the network. In \emph{point-to-point channels}~\cite{cachinBook}, there is a pairwise connection between all nodes which is both reliable and authenticated.
In the peer-to-peer (p2p) messaging model, a node `diffuses' a message into the network, which is expected to eventually reach all honest nodes with some probability~\cite{fraudproofs}. Every node knows a set of other nodes (\emph{peers})---when a message is received, nodes diffuse it by passing it on to their peers. A node may not be aware of the identities or number of other nodes in the network. \emph{Gossip-based protocols}~\cite{avalanche} rely on this assumption by considering that each node has a point-to-point connection with at least a subset of the network; the size of that subset is a security parameter.

\para{Communication complexity} 
The communication complexity of a consensus protocol refers to the maximum number of messages exchanged between the nodes in a single run of the consensus protocol. Note that a single run might involve multiple rounds of message exchanges before it completes (\ie consensus is reached).  

\para{Performance} The performance of consensus protocols is usually defined in terms of \emph{throughput} (\ie the maximum rate at which values can be agreed upon by the consensus protocol), \emph{scalability} (\ie the system's ability to achieve greater throughput when consensus involves a larger number of nodes) and \emph{latency} (\ie the time it takes from when a value is proposed, until when it is totally ordered).

\para{Adversary model for consensus} 
The adversary model is the fraction of malicious or faulty nodes that the consensus protocol can tolerate (\ie it will operate correctly despite the presence of such nodes).
This is usually referred to as the \emph{failure model} in the distributed systems literature.  In the \emph{crash failure} model, nodes may fail at any time---but they fail by stopping to process, emit or receive messages.  Usually failed nodes remain silent forever, although a number of distributed protocols consider recovery. On the other hand, in the \emph{byzantine failures} model, failed nodes may take arbitrary actions---including sending and receiving sequences of messages that are specially crafted to defeat properties of the consensus protocol. 
Another failure model in the context of consensus protocols relates to \emph{network partition}: when network devices fail (or are attacked) such that the network splits into two or more relatively independent subnets.
 
\para{Adversary model for blockchain consensus} 
Blockchain consensus has extended the adversarial model to include several new threats.  In consensus protocols with weak consistency guarantees, nodes might end up having different views of the blockchain (\emph{forks}) because of latency in propagation of transactions, and faulty or malicious nodes. 
A related concept is that of \emph{double-spending} where a transaction consumes an asset which has already been consumed by a previous transaction.  \emph{\dos resistance} defines resilience of the node(s) involved in consensus to denial-of-service (\dos) attacks.  
In the context of permissionless blockchains, \emph{Sybil attacks} refer to an attacker's ability to create fake identities or subvert existing nodes, and take over majority of the network~\cite{douceur2002sybil}. 
Sybil attacks assume two types of adversary models: static and adaptive.  A \emph{static adversary} has corrupted a fixed number of nodes in advance---it cannot corrupt new nodes or create new identities over time. An \emph{adaptive adversary} has flexibility in the nodes it can corrupt and the new identities it can create over time, to improve its probability of controlling majority of the network.

\para{Decentralization} 
This is a key property of the blockchain that enables a number of other properties such as censorship resistance, attack resistance and fault-tolerance.  
Decentralization has no formal definition, but generally~\cite{Buterin:2017} refers to a system that: \first is run by multiple machines and has no architectural choke point (\emph{architectural decentralization});
\second is run by multiple independent individuals or organizations (\emph{political decentralization}); and \third comprises multiple interfaces and data structures that can fully operate independently, instead of acting as a single whole (\emph{logical decentralization}).  
\Cref{sec:literature-review:consensus} discusses the impact of different consensus design choices on decentralization; however, a detailed discussion is beyond the scope of this work and we refer interested readers to the work by Troncoso~\etal~\cite{troncoso2017systematizing}.

%% file: chapters/literature-review/sections/consensus.tex
% =========
% Consensus
% =========
\section{Consensus in the Age of Blockchains} \label{sec:literature-review:consensus}
This section summarizes the evolution of blockchain consensus starting from Bitcoin~\cite{bitcoin} up to state-of-the-art sharded systems. This provides context and explains the reasons that led to the design choices of \chainspace (\Cref{chainspace}) and \byzcuit (\Cref{byzcuit}), and provides background on the core of \fastpay (\Cref{fastpay}).

Systems like Bitcoin~\cite{bitcoin, ethereum, zcash} probabilistically elect a single node which can extend the blockchain; they assume synchrony for safety, have probabilistic finality (\ie, forks can exist and be eventually accepted) and low performance (\ie, high latency and low throughput). 
As explained in \Cref{introduction}, \bitcoin suffers from scalability issues: with a current block size of 1MB and 10 minute inter-block interval, throughput is capped at about 7 transactions per second, and a client that creates a transaction has to wait for about 10 minutes for confirmation.
For those reasons, the community shifted to committee-based designs~\cite{sok-consensus} where a group of nodes collectively extends the blockchain typically \via classical Byzantine fault tolerance (BFT) consensus protocols such as PBFT~\cite{bft}. While these systems offer better performance, single-committee consensus is not scalable---as every node handles every transaction, adding more nodes to the committee decreases throughput due to the increased communication overhead. 
This motivated the design of \emph{sharded} systems, where multiple committees handle a subset of all the transactions---allowing parallel execution of transactions. 

% ======
\subsection{Classical Consensus} \label{sec:literature-review:classical-consensus}
Consensus protocols have been studied in the distributed systems community since the 1970s~\cite{lamport1978time}. These protocols were intended for closed, small groups of nodes.
We provide an overview of key themes in classical consensus literature, with the goal to contextualize the rest of the section. We will revisit some of these concepts when discussing committee-based consensus (Sections~\ref{sec:literature-review:hybrid-single}~and~\ref{sec:literature-review:hybrid-multi}).

\para{Two-Phase commit} 
Jim Gray, in 1978, proposed the two-phase commit protocol~\cite{gray1978notes}, allowing a \emph{transaction manager} to atomically commit a transaction, depending on different resources held by a distributed set of servers called \emph{resource managers}. Transaction commit protocols enable distributed processing, and thus scalability---but do not provide resilience against faulty resource managers, or more generally nodes. 
In fact, two-phase commit suffers a deadlock in case a resource manager fails to complete the protocol, requiring the introduction of more complex three-round protocols allowing recovery~\cite{skeen1981nonblocking}---\ie the distributed resource managers being able to release the locks held on resources.  Since potentially a crucial resource may only be available on a single resource manager, any failures inhibit progress towards accepting dependent transactions. \Cref{byzcuit} combines this primitive with Byzantine agreement in a novel way to design a scalable consensus protocol.

\para{Consensus, atomic broadcast and state machine replication} 
The need for consensus, or atomic broadcast protocols, in distributed systems originates from the need to provide resilience against failures across multiple nodes holding \emph{replicas} of a database. \emph{Atomic broadcast}~\cite{atomic} allows a set of servers to agree on a value associated with an instance of the protocol; and \emph{consensus protocols} extend this to agreeing on a sequence of values. 
This primitive is closely associated with the state machine replication paradigm~\cite{schneider1990implementing} for building reliable distributed computations: any computation is expressed as a state machine, accepting messages to mutate its state.  Given that a set of replicas start at the same initial state, and can agree on a common sequence of messages, then they may all locally evolve the state of the computation and correctly maintain consistency across the replicated databases they hold, despite failures or network variations. 
The underlying consensus protocols are characterized by the communication model, as well as the failure model, assumed (\Cref{sec:literature-review:terminology}).
Fischer~\etal~\cite{fischer1985impossibility} show that deterministic protocols for consensus are impossible in the fully asynchronous case, and have known solutions in the synchronous case (also known as the ``Byzantines General's Problem''~\cite{Lamport:1982}). 

\para{Key protocols}
In the network security literature Byzantine nodes would be considered malicious or collectively controlled by an adversary. Thus the Byzantine setting is of relevance to security-critical settings, and traditional consensus protocols tolerating only crash failures such as Paxos~\cite{lamport1998part}, viewstamped replication~\cite{oki1988viewstamped}, Raft~\cite{ongaro2014search}, or Zab~\cite{junqueira2011zab} cannot be used, unmodified, in adversarial settings. Practical Byzantine Fault Tolerance (PBFT) by Castro and Liskov~\cite{pbft} is the canonical protocol implementing consensus in the Byzantine and eventually synchronous setting. 

\para{Byzantine consistent broadcast}
Byzantine consistent broadcast is a protocol allowing a sender to broadcast a message to nodes while guaranteeing the following properties in the presence of Byzantine faults: \emph{validity}, \emph{no duplication}, \emph{integrity}, and \emph{consistency}~\cite{cachinBook}. Validity means that if the sender is honest, all honest nodes eventually deliver the message; no duplication ensures that every honest node delivers at most one message; integrity means that nodes deliver a message only if it originates from the sender (\ie the sender is not impersonated); and consistency ensures that if two honest nodes deliver a message, it is the same message (no honest nodes deliver different messages).
The properties of the protocol are guaranteed if $n = 3f+1$, where $n$ is the total number of nodes and $f$ is the number of Byzantine nodes.
Byzantine consistent broadcast can be implemented in a number of ways, the most notable is called \emph{Signed Echo Broadcast} which works in three steps~\cite{cachinBook}. In the first step, the sender disseminates a digitally signed message to all nodes (best effort broadcast), then all honest nodes witness the message by replying with a signed acknowledgement. Finally, the sender collects these signed acknowledgements and relays them in a third communication step to all nodes. Byzantine consistent broadcast is the primitive at the core of \fastpay (\Cref{fastpay}).

\para{Practical Byzantine Fault Tolerance (PBFT)} 
PBFT operates in a sequence of views, each coordinated by a leader---a pattern also used in Paxos~\cite{lamport1998part}. Within each view the leader orders messages, and propagates them through a three-step reliable broadcast to the replicas. Replicas monitor the leader for safety, as well as for liveness, and can propose a \emph{view change} in case the leader is unavailable or malicious. 
Safety is guaranteed within the asynchronous network setting; liveness on the other hand is only guaranteed within a partially synchronous setting, since replicas rely on time-outs to detect a faulty leader. The key complexity of PBFT lies in the view-change sub-protocol, that needs to ensure agreement on the new leader and view, as well as guarantee safety of messages agreed in previous views. 
The basic protocol requires $\mathcal{O}(n^2)$ messages for $n$ replicas to achieve consensus, where $n$ is the number of nodes. The properties of the protocol are guaranteed if $n = 3f+1$, where $f$ is the number of Byzantine nodes. Hotstuff~\cite{hotstuff-2019} is one of the latest successors of PBFT. It builds on a body of works~\cite{tendermint, buterin2017casper} to achieve a number of improvements over PBFT; notably, it operates with a communication complexity of $\mathcal{O}(n)$ messages and has a simpler view change protocol.

\para{Limitations of classical consensus}
PBFT and other consensus protocols employ replication to achieve resilience against failures, not scalability. In fact the traditional literature on Byzantine consensus does not discuss distribution of resources, in the context of a distributed or sharded database, with the exception of a less known joint work by Gray and Lamport on combining atomic broadcast with atomic commit~\cite{gray2006consensus}.
As a result, one expects systems employing Byzantine consensus to see this protocol become a bottleneck, since its trivial application would require all transactions to be sequenced by the quorum of $n$ nodes---using protocols that are slower than asking a single processor to sequence them. \Cref{byzcuit} discusses sharded consensus, where multiple quorums only handle a subset of the transactions.

%Some newer BFT protocols~\cite{miller2016honey,Duan:2018} even circumvent impossibility results~\cite{fischer1985impossibility}, and provide both safety and liveness in a fully asynchronous setting, through a randomized consensus algorithm. 
%This breakthrough, building upon the earlier work by Cachin~\etal~\cite{cachin2000random}, is of notable theoretical value---however, it cannot be extended to permissionless blockchains having open node participation.  Such randomized BFT protocols have traditionally been more expensive than deterministic ones, both in terms of communication and cryptographic operation costs.

% ======
\subsection{Elected Leader Consensus}
The need to achieve consensus in open, decentralized networks motivated the design of protocols based on \emph{elected leaders} that write to the blockchain.  This may
involve a combination of steps, usually applied sequentially: \first \emph{selection resource} refers to selecting a set of nodes based on some resource they own, for example \via mining power in proof-of-work (\eg \bitcoin~\cite{bitcoin}), stakes in proof-of-stake (\eg Cardano~\cite{cardano, ouroboros, ouroborospraos}), trusted hardware~\etc; and \second \emph{selection mechanism} refers to a technique that is used to non-deterministically elect the leader. 
This typically takes the form of a cryptographic \emph{lottery}---\eg a random beacon, a periodically generated pseudo-random number, which allows the nodes to determine if they have been elected as the leader. 
We briefly describe \pow and \pos consensus, and refer to Bano~\etal~\cite{sok-consensus} for further details.

\para{Proof-of-Work consensus}
Proof-of-Work consensus protocols rely on a computational puzzle to elect a leader that writes to the blockchain. As finding a solution to the puzzle requires a significant amount of computational work, so a valid solution is considered to be a proof-of-work (\pow). 
\pow was first presented by Dwork and Naor in 1993~\cite{dwork1992pricing} as a technique for combatting spam mail, by requiring the email sender to compute the solution to a mathematical puzzle to prove that some computational work was performed~\cite{dwork1993crypto}. \pow was independently proposed in 1997 for Hashcash by Back, another system for fighting spam~\cite{hashcash}. 
In 2008, \bitcoin~\cite{bitcoin} was published by a pseudonymous author Satoshi Nakamoto. Their key innovation is the use of \pow as a \sybil-resistance mechanism, combined with a rule to choose between different versions of the blockchain (fork-choice rule), to achieve consensus---originator---in an open, permissionless network. 
It was not until 2015---7 years after Bitcoin was first released---that it was formally proved that Bitcoin \pow is a consensus protocol~\cite{garay2015bitcoin}.
While the technical components of \bitcoin originate in previous literature~\cite{Narayanan:2017}, their composition in \bitcoin to achieve consensus is novel.

Nakamoto consensus probabilistically elects a single node which can extend the blockchain: it is based on a \pow puzzle derived from Hashcash~\cite{hashcash}, which requires finding a hash of a block that is less than a target integer value $t$. As the hashing algorithm is \preimage resistant, the puzzle can be solved only by including random nonces in the block until the resulting hash is valid (\ie less than $t$). The difficulty of the puzzle is therefore adjustable: decreasing $t$ increases the number of guesses (and thus work) required to generate a valid hash. 
The nodes that generate hashes are called \emph{miners} and the process is referred to as \emph{mining}. Miners calculate hashes of candidate blocks of transactions to be added to the blockchain.  
% SoK insight
Nakamoto consensus relies on the cryptographic paradigm of provers and verifiers. Miners take on the role of provers who mint blocks, and every other node is a verifier who validates (and potentially rejects) blocks according to a list of globally agreed consensus rules. This is the `trust, but verify' paradigm.

% SoK insight
Nakamoto consensus is a fork-tolerant protocol as all nodes reach eventual consistency about the blockchain's content, whereas classical consensus focuses on fork-avoidance protocols as nodes must have a consistent view after every epoch. 
A fork occurs if two miners find two different blocks that build on the same previous block.
An attacker must have sufficient computing power to be able to create a fork of the blockchain that has more accumulated work than the chain that is to be overridden. Thus the threat model assumes an adversary that has the majority of the computing power on the network (referred to as a \emph{51\% attack}).
The \emph{security threshold} of the network is the percentage of computing power required to conduct a 51\% attack.
Nakamoto consensus resolves forks by accepting the `longest chain, which has the greatest \pow effort invested in it' as the correct one. 
However, such systems assume synchrony, have probabilistic finality (\ie, forks can exist and be eventually accepted) and low performance (\ie, high latency and low throughput).

\para{Proof-of-Stake consensus}
One of the biggest criticisms of \bitcoin is that it is based on power-intensive \pow that has no external utility. 
In \poxStake, participants vote on new blocks weighted by their in-band investment such as the amount of currency held in the blockchain (\emph{stake}). A number of systems have provably secure \poxStake protocols~\cite{ouroborospraos,ouroboros,snow-white}.
A common theme in these systems is to randomly elect a leader from among the \stakeholders (participants) \via lottery, which then appends a block to the blockchain.  Leader election may be public, that is the outcome is visible to all the participants~\cite{ouroboros,snow-white}.  Alternatively, in a private election the participants use private information to check if they have been selected as the leader, which can be verified by all other participants using public information~\cite{ouroborospraos}.
% SoK insight
Leader election based on private lottery is resilient to \dos attacks because participants privately check if they are elected before revealing it publicly in their blocks, at which point it is too late to attack them.

The nature of the lottery varies across different systems, but broadly it is either collaborative (\ie requires coordination between the participants) or independent.  In \ouroboros~\cite{ouroboros}, the participants (a random subset of all \stakeholders) run a multiparty coin-tossing protocol to agree on a random seed.  
The participants then feed this seed to a pseudo-random function defined by the protocol, that elects the leader from among the participants in proportion to their stake.  The same random seed is used to elect the next set of participants for the next epoch. 
In \ouroborospraos~\cite{ouroborospraos} and \snowwhite~\cite{snow-white} participants independently determine if they have been elected. \snowwhite selects participants for each epoch based on the previous state of the blockchain, who independently check if they have been elected as the leader.  \snowwhite uses similar criteria for leader election as \bitcoin, that is finding a \preimage that produces a hash below some target.  
However, participants are limited to compute only one hash per time step (assuming access to a weakly synchronized clock) and the target takes into account each participant's amount of stake. In \ouroborospraos, participants generate a random number using a verifiable random function (\vrf).  If the random number is below a threshold, it indicates that the participant has been elected as the leader, who then broadcasts the block along with the associated proof
generated by the \vrf to the network. 
\ouroboros and \ouroborospraos distribute rewards among all the participants regardless of whether or not they win the election.
% SoK insight
\pow's leader election eligibility is out-of-band, and all nodes verify the leader election's result only so far as to find the longest and heaviest chain. Whereas in \poxStake the entire
leader-election protocol transcript is recorded in-band which increases the nodes' storage, bandwidth and validation overhead for every block.

A challenge for \poxStake systems is to keep track of the changing stakes of the \stakeholders.  \ouroboros requires that shift in stakes is bounded, meaning the statistical
distance is limited over a certain number of epochs.  Additionally, \snowwhite looks at stakes sufficiently far back in time to ensure that everyone has agreed on the stake distribution.
Outside academia, some deployed cryptocurrencies incorporate \poxStake (\eg Peercoin), but their designs have not been rigorously studied. \ethereum Foundation has been considering using \poxStake for some time~\cite{casper}, and some systems like EOS~\cite{eos} use delegated \poxStake, where participants elect delegates of their choice for mining. 
\poxStake systems are subject to a number of attacks. In long-range attacks~\cite{deirmentzoglou2019survey}, old  nodes' signature keys are compromised and used to re-write the blockchain history; in stake bleeding attacks~\cite{gavzi2018stake}, an adversary accumulates the rewards associated with creation of new blocks in a forking chain in order to inflate its stake and eventually accumulate enough to confirm an inconsistent fork; in nothing-at-stake attacks~\cite{buterin2017casper} nodes mine on every fork of the chain to benefit from the rewards of whichever fork wins.

% ======
\subsection{Hybrid Consensus: Single Committee} \label{sec:literature-review:hybrid-single}
The elected leader approach suffers from poor performance as well as weak consistency and slow finality.
This has resulted in a shift towards consensus protocols where a \emph{committee}---rather than a single node---collectively drives the consensus.

\para{Intra-Committee consensus protocol}
The intra-committee consensus protocol ensures that the committee members reach agreement on state of the blockchain. A number of committee-based blockchains~\cite{omniledger,chainspace,elastico} use \pbft. The messaging complexity of \pbft's MAC-authenticated all-to-all communication is $O(n^2)$. This is problematic for permissionless blockchains where a committee can potentially have thousands of nodes.
Another approach to reach intra-committee consensus is based on gossip protocols. These protocols are suited to permissionless blockchains as point-to-point connections between the $n$ nodes of the committee are no longer needed. Upon reception of a new transaction, nodes query a subset of $k$ randomly selected others nodes; each of those nodes replies with its view of the state of the system, and initiates a similar query.  The requesting node weights the replies and potentially updates its own view of the state; this process is repeated until global consensus is reached (with high probability). 
\avalanche~\cite{avalanche} proposes a gossip-based family of \bft protocols that have a communication complexity of $O(k \times n)$, where $k << n$ is a security parameter.  These protocols are leaderless and claim strong Denial of Service (DoS) and censorship resistance.

\para{Committee leadership}
Traditional \bft protocols proceed in rounds, where consensus in each round is led by a committee leader. The concept of a committee leader is not compatible with permissionless blockchains that aspire to achieve the design goal of decentralization.  An adversary can concentrate its DoS attack on committee leaders which are easy to discover by joining the committee.  As the leader is responsible for proposing transactions, a malicious leader can prioritize transactions from which it can benefit. While committee members can potentially detect a malicious leader and trigger leader re-election (\ie view change), this severely degrades performance~\cite{Amir:2011}.  
In \solidus, the leader is external to the committee and can propose transactions and \pow to nominate itself as a committee member only once to the committee. If the committee agrees, they approve the proposed transactions and allow the miner to join the committee in the next round. The proposal, that has now become a decision, also serves as the next puzzle and is propagated to all miners. 
\solidus highlights a safety problem in \pbft's `stable' leader which can potentially manipulate reconfiguration by waiting for a malicious miner to solve the puzzle, and later nominating it on to the committee---allowing the committee to gradually become dominated by corrupt members. 
% insight:
The concept of a leader in committee-based blockchains introduces a number of challenges with respect to scalability, and security (\dos attack, transaction censorship, and centralization). This has motivated the design of leaderless consensus protocols~\cite{hashgraph, blockmania}.

\para{Optimizations} 
A number of optimizations have been proposed to improve the performance of \bft consensus protocols.
\emph{Scheduling optimizations} involve techniques to identify and execute non-conflicting transactions in parallel (and thus achieve high throughput) by leveraging application-specific information~\cite{Kotla:2004}. \emph{Execution optimizations} reduce latency by allowing clients~\cite{WesterCNCFL09} or replicas~\cite{KotlaADCW09} to speculatively execute transactions based on predicted results---if a fault is detected (\ie the speculation turns out to be incorrect), the client/replica rolls back its state to the last checkpoint and re-executes the transactions based on the correct results. \emph{Protocol optimizations} refer to the committee's ability to switch between suitable \bft protocols according to varying network conditions and performance requirements~\cite{Guerraoui:2010}. 
\hyperledger uses \emph{pluggable and modular} consensus in which the consensus protocol can be specified by the smart contract policy. \emph{Cryptographic optimizations} leverage advances in cryptography to optimize the communication complexity of \bft. \byzcoin organizes the consensus committee into a communication tree that uses a primitive called scalable collective signing~\cite{cosi} which reduces \pbft's messaging complexity to $O(n)$.
The XFT~\cite{liu2016xft} protocol, on the other hand, improves the efficiency of consensus by relaxing the threat model. It considers that Byzantine nodes may act arbitrarily, however links between honest nodes are reliable and eventually synchronous. This leads to a simplification of the view change and steady state \bft protocol.
\emph{Hardware optimizations} enable consensus protocols to achieve high performance by exploiting advances in hardware. The Intel Sawtooth lake system uses the Intel SGX and related trusted execution environments to perform the duties related to ordering transactions, while ensuring safety and liveness~\cite{prisco2016intel}. 
Finally, \emph{architectural optimizations} improve performance by distributing different consensus duties across independent subsets of replicas.  A useful paradigm (employed by \hyperledger) is to separate ordering from execution~\cite{Yin:2003}, which allows for a modular design where transaction validation is performed by the fully trusted nodes (or endorsers) while the semi-trusted nodes (ordering nodes) order the transactions and add these to the blockchain.  
In \hyperledger, clients first submit their transactions to the endorsers who execute the smart contract. A transaction is only submitted to a subset of endorsers according to the policy of the respective smart contract. As different smart contracts can designate different endorsers, execution can take place in parallel.  Clients collect matching signed results and smart contract state updates from sufficient number of endorsers, and submit these to the ordering nodes which append it to the blockchain using a consensus protocol.
Others~\cite{Vukolic:2017:RPB} argue that distributed ledgers can
decouple the ordering---performed in public on cryptographic
commitments of transactions---from the validation containing private
information, that is only checked by a trusted cabal. 
% SoK insight:
Separating ordering from execution~\cite{Yin:2003} allows
committee-based blockchains to scale at the same rate as the core
ordering protocol, but providing universal end-to-end
verifiability and decentralization in this setting remains an open challenge.

% ======
\subsection{Hybrid Consensus: Multiple Committees} \label{sec:literature-review:hybrid-multi}
Single-committee consensus is not scalable and adding more nodes to the committee decreases throughput---leading to the design of consensus based on multiple committees. Transactions are split among multiple committees (called \emph{shards}) which then process these transactions in parallel. Every committee has its own blockchain and set of objects (or unspent transaction outputs, UTXO) that they manage. Committees run an \emph{intra-shard consensus protocol} (\eg, PBFT) within themselves, and extend their blockchain in parallel. 
This incurs additional coordination between the committees \via \emph{inter-committee consensus} to reach agreement on a value among nodes across multiple committees. 
The inter-committee consensus protocol may be run entirely by the committees (\emph{non-mediated}), or may be mediated by an external party (\emph{mediated}). \emph{Inter-committee configuration} defines how nodes are assigned to the committees in a multiple committees setting; it can be static or dynamic.
When multiple committees are involved in consensus, an important consideration is how they will be organized in terms of \emph{topology}.

\para{Committee topology}
\chainspace (presented in \Cref{chainspace}) and \omniledger~\cite{omniledger} have flat topologies, that is all committees are at the same level.
\elastico~\cite{elastico} has a hierarchical topology in which a number of `normal' committees validate transactions, and a leader committee orders these transactions and extends the blockchain. 
In \rscoin~\cite{rscoin}, a permissioned blockchain, the central bank controls all monetary supply while committees (called mintettes) authorized by the bank validate a subset (shard) of transactions.  The transactions that pass validation are submitted to the central bank which adds them to the blockchain.
% Sok insight
Hierarchical topology facilitates configuration and
management of committees in multi-committee blockchains, but
undermines decentralization.

\para{Inter-committee consensus}
In a multi-committee system, some transactions might manipulate state that is handled by different committees. The inter-committee consensus ensures that such transactions are processed consistently and atomically across all the concerned committees.  One approach is to mediate the inter-committee consensus protocol \via the client. \omniledger uses an atomic commit protocol to process transactions across committees (see \Cref{sec:byzcuit:atomix-summary}). 
Client-driven inter-committee consensus protocols make the assumption that clients are incentivized to proceed to the unlock phase. Such incentives may exist in a cryptocurrency application where an unresponsive client will lose its own coins if the inputs are permanently locked, but do not hold for a general-purpose platform where transaction inputs may have shared ownership.
% SoK insight
Client-driven inter-committee consensus protocols are vulnerable to \dos attack if the client stops participating midway, resulting in the transaction inputs being locked forever.
Another approach is to run an atomic commit protocol collaboratively between all the concerned committees.  This is achieved by making the entire committees act as resource managers for the transactions they manage (see \Cref{sec:byzcuit:sbac-summary}).
% SoK gap
Inter-committee consensus protocols are relatively immature, and their security has not been rigorously evaluated.
For example, \Cref{byzcuit} shows the susceptibility of these protocols to replay attacks that allow an attacker to double-spend resources with minimal effort, and without colluding with any nodes.
The attacker records a target committee's responses to the consensus protocol, and replays them during another instance of the protocol.

%% file: chapters/literature-review/sections/xshard.tex
%Replay attacks in general have seen extensive study in the security literature. This work is the first that presents replay attacks on cross-shard consensus protocols. Traditionally, the most stringent threat models assumed by consensus protocols involve Byzantine adversaries who are able to control or subvert consensus nodes and cause them to behave arbitrarily.  Repurposing those protocols to cross-shard distributed ledgers (\eg, blockchains) opens up new attack avenues such as replay attacks as shown in Sections~\ref{sec:byzcuit:shard-led-consensus} and~\ref{sec:byzcuit:client-led-consensus}.

\section{Cross-Shard Consensus Protocols} \label{sec:literature-review:xshard}
\Cref{sec:literature-review:hybrid-multi} presented a brief overview of inter-committee consensus; this section provides a deeper background on cross-shard consensus protocols as it is a core component of \chainspace (presented in \Cref{chainspace}) and the subject of \Cref{byzcuit}.

Multi-committee systems are often called \emph{sharded systems}, and each committee is called \emph{shard}.
In sharded systems, some transactions may operate on objects handled by different shards, effectively requiring the relevant shards to additionally run a \emph{cross-shard consensus protocol} to enable agreement across the shards. Specifically, if any of the shards relevant to the transaction rejects it, all the other shards should likewise reject the transaction to ensure atomicity.

\subsection{Two-Phases Atomic Commit Protocols}
A typical choice for implementing cross-shard consensus is the two-phase atomic commit protocol~\cite{gray1978notes} (see \Cref{sec:literature-review:classical-consensus}). This protocol has two phases which are run by a \emph{coordinator}. In the first \emph{voting} phase, the nodes tentatively write changes locally, lock resources, and report their status to the coordinator. If the coordinator does not receive status message from a node (\eg, because the node crashed or the status message was lost), it assumes that the node's local write failed and sends a rollback message to all the nodes to ensure any local changes are reversed, and locks released. If the coordinator receives status messages from all the nodes, it initiates the second \emph{commit} phase and sends a commit message to all the nodes so they can permanently write the changes and unlock resources. In the context of sharded blockchains, the atomic commit protocol operates on shards (which make the local changes associated with the voting phase \via an intra-shard consensus protocol like PBFT), rather than individual nodes. A further consideration is who assumes the role of the coordinator; \emph{\clientled protocols} rely on the client to take the role of coordinator, and \emph{\shardled protocols} rely on entire shard.

There are currently two key approaches to cross-shard atomic commit protocols. The first approach involves \clientled protocols (such as \atomix~\cite{omniledger} and \rscoin~\cite{rscoin}), where the client acts as a coordinator. These protocols assume that clients are incentivized to proceed to the unlock phase. 
The second approach involves \shardled protocols (such as \sbac~\cite{chainspace} and Elastico~\cite{elastico}), where shards collectively assume the role of a coordinator.
All the concerned shards collaboratively run the protocol between them. This is achieved by making the entire shard act as a `resource manager' for the transactions it handles.
\Cref{byzcuit} describes a family of replay attacks in the context of two representative systems: \sbac~\cite{chainspace} as an example of \shardled protocols (\Cref{sec:byzcuit:shard-led-consensus}); and \atomix as an example of \clientled protocols (\Cref{sec:byzcuit:client-led-consensus}); these attacks can compromise both system liveness and safety. 

\subsection{Mutex-Based Consensus Protocols}
Contrarily to \sbac and \atomix that achieve cross-shard consensus through a two-phase atomic commit protocol,
mutex-based schemes for cross-shard transactions, such as \rapidchain~\cite{rapidchain} and \ethereum's cross-shard `yanking' proposal~\cite{yanking}, adopt mutex-based schemes for transactions that involve objects managed by different shards.
The key idea is to require all objects that a transaction reads or writes  to be in the same shard (\ie, all locks for a transaction are local to the shard). Cross-shard transactions are enabled by transferring the concerned objects between shards, effectively giving shards a lock on those objects.  When shard 1 transfers an object to shard 2, shard 1 includes a transfer `receipt' in its blockchain. A client can then send to shard 2 a Merkle proof of this receipt being included in shard 1's blockchain, which makes the object active in shard 2.

Mutex-based schemes also need to consider replay attacks. Clients can claim the same receipt multiple times, unless shards store information about previously claimed receipts.  Na\"{\i}vely, shards have to store information about all previously claimed receipts permanently.
However, two intermediate options with trade-offs have been proposed~\cite{yanking}:
\begin{itemize}
    \item Shards only store information about receipts for $l$ blocks, and clients can only claim receipts within $l$ blocks; objects are permanently lost if not claimed in time (this introduces a synchrony assumption).
    \item Shards only store information about receipts for $l$ blocks, and include the root of a Merkle tree of claimed receipts in their blockchain every $l$ blocks. If a receipt is not claimed within $l$ blocks, the client must provide one Merkle proof every $l$ blocks that have passed to show that the receipt has not been previously claimed, in order to claim it. The longer the receipt was not claimed, the greater the number of proofs that are needed to claim a receipt. These proofs need to be also stored on-chain to allow other nodes to validate them.
\end{itemize}
\Cref{byzcuit} presents a system, called \byzcuit, that forgoes the need for shards to store information about old state (such as inactive objects or old receipts) as shards only need to know the set of active objects they manage, and does not impose a trade-off between the amount of information about old state that needs to be stored and the cost of recovering old state that was held up in an incomplete cross-shard transaction (\ie, an unclaimed receipt).

%% file: chapters/literature-review/sections/sybil.tex
% =========
% Sybil
% =========
\section{Sybil Resistance and Committee Management} \label{sec:literature-review:sybil}
A limitation of using \pow or \pox for \sybil resistance in permissionless committees is that the biggest miners will have a greater likelihood of dominating the committee, though at the cost of significantly more hashing power than required for single-leader \pow systems. Other \pox alternatives, relying on space, memory, or
space-time, have been proposed but these suffer from similar issues.
% SoK gap
Protocols have also been proposed for \sybil detection based on the analysis of social networks and trust graphs~\cite{Alvisi:2013}, but those have not been adapted to blockchains, besides the definitional framework for Federated Byzantine Agreement Systems proposed by \stellar~\cite{mazieres2015stellar}. Sybilquorum~\cite{sybilquorum} makes a first attempt to bridge this gap by proposing mechanism to detect sybils through a stake-weighted social networks analysis.

% =========
\subsection{Sybil Resistance}
Committee formation refers to the criteria used to allow nodes to join a committee. \emph{Permissioned} blockchains like~\hyperledger~\cite{hyperledger} operate in a trusted environment where nodes are granted committee membership based on the organizational policy. 
In permissionless blockchains, the committee is formed so as to thwart \sybil attacks. Nodes are usually allowed to join the committee based on a \emph{selection resource} such as \emph{\pow}. In \byzcoin~\cite{byzcoin}, the consensus committee is dynamically formed by a window of recent miners. Each miner has voting power proportional to its number of mining blocks in the current window, which is proportional to its hash power. 
When a miner finds a solution to the puzzle, it becomes a member of the committee and receives a share in the consensus. In addition to \pow, \omniledger~\cite{omniledger} also supports \emph{\poxStake} to allocate committee membership based on directly invested stake.
Some permissionless blockchains employ a further \emph{selection mechanism} such as a \emph{lottery} to form the committee. In \algorand~\cite{algorand}, all the nodes that have \pox run a verifiable random function---they are promoted to the committee if the output is below a certain value.  

\para{Coercion resistance}
Another consideration for bootstrapping committees is to achieve coercion resistance, in the form of requiring enormous effort for an adversary to suppress the overall operation of the system. Coercion resistance properties are also key to the success of other decentralized systems, such as BitTorrent~\cite{cohen2003incentives},
that are subject to take-down pressures by publishers. Bitcoin itself was coincidentally proposed in 2008~\cite{nakamoto-post}, the same year when E-gold~\cite{egold} was declared illegal by the US Department of Justice and taken offline---illustrating that value exchange systems, and monetary systems that are transnational and unregulated, will come under fire by national monetary and law enforcement authorities.  
Systems such as Tor~\cite{tor} have survived in a highly adversarial environment despite parts of its infrastructure, namely directory authorities, being a closed consensus group. 
These authorities are distributed geographically, and are under different jurisdictions and managed by different organizations to preclude both collusion and single jurisdiction attacks.
% SoK insight
Single-committee blockchains may, through careful selection of nodes, achieve coercion resistance~\cite{sok-consensus}.

% =========
%\subsection{Committee Formation: Multiple Committees}
\para{Multiple committees}
Multi-committee blockchains raise the additional issue of how to map nodes to committees. In permissioned systems, the process of assigning nodes to committees is usually done \emph{statically} according to the policy of the federation.  Another approach is to \emph{dynamically} allocate nodes to committees. 
Permissioned systems like \rscoin can use a trusted source of randomness for committee reconfiguration, but this can be problematic in a permissionless setting which would require a shared random coin~\cite{Corbett:2013,Glendenning:2011}.
Generating good randomness in a distributed way is a known hard problem: current best solutions tolerate up to $1/6$ fraction of
Byzantine nodes, while incurring a high message complexity~\cite{Awerbuch:2006}. Among the more recent solutions, \randhound~\cite{randherd} provides a more scalable, secure multi-party computation protocol that offers unbiasable, decentralized randomness while tolerating a third of Byzantine faults.  It brings down the communication complexity to $O(c^2 \textrm{log}(n))$, where $c$ is the size the subgroups it uses.  
% Sok insight
In multi-committee blockchains, nodes should be assigned to committees in a non-deterministic way to stop an adversary from concentrating its presence in one committee and exceeding the Byzantine-tolerance threshold.

\para{Securing committees}
The idea of scaling services built on state machine replication (\smr) by splitting state (or sharding) among multiple committees (also called partitions or shards) has been well-studied in the context of traditional distributed systems~\cite{Corbett:2013,Glendenning:2011,Le:2016:SSMR}.
The key challenge in these systems is to ensure linearizability by atomically executing operations that span multiple committees. 
These systems employ fault-tolerant \bft protocols at their core as the nodes are controlled by a single entity or a group of entities that collectively govern the system. Due to similar governance assumptions, these techniques can be extended to permissioned blockchains.
Sharding permissionless blockchains with Byzantine adversaries is challenging and tackled by only a few recent systems~\cite{chainspace,omniledger,elastico,rapidchain}. Individual committees can tolerate up to 33\% of malicious members, otherwise the malicious committee can compromise all the transactions that touch the bad committee. 
%\chainspace starts mitigating this issue by making the author of the smart contract responsible for designating the parts of the infrastructure that are trusted to maintain the integrity of its contract. 
%Moreover, \chainspace provides an auditing mechanism allowing honest node in honest committees to detect inconsistencies and discover the malicious committee. 
% Sok gap
%In multi-committee blockchains, a single malicious committee can compromise security of the entire system.  There is some preliminary work on detecting a malicious shard, however there are no systems today providing a recovery mechanism.

\para{Committee governance}
Randomly mapping nodes to committees improves security, but prohibits finer governance.  General-purpose platforms like \chainspace might have different policies within committees; for example some committees can be permissioned while others can be permissionless. In this case it might be useful to enforce node-to-shard mapping \via smart contracts that allow a node to join a committee trusted by the smart contract provider.  

%\para{Coercion resistance}
%It is crucial for a value exchange system based on blockchains that its clients believe it will exist as a medium of value in the future---and thus the potential for future disruption of the network, would reduce its value significantly. Thus, if the key feature of \pow schemes is the robustness and coercion resistance resulting from their openness, multi-committee blockchains that sacrifice this property may fail regardless of their superior performance. 
%Similar to single-committee systems, multi-committee systems have to consider forming committees in such a way that the system is resilient against coercion. 
% Sok insight
%Multi-committee blockchains can achieve coercion resistance within each committee \via careful selection of nodes; and across committees by creating `backup' committees that mirror the `primary' committees, and can replace primaries that are taken down.

% =========
\subsection{Committee Reconfiguration}
Intra-committee configuration means how nodes are assigned to the committee. 
In \emph{static} configuration, nodes are statically assigned to the committee, and are allowed to stay on indefinitely.
Static configuration is typically employed in permissioned blockchains like \hyperledger and \rscoin. In \emph{dynamic} configuration, committee members are changed periodically. This model is typically used in permissionless blockchains as this helps thwart \sybil attacks. Dynamic committee membership can take three forms. \first In \emph{rolling (single)} membership, the committee is updated in a sliding window fashion, \ie a new node replaces the oldest committee member periodically. 
In \byzcoin, when a miner finds a solution to the puzzle, it becomes a member of the committee and receives a share in the current consensus window which moves one step forward (ejecting the oldest miner). \second \emph{Rolling (multiple)} committee membership is a similar concept, where multiple committee members are replaced periodically.  \omniledger uses cryptographic sortition to select a subset of the committee to be swapped out and replaced with new members. \third Some systems replace the \emph{full} committee, \eg \algorand and \snowwhite select the committee members for each epoch using randomness generated based on previous blocks.

\para{Liveness in dynamic committee configuration}
Dynamic committee configuration improves security by raising the bar for sybil attacks, but introduces a new challenge: how is liveness maintained during reconfiguration? One approach is to only do rolling configuration, which has the benefit that the committee is operational during reconfiguration as the operational members can continue to process transactions while a fraction of the committee is being reconfigured and bootstrapped.  
\omniledger uses cryptographic sortition to select a subset of the committee to be swapped out and replaced with new members. This is done in such a way that the ratio between honest and Byzantine members in a committee is maintained. In \solidus~\cite{solidus}, a new miner joining the committee can propose transactions only once.  This binds transaction proposals to reconfiguration, so it is no longer possible for an old committee to approve transactions concurrent to a reconfiguration event.  
% SoK gap
%Maintaining liveness and security in dynamic intra-committee configuration setting is an open and neglected research area in the design of single-committee blockchains.

%\para{Denial-of-Service (DoS) attack}   
%An adversary can launch \dos attack on the blockchain committee. Small statically configured committees are particularly vulnerable.
%Systems that do a full swap~\cite{elastico,ouroborospraos,algorand} of the committee, and have small epochs are highly resilient against \dos attacks---but this may introduce liveness challenges. Rolling configuration~\cite{byzcoin,omniledger} can be tuned to provide optimal tradeoffs between \dos resilience and liveness.  
%A blockchain committee's vulnerability to \dos attacks is directly related to the intra-committee configuration, \ie how frequently and what fraction of committee membership changes (epoch, dynamism).

% =========
\para{Reconfiguration of multiple committees systems}
Inter-committee configuration means whether node assignment to committees in a multi-committee blockchain remains \emph{static} or is periodically changed (\emph{dynamic}).
\omniledger periodically reconfigures committees to ensure that a committee is never compromised. This is achieved by a secure shard reconfiguration protocol, based on \randhound, that committee members run periodically and  autonomously.  In every epoch, a random subset of members is replaced with new set of members that registered their interest in the previous epoch. 
The swap operation is done such that liveness is maintained during reconfiguration events because a subset of committee members continues to be operational.
% Sok insight 
Dynamic inter-committee configuration prevents an adversary from subverting existing nodes in a committee and exceeding the Byzantine-tolerance threshold.
%
%\elastico operates in epochs: assignment of nodes to committees is valid only for duration of the epoch. At the end of the epoch, nodes compute solution to a puzzle seeded by a random string generated by the final committee and send the solution to the final committee to be assigned to a committee. As a result, in each epoch a node is paired with different nodes in a committee managing a different set of transactions. The number of committees scales linearly in the amount of computational power available in the system, but the number of nodes within a committee is fixed. 
\chainspace has abstracted details of committee reconfiguration and it is up to policy enforced \via a smart contract to decide how nodes are allocated to committees. Nodes can be added (and removed) to committees by their members through majority voting.

%% file: chapters/literature-review/sections/crypto.tex
\section{Selective Disclosure Credentials} \label{sec:literature-review:crypto}
This section provides some background on selective disclosure credentials and on some useful cryptographic building blocks. \Cref{coconut} presents \coconut, a novel selective disclosure credential scheme that is used to build a number of decentralized and scalable privacy-preserving applications on top of \chainspace. 
Selective disclosure credentials (or anonymous credentials)~\cite{cl, amac, pointcheval} allow the issuance of credentials to users, and the subsequent unlinkable revelation to a verifier. They allow users to be known in different contexts by different pseudonyms; a user might be known to one service under one pseudonym and to another service under a different pseudonym, and yet be unlinkable accross services. Users can selectively disclose some of the attributes embedded in the credential or specific functions of these attributes. At the high level, anonymous credentials scheme have two main protocols: \first the \emph{issuing} protocol where the issuing authorities provide the user with a credential embedding a number of attributes, and \second the \emph{showing} protocol where the user shows its credential (and potentially some of its attributes) to a verifier.

% ============
\subsection{Cryptographic Building Blocks} \label{sec:literature-review:crypto-blocks}
We describe the security properties of anonymous credentials and provide background on zero-knowledge proofs that are extensively used at the core of any anonymous credential scheme. We then present the cryptographic assumptions on which \coconut as well as many of its predecessors (see \Cref{sec:literature-review:coconut-related}) rely.

\para{Security properties}
Anonymous credentials scheme satisfy \emph{unforgeability}, \emph{blindness}, and \emph{unlinkability} (or \emph{zero-knowledge}). Unforgeability means it is unfeasible for an adversarial user to convince an honest verifier that they are in possession of a credential if they are in fact not (\ie, if they have not received a valid credential from the issuing authority). Blindness ensures it is unfeasible for an adversarial authority to learn any information about the user attributes during the credentials issuing protocol, except for what is explicitly revealed by the user. Finally, unlinkability (or Zero-knowledge) ensures it is unfeasible for an adversarial verifier (potentially working with an adversarial authority) to learn anything about the user attributes (except for what is explicitly revealed by the user) during the execution of the showing protocol, or to link the execution of the showing protocol with either another execution of that protocol or with the credentials issuing protocol.

\para{Zero-knowledge proofs}
Zero-knowledge proofs are protocols allowing a \emph{prover} to convince a \emph{verifier} that it knows a secret value $x$, without revealing any information about that value. The prover can also convince the verifier that they know a secret value $x$ satisfying some statements $\phi$.
Anonymous credentials extensively employ zero-knowledge proofs to provide users with certified secret values; users are successively able to prove to third party verifiers that they hold secret values certified by specific credentials issuers, and prove statements about those values without disclosing them.
This enables, for instance, the property of \emph{provable personal properties}.
A credential issuer may provide a user with a secret value $x=20$ representing their age; the user can then prove in zero-knowledge to a verifier that a specific credential issuer certified that their age is larger than 18, without revealing their real age $x$.

\coconut uses non-interactive zero-knowledge proofs (NIZK) to assert knowledge and relations over discrete logarithm values.
These proofs can be efficiently implemented without trusted setups using sigma protocols~\cite{schnorr}, which can be made non-interactive using the Fiat-Shamir heuristic~\cite{fiat1986prove} in the random oracle model.

\para{Cryptographic assumptions}
\coconut requires groups $(\mathbb{G}_1,\mathbb{G}_2,\mathbb{G}_T)$ of prime order $p$ with a bilinear map $e:\mathbb{G}_1 \times \mathbb{G}_2 \rightarrow \mathbb{G}_T$ and satisfying the following properties: \first\emph{Bilinearity} means that for all $g_1 \in \mathbb{G}_1$, $g_2 \in \mathbb{G}_2$ and $(a,b) \in \mathbb{F}_p^2$, \; $e(g_1^a,g_2^b) = e(g_1,g_2)^{ab}$; \second\emph{Non-degeneracy} means that for all $g_1 \in \mathbb{G}_1$, $g_2 \in \mathbb{G}_2$, $e(g_1,g_2) \neq 1$; \third\emph{Efficiency} implies the map $e$ is efficiently computable; \fourth furthermore, $\mathbb{G}_1 \neq \mathbb{G}_2$, and there is no efficient homomorphism between $\mathbb{G}_1$ and $\mathbb{G}_2$.
The type-3 pairings are efficient~\cite{galbraith2008pairings}. They support the XDH assumption which implies the difficulty of the Computational co-Diffie-Hellman (co-CDH) problem in $\mathbb{G}_1$ and $\mathbb{G}_2$, and the difficulty of the Decisional Diffie-Hellman (DDH) problem in $\mathbb{G}_1$~\cite{bls}.
\coconut also relies on a cryptographically secure hash function $\hashtopoint$, hashing an element $\mathbb{G}_1$ into an other element of $\mathbb{G}_1$, namely $\hashtopoint: \mathbb{G}_1\rightarrow\mathbb{G}_1 $. We implement this function by serializing the $(x,y)$ coordinates of the input point and applying a full-domain hash function to hash this string into an element of $\mathbb{G}_1$ (as Boneh~\etal~\cite{bls}).

% ============
\subsection{The Predecessors of \coconut} \label{sec:literature-review:coconut-related}
This section describes the body of works on top of which \coconut is built, starting with some relevant signature schemes (which provide unforgeability but neither blindness nor unlinkability, see \Cref{sec:literature-review:crypto-blocks}).

\para{Short and aggregable signatures} 
The Waters signature scheme~\cite{waters} provides the bone structure of our primitive, and introduces a clever solution to aggregate multiple attributes into short signatures. 
However, the original Waters signatures do not allow blind issuance or unlinkability, and are not aggregable as they have not been built for use in a multi-authority setting. 
Lu~\etal scheme, commonly known as LOSSW signature scheme~\cite{lossw}, is also based on Waters scheme and comes with the improvement of being sequentially aggregable. 
In a sequential aggregate signature scheme, the aggregate signature is built in turns by each signing authority; this requires the authorities to communicate with each other resulting in increased latency and cost. The BGLS signature~\cite{bgls} scheme is built upon BLS signatures and is remarkable because of its short signature size---signatures are composed of only one group element. The BGLS scheme has a number of desirable properties as it is aggregable without needing coordination between the signing authorities, and can be extended to work in a threshold setting~\cite{boldyreva2002efficient}. Moreover, Boneh~\etal show how to build verifiably encrypted signatures~\cite{bgls} which is close to our requirements, but not suitable for anonymous credentials. 

\para{Anonymous credentials} 
CL Signatures~\cite{cl,lee2013aggregating} and Idemix~\cite{idemix} are amongst the most well-known building blocks that inspired applications going from direct anonymous attestations~\cite{chen2010design, bernhard2013anonymous} to electronic cash~\cite{canard2015divisible}. They provide blind issuance and unlinkability through randomization; but come with significant computational overhead and credentials are not short as their size grows linearly with the number of signed attributes, and are not aggregable. U-Prove~\cite{uprove} and Anonymous Credentials Light (ACL)~\cite{acl} are computationally efficient credentials that can be used once unlinkably; therefore the size of the credentials is linear in the number of unlinkable uses. Pointcheval and Sanders~\cite{pointcheval} present a construction which is the missing piece of the BGLS signature scheme; it achieves blindness by allowing signatures on committed values and unlinkability through signature randomization. However, it only supports sequential aggregation and does not provide threshold aggregation. For anonymous credentials in a setting where the signing authorities are also verifiers (i.e., without public verifiability), Chase~\etal~\cite{amac} develop an efficient protocol. Its `GGM' variant has a similar structure to \coconut, but forgoes the pairing operation by using message authentication codes (MACs). None of the above schemes support threshold issuance.

While the scheme of Garman~\etal~\cite{dac} does not specifically focus on threshold issuance of credentials or on general purpose credentials, it provides the ability to issue credentials without central issuers supporting private attributes, blind issuance, and unlinkable multi-show selective disclosure. To obtain a credential, users build a vector commitment to their secret key and a set of attributes; and append it to a ledger along with a pseudonym built from the same secret key, and a zk-proof asserting the correctness of the vector commitment and of the pseudonym. To show a credential under a different pseudonym, users scan the ledger for all credentials and build an RSA accumulator; they provide a zk-proof that they know a credential embedded in the accumulator. Similarly to Zerocoin~\cite{zerocoin}, showing credentials requires an expensive double  discrete-logarithm proof (about 50KB~\cite{dac}); and the security of the credentials scheme relies on the security of the ledger. \coconut addresses the two open questions left as future work by Garman~\etal~\cite{dac}; \first the security of \coconut credentials do not depend on the security of a transaction ledger as they are general purpose credentials, and \second \coconut enjoys short and efficient proofs as it builds from blind signatures and does not require cryptographic accumulators.

\para{Short and threshold issuance anonymous credentials} 
\coconut (\Cref{coconut}) extends these previous works by presenting a short, aggregable, and randomizable credential scheme; allowing threshold and blind issuance, and a multi-authority anonymous credentials scheme. \coconut primitives do not require  sequential aggregation, meaning the aggregate operation does not have to be performed by each signer in turn. Any independent party can aggregate any threshold number of partial signatures into a single aggregate credential, and verify its validity. 

%% file: chapters/literature-review/sections/data-structures.tex
\section{Blockchains Data Structures} \label{sec:literature-review:data-structures}
There are a number of ways in which blockchain systems organize their state and the records of processed transactions; the most common strategies are linear chains, block DAGs, and Merkle trees.

\para{Linear chain}
The most common strategy to store blockchain records is by grouping several transactions into sets called \emph{blocks}~\cite{bitcoin, ethereum}. Each block contains a cryptographic hash to the previous block, thus forming a chain of blocks. This chain provides high-integrity as it is unfeasible to alter a block without invalidating the rest of the chain.
Adding records to a linear chain is efficient but searching the chain for specific transactions can be expensive.

\para{Block DAG}
\chainspace (\Cref{chainspace}) organizes records as a Directed Acyclic Graph (DAG). A DAG  is a directed graph with no directed cycles. It consists of vertices and edges with each edge directed from one vertex to another, such that following those directions never forms a closed loop. \chainspace proposes a new data model where objects naturally forms a DAG (see \Cref{sec:chainspace:design}). This model has no disadvantages with respect to traditional linear chains, and allows to easily verify the correct creation of each object.

\para{Merkle tree}
A Merkle tree is a binary tree where every non-leaf node is  the cryptographic hash of the concatenation of its children nodes~\cite{szydlo2004merkle}. As a result, the root of the tree is a commitment to all of its leaf nodes. A prover can convince a verifier knowing only the root of tree that a specific leaf is included in the Merkle tree. The size and verification time of this proof is logarithmic in the number of leaves.
A number of blockchains, such as Libra~\cite{libra}, store their records using Merkle trees to allow nodes to efficiently convince clients that a specific transaction has been committed. On the downside, expanding a blockchain stored as a Merkle tree is more expensive than growing a linear chain.

%% file: chapters/literature-review/sections/conclusion.tex
\section{Chapter Summary}
This chapter laid the foundations of the next chapters, and inserted the contributions of the core chapters into their respective line of work. 
It first introduced the main terminology and assumptions used throughout this thesis. 
It then motivated the design of multi-committee systems such as \chainspace (\Cref{chainspace}) by summarizing the evolution of blockchain consensus through history, and explained why they are considered today's state-of-the-art.
This chapter provided background on 2-phase atomic commit protocols and Byzantine consensus, the primitives at the core of \byzcuit (\Cref{byzcuit}); Byzantine consistent broadcast, the primitive at the core of \fastpay (\Cref{fastpay}); and described multiple cryptographic building blocks used by \coconut (\Cref{coconut}) to build the first threshold-issuance selective disclosure credentials scheme.

%% file: chapters/chainspace/chainspace.tex
\chapter{\chainspace: A Sharded Smart Contracts Platform} \label{chainspace}

%\input{chapters/chainspace/sections/abstract.tex}
\input{chapters/chainspace/sections/01Introduction.tex}

\input{chapters/chainspace/sections/02Overview.tex}
\input{chapters/chainspace/sections/03Interface.tex}
\input{chapters/chainspace/sections/04Design.tex}

\input{chapters/chainspace/sections/05Security.tex}
\input{chapters/chainspace/sections/06Contracts.tex}
\input{chapters/chainspace/sections/07Evaluation.tex}

\input{chapters/chainspace/sections/08Limitations.tex}
\input{chapters/chainspace/sections/09Related.tex}
\input{chapters/chainspace/sections/10Conclusion.tex}

%% file: chapters/chainspace/sections/01Introduction.tex
% =========
% Introduction
% =========

%
\begin{figure}[t]
\centering
\includegraphics[width=.7\textwidth]{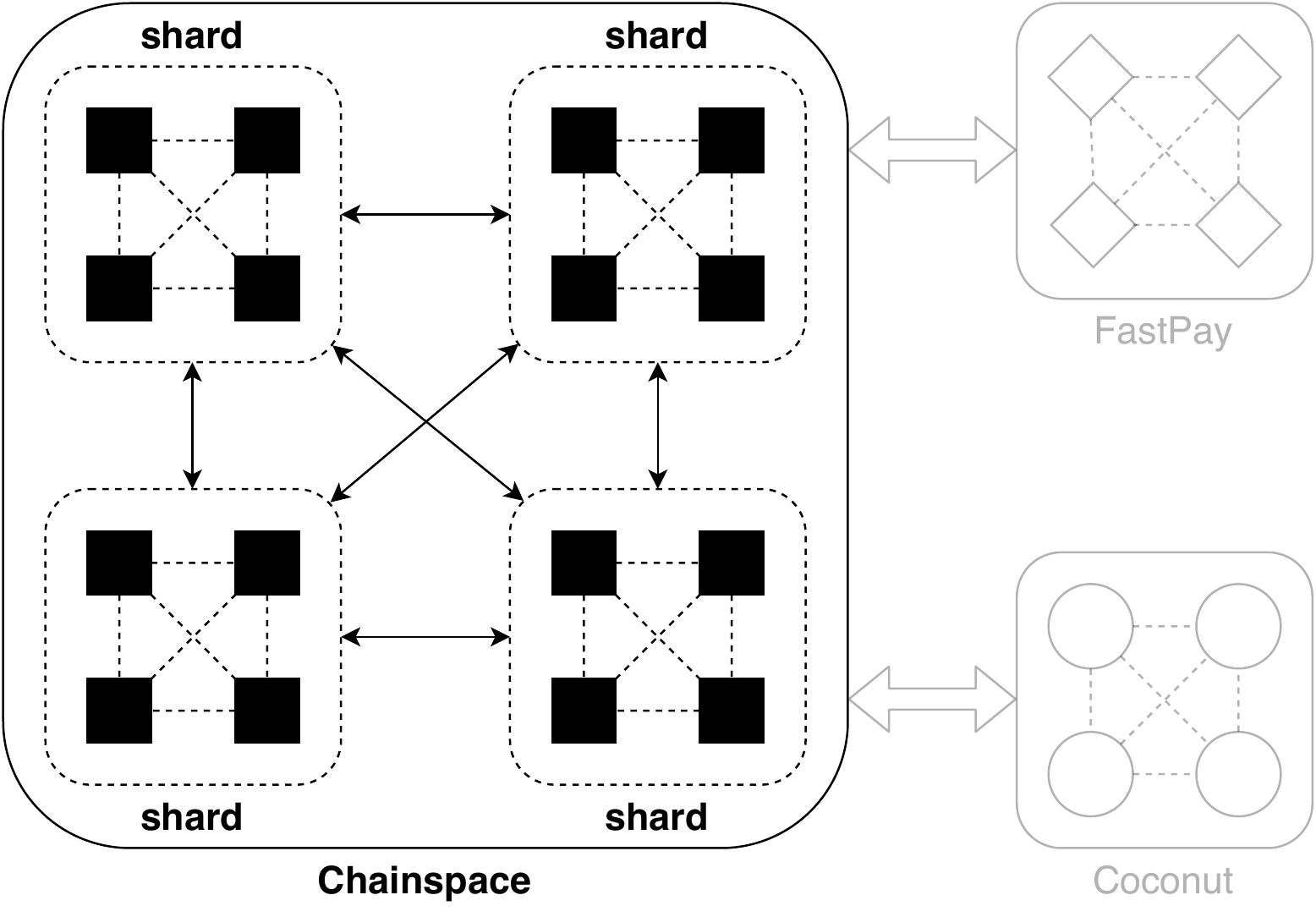}
\caption[Global overview: \chainspace.]{Example of instantiation of \chainspace with four shards containing four nodes each. Nodes are represented by black squares and the dashed lines connecting them to each other indicate they run a BFT protocol. The solid arrows connecting each shards are a \emph{cross-shard consensus protocol} allowing to coordinate shards operations.}
\label{fig:chainspace:global-overview}
\end{figure}
%

%This thesis introduces technologies to make blockchains practical, specifically aiming to overcome the following limitations: \first poor scalability, \second high latency, and \third difficulty to operate on secret values (privacy). 
%This chapter tackles blockchain scalability and privacy by presenting 
This chapter presents \chainspace, a sharded distributed ledger with integrated privacy support. \chainspace is the main component of \Cref{fig:introduction:overview} and is highlighted in \Cref{fig:chainspace:global-overview}. 
%\chainspace is a sharded distributed ledger; sharding is one of the main approaches to address blockchain scalability issues. 
The key idea of \chainspace is to create groups (called \emph{shards}) of nodes that handle only a subset of all transactions and system state, relying on classical Byzantine Fault Tolerance (BFT) protocols for reaching consensus amongst each shard. 
\Cref{fig:chainspace:global-overview} illustrates an instantiation of \chainspace with four shards containing four nodes each. Nodes are represented by black squares and the dashed lines connecting them to each other indicate they run a BFT protocol. The solid arrows connecting each shards are a \emph{cross-shard consensus protocol} allowing to coordinate shards operations; this protocol is discussed in \Cref{sec:chainspace:sbac} and is the subject of \Cref{byzcuit}.
%
%\chainspace is a distributed ledger platform for high-integrity and transparent processing of transactions within a decentralized system. 
%Unlike application specific blockchains, such as Bitcoin~\cite{bitcoin} for a currency, Google Spanner~\cite{google-spanner} for a relational database, or certificate transparency~\cite{laurie2014certificate} for certificate verification, 
\chainspace offers extensibility though smart contracts, like Ethereum~\cite{ethereum}. 
%However, 
%Users expose to \chainspace enough information about contracts and transaction semantics to provide higher scalability through sharding across infrastructure nodes. 
Furthermore, it is agnostic as to the smart contract language, or identity infrastructure, and supports privacy features through modern zero-knowledge techniques~\cite{bootle2016efficient,danezis2014square}. \chainspace supports privacy-preserving applications by design; it revisits the execution of smart contracts on blockchains, and proposes a system where the transaction is executed by the client, and the smart contract platform only verifies the correctness of the execution.

Unlike other scalable but `permissioned' smart contract platforms, such as \hyperledgerfabric~\cite{hyperledger} or \bigchaindb~\cite{mcconaghy2016bigchaindb}, \chainspace aims to be an `open' system: it allows anyone to author a smart contract, anyone to provide infrastructure on which smart contract code and state runs, and any user to access calls to smart contracts. Further, it provides ecosystem features, by allowing composition of smart contracts from different authors. We integrate a value system, named \cscoin, as a system smart contract to allow for accounting between those parties.
However, the security model of \chainspace, is different from traditional permissionless blockchains, that rely on proof-of-work and global replication of state, such as \ethereum. In \chainspace smart contract authors designate the parts of the infrastructure that are trusted to maintain the integrity of their contract---and only depend on their correctness, as well as the correctness of contract sub-calls. This provides fine grained control of which part of the infrastructure need to be trusted on a per-contract basis, and also allows for horizontal scalability.

% =========
\subsection*{Contributions} 
This chapter makes the following key contributions:
\begin{itemize}
\item It presents \chainspace, a system that can scale arbitrarily as the number of nodes increase, tolerates Byzantine failures, and can be publicly audited.
    
%\item It presents a novel distributed atomic commit protocol, called \sbac, for sharding generic smart contract transactions across multiple Byzantine nodes, and correctly coordinating those nodes to ensure safety, liveness and security. 

\item It introduces a distinction between parts of the smart contract that execute a computation, and those that check the computation and discusses how that distinction is key to supporting privacy-friendly smart-contracts.

%\item It evaluates the performance of the Byzantine distributed commit protocol, \sbac, on a real distributed set of nodes and under varying transaction loads.

\item It presents a number of key system and application smart contracts and evaluates their performance. The contracts for privacy-friendly smart-metering and privacy-friendly polls illustrate and validate support for high-integrity and high-privacy applications.
\end{itemize}

% =========
\subsection*{Outline}
\Cref{sec:chainspace:overview} presents an overview of \chainspace; \Cref{sec:chainspace:interface} presents the client-facing application interface; \Cref{sec:chainspace:design} presents the design of internal data structures guaranteeing integrity, the distributed architecture, and smart contract definition and composition. \Cref{sec:chainspace:theorems} argues the correctness and security; specific smart contracts and their evaluations are presented in \Cref{sec:chainspace:applications}; \Cref{sec:chainspace:evaluation} presents an evaluation of the smart contract performance; \Cref{sec:chainspace:limits} presents limitation and \Cref{sec:chainspace:conclusions} concludes the chapter.

%% file: chapters/chainspace/sections/02Overview.tex
% =========
% Overview
% =========
\section{Overview}
\label{sec:chainspace:overview}

\chainspace allows applications developers to implement distributed ledger applications by defining and calling procedures of smart contracts operating on controlled objects, and abstracts the details of how the ledger works and scales. In this section, we first describe data model of \chainspace, followed by an overview of the system design, its threat model and security properties. 

% =========
\subsection{Data Model: Objects, Contracts, Transactions.}
\label{sec:chainspace:data-model}

\chainspace applies aggressively the end-to-end principle~\cite{saltzer1984end} in relying on untrusted end-user applications to build transactions to be checked and executed. We describe below key concepts within the \chainspace data model. %that developers need to grasp to use the system.

\emph{Objects} are atoms that hold state in the \chainspace system. We usually refer to an object through the letter $o$, and a set of objects as $\csobject \in O$. All objects have a cryptographically derived unique identifier used to unambiguously refer to the object, that we denote $\id(\csobject)$. Objects also have a type, denoted as $\type(\csobject)$, that determines the unique identifier of the smart contract that defines them, and a type name. In \chainspace object state is immutable. Objects may be in two meta-states, either \emph{active} or \emph{inactive}. Active objects are available to be operated on through smart contract procedures, while inactive ones are retained for the purposes of audit only.
\emph{Contracts} are special types of objects, that contain executable information on how other objects of types defined by the contract may be manipulated. They define a set of initial objects that are created when the contract is first created within \chainspace. A contract $\contract$ defines a \emph{namespace} within which \emph{types} (denoted as $\types(\contract)$) and a \emph{checker} $\checker$ for \emph{procedures} (denoted as $\procedures(\contract)$) are defined.
A \emph{procedure}, $\csprocedure$, defines the logic by which a number of objects, that may be \emph{inputs} or \emph{references}, are processed by some logic and \emph{local parameters} and \emph{local return values} (denoted as $\lparams$ and $\lreturns$), to generate a number of object \emph{outputs}. Notionally, input objects, denoted as a vector $\pinputs$, represent state that is invalidated by the procedure; references, denoted as $\preferences$ represent state that is only read; and outputs are objects, or $\poutputs$ are created by the procedure. Some of the local parameters or local returns may be secrets, and require confidentiality. We denote those as $\sparams$ and $\sreturns$ respectively.
We denote the execution of such a procedure as: 
\begin{equation}\label{eq:procedure}
\contract.\csprocedure(\pinputs, \preferences, \lparams, \sparams) \rightarrow \poutputs, \lreturns, \sreturns
\end{equation}
for $\pinputs, \preferences, \poutputs \in O$ and $\csprocedure \in \procedures(\contract)$. We restrict the type of all objects (inputs $\pinputs$, outputs $\poutputs$ and references $\preferences$) to have types defined by the same contract $\contract$ as the procedure $\csprocedure$ (formally: $\forall \csobject \in \pinputs \cup \poutputs \cup \preferences . \type(\csobject) \in \types(\contract)$). However, public locals (both $\lparams$ and $\lreturns$) may refer to objects that are from different contracts through their identifiers. We further require a procedure that outputs an non empty set of objects $\poutputs$, to also take as parameters a non-empty set of input objects $\pinputs$. Transactions that create no outputs are allowed to just take locals and references $\preferences$.

Associated with each smart contract $\contract$, we define a \emph{checker} denoted as $\checker$. Those checkers are pure functions (ie.\ deterministic, and have no side-effects), and return a Boolean value. A checker $\checker$ is defined by a contract, and takes as parameters a procedure $\csprocedure$, as well as inputs, outputs, references and locals.
\begin{equation}\label{eq:chainspace:checker}
\contract.\checker(\csprocedure, \pinputs, \preferences, \lparams, \poutputs, \lreturns, \dependencies) \rightarrow \{\mathsf{true}, \mathsf{false}\}
\end{equation}
Note that checkers do not take any secret local parameters ($\sparams$ or $\sreturns$). A checker for a smart contract returns $\mathsf{true}$ only if there exist some secret parameters $\sparams$ or $\sreturns$, such that an execution of the contract procedure $\csprocedure$, with the parameters passed to the checker alongside $\sparams$ or $\sreturns$, is possible as defined in \Cref{eq:procedure}.
The variable $\dependencies$ represent the context in which the procedure is called: namely information about other procedure executions. This supports composition, as we discuss in detail in the next section.
We note that procedures, unlike checkers, do not have to be pure functions, and may be randomized, keep state or have side effects. A smart contract defines explicitly the checker $\contract.\checker$, but does not have to define procedures \emph{per se}. The \chainspace system is oblivious to procedures, and relies merely on checkers. Yet, applications may use procedures to create valid transactions. The distinction between procedures and checkers---that do not take secrets---is key to implementing privacy-friendly contracts.

\emph{Transactions} represent the atomic application of one or more valid procedures to active input objects, and possibly some referenced objects, to create a number of new active output objects. The design of \chainspace is user-centric, in that a user client executes all the computations necessary to determine the outputs of one or more procedures forming a transaction, and provides enough evidence to the system to check the validity of the execution and the new objects.
Once a transaction is accepted in the system it `consumes' the input objects, that become inactive, and brings to life all new output objects that start their life by being active. References on the other hand must be active for the transaction to succeed, and remain active once a transaction has been successfully committed.
A client packages enough information about the execution of those procedures to allow \chainspace to safely \emph{serialize} its execution, and \emph{atomically} commit it only if all transactions are valid according to relevant smart contract checkers.

% =========
\subsection{System Design, Threat Model and Security Properties}
\label{sec:chainspace:threat-security}

\begin{figure}[!t]
\centering
\includegraphics[width=.8\textwidth]{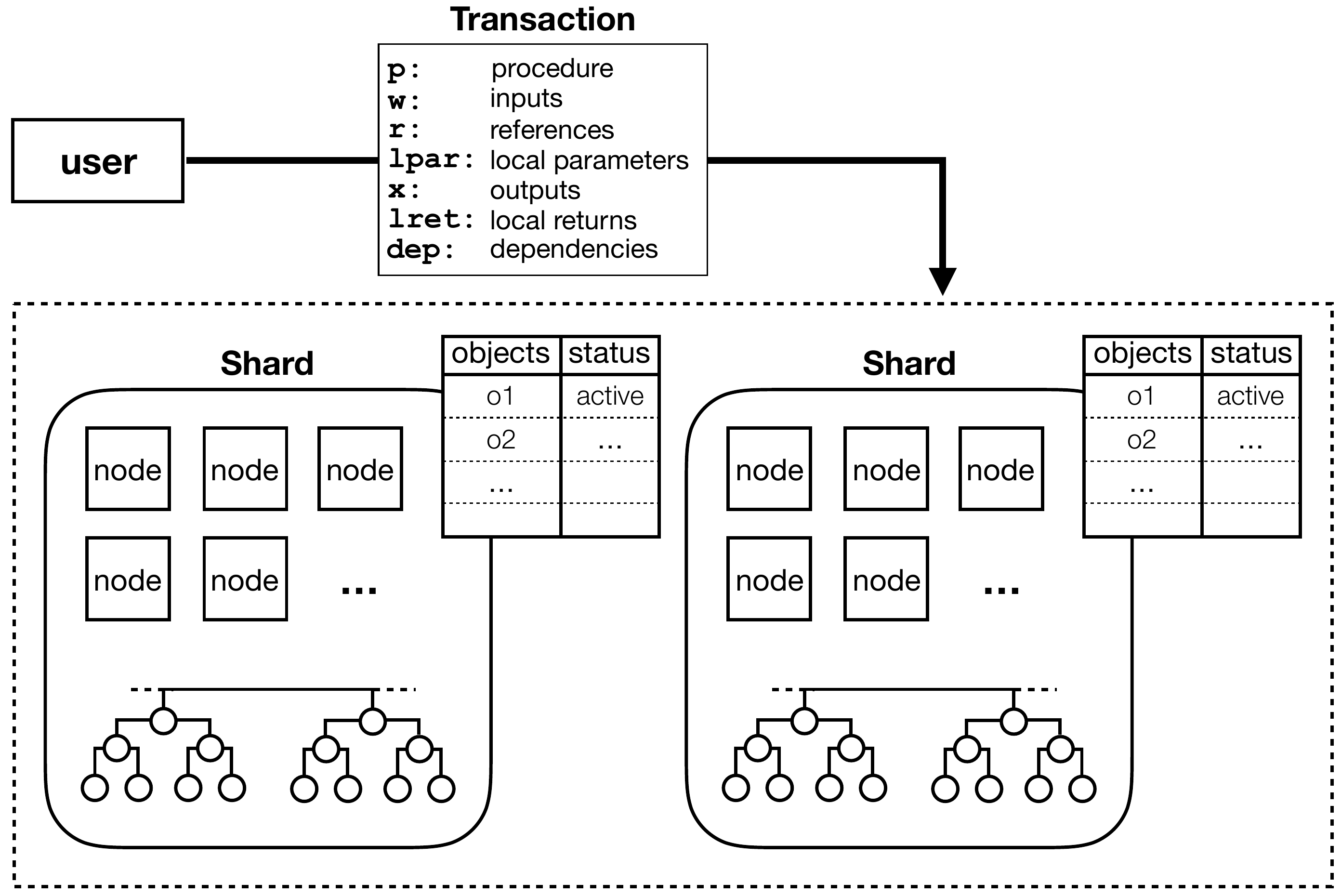}
\caption[Design overview of \chainspace.]{Design overview of \chainspace system, showing the interaction between users, transactions, objects and nodes in shards. \chainspace is made of a network of infrastructure \emph{nodes} that manages valid objects, and ensures that only valid transactions are committed.}
\label{fig:chainspace:design}
\end{figure}

We provide an overview of the system design, illustrated in \Cref{fig:chainspace:design}. \chainspace is comprised of a network of infrastructure \emph{nodes} that manages valid objects, and ensures that only valid transactions are committed. A key design goal is to achieve scalability in terms of high transaction throughput and low latency. To this end, nodes are organized into shards that manage the state of objects, keep track of their validity, and record transactions aborted or committed. Within each shard all honest nodes ensure they consistently agree whether to accept or reject a transaction: whether an object is active or inactive at any point, and whether traces from contracts they know check. Across shards, nodes must ensure that transactions are \emph{committed} if all shards are willing to commit the transaction, and rejected (or \emph{aborted}) if any shards decide to abort the transaction---due to checkers returning \textsf{false} or objects being inactive. To satisfy these requirements, \chainspace relies on a \emph{cross-shard consensus protocol}.
%\chainspace implements \sbac---a protocol that composes existing Byzantine agreement and atomic commit primitives in a novel way. 
Consensus on committing (or aborting) transactions takes place in parallel across different shards. For transparency and auditability, nodes in each shard periodically publish a signed hash chain of \emph{checkpoints}: shards add a block (Merkle tree) of evidence including transactions processed in the current epoch, and signed promises from other nodes, to the hash chain.     
\chainspace supports security properties against two distinct types of adversaries, both polynomial time bounded:

\begin{itemize}
\item  \textbf{Honest Shards (HS).} The first adversary may create arbitrary contracts, and input arbitrary transactions into \chainspace, however they are bound to only control up to $\faulty$ faulty nodes in any shard. As a result, and to ensure the correctness and liveness properties of Byzantine consensus, each shard must have a size of at least $3\faulty + 1$ nodes.

\item \textbf{Dishonest Shards (DS).} The second adversary has, additionally to HS, managed to gain control of one or more shards, meaning that they control over $\faulty$ nodes in those shards. Thus, its correctness or liveness may not be guaranteed.
\end{itemize}

Faulty nodes in shards may behave arbitrarily, and collude to violate any of the security, safely or liveness properties of the system. They may emit incorrect or contradictory messages, as well as not respond to any or some requests. 
Given this threat model, \chainspace supports the following security properties: 

\begin{itemize}
\item \textbf{Transparency.} \chainspace ensures that anyone in possession of the identity of a valid object may authenticate the full history of transactions and objects that led to the creation of the object. No transactions may be inserted, modified or deleted from that causal chain or tree. Objects may be used to self-authenticate its full history---this holds under both the HS and DS threat models.

\item \textbf{Integrity.} Subject to the HS threat model, when one or more transactions are submitted only a set of valid non-conflicting transactions will be committed within the system. This includes resolving conflicts---in terms of multiple transactions using the same objects---ensuring the validity of the transactions, and also making sure that all new objects are registered as active. Ultimately, \chainspace transactions are accepted, and the set of active objects changes, as if executed sequentially---however, unlike other systems such as Ethereum~\cite{ethereum}, this is merely an abstraction and high levels of concurrency are supported.

\item \textbf{Encapsulation.} The smart contract checking system of \chainspace enforces strict isolation between smart contracts and their state---thus prohibiting one smart contract from directly interfering with objects from other contracts. Under both the HS and DS threat models. However, cross-contract calls are supported but mediated by well defined interfaces providing encapsulation.

\item \textbf{Non-repudiation.} In case conflicting or otherwise invalid transactions were to be accepted in honest shards (in the case of the DS threat model), then evidence exists to pinpoint the parties or shards in the system that allowed the inconsistency to occur. Thus, failures outside the HS threat model, are detectable; the guilty parties may be banned; and appropriate off-line recovery mechanisms could be deployed.
\end{itemize}

%% file: chapters/chainspace/sections/03Interface.tex
% =========
% Interface
% =========
\section{The \chainspace Application Interface} \label{sec:chainspace:interface}
Smart contract developers in \chainspace register a smart contract $c$ into the distributed system managing \chainspace, by defining a  checker for the contract and some initial objects. Users may then submit transactions to operate on those objects in ways allowed by the checkers. Transactions represent the execution of one or more procedures from one or more smart contracts. It is necessary for all inputs to all procedures within the transaction to be active for a transaction to be executed and produce any output objects. 
Transactions are \emph{atomic}: either all their procedures run, and produce outputs, or none of them do. Transactions are also \emph{consistent}: in case two transactions are submitted to the system using the same active object inputs, at most one of them will eventually be executed to produce outputs. Other transactions, called \emph{conflicting}, will be aborted.

\para{Representation of transactions} A transaction within \chainspace is represented by sequence of \emph{traces} of the executions of the procedures that compose it, and their interdependencies. These are computed and packaged by end-user clients, and contain all the information  a checker needs to establish its correctness. A Transaction is a data structure such that:
\begin{align*}
\text{type}\ &\textit{Transaction}: \textit{Trace}\ \text{list}\\
\text{type}\ &\textit{Trace}: \text{Record}\ \{ \\
&\contract: \id(\csobject), \quad
\csprocedure: \text{string},\\
&\pinputs, \preferences, \poutputs: \id(\csobject)\ \text{list},\\
&\lparams, \lreturns: \text{arbitrary data}, \\
&\dependencies: \textit{Trace}\ \text{list} \}
\end{align*}
To generate a set of traces composing the transaction, a \emph{user executes on the client side all the smart contract procedures} required on the input objects, references and local parameters, and generates the output objects and local returns for every procedure---potentially also using secret parameters and returns. Thus the actual computation behind the transactions is performed by the user, and the traces forming the transaction already contain the output objects and return parameters, and sufficient information to check their validity through smart contract checkers. This design pattern is related to traditional \emph{optimistic concurrency control}.

\begin{figure*}[t]
\centering
\includegraphics[width=\textwidth]{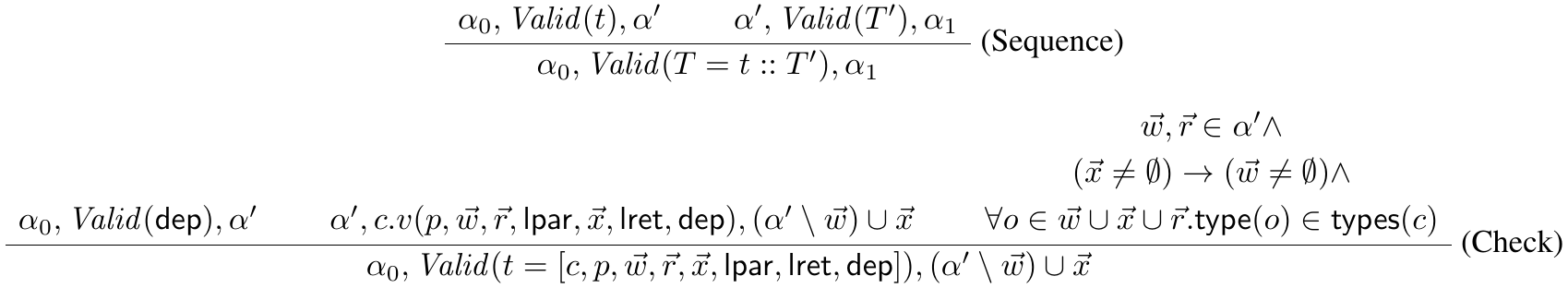}
\caption[\chainspace's sequencing and validity rules for transactions.]{The sequencing and checking validity rules for transactions.The \emph{Sequence} rule checks that each trace is valid in sequence. The \emph{Check} rule verifies that the dependencies of the trace are valid, and that the checker returns $\mathsf{true}$. It further checks that \first the inputs and references are active, \second that if the trace produces any output objects it also contains some input objects, and \third that all objects passed to the checker are of types defined by the smart contract of this checker.}
\label{fig:chainspace:rules}
\end{figure*}

Only valid transactions are eventually committed into the \chainspace system, as specified by two validity rules \emph{sequencing} and \emph{checking} presented in \Cref{fig:chainspace:rules}. Transactions are considered valid within a context of a set of active objects maintained by \chainspace, denoted with $\alpha$. Valid transactions lead to a new context of active objects (\eg $\alpha'$). We denote this through the triplet ($\alpha, \textit{Valid}(T), \alpha'$), which is true if the execution of transaction $T$ is valid within the context of active objects $\alpha$ and generates a new context of active objects $\alpha'$. The two rules are as follows:

\begin{itemize}
\item (Sequence rule). A `\textit{Trace} list' (within a `\textit{Transaction}' or list of dependencies) is valid if each of the traces are valid in sequence (see \Cref{fig:chainspace:rules} rule for sequencing). Further, the active objects set is updated in sequence before considering the validity of each trace.

\item (Check rule).  A particular `\textit{Trace}' is valid, if the sequence of its dependencies are valid, and then in the resulting active object context, the checker for it returns $\mathsf{true}$. A further three side conditions must hold: \first inputs and references must be active; \second if the trace produces any output objects it must also contain some input objects; and \third all objects passed to the checker must be of types defined by the smart contract of this checker (see \Cref{fig:chainspace:rules} rule for checking).
\end{itemize}

The ordering of active object sets in the validation rules result in a depth-first validation of all traces, which represents a depth-first execution and data flow dependency between them. It is also noteworthy that only the active set of objects needs to be tracked to determine the validity of new transactions, which is in the order of magnitude of active objects in the system. The much longer list of inactive objects, which grows to encompass the full history of every object in the system is not needed---which we leverage to enable scaling when validating transactions. It also results in a smaller amount of working memory to perform incremental audits.
A valid transaction is executed in a serialized manner, and committed or aborted atomically. If it is committed, the new set of active objects replaces the previous set; if not the set of active objects does not change. Determining whether a transaction may commit involves ensuring all the input objects are active, and all are consumed as a result of the transaction executing, as well as all new objects becoming available for processing (references however remain active). \chainspace ensures this through a distributed atomic commit protocol (\Cref{sec:chainspace:sbac}).

\para{Smart contract composition} A contract procedure may call a transaction of another smart contract, with specific parameters and rely upon returned values. This is achieved through passing the $\dependencies$ variable to a smart contract checker, a validated list of traces of all the sub-calls performed. The checker can ensure that the parameters and return values are as expected, and those dependencies are checked for validity by \chainspace.
Composition of smart contracts is a key feature of a transparent and auditable computation platform. It allows the creation of a library of smart contracts that act as utilities for other higher-level contracts: for example, a simple contract can implement a cryptographic currency, and other contracts---for e-commerce for example---can use this currency as part of their logic. 
\Cref{sec:coconut:smart_contract_library} shows how to take advantage of this feature to build a smart contract library allowing to issue and verify selective disclosure credentials.
Furthermore, we compose smart contracts, in order to build some of the functionality of \chainspace itself as a set of `system' smart contracts, including management of shards mapping to nodes, key management of shard nodes, and governance.

\chainspace also supports the atomic batch execution of multiple procedures for efficiency, that are not dependent on each other.

\para{Reads} Besides executing transactions, \chainspace clients, need to read the state of objects, if anything, to correctly form transactions. Reads, by themselves, cannot lead to inconsistent state being accepted into the system, even if they are used as inputs or references to transactions. This is a result of the system checking the validity rules before accepting a transaction, which will reject any stale state. 
Thus, any mechanism may be used to expose the state of objects to clients, including traditional relational databases, or `no-SQL' alternatives. Additionally, any indexing mechanism may be used to allow clients to retrieve objects with specific characteristics faster. Decentralized, read-only stores have been extensively studied, so we do not address the question of reads further in this work.

\para{Privacy by design} Defining smart contract logic as checkers allows \chainspace to support privacy friendly-contracts by design. In such contracts some information in objects is not in the clear, but instead either encrypted using a public key, or committed using a secure commitment scheme as~\cite{pedersen}. The transaction only contains a valid proof that the logic or invariants of the smart contract procedure were applied correctly or hold respectively, and can take the form of a zero-knowledge proof, or a Succinct Argument of Knowledge (SNARK). Then, generalizing the approach of Miers~\etal~\cite{zerocoin}, the checker runs the verifier part of the proof or SNARK that validates the invariants of the transactions, without revealing the secrets within the objects to the verifiers. 

%% file: chapters/chainspace/sections/04Design.tex
% =========
% Design
% =========
\section{The \chainspace System Design} \label{sec:chainspace:design}
In \chainspace a network of infrastructure \emph{nodes} manages valid objects, and ensure key invariants: namely that only valid transactions are committed. We discuss the data structures nodes use collectively and locally to ensure high integrity; and the distributed protocols they employ to reach consensus on the accepted transactions.

% =========
\subsection{High-Integrity Data Structures}

\chainspace employs a number of high-integrity data structures. They enable those in possession of a valid object or its identifier
to verify all operations that lead to its creation; they are also used to support \emph{non-equivocation}---preventing \chainspace nodes from providing a split view of the state they hold without detection.

\para{Hash-DAG structure} Objects and transactions naturally form a directed acyclic graph (DAG): given an initial state of active objects a number of transactions render their inputs invalid, and create a new set of outputs as active objects. These may be represented as a directed graph between objects, transactions and new objects and so on. Each object may only be created by a single transaction trace, thus cycles between future transactions and previous objects never occur. We prove that output object identifiers resulting from valid transactions are fresh (see \Cref{th:chainspace:thm1}). Hence, the graph of objects inputs, transactions and objects outputs form a DAG, that may be indexed by their identifiers.
We leverage this DAG structure, and augment it to provide a high-integrity data structure. Our principal aim is to ensure that given an object, and its identifier, it is possible to unambiguously and unequivocally check all transactions and previous (now inactive) objects and references that contribute to the existence of the object. To achieve this we define as an identifier for all objects and transactions a cryptographic hash that directly or indirectly depends on the identifiers of all state that contributed to the creation of the object.

Specifically, we define a function $\id(\mathit{Trace})$ as the identifier of a trace contained in transaction $T$. The identifier of a trace is a cryptographic hash function over the name of contract and the procedure producing the trace; as well as serialization of the input object identifiers, the reference object identifiers, and all local state of the transaction (but not the secret state of the procedures); the identifiers of the trace's dependencies are also included. Thus all information contributing to defining the Trace is included in the identifier, except the output object identifiers.
We also define the $\id(\csobject)$ as the identifier of an object $\csobject$. We derive this identifier through the application of a cryptographic hash function, to the identifier of the trace that created the object $\csobject$, as well as a unique name assigned by the procedures creating the trace, to this output object. (Unique in the context of the outputs of this procedure call, not globally, such as a local counter.)
An object identifier $\id(\csobject)$ is a high-integrity handle that may be used to authenticate the full history that led to the existence of the object $\csobject$. Due to the collision resistance properties of secure cryptographic hash functions an adversary is not able to forge a past set of objects or transactions that leads to an object with the same identifier. Thus, given  $\id(\csobject)$ anyone can verify the authenticity of a trace that led to the existence of $\csobject$.
A very important property of object identifiers is that future transactions cannot re-create an object that has already become inactive. Thus checking object validity only requires maintaining a list of active objects, and not a list of past inactive objects.

\begin{theorem} \label{th:chainspace:thm1}
No sequence of valid transactions, by a polynomial time constrained adversary, may re-create an object with the same identifier with an object that has already been active in the system.
\end{theorem}
\begin{proof}
\small
We argue this property by induction on the serialized application of valid transactions, and for each transaction by structural induction on the two validity rules. Assuming a history of $n-1$ transactions for which this property holds we consider transaction $n$. Within transaction $n$ we sequence all traces and their dependencies, and follow the data flow of the creation of new objects by the `check' rule. For two objects to have the same $\id(\csobject)$ there need to be two invocations of the check rule with the same contract, procedure, inputs and references. However, this leads to a contradiction: once the first trace is checked and considered valid the active input objects are removed from the active set, and the second invocation becomes invalid. Thus, as long as object creation procedures have at least one input (which is ensured by the side condition) the theorem holds, unless an adversary can produce a hash collision. The inductive base case involves assuming that no initial objects start with the same identifier---which we can ensure axiomatically.
\end{proof}

We call this directed acyclic graph with identifiers derived using cryptographic functions a Hash-DAG, and we make extensive use of the identifiers of objects and their properties in \chainspace.

\para{Node hash-chains} Each node in \chainspace, that is entrusted with preserving integrity, associates with its shard a hash chain. Periodically, peers within a shard consistently agree to seal a \emph{checkpoint}, as a block of transactions into their hash chains. They each form a Merkle tree containing all transactions that have been accepted or rejected in sequence by the shard since the last checkpoint was sealed. Then, they extend their hash chain by hashing the root of this Merkle tree and a block sequence number, with the head hash of the chain so far, to create the new head of the hash chain. Each peer signs the new head of their chain, and shares it with all other peers in the shard, and anyone who requests it. For strong auditability additional information, besides committed or aborted transactions, has to be included in the Merkle tree: node should log any promise to either commit or abort a transaction from any other peer in any shard (see \Cref{sec:chainspace:sbac}). 
All honest nodes within a shard independently create the same chain for a checkpoint, and a signature on it---as long as the consensus protocols within the shards are correct. We say that a checkpoint represents the decision of a shard, for a specific sequence number, if at least $\faulty + 1$ signatures of shard nodes sign it. On the basis of these hash chains we define a \emph{partial audit} and a \emph{full audit} of the \chainspace system.

In a \emph{partial audit} a client is provided evidence that a transaction has been either committed or aborted by a shard. A client performing the partial audit may request from any node of the shard evidence for a transaction \transaction. The shard peer will present a block representing the decision of the shard, with $\faulty + 1$ signatures, and a proof of inclusion of a commit or abort for the transaction, or a signed statement the transaction is unknown. A partial audit provides evidence to a client of the fate of their transaction, and may be used to detect past of future violations of integrity. A partial audit is an efficient operation since the evidence has size $O(s + \log N)$ in $N$ the number of transactions in the checkpoint and $s$ the size of the shard---thanks to the efficiency of proving inclusion in a Merkle tree, and checking signatures.

A \emph{full audit} involves replaying all transactions processed by the shard, and ensuring that \first all transactions were valid according to the checkers the shard executed; \second the objects input or references of all committed transactions were all active (see rules in \Cref{fig:chainspace:rules}); and \third the evidence received from other shards supports committing or aborting the transactions. To do so an auditor downloads the full hash-chain representing the decisions of the shard from the beginning of time, and re-executes all the transactions in sequence. This is possible, since---besides their secret signing keys---peers in shards have no secrets, and their execution is deterministic once the sequence of transactions is defined. Thus, an auditor can re-execute all transactions in sequence, and check that their decision to commit or abort them is consistent with the decision of the shard. Doing this, requires any inter-shard communication (namely the promises from other shards to commit or abort transactions) to be logged in the hash-chain, and used by the auditor to guide the re-execution of the transactions. A full audit needs to re-execute all transactions and requires evidence of size $O(N)$ in the number $N$ of transactions. This is costly, but may be done incrementally as new blocks of shard decisions are created. 

% =========
\subsection{Distributed Architecture \& Consensus}
A network of \emph{nodes} manages the state of \chainspace objects, keeps track of their validity, and record transactions that are seen or that are accepted as being committed. 

\chainspace uses sharding strategies to ensure scalability: a public function $\shard(\csobject)$ maps each object $\csobject$ to a set of nodes, we call a \emph{shard}. These nodes collectively are entrusted to manage the state of the object, keep track of its validity, record transactions that involve the object, and eventually commit at most one transaction consuming the object as input and rendering it inactive. However, nodes must only record such a transaction as committed if they have certainty that all other nodes have, or will in the future, record the same transaction as consuming the object. We call this distributed algorithm the \emph{consensus} algorithm within the shard.
For a transaction $\transaction$ we define a set of \emph{concerned nodes}, $\Phi(\transaction)$ for a transaction structure $\transaction$. We first denote as $\zeta$ the set of all objects identifiers that are input into or referenced by any trace contained in $\transaction$. We also denote as $\xi$ the set of all objects that are output by any trace in $\transaction$. The function $\Phi(\transaction)$ represents the set of nodes that are managing objects that should exist, and be active, in the system for $\transaction$ to succeed. More mathematically, $\Phi(\transaction) = \bigcup \{ \phi(\csobject_i) | \csobject_i \in \zeta \setminus \xi \} $, where $\zeta \setminus \xi$ represents the set of objects input but not output by the transaction itself (its free variables). The set of concerned peers thus includes all shard nodes managing objects that already exist in \chainspace that the transaction uses as references or inputs. 

An important property of this set of nodes holds, that ensures that all smart contracts involved in a transaction will be mapped to some concerned nodes that manage state from this contract:

\begin{theorem}
\label{thm:chainspace:partition}
If a contract $\contract$ appears in any trace within a transaction $\transaction$, then the concerned nodes set $\Phi(\transaction)$ will contain nodes in a shard managing an object $\csobject$ of a type from contract $\contract$. I.e.\ $\exists o. \type(o) \in \types(\contract) \land \shard(o) \cap \Phi(\transaction) \neq \emptyset$.
\end{theorem}
\begin{proof}
\small
Consider any trace $t$ within $\transaction$, from contract $c$. If the inputs or references to this trace are not in $\xi$---the set of objects that were created within $\transaction$---then their shards will be included within $\Phi(\transaction)$. Since those are of types within $\contract$ the theorem holds. If on the other hand the inputs or references are in $\xi$, it means that there exists another trace within $\transaction$ from the same contract $\contract$ that generated those outputs. We then recursively apply the case above to this trace from the same $\contract$. The process will terminate with some objects of types in $\contract$ and shard managing them within the concerned nodes set---and this is guarantee to terminate due to the Hash-DAG structure of the transactions (that may have no loops).
\end{proof}

Security Theorem~\ref{thm:chainspace:partition} ensures that the set of concerned nodes, includes nodes that manage objects from all contracts represented in a transaction. \chainspace leverages this to distribute the process of rule validation across peers in two ways:

\begin{itemize}
\item For any existing object $\csobject$ in the system, used as a reference or input within a transaction $\transaction$, only the shard nodes managing it, namely in $\shard(\csobject)$, need to check that it is active (as part of the `check' rule in \Cref{fig:chainspace:rules}).

\item For any trace $t$ from contract $\contract$ within a transaction $\transaction$, only shards of concerned nodes that manage objects of types within $\contract$ need to run the checker of that contract to validate the trace (again as part of the `check' rule), and that all input, output and reference objects are of types within $\contract$.
\end{itemize}

However, all shards containing concerned nodes for $\transaction$ need to ensure that all others have performed the necessary checks before committing the transaction, and creating new objects.
There are many options for ensuring that concerned nodes in each shards do not reach an inconsistent state for the accepted transactions, such as Nakamoto consensus through proof-of-work~\cite{bitcoin}, two-phase commit protocols, and classical consensus protocols like Paxos~\cite{lamport2001paxos}, PBFT~\cite{pbft}, or xPaxos~\cite{liu2016xft}. However, these approaches lack in performance, scalability, and/or security. 
%We design an open, scalable and decentralized mechanism to perform \emph{Sharded Byzantine Atomic Commit} (\sbac), which is described in the next section.
\chainspace requires a scalable and decentralized mechanism to achieve consensus across shards.

% =========
\subsection{Cross-Shard Consensus Protocol} \label{sec:chainspace:sbac}
\chainspace entrusts each object to a shard of nodes, that keep track of whether it exists, it is active or inactive. Within each shard all honest nodes must ensure they consistently agree whether to accept or reject a transaction: whether an object is active or inactive at any point, and whether traces from contracts they know check. Across shards, nodes must ensure that transactions are \emph{committed} if all shards are willing to commit the transaction, and rejected (or \emph{aborted}) if any shards decide to abort the transaction---due to it being invalid or objects involved being inactive. 
\chainspace requires a sharded consensus protocol for transaction processing in the \emph{Byzantine} and \emph{asynchronous} setting. Such protocols are delicate to build, and are the subject of \Cref{byzcuit}. For the rest of this chapter, we assume \chainspace implements a cross-shard consensus protocol providing liveness, consistency and validity (see \Cref{sec:chainspace:theorems}), and treats it as a black box.

% =========
%\subsection{Concurrency \& Scalability}
Each transaction $\transaction$ involves a fixed number of \emph{concerned nodes} $\Phi(\transaction)$ within \chainspace, corresponding to the shards managing its inputs and references. If two transactions $\transaction_0$ and $\transaction_1$ have disjoint sets of concerned nodes ($\Phi(T_0) \cap \Phi(T_1) = \emptyset$) they cannot conflict, and are executed in parallel or in any arbitrary order. Even if this set is not empty, but instead the input objects of those two transactions are disjoint, this property holds.
If however, two transactions have common input objects, only one of them is accepted by all nodes. This is achieved through the cross-shard consensus protocol. It is local, in that it concerns only nodes managing the conflicting transactions, and does not require a global consensus.
From the point of view of scalability, this protocol allows \chainspace's capacity grows linearly as more shards are added, subject to transactions having on average a constant, or sub-linear, number of inputs and references. Furthermore, those inputs must be managed by different nodes within the system to ensure that load of accepting transactions is distributed across them. 
How such distribution is achieved depends on the threat model of the application. Some may opt for a totally peer to peer model, where an ad-hoc random quorum of peers manages each object. Other application may opt for a small set of well-known authorities managing peers, with each object being managed by at least one representative peer from each authority.

%% file: chapters/chainspace/sections/05Security.tex
% =========
% Security
% =========
\section{Security and Correctness}
\label{sec:chainspace:theorems}
\chainspace relies on a black box cross-shard consensus protocol (see \Cref{byzcuit}), on which rest the security of \chainspace, namely \emph{liveness}, \emph{consistency}, and \emph{validity}. Those properties hold under the `honest shards' threat model (see \Cref{sec:chainspace:threat-security}). Liveness ensures that transactions make progress, and no locks are held indefinitely on objects, preventing other transactions from making progress. Consistency ensures that the execution of valid transactions could be serialized, and thus is correct. Validity ensures a transaction may only be committed if it is valid according to the smart contract checkers matching the traces of the procedures it executes. 

% =========
\subsection{Auditability}
In the previous sections we show that if each shard contains at most $\faulty$ faulty nodes (honest shard model), the cross-shard consensus protocol guarantees consistency and validity. In this section we argue that if this assumption is violated, \ie one or more shards contain more than $\faulty$ byzantine nodes each, then honest shards can detect faulty shards. Namely, enough auditing information is maintained by honest nodes in \chainspace to detect inconsistencies and attribute them to specific shards (or nodes within them).
The rules for transaction validity are summarized in \Cref{fig:chainspace:rules}. Those rules are checked in a distributed manner: each shard keeps and checks the active or inactive state of objects assigned to it; and also only the contract checkers corresponding to the type of those objects. An honest shard emits a \preacceptt for a transaction \transaction only if those checks pass, and \preabortt otherwise or if there is a lock on a relevant object. A dishonest shard may emit any of those messages arbitrarily without checking the validity rules. By definition, an invalid transaction is one that does not pass one or more of the checks defined in \Cref{fig:chainspace:rules} at a shard, for which the shard has erroneously emitted a \preacceptt or \preabortt message. 

\begin{theorem}[Auditability]
A malicious shard (with more than \faulty\ faulty nodes) that attempts to introduce an invalid transaction or object into the state of one or more honest shards, can be detected by an auditor performing a full audit of the \chainspace system.
\end{theorem}
\begin{proof}
\small
We consider two hash-chains from two distinct shards. We define the pair of them as being valid if \first they are each valid under full audit, meaning that a re-execution of all their transactions under the messages received yields the same decisions to commit or abort all transactions; and \second if all \preacceptt and \preabortt messages in one chain are compatible with all messages seen in the other chain. In this context `compatible' means that all \preacceptt and \preabortt statements received in one shard from the other represent the `correct' decision to commit or abort the transaction \transaction in the other shard. An example of incompatible message would result in observing a \preacceptt message being emitted from the first shard to the second, when in fact the first shard should have aborted the transaction, due to the checker showing it is invalid or an input being inactive.

Due to the property of digital signatures (unforgeability and non-repudiation), if two hash-chains are found to be `incompatible', one belonging to an honest shard and one belonging to a dishonest shard, it is possible for everyone to determine which shard is the dishonest one. To do so it suffices to isolate all statements that are signed by each shard (or a peer in the shard)---all of which should be self-consistent. It is then possible to show that within those statements there is an inconsistency---unambiguously implicating one of the two shards in the cheating. Thus, given two hash-chains it is possible to either establish their consistency, under a full audit, or determine which belongs to a malicious shard. 
\end{proof}

Note that the mechanism underlying tracing dishonest shards is an instance of the age-old double-entry book keeping\footnote{The first reported use is 1340AD~\cite{lauwers1994five}.}: shards keep records of their operations as a non-repudiable signed hash-chain of checkpoints---with a view to prove the correctness of their operations. They also provide non-repudiable statements about their decisions in the form of signed \preacceptt and \preabortt statements to other shards. The two forms of evidence must be both correct and consistent---otherwise their misbehavior is detected.

%% file: chapters/chainspace/sections/06Contracts.tex
% =========
% Smart Contracts
% =========
\section{System and Applications Smart Contracts}
\label{sec:chainspace:applications}
We present a number of key system and application smart contracts; the contracts for privacy-friendly smart-metering and privacy-friendly polls illustrate and validate support for high-integrity and high-privacy applications.

% =========
\subsection{System Contracts} \label{sec:chainspace:system-contracts}
The operation of a \chainspace distributed ledger itself requires the maintenance of a number of high-integrity high-availability data structures. Instead of employing an ad-hoc mechanism, \chainspace employs a number of \emph{system smart contracts} to implement those.  Effectively, an instantiation of \chainspace is the combination of nodes running a consensus protocol, as well as a set of system smart contracts providing flexible policies about managing shards, smart contract creation, auditing and accounting. This section provides an overview of system smart contracts.

\para{Shard management} The discussion of \chainspace so far, has assumed a function $\shard(\csobject)$ mapping an object $\csobject$ to nodes forming a shard. However, how those shards are constituted has been abstracted. A smart contract \textsf{ManageShards} is responsible for mapping nodes to shards. \textsf{ManageShards} initializes a singleton object of type \textsf{MS.Token} and provides three procedures: \textsf{MS.create} takes as input a singleton object, and a list of node descriptors (names, network addresses and public verification keys), and creates a new singleton object and a \textsf{MS.Shard} object representing a new shard; \textsf{MS.update} takes an existing shard object, a new list of nodes, and $2\faulty+1$ signatures from nodes in the shard, and creates a new shard object representing the updated shard. Finally, the \textsf{MS.object} procedure takes a shard object, and a non-repudiable record of malpractice from one of the nodes in the shard, and creates a new shard object omitting the malicious shard node---after validating the misbehavior. Note that \chainspace is `open' in the sense that any nodes may form a shard; and anyone may object to a malicious node and exclude it from a shard.

\para{Smart-contract management} \chainspace is also `open' in the sense that anyone may create a new smart contract, and this process is implemented using the \textsf{ManageContracts} smart contract. \textsf{ManageContracts} implements three types: \textsf{MC.Token}, \textsf{MC.Mapping} and \textsf{MC.Contract}. It also implements at least one procedure, \textsf{MC.create} that takes a binary representing a checker for the contract, an initialization procedure name that creates initial objects for the contract, and the singleton token object. It then creates a number of outputs: one object of type \textsf{MC.Token} for use to create further contracts; an object of type \textsf{MC.Contract} representing the contract, and containing the checker code, and a mapping object \textsf{MC.mapping} encoding the mapping between objects of the contract and shards within the system. Furthermore, the procedure \textsf{MC.create} calls the initialization function of the contract, with the contract itself as reference, and the singleton token, and creates the initial objects for the contract.
Note that this simple implementation for \textsf{ManageContracts} does not allow for updating contracts. The semantics of such an update are delicate, particularly in relation to governance and backwards compatibility with existing objects. The definitions of more complex, but correct, contracts for managing contracts are left for future work.

\para{Payments for processing transactions} \chainspace is an open system, and requires protection against abuse resulting from overuse. To achieve this we implement a method for tracking value through a contract called \textsf{CSCoin}.
The \textsf{CSCoin} contract creates a fixed initial supply of coins---a set of objects of type The \textsf{CSCoin.Account} that may only be accessed by a user producing a signature verified by a public key denoted in the object. A \textsf{CSCoin.transfer} procedure allows a user to input a number of accounts, and transfer value between them, by producing the appropriate signature from incoming accounts. It produces a new version of each account object with updated balances. This contract has been implemented in Python with approximately 200 lines of code.
The \textsf{CSCoin} contract is designed to be composed with other procedures, to enable payments for processing transactions. The transfer procedure outputs a number of local returns with information about the value flows, that may be used in calling contracts to perform actions conditionally on those flows. Shards may advertise that they will only consider actions valid if some value of \textsf{CSCoin} is transferred to their constituent nodes. This may apply to system contracts and application contracts.

% =========
\subsection{Application Level Smart Contracts}
This section describes some examples of privacy-friendly smart contracts and showcases how smart contract creators may use \chainspace to implement advanced privacy mechanisms.

\para{Sensor---`Hello World' contract}
To illustrate \chainspace's applications, we implement a simple 150 lines contract aggregating data from different sensors, called \textsf{Sensor}. This contract defines the types \textsf{Sensor.Token} and \textsf{Sensor.Data}, and two procedures  \textsf{Sensor.createSensor} and \textsf{Sensor.addData}. The procedure \textsf{Sensor.createSensor} takes as input the singleton token (that is created upon contract creation), and outputs a fresh \textsf{Sensor.Data} object with initially no date. The \textsf{Sensor.addData} procedure is applied on a \textsf{Sensor.Data} object with some new sensor's data as parameter; it then creates a new object \textsf{Sensor.Data} appending the list of new data to the previous ones.

\para{Smart-Meter private billing} 
A body of work~\cite{DBLP:conf/isse/RialD12} examines how to achieve privacy-friendly time of use billing for smart meter deployments---a use-case requiring both high-integrity, and also privacy. Thus it showcases how smart contract creators may use \chainspace to implement advanced privacy mechanisms.

We implement a basic private smart-meter billing mechanism~\cite{DBLP:conf/pet/JawurekJK11,DBLP:conf/isse/RialD12} using the contract \textsf{SMet}: it implements three types \textsf{SMet.Token}, \textsf{SMet.Meter} and \textsf{SMet.Bill}; and three procedures, \textsf{SMet.createMeter}, \textsf{SMet.AddReading}, and \textsf{SMet.computeBill}. The procedure \textsf{SMet.createMeter} takes as input the singleton token and a public key and signature as local parameters, and it outputs a \textsf{SMet.Meter} object tied to this meter public key if the signature matches. \textsf{SMet.Meter} objects represent a collection of readings and some meta-data about the meter. Subsequently, the meter may invoke \textsf{SMet.addReading} on a \textsf{SMet.Meter} with a set of cryptographic commitments readings and a period identifier as local parameters, and a valid signature on them. A signature is also included and checked to ensure authenticity from the meter. A new object \textsf{SMet.Meter} is output appending the list of new readings to the previous ones. Finally, a procedure \textsf{SMet.computeBill} is invoked with a \textsf{SMet.Meter} and local parameters a period identifier, a set of tariffs for each reading in the period, and a zero-knowledge proof of correctness of the bill computation. The procedure outputs a \textsf{SMet.Bill} object, representing the final bill in plain text and the meter and period information.
This proof of correctness is provided to the checker---rather than the secret readings---which proves that the readings matching the available commitments and the tariffs provided yield the bill object. The role of the checker, which checks public data, in both those cases is very different from the role of the procedure that is passed secrets not available to the checkers to protect privacy. 
This contract is implemented in about 200 lines of Python and is evaluated in section \Cref{sec:chainspace:evaluation}.

\para{A Platform for decision making}
An additional example of \chainspace's privacy-friendly application is a smart voting system. We implement the contract \textsf{SVote} with three types, \textsf{SVote.Token}, \textsf{SVote.Vote} and \textsf{SVote.Tally}; and three procedures. 
\textsf{SVote.createElection}, consumes a singleton token and takes as local parameters the options, a list of all voter's public key, the tally's public key, and a signature on them from the tally. It outputs a fresh \textsf{SVote.Vote} object, representing the initial stage of the election (all candidates having a score of zero) along with a zero-knowledge proof asserting the correctness of the initial stage.
\textsf{SVote.addVote}, is called on a \textsf{SVote.Vote} object and takes as local parameters a new vote to add, homomorphically encrypted and signed by the voter. In addition, the voter provides a zero-knowledge proof certifying that her vote is a binary value and that she voted for exactly one option. The voter's public key is then removed from the list of participants to ensure that she cannot vote more than once. If all proofs are verified by the checker and the voter's public key appears in the list, a new \textsf{SVote.Vote} object is created as the homomorphic addition of the previous votes with the new one. Note that the checker does not need to know the clear value of the votes to assert their correctness since it only has to verify the associated signatures and zero-knowledge proofs.
Finally, the procedure \textsf{SVote.tally} is called to threshold decrypt the aggregated votes and provide a \textsf{SVote.Tally} object representing the final election's result in plain text, along with a proof of correct decryption from the tally. The \textsf{SVote} contract's size is approximately 400 lines of Python and is also evaluated in \Cref{sec:chainspace:evaluation}.

%% file: chapters/chainspace/sections/07Evaluation.tex
% =========
% Evaluation
% =========
\section{Smart Contract Evaluation}
\label{sec:chainspace:evaluation}
We build a Python contracts environment allowing developers to write, deploy and test smart contracts. These are deployed on each node by running the Python script for the contract, which starts a local web service for the contract's checker. The contract's checker is then called though the web service. The environment provides a framework to allow developers to write smart contracts with little worry about the underlying implementation, and provides an auto-generated checker for simple contracts. We are releasing the code as an open-source project\footnote{\url{https://github.com/chainspace/chainspace}}.
We evaluate the cost and performance of some smart contracts described in \Cref{sec:chainspace:system-contracts}. We compute the mean and standard deviation of the execution of each procedure (denoted as [g]) and checker (denoted as [c]) in the contracts. Each figure is the result of 10,000 runs measured on a dual-core Apple MacBook Pro 4.1, 2.7GHz Intel Core i7. The last column indicates the transaction's size resulting from executing the procedure. All cryptographic operations as digital signatures and zero-knowledge proofs are implemented using the Python library petlib~\cite{petlib}, wrapping OpenSSL.

\begin{table}[t]
\centering
\begin{tabular}{lcccc}
\toprule
\multicolumn{3}{l}{ \textsf{Sensor}---Contract size: $\sim$150 lines}\\
\textbf{Operation} && \textbf{Mean (ms)} & \textbf{Std. (ms)} & \textbf{Size (B)}\\

\midrule
\textsf{createSensor} & [g] & 0.139 & $\pm$ 0.039 & 416\\ 
& [c] & 0.021 & $\pm$ 0.008 & - \\ 
\textsf{addData} & [g] & 0.105 & $\pm$ 0.085 & 417\\  
& [c] & 0.036 & $\pm$ 0.015 & - \\  
\bottomrule
\end{tabular}
\caption[\chainspace \textsf{Sensor} contract.]{\chainspace \textsf{Sensor} contract. The notation [g] denotes the execution the procedure and [c] denotes the execution of the checker. Each measure is the result of 10,000 runs.}
\label{tab:chainspace:sensor}
\end{table}
The \textsf{Sensor} contract (\Cref{tab:chainspace:sensor}) is very cheap since it does not contain any cryptographic operation; the checker only has to verify the format of the transaction. The resulting micro benchmarks for \textsf{createSensor} is therefore a good indication of just the overhead of the \chainspace system---to execute a procedure or a checker.
The user needs to generate a signing key pair to create an account in the \textsf{CSCoin} contract, which takes about 5~ms. However, verifying the account creation only requires to check the transaction's format, and it is therefore very fast. Transferring money is a little more expensive due to the need to sign the amount transferred and the beneficiary, and verifying the signature in the checker.

\begin{table}[t]
\centering
\begin{tabular}{lcccc}
\toprule
\multicolumn{3}{l}{ \textsf{CSCoin}---Contract size: $\sim$200 lines}\\
\textbf{Operation} && \textbf{Mean (ms)} & \textbf{Std. (ms)} & \textbf{Size (B)}\\
\midrule
\textsf{createAccount} & [g] & 4.845 & $\pm$ 0.683 & 512\\ 
& [c] & 0.022 & $\pm$ 0.005 & -\\ 
\textsf{authTransfer} & [g] & 4.986 & $\pm$ 0.684 & 1114\\  
&[c] & 5.750 & $\pm$ 0.474 & -\\  
\bottomrule
\end{tabular}
\caption[\chainspace \textsf{CSCoin} contract]{\chainspace \textsf{CSCoin} contract. The notation [g] denotes the execution the procedure and [c] denotes the execution of the checker. Each measure is the result of 10,000 runs.}
\label{tab:chainspace:cscoin}
\end{table}
\begin{table}[t]
\centering
\begin{tabular}{lcccc}
\toprule
\multicolumn{3}{l}{ \textsf{SMet}---Contract size: $\sim$200 lines}\\
\textbf{Operation} && \textbf{Mean (ms)} & \textbf{Std. (ms)} & \textbf{Size (B)}\\
\midrule
\textsf{createMeter} & [g] & 4.786 & $\pm$ 0.480 & $\sim$600\\ 
& [c] & 0.060 & $\pm$ 0.003 & -\\ 
\textsf{addReading} & [g] & 5.286 & $\pm$ 0.506 & $\sim$1100\\  
& [c] & 5.965 & $\pm$ 0.697 & -\\  
\textsf{computeBill} & [g] & 5.043 & $\pm$ 0.513 & $\sim$1100\\  
& [c] & 5.870 & $\pm$ 0.603 & -\\  
\bottomrule
\end{tabular}
\caption[\chainspace \textsf{SMet} contract.]{\chainspace \textsf{SMet} contract. The notation [g] denotes the execution the procedure and [c] denotes the execution of the checker. Each measure is the result of 10,000 runs.}
\label{tab:chainspace:smart-meter}
\end{table}

Similarly to \textsf{CSCoin} (\Cref{tab:chainspace:cscoin}), creating a meter (\Cref{tab:chainspace:smart-meter}) requires generating a cryptographic key pair which takes about 5~ms, while verifying the meter's creation is faster and only requires checking the transaction's format. Adding new readings takes about 5~ms, as the user needs to create a signed commitment of the readings which requires elliptic curve operations and an ECDSA signature. Computing the bill takes slightly longer (5.8~ms), and involves homomorphic additions, and verifying the bill involves checking a zero-knowledge proof of the billing calculation.

\begin{table}[t]
\centering
\begin{tabular}{lcccc}
\toprule
\multicolumn{3}{l}{ \textsf{SVote}---Contract size: $\sim$400 lines}\\
\textbf{Operation} && \textbf{Mean (ms)} & \textbf{Std. (ms)} & \textbf{Size (B)}\\
\midrule
\textsf{createElection} & [g] & 11.733 & $\pm$ 1.028 & $\sim$1227\\ 
& [c] & 11.327 & $\pm$ 0.782 & -\\ 
\textsf{addVote} & [g] & 14.086 & $\pm$ 1.043 & $\sim$2758\\  
& [c] & 28.178 & $\pm$ 1.433 & -\\ 
\textsf{tally} & [g] & 253.286 & $\pm$ 7.793 & $\sim$1264\\  
& [c] & 11.589 & $\pm$ 0.937 & -\\ 
\bottomrule
\end{tabular}
\caption[\chainspace \textsf{SVote} contract.]{\chainspace \textsf{SVote} contract. The notation [g] denotes the execution the procedure and [c] denotes the execution of the checker. Each measure is the result of 10,000 runs.}
\label{tab:chainspace:vote}
\end{table}

The \textsf{SVote} contract (\Cref{tab:chainspace:vote}) is more expensive than the others since it extensively uses zero-knowledge proofs and more advanced cryptography. For simplicity, this smart contract is tested with three voters and two options. First, creating a new election event requires building a signed homomorphic encryption of the initial value for each option, and a zero-knowledge proof asserting that the encrypted value is zero; this takes roughly 11~ms to generate the transaction and to run the checker. Next, each time a vote is added, the user proves two zero-knowledge statements---one asserting that she votes for exactly one option and one proving that her vote is a binary value---and computes an ECDSA signature on her vote, which takes about 11~ms and generates a transaction of about 2.7~kB. Verifying the signature and the two zero-knowledge proofs are slower and takes about 30~ms. Finally, tallying is the slowest operation since it requires to threshold decrypt the homomorphic encryption of the votes' sum.

%% file: chapters/chainspace/sections/08Limitations.tex
% =========
% Limitations
% =========
\section{Limitations}
\label{sec:chainspace:limits}
\chainspace has a number of limitations, that are beyond the scope of this work to  tackle, and deferred to future work.
The integrity properties of \chainspace rely on all shards managing objects being honest, namely containing at most $\faulty$  fault nodes each. We allow any set of nodes to create a shard. However, this means that the function $\phi(o)$ mapping objects to shards must avoid dishonest shards. Our isolation properties ensure that a dishonest shard can at worse affect state from contracts that have objects mapped to it. Thus, in \chainspace, we opt to allow the contract creator to designate which shards manage objects from their contract. This embodies specific trust assumptions where users have to trust the contract creator both for the code (which is auditable) and also for the choice of shards to involve in transactions---which is also public.
In case one or more shards are malicious, we provide an auditing mechanism for honest nodes in honest shards to detect 
the inconsistency and to trace the malicious shard. Through the Hash-DAG structure it is also possible to fully audit the histories of two objects, and to ensure that the validity rules hold jointly---in particular the double-use rules. However, it is not clear how to automatically recover from detecting such an inconsistency. Options include: forcing a fork into one or many consistent worlds; applying a rule to collectively agree the canonical version; patching past transactions to recover consistency; or agree on a minimal common consistent state. Which of those options is viable or best is left as future work.
Checkers involved in validating transactions can be costly. For this reason we allow peers in a shard to accept transactions subject to a \textsf{SCCoin} payment to the peers. However, this `flat' fee is not dependent on the cost or complexity of running the checker which might be more or less expensive. \ethereum~\cite{ethereum} instead charges `gas' according to the cost of executing the contract procedure---at the cost of implementing their own virtual machine and language.
%
%The \sbac protocol ensures correctness in all cases. However, under high contention for the same object the rate of aborted transactions rises. This is expected, since the \sbac protocol in effect implements a variant of optimistic concurrency control, that is known to result in aborts under high contention. There are strategies for dealing with this in the distributed systems literature, such as locking objects in some conventional order---however none is immediately applicable to the byzantine setting. 
%
Finally, we present in \Cref{sec:chainspace:system-contracts} `payments for processing transactions' to address DoS through transaction fees. However, this mechanism only offers a partial solutions, since operations with higher costs to the system would need to cost more, and all infrastructure nodes would have to be proper incentives to continue participating in the protocols. A through study of such mechanisms is left for future work.

%% file: chapters/chainspace/sections/09Related.tex
\section{Comparison with Related Work} \label{sec:chainspace:related}
We present some recent systems that provide a transparent platform based on blockchains for smart contracts, and compare it with \chainspace.  

\para{\omniledger}
The most comparable system to \chainspace is \omniledger~\cite{omniledger}---that was developed concurrently---and provides a scalable distributed ledger for a cryptocurrency, and cannot support generic smart contracts (in contrast with \chainspace which is a smart contract platform). \omniledger assigns nodes (selected using a Sybil-attack resistant mechanism) into shards among which state, representing coins, is split. The node-to-shard assignment is done every epoch using a bias-resistant decentralized randomness protocol~\cite{randherd} to prevent an adversary from compromising individual shards. Similarly to \chainspace, a block-DAG (Directed Acyclic Graph) structure is maintained in each shard rather than a single blockchain, effectively creating multiple blockchains in which consensus of transactions can take place in parallel. Nodes within shards reach consensus through the PBFT protocol~\cite{pbft} with \byzcoin~\cite{byzcoin}'s modifications that enable $O(n)$ messaging complexity. In contrast, \chainspace relies on a black box consensus protocols that can be implemented with any PBFT variant without breaking any security assumptions.

%\para{\hyperledger fabric}
%\hyperledger Fabric~\cite{hyperledger} is a permissioned blockchain to setup private infrastructures for smart contracts. It is designed around the idea of a `consortium' blockchain, where a specific set of nodes are designated to validate transactions, rather than random nodes in a decentralized network. Each smart contract (called \textit{chaincode}) has its own set of \textit{endorsers} that re-execute submitted transactions to validate them. A \textit{consensus service} then orders transactions and filters out those endorsed by too few. It uses \textit{modular consensus}, which is replaceable depending on the requirements (\eg Apache Kafka or SBFT).

\para{\ethereum} 
\ethereum~\cite{ethereum} provides a decentralized virtual machine, called EVM, for executing smart-contracts. Its main scalability limitation results from every node having to process every transaction, as Bitcoin. On the other hand, \chainspace's sharded architecture allows for a ledger linearly scalable since only the nodes concerned by the transaction---that is, managing the transaction's inputs or references---have to process it. Ethereum plans to improve scalability through sharding techniques~\cite{buterin2015notes}, but their work is still theoretical and is not implemented or measured. One major difference with \chainspace is that \ethereum's smart contract are executed by the node, contrarily to the user providing the outputs of each transaction. \chainspace also supports smart contracts written in any kind of language as long as checkers are pure functions, and there are no limitations for the code creating transactions. Some industrial systems~\cite{corda, rootstock} implement similar functionalities as \chainspace, but with little empirical performance evaluation.
In terms of security policy, \chainspace system implements a platform that enforces high-integrity by embodying a variant of the Clark-Wilson~\cite{clark1987comparison}, proposed before smart contracts were heard of. 

\para{Tezos} 
Tezos~\cite{tezos} has a strong type checking system and implements its cryptocurrency as an `account smart contract. However, Tezos's smart contracts are statefull and updating a balance requires to rewrite the contract's storage space. This introduces many complications to prevent replay attacks and transaction's validity's check. \chainspace avoid this situation by producing new version of account objects as specified in the Security Theorem 1 (see \Cref{sec:chainspace:theorems}). Moreover, Tezos implements an \ethereum-like gas system to pay nodes to execute the smart contract.

%\para{Rootstock} 
%Rootstock is a smart-contract platform that incorporate a Turing Complete virtual machine to Bitcoin. There are a number of Rootstock's smart-contracts illustrated in~\cite{rootstock}, but without any scientific evidence of implementation. 

%% file: chapters/chainspace/sections/10Conclusion.tex
% =========
% Conclusion
% =========
\section{Chapter Summary}
\label{sec:chainspace:conclusions}
\chainspace is an open, distributed ledger platform for high-integrity and transparent processing of transactions. \chainspace offers extensibility though privacy-friendly smart contracts. We presented an instantiation of \chainspace by parameterizing it with a number of `system' and `application' contracts, along with their evaluation. However, unlike existing smart-contract based systems such as \ethereum~\cite{ethereum}, it offers high scalability through sharding across nodes, while offering high auditability. As such it offers a competitive alternative to both centralized and permissioned systems, as well as fully peer-to-peer, but unscalable systems like \ethereum.

%% file: chapters/byzcuit/byzcuit.tex
\chapter{Replay Attacks and Defenses Against Cross-shard Consensus}
\label{byzcuit}

\input{chapters/byzcuit/sections/introduction.tex}

%\input{chapters/byzcuit/sections/background.tex} % Moved to background chapter
\input{chapters/byzcuit/sections/overview.tex}
\input{chapters/byzcuit/sections/chainspace.tex}
\input{chapters/byzcuit/sections/omniledger.tex}
\input{chapters/byzcuit/sections/prerecoring.tex}
\input{chapters/byzcuit/sections/defenses.tex}
\input{chapters/byzcuit/sections/byzcuit.tex}
\input{chapters/byzcuit/sections/implementation.tex}

\input{chapters/byzcuit/sections/conclusion.tex}

%% file: chapters/byzcuit/sections/introduction.tex
% =========
% Introduction
% =========

%
\begin{figure}[t]
\centering
\includegraphics[width=.7\textwidth]{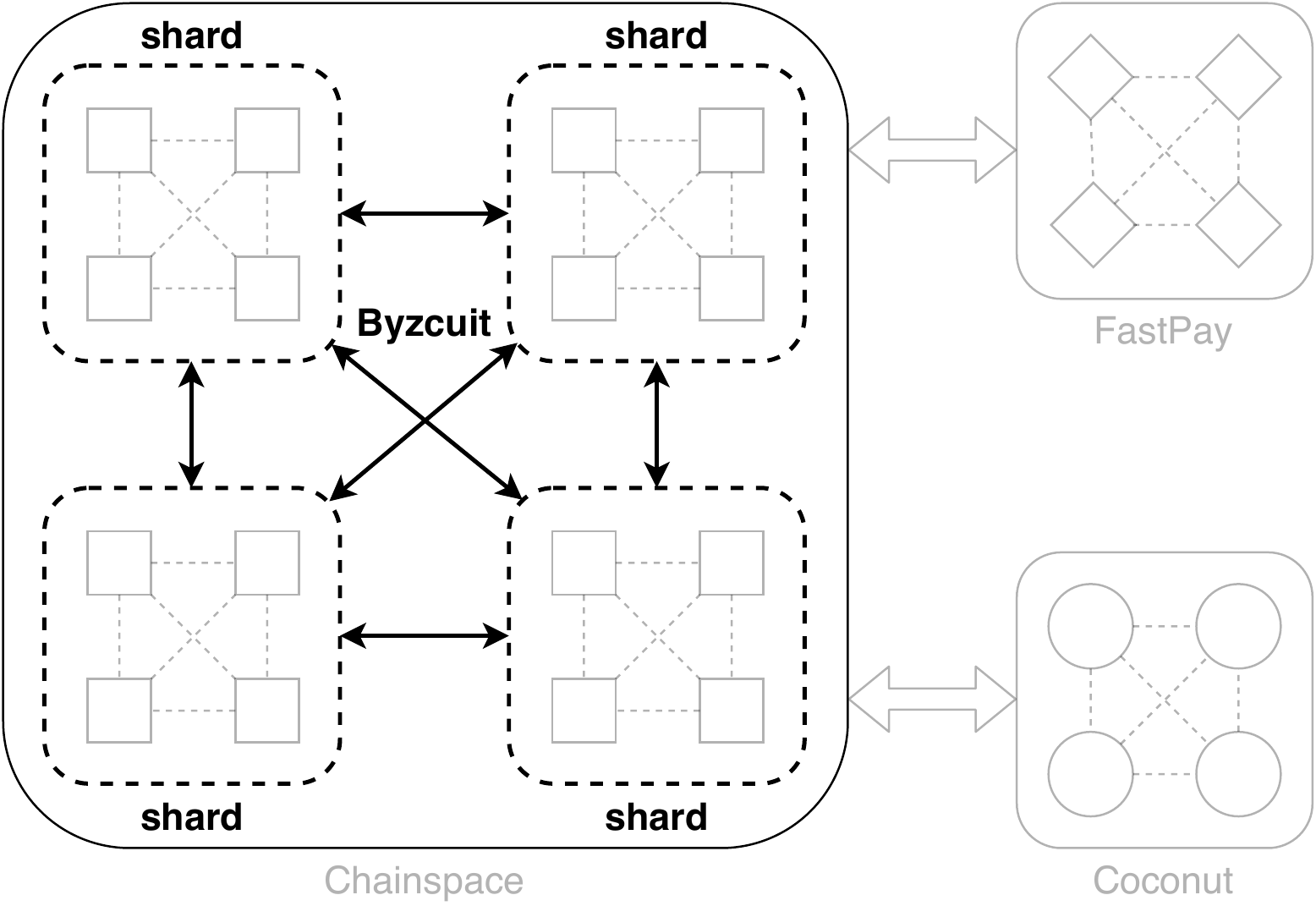}
\caption[Global overview: \byzcuit.]{Example of instantiation of \byzcuit with four shards. The solid arrows connecting each shards are the \byzcuit cross-shard consensus protocol allowing to coordinate shards.}
\label{fig:byzcuit:global-overview}
\end{figure}

\Cref{chainspace} describes \chainspace assuming it implements a black box cross-shard consensus protocol; this chapter opens that black box.
We first present a family of replay attacks against existing cross-shard consensus protocols, illustrating that designing such protocols is a delicate task. We then describe the issues that lead to these vulnerabilities, and present \byzcuit, a novel cross-shard consensus protocol that is immune to those attacks. \Cref{fig:byzcuit:global-overview} highlights \byzcuit, the cross-shard consensus protocol run at the heart of \chainspace.

%Sharding is one of the key approaches to address blockchain scalability issues~\cite{sok-consensus}, and a growing number of systems are implementing sharded blockchains~\cite{chainspace, omniledger, rscoin, rapidchain, elastico, sok-consensus}.  
As explained in the previous chapters, the key idea is to create groups (or shards) of nodes that handle only a subset of all transactions and system state, relying on classical Byzantine Fault Tolerance (BFT) protocols for reaching \emph{intra-shard consensus}. These systems achieve optimal performance and scalability because: \first non-conflicting transactions can be processed in parallel by multiple shards; and \second the system can scale up by adding new shards. 
However, this separation of transaction handling across shards is not perfectly `clean'---a transaction might rely on data managed by multiple shards, requiring an additional step of \emph{cross-shard consensus} across the concerned shards (see arrows on \Cref{fig:byzcuit:global-overview}).
An atomic commit protocol (such as the two-phase commit protocol~\cite{gray1978notes}) typically runs across all the concerned shards to ensure the transaction is accepted by all or none of those shards.

\para{Vulnerabilities in previous systems}
This chapter presents the first replay attacks on cross-shard consensus in sharded protocols. 
An \attacker can launch these attacks with minimal effort, without subverting any nodes, and assuming a weakly synchronous network (and in some cases, without relying on any network assumption)---even when the byzantine safety assumptions are satisfied.
These attacks compromise key system properties of safety and liveness, effectively enabling the \attacker to double-spend coins (or any other objects managed by the blockchain) and create coins out of thin air.
%Our attacks apply to the two main approaches to achieve cross-shard consensus~\cite{sok-consensus}: \first \shardled protocols that only involve the concerned shards, and require no external entity for coordination (\Cref{sec:byzcuit:shard-led-consensus}); and \second \clientled protocols that are coordinated by the client (\Cref{sec:byzcuit:client-led-consensus}).
%
We concretely sketch the replay attacks in the context of two representative systems: \sbac~\cite{chainspace} 
%as an example of \shardled protocols;
and \atomix~\cite{omniledger} .
%as an example of \clientled protocols. 
%Not only those systems were recently presented at top security conferences, but they form the basis of numerous start-ups and open-source projects such as chainspace.io\footnote{\url{https://chainspace.io}} and Harmony\footnote{\url{https://harmony.one}}. 
%For each of the two cross-shard consensus approaches, 
\Cref{sec:byzcuit:prerecoring} describes how an \attacker can actively stage the attack by eliciting from the system the messages to replay (in contrast to passively observing the network traffic, and waiting to detect and record the target messages).
We also discuss the feasibility of these attacks and their real-world impact.
The replay attacks we present are generic and apply to other systems that are based on similar models, like \rapidchain~\cite{rapidchain}. 
%\ethereum's cross-shard ``yanking'' proposal~\cite{yanking} also faces similar challenges; \Cref{sec:byzcuit:background} describes their cross-shard consensus protocol and compares their current proposal to mitigate cross-shard replay attacks with this work. 
%We note that account-based blockchains like \ethereum defend against transactions replay using account sequence numbers, in an entirely different context; \ie each account holds a monotonically increasing counter to prevent attackers from re-submitting old transactions to the network.
%On the other hand, this work focuses on attacks due to replaying messages in cross-shard atomic commit protocols. 
Based on our detailed analysis of replay attacks, we develop a defense strategy (\Cref{sec:byzcuit:defenses}).

%Drawing insights from our analysis of performance trade-offs and replay attack vulnerabilities in existing \shardled and \clientled cross-shard consensus protocols, we present a hybrid system, \byzcuit (\Cref{sec:byzcuit:byzcuit}). It combines useful features from both these design approaches to achieve high performance and scalability, and leverages our proposed defense to achieve resilience against replay attacks. \byzcuit employs a Transaction Manager to coordinate cross-shard communication, reducing its cost to $O(n)$ communication, amongst $n$ shards, in the absence of faults.

\para{\byzcuit}
Drawing insights from our analysis of performance trade-offs and replay attack vulnerabilities in existing cross-shard consensus protocols, we present a hybrid system, \byzcuit (\Cref{sec:byzcuit:byzcuit}). It combines useful features from previous designs to achieve high performance and scalability, and leverages our proposed defense to achieve resilience against replay attacks. \byzcuit employs a Transaction Manager to coordinate cross-shard communication, reducing its cost to $O(n)$ communication, amongst $n$ shards, in the absence of faults.
We implement a prototype of \byzcuit 
%in Java as a fork of the \chainspace code~\cite{chainspace}
, and release it as an open-source project\footnote{\url{https://github.com/sheharbano/byzcuit}}. We evaluate \byzcuit on a real cloud-based testbed under varying transaction loads and show that \byzcuit has a client-perceived latency of less than a second, even for a system load of 1,000 transactions per second (tps). \byzcuit's transaction throughput scales linearly with the number of shards by 250--300 tps for each shard added, handling up to 1,550 tps with 10 shards. We quantify the overhead of our replay defenses and find that as expected those reduce the throughput by 20--250 tps.

\subsection*{Contributions}
This chapter makes the following key contributions: 
\begin{itemize}
\item It develops the first replay attacks against cross-shard consensus protocols, and illustrate their impact on important academic and implemented designs
    
\item It presents defenses against those replay attacks, and discusses the issues that lead to these vulnerabilities. 

\item It designs a hybrid, new system \byzcuit with improved performance trade-offs, and which integrates our proposed defense to achieve resilience against the replay attacks; 

\item It implements a prototype of \byzcuit and evaluate its performance and scalability on a real distributed set of nodes and under varying transaction loads, and illustrate how it is superior to previous proposals.
\end{itemize}

% =========
\subsection*{Outline}
\Cref{sec:byzcuit:attack-overview} presents an overview of the replay attacks; \Cref{sec:byzcuit:shard-led-consensus} describes replay attacks on \shardled cross-shard consensus protocols; and \Cref{sec:byzcuit:client-led-consensus} describe replay attacks on \clientled cross-shard consensus protocols. \Cref{sec:byzcuit:defenses} discusses the issues that lead to the replay attacks describes how to fix them; \cref{sec:byzcuit:byzcuit} presents \byzcuit, a system achieving resilience against the replay attacks; \Cref{sec:byzcuit:implementation} presents an evaluation of \byzcuit; and \Cref{sec:byzcuit:conclusion} concludes the chapter.

%% file: chapters/byzcuit/sections/overview.tex
% =========
% Attack Overview
% =========
\section{Attack Overview} \label{sec:byzcuit:attack-overview}
Sections~\ref{sec:byzcuit:shard-led-consensus}~and~\ref{sec:byzcuit:client-led-consensus} discuss replay attacks on both \shardled and \clientled cross-shard consensus protocols, respectively. We provide a high-level description of these attacks and the threat model, and describe the notation we use.

\para{Replay attacks on cross-shard consensus}
The \attacker records a target shard's responses to the atomic commit protocol, and replays them during another instance of the protocol. We present two families of replay attacks: 
\first attacks against the first phase (\emph{voting}), and \second attacks against the second phase (\emph{commit}) of the atomic commit protocol.
To attack the first phase (\emph{voting}) of the atomic commit protocol, the \attacker replaces messages generated by the target shard by replaying \prerecorded messages. In practice, the \attacker does not \emph{replace} those messages---it achieves a similar result by making its replayed messages arrive at the coordinator faster (racing the target shard's original message), exploiting the fact that the coordinator makes progress based on the first message it receives. Replaying messages in this fashion enables the \attacker to compromise the system safety (by creating inconsistent state on the shards) and/or liveness (by causing valid transactions to be rejected). 
To attack the second phase (\emph{commit}) of the atomic commit protocol, the \attacker simply replays prerecorded messages to target shards, and compromises consistency. The \attacker can replay those messages at any time of its choice, and does not rely on any racing condition as in the previous case.

\para{Threat model} 
The \attacker can successfully launch the described attacks without colluding with any shard nodes, and under the BFT honest majority safety assumption for nodes within shards (\ie, the attacks are effective even if \emph{all} nodes are honest). We assume an \attacker that can observe and record messages generated by shards; this can be achieved by \first monitoring the network, or \second reading the blockchain (which is more practical).
The \attacker can be an external observer that passively collects the target messages at the level of the network, or it can act as a client and actively interact with the system to elicit the target messages.
The attacks against the first phase of the atomic commit protocol (Sections~\ref{sec:byzcuit:chainspace-attack-1} and \ref{sec:byzcuit:omniledger-attack-1}) assume a weakly synchronous network  in which an \attacker may delay messages and race target shards by replaying pre-recorded messages. 
The attacks against the second phase of the atomic commit protocol (Section~\ref{sec:byzcuit:chainspace-attack-2}~and~\ref{sec:byzcuit:omniledger-attack-2}) do not make any such assumptions on the underlying network.  

\para{Attacks Implementation}
We implemented a demo of the replay attacks against \sbac (attacks against \atomix can be similarly implemented) in Java. 
%We are open-sourcing our demo\footnote{\url{https://github.com/sheharbano/byzcuit/tree/replay-attacks}}, and a document describing a step-by-step tutorial to execute it\footnote{\url{https://github.com/sheharbano/byzcuit/tree/master/docs}}. 
The demo shows, in the context of a simple payment application that supports account creation and coin transfer, how the replay attacks described in this chapter can be used to create coins out of thin air. 
We show that the attacks do not rely on any strict timing assumptions---the same entity could control the accounts of both payer and payee, as well as the client, and generate coins out of thin air.

\para{Notation}
Operations on the blockchain are specified as \emph{transactions}. A transaction defines some transformation on the blockchain state, and has input and output \emph{objects} (such as UTXO entries). An object is some data managed by the blockchain, such as a bank account, a specific coin, or a hotel room. 
For example, $T(x_1,x_2)\rightarrow (y_1,y_2,y_3)$ represents a transaction with two inputs, $x_1$ managed by shard 1 and $x_2$ managed by shard 2; and three outputs, $y_1$ managed by shard 1, $y_2$ managed by shard 2, and $y_3$ managed by shard 3.
We call the shards that manage the input objects \emph{input shards}, and the shards that manage the output objects \emph{output shards}. It is possible for a shard to be both the input and output shard.
Objects can be in two states: \emph{active} (on unspent) objects are available for being processed by a transaction; and \emph{inactive} (or spent) objects cannot be processed by any transaction.
Additionally, some systems also associate \emph{locked} state with objects that are currently being processed by a transaction to protect against manipulation by other concurrent transactions involving those objects. 
The attacks we describe in this chapter generalize to transactions with $k$ inputs and $k'$ outputs managed by an arbitrary number of shards.

%% file: chapters/byzcuit/sections/chainspace.tex
% =========
% Shard-Led
% =========
\section{\Shardled Cross-Shard Consensus Protocol} \label{sec:byzcuit:shard-led-consensus}
In \shardled cross-shard consensus protocols, the shards collectively take on the role of the coordinator in the atomic commit protocol. We describe replay attacks on \shardled cross-shard consensus protocols.  To make the discussion concrete, we illustrate these attacks in the context of \sbac~\cite{chainspace}, though we note that these attacks can be generalized to other similar systems. We discuss how the \attacker can record shard messages to replay in future attacks (\Cref{sec:byzcuit:chainspace-message-recording}). In Sections~\ref{sec:byzcuit:chainspace-attack-1}~and~\ref{sec:byzcuit:chainspace-attack-2}, we describe replay attacks on the first and second phase of the cross-shard consensus protocol, and discuss the real-world impact of these attacks (\Cref{sec:byzcuit:chainspace-discussion}).  

% =========
\subsection{\sbac Overview} \label{sec:byzcuit:sbac-summary}
\Cref{chainspace} presents \chainspace by treating the cross-shard consensus protocol as a black box.
A preliminary version of \chainspace~\cite{chainspace} relied on a \shardled cross-shard consensus protocol called \sbac, which was insecure; since then, \chainspace has been updated to use the protocol described in \Cref{sec:byzcuit:byzcuit}.

In \sbac, the client submits a transaction to the input shards.
Each shard internally runs a BFT protocol to tentatively decide whether to accept or abort the transaction locally, and broadcasts its local decision (\preacceptt or \preabortt) to other relevant shards.
\Cref{fig:byzcuit:sbac-original-fsm} shows the state machine representing the life cycle of objects. A shard generates \preabortt if the transaction fails local checks (\eg, if any of the input objects are \inactiveObj or \locked).
If a shard generates \preacceptt, it changes the state of the input objects to \locked.
This is the first step of \sbac, and is equivalent to the voting phase in the two-phase atomic commit protocol (\Cref{sec:literature-review:xshard}).

Each shard collects responses from other relevant shards, and commits the transaction if all shards respond with \preacceptt, or aborts the transaction otherwise.
This is the second step of \sbac, and is equivalent to the commit phase in the two-phase atomic commit protocol (\Cref{sec:literature-review:xshard}).
The shards communicate this decision to the client as well as the output shards by sending them the \acceptt or \abortt messages.
If the shard's decision is \acceptt, it changes the input object state to \inactiveObj.
If the shard's decision is \abortt, it changes the input object state to \activeObj (effectively unlocking it).
Upon receiving \acceptt, the client concludes that the transaction was committed, and the output shards create the output objects (with the state \activeObj) of the transaction.

\begin{figure}[t]
\centering
\includegraphics[width=\textwidth]{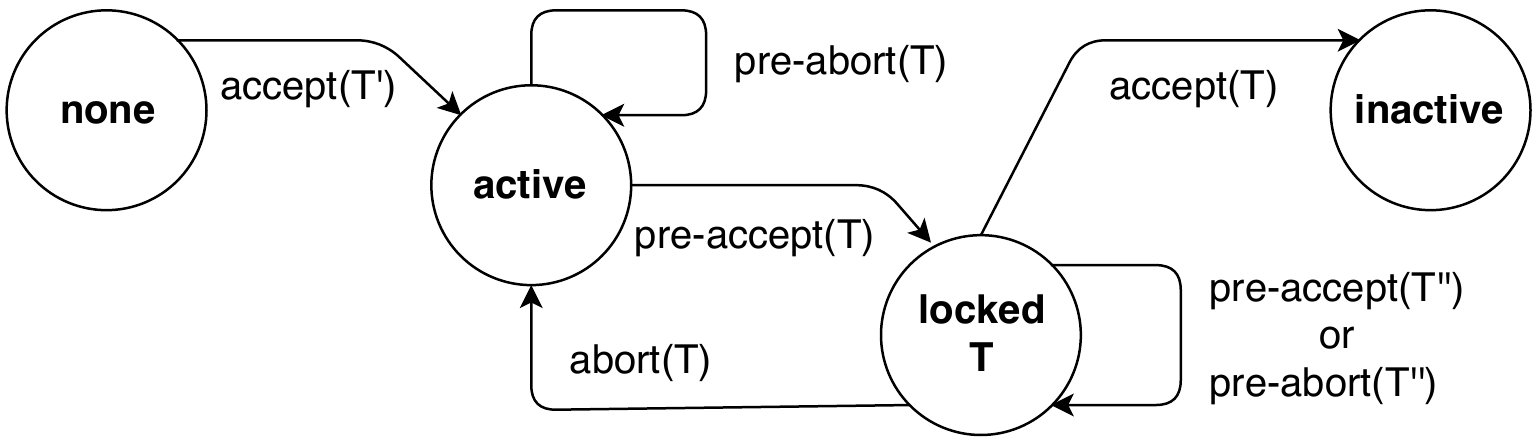}
\caption[State machine representing the life cycle of objects handled by \sbac.]{\footnotesize State machine representing the life cycle of objects handled by \sbac. An object becomes \activeObj as a result of a previous successful transaction. The object state changes to \locked if a shard locally emits \preacceptt in the first phase of the cross-shard consensus protocol for a transaction $T$. A locked object cannot be processed by other transactions $T{''}$. If the second phase of the protocol results in \acceptt, the object becomes \inactiveObj; alternatively, if the result is \abortt the object becomes \activeObj again and is available for being processed by other transactions.}
\label{fig:byzcuit:sbac-original-fsm}
\end{figure}
\begin{figure}[t]
\centering
\includegraphics[width=\textwidth]{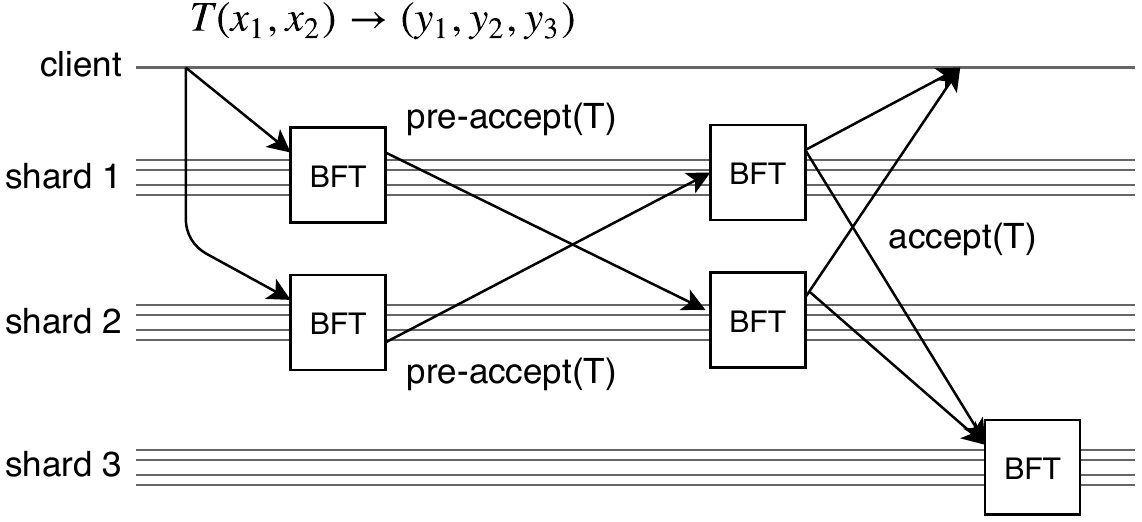}
\caption[An example execution of \sbac]{\footnotesize An example execution of \sbac for a valid transaction $T(x_1,x_2)\rightarrow (y_1,y_2,y_3)$ with two inputs ($x_1$ and $x_2$, both are active) and three outputs ($y_1,y_2,y_3$), where the final decision is \acceptt. All cross-shard arrows represent a multicast of all nodes in one shard to all nodes in another.}
\label{fig:chainspace:sbac}
\end{figure}
%

%\Cref{fig:chainspace:sbac} illustrates a simple example of the \sbac protocol to commit a single transaction with two inputs and three outputs that we may use as an example. 
\Cref{fig:chainspace:sbac} shows an example execution of \sbac for a valid transaction $T(x_1,x_2)\rightarrow (y_1,y_2,y_3)$ with two inputs ($x_1$ and $x_2$, both are active) and three outputs ($y_1,y_2,y_3$), where the final decision is \acceptt. 
The client submits $T$ to shard 1 and shard 2. 
Upon receiving $T$, both shard 1 and shard 2 confirm that the transaction is to commit, and emit \preacceptt at the end of the first phase of \sbac. 
Each shard receives \preacceptt from the other shard, and emits \acceptt at the end of the second phase of \sbac. 
As a result, the input objects $x_1$ and $x_2$ become inactive, and the output shards respectively create objects $y_1$, $y_2$, and $y_3$.

% =========
\subsection{Message Recording} \label{sec:byzcuit:chainspace-message-recording}
Prior to the replay attacks, the \attacker records responses generated by shards. The \attacker can record shard responses in the first phase of \sbac (\ie, \preacceptt or \preabortt), enabling the family of attacks described in \Cref{sec:byzcuit:chainspace-attack-1}. The \attacker can also record shard responses in the second phase of \sbac (\ie, \acceptt or \abortt), enabling the family of attacks described in \Cref{sec:byzcuit:chainspace-attack-2}. 
In the general case, the attacker passively collects the messages either by sniffing the network on protocol executions, or by downloading the blockchain and selecting the messages to replay\footnote{Since those messages need to be recorded on chain for verification, just using transport layer encryption between nodes is not effective.}. \Cref{sec:byzcuit:prerecoring-chainspace} shows how the \attacker can act as client to actively elicit the messages necessary for the attacks, to record and later replay---this empowers the \attacker to actively orchestrate the attacks.

% =========
\subsection{Attacks on the First Phase of \sbac} \label{sec:byzcuit:chainspace-attack-1}

\input{chapters/byzcuit/tables/sbac-1}

\begin{figure}[t]
\centering
\includegraphics[width=\textwidth]{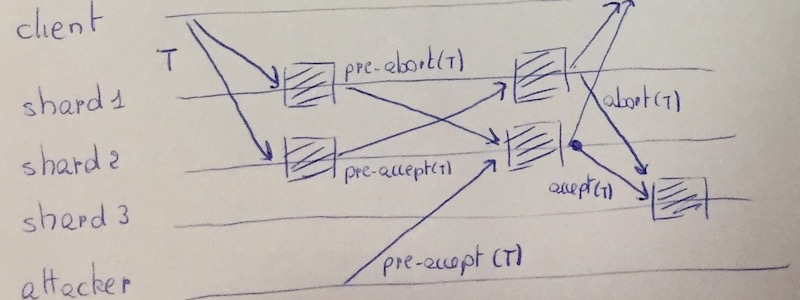}
\caption[An example of replay attack against \sbac.]{\footnotesize 
Illustration of the replay attack depicted in row \textsf{6} of \Cref{tab:byzcuit:sbac-attack}. The \attacker replays to shard 2 a \prerecorded \preacceptt message (shown as a bold line) from shard 1, which precludes shard 1's \preabortt message (shown as a dotted line). 
}
\label{fig:byzcuit:sbac-attack-example-2}
\end{figure}

We present replay attacks on the first phase of \sbac by taking the example of a transaction $T(x_1,x_2)\rightarrow (y_1,y_2,y_3)$ as described in \Cref{sec:byzcuit:attack-overview}. These attacks easily generalize to transactions with $k$ inputs and $k'$ outputs managed by an arbitrary number of shards. The replay attacks work in two steps; \first the \attacker records \preacceptt or \preabortt messages (as described in \Cref{sec:byzcuit:chainspace-message-recording} and \Cref{sec:byzcuit:prerecoring-chainspace}); and \second then replays those messages.
\Cref{tab:byzcuit:sbac-attack} shows the replay attacks that the \attacker can launch, for all possible combinations of messages emitted by shard 1 and shard 2 in the first phase of \sbac. The caption includes details about how to interpret this table. All attacks exploit the parallel composition of multiple \sbac instances, and insufficient binding of messages to its \sbac instance.
We describe \myrow{6} of \Cref{tab:byzcuit:sbac-attack}, to help readers interpret rest of the table on their own.
In the correct execution (\myrow{5}), shard 1 and shard 2 emit \preabortt (because $x_1$ is not active) and \preacceptt in the first phase, respectively. In the second phase, both shards emit \abortt and the protocol terminates.   
\Cref{fig:byzcuit:sbac-attack-example-2} illustrates the replay attack corresponding to row \textsf{6} of \Cref{tab:byzcuit:sbac-attack}.
The attacker races shard 1 by sending to shard 2 the \prerecorded \preacceptt message from shard 1. As a result, shard 2 emits \acceptt, inactivates object $x_2$ and creates object $y_2$. This leads to inconsistent state across the shards. In a correct execution: \first if $T$ is accepted all its inputs ($x_1$ and $x_2$) should become inactive, and all the outputs ($y_1$, $y_2$, $y_3$) should be created; and \second if $T$ is aborted, all its inputs ($x_1$ and $x_2$) should become active again, and none of the outputs ($y_1$, $y_2$, $y_3$) should be created. However, here we have an incorrect termination of \sbac: at the end of the protocol $x_1$ could be active and $x_2$ is inactive; $y_1$ is not created, $y_2$ and $y_3$ are created.

\Cref{tab:byzcuit:sbac-attack} shows that through careful selection of the messages to replay from different \sbac instances, the attacks can be effective against any shard.
All the attacks (except \myrow{4}) compromise consistency; the \attacker can trick the input shards to inactivate arbitrary objects, and trick the output shards into creating new objects in violation of the protocol. 
The attack depicted in \myrow{4} only affects availability.

% =========
\subsection{Attacks on the Second Phase of \sbac} \label{sec:byzcuit:chainspace-attack-2}

\input{chapters/byzcuit/tables/sbac-2.tex}

We present replay attacks on the second phase of \sbac. The \attacker \prerecords \acceptt messages as described in \Cref{sec:byzcuit:chainspace-message-recording} and \Cref{sec:byzcuit:prerecoring-chainspace}. 
\Cref{tab:byzcuit:sbac-attack:2} shows replay attacks for all possible combinations of messages emitted by shard 1 and shard 2 in the second phase. Since the attacks we describe in this section assume that the first phase of \sbac concluded correctly (\ie, all the relevant shards unanimously decide to accept or reject a transaction), both the shards generate \abortt (\myrow{1}) or \acceptt (\myrow{5}).
The caption includes details about how to interpret this table.
We describe \myrow{6} of \Cref{tab:byzcuit:sbac-attack:2}, to help readers interpret rest of the table on their own.
In the correct execution (\myrow{5}), both shards emit \abortt and no output objects are created.
In the attack in \myrow{6}, the \attacker replays a \prerecorded \acceptt from shard 1 to all the relevant shards (in this case shard 3).
Upon receiving this message, shard 3 (incorrectly) creates $y_3$.

The potential victims of replay attacks corresponding to the second phase of \sbac are the shards that \emph{only} act as output shards (\ie, do not simultaneously act as input shards). 
The \attacker can replay \acceptt multiple times tricking shard 3 into creating $y_3$ multiple times. 
These attacks are possible because shards do not keep records of inactive objects (following the UTXO model) for scalability reasons\footnote{Requiring shards to remember the full history of inactive objects would increase their memory requirements monotonically over time, reaching at some point memory limits preventing further operations. Thus this is a poor mitigation for the attacks presented.}, and because shard 3 takes part in only the second phase of \sbac.
The \attacker can double-spend $y_3$ repeatedly by replaying a single \prerecorded message multiple times, and spending the object (and effectively purging it from shard 3's UTXO) before each replay.
Contrarily to the attacks against the first phase of \sbac (\Cref{sec:byzcuit:chainspace-attack-1}), these attacks do not rely on any racing conditions;  there is no need to race any honest messages.

% =========
\subsection{Real-world Impact} \label{sec:byzcuit:chainspace-discussion}
The real-world impact and attacker incentives to conduct these attacks depends on the nature and implementation of the smart contract handling the target objects. We discuss the impact of these attacks in the context of two common smart contract applications, which are also described in the \chainspace paper~\cite{chainspace}. To take a concrete example, we illustrate the attack depicted in \myrow{3} of \Cref{tab:byzcuit:sbac-attack}, but similar results can be obtained with the other attacks described in \Cref{tab:byzcuit:sbac-attack} and \Cref{tab:byzcuit:sbac-attack:2}.

One of the most common blockchain application is to manage cryptocurrency (or coins) and enable payments for processing transactions, implemented by the \textsf{CSCoin} smart contract in \chainspace.
Lets suppose object $x_1$ (handled by shard 1) represents Alice's account, and object $x_2$ (handled by shard 2) represents Bob's account. To transfer $v$ coins to Bob, Alice submits a transaction $T(x_1,x_2)\rightarrow(y_1,y_2)$, where $y_1$ and $y_2$ respectively represent the new account objects of Alice and Bob, with updated account balances. By executing the attack described in \myrow{3} of \Cref{tab:byzcuit:sbac-attack}, an \attacker can trick shard 1 to abort the transaction and unlock $x_1$ (thus reestablishing Alice's account balance as it was prior to the coin transfer), and shard 2 to accept the transaction and create $y_2$ (thus adding $v$ coins to Bob's account). This attack effectively allows any attacker to double-spend coins on the ledger; and shows how to create $v$ coins out of thin air.
Another common blockchain use case is a platform for decision making (or electronic petitions), implemented by the \textsf{SVote} smart contract in \chainspace. Upon initialization, the \textsf{SVote} contract creates two objects: \first $x_1$ representing the tally's public key, a list of all voters' public keys, and the tally's signature on these; and \second $x_2$ representing a vote object at the initial stage of the election (all candidates having a score of zero) along with a zero-knowledge proof asserting the correctness of the initial stage. To vote, clients submit a transaction $T(x_1,x_2)\rightarrow(y_1,y_2)$, where $y_1$ and $y_2$ are respectively the updated voting list (\ie, the voting list without the client's public key), and the election stage updated with the client's vote. By executing the attack described by \myrow{3} of \Cref{tab:byzcuit:sbac-attack}, an \attacker can trick shard 1 to abort the transaction and thus not update the voting list, and shard 2 to accept the transaction and thus update the election stage. This effectively allows any client to vote multiple times during an election while remaining undetected (due to the privacy-preserving properties of the \textsf{SVote} smart contract).

%% file: chapters/byzcuit/tables/sbac-1.tex
\begin{table*}[t]
\centering
%\footnotesize
\resizebox{\textwidth}{!}{\begin{tabular}{c c C  c c c}
\toprule
\multicolumn{3}{c}{\small Phase 1 of \sbac} & \multicolumn{3}{c}{\small Phase 2 of \sbac}\\
\noalign{\smallskip}
&\Centerstack{\textbf{Shard 1} \\ (potential victim)} & \Centerstack{\textbf{Shard 2} \\ (potential victim)} & \Centerstack{\textbf{Shard 1} \\ (potential victim)} & \Centerstack{\textbf{Shard 2} \\ (potential victim)} & \Centerstack{\textbf{Shard 3} \\ (potential victim)}\\
\noalign{\smallskip}
\toprule

\rowcolor{verylightgray}
1 & \Centerstack{\preacceptt \\ lock $x_1$} & \Centerstack{\preacceptt \\ lock $x_2$} &
\Centerstack{\acceptt \\ create $y_1$; inactivate $x_1$} & \Centerstack{\acceptt \\ create $y_2$; inactivate $x_2$} & \Centerstack{- \\ create $y_3$} \\
\midrule
2 & $\rhd$\preabortt & & \Centerstack{\acceptt \\ create $y_1$; inactivate $x_1$} & \Centerstack{\abortt \\ unlock $x_2$} & \Centerstack{- \\ create $y_3$} \\
\midrule
3 & & $\rhd$\preabortt & \Centerstack{\abortt \\ unlock $x_1$} & \Centerstack{\acceptt \\ create $y_2$; inactivate $x_2$} & \Centerstack{- \\ create $y_3$} \\
\midrule
4 & $\rhd$\preabortt & $\rhd$\preabortt & \Centerstack{\abortt \\ unlock $x_1$} & \Centerstack{\abortt \\ unlock $x_2$} & - \\
\noalign{\smallskip}
\toprule

\rowcolor{verylightgray}
5 & \Centerstack{\preabortt \\ -} & \Centerstack{\preacceptt \\ lock $x_2$} & \Centerstack{\abortt \\ -} & \Centerstack{\abortt \\ unlock $x_2$} & - \\
\midrule
6 & $\rhd$\preacceptt & & \Centerstack{\abortt \\ -} & \Centerstack{\acceptt \\ create $y_2$; inactivate $x_2$} &  \Centerstack{- \\ create $y_3$} \\
\noalign{\smallskip}
\toprule

\rowcolor{verylightgray}
7 & \Centerstack{\preacceptt \\ lock $x_1$} & \Centerstack{\preabortt \\ -} & \Centerstack{\abortt \\ unlock $x_1$} & \Centerstack{\abortt \\ -} & - \\
\midrule
8 & & $\rhd$ \preacceptt& \Centerstack{\acceptt \\ create $y_1$; inactivate $x_1$} & \Centerstack{\abortt \\ -} & \Centerstack{- \\ create $y_3$} \\
\noalign{\smallskip}
\toprule

\rowcolor{verylightgray}
9 & \Centerstack{\preabortt \\ -} & \Centerstack{\preabortt \\ -} & \Centerstack{\abortt \\ -} & \Centerstack{\abortt \\ -} & - \\
\bottomrule
\end{tabular}}
\caption[List of replay attacks against the first phase of \sbac.]{\footnotesize List of replay attacks against the first phase of \sbac for all possible executions of the transaction $T(x_1,x_2)\rightarrow (y_1,y_2,y_3)$ as described in \Cref{sec:byzcuit:attack-overview}. 
The highlighted rows indicate correct executions of \sbac (\ie, without the \attacker), and the other rows indicate incorrect executions due to the replay attacks.
In multirows, the top sub-rows show the protocol messages emitted by shards, and the bottom sub-rows indicate local shard actions as a result of emitting those messages.
For example, (\mycolumn{3}, \myrow{2}) means that shard 1 emits \acceptt (top sub-row), and creates a new object $y_1$ and inactivates $x_1$ (bottom sub-row).  
The first two columns indicate the messages emitted by each shard at the end of the first phase of \sbac.
The \attacker races shards at the end of the first phase of \sbac by replaying \prerecorded messages, marked with the symbol $\rhd$ in the first two columns of \Cref{tab:byzcuit:sbac-attack}.
%Replayed messages are marked with the symbol $\rhd$, 
For example $\rhd$\preabortt at (\mycolumn{1}, \myrow{2}) means that the \attacker sends to other relevant shards (in this case shard 2) a \prerecorded \preabortt\xspace message impersonating shard 1 that races the original \preacceptt (\mycolumn{1}, \myrow{1}) emitted by shard 1. 
The last three columns indicate the messages emitted at the end of the second phase of \sbac.}
\label{tab:byzcuit:sbac-attack}
\end{table*}

%% file: chapters/byzcuit/tables/sbac-2.tex
\begin{table*}[t]
\centering
%\footnotesize
\resizebox{\textwidth}{!}{\begin{tabular}{c C C C}
\toprule
\multicolumn{4}{c}{\small Phase 2 of \sbac}\\
& \textbf{Shard 1} & \textbf{Shard 2} & \shortstack{\textbf{Shard 3} \\ (potential victim)}\\
\midrule

\rowcolor{verylightgray}
1 & \Centerstack{\acceptt \\ create $y_1$; inactivate $x_1$} & \Centerstack{\acceptt \\ create $y_2$; inactivate $x_2$} & \Centerstack{- \\ create $y_3$} \\
2 & $\rhd$\acceptt & & create $y_3$ \\
3 & & $\rhd$\acceptt & create $y_3$ \\
4 & $\rhd$\acceptt & $\rhd$\acceptt & create $y_3$ \\
\midrule
\rowcolor{verylightgray}
5 & \Centerstack{\abortt \\ (unlock $x_1$)} & \Centerstack{\abortt \\ (unlock $x_2$)} & \Centerstack{- \\ -} \\
6 & $\rhd$\acceptt & & create $y_3$ \\
7 & & $\rhd$\acceptt & create $y_3$ \\
8 & $\rhd$\acceptt & $\rhd$\acceptt & create $y_3$ \\

\bottomrule
\end{tabular}}
\caption[List of replay attacks against the second phase of \sbac.]{\footnotesize  List of replay attacks against the second phase of \sbac for all possible executions of the transaction $T(x_1,x_2)\rightarrow (y_1,y_2,y_3)$ as described in \Cref{sec:byzcuit:attack-overview}. 
The highlighted rows indicate correct executions of \sbac (\ie, without the \attacker), and the other rows indicate incorrect executions due to the replay attacks.
In multirows, the top sub-rows show the protocol messages emitted by shards, and the bottom sub-rows indicate local shard actions as a result of emitting those messages.
 For example, (\mycolumn{1}, \myrow{1}) means that shard 1 emits \acceptt (top sub-row), and creates a new object $y_1$ and inactivates $x_1$ (bottom sub-row).
The first two columns indicate the messages emitted by each shard at the end of the second phase of \sbac, and the last column shows the effect of these messages on the output shard 3.
Replayed messages are marked with the symbol $\rhd$.
For example $\rhd$\acceptt at (\mycolumn{1}, \myrow{2}) means that the \attacker sends to other relevant shards (in this case shard 3) a \prerecorded \acceptt message impersonating shard 1.} 
\label{tab:byzcuit:sbac-attack:2}
\end{table*}

%% file: chapters/byzcuit/sections/omniledger.tex
% =========
% Client Led
% =========
\section{\Clientled Cross-shard Consensus Protocol} \label{sec:byzcuit:client-led-consensus}
We describe replay attacks on \clientled cross-shard consensus protocols. To make the discussion concrete, we illustrate these attacks in the context of \atomix, the cross-shard protocol at the heart of \omniledger~\cite{omniledger}. However, we note that these attacks can be generalized to other similar systems. 
We discuss how the \attacker can record shard messages to replay in future attacks (\Cref{sec:byzcuit:omniledger-message-recording}). In Sections~\ref{sec:byzcuit:omniledger-attack-1}~and~\ref{sec:byzcuit:omniledger-attack-2}, we describe replay attacks on the first and second phase of the cross-shard consensus protocol. Finally, we discuss the real-world impact of these attacks (\Cref{sec:byzcuit:omniledger-discussion}). 

% =========
\subsection{\atomix Overview} \label{sec:byzcuit:atomix-summary}

\begin{figure}[t]
\centering
\includegraphics[width=\textwidth]{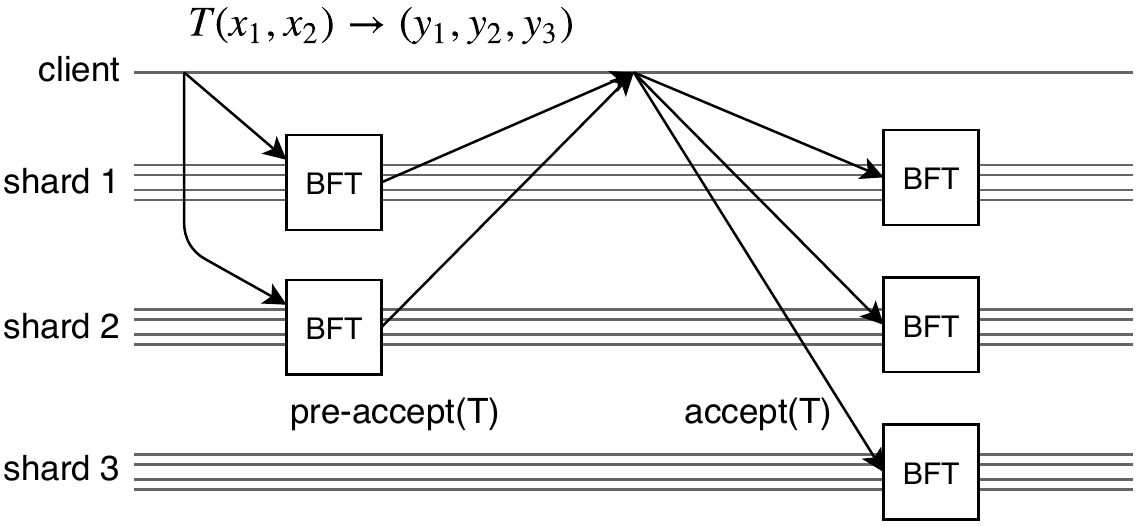}
\caption[An example execution of \atomix.]{\footnotesize 
An example execution of \atomix for a valid transaction $T(x_1,x_2)\rightarrow (y_1,y_2,y_3)$ with two inputs ($x_1$ and $x_2$, both are active) and three outputs $(y_1,y_2,y_3)$, where the final decision is \acceptt.}
\label{fig:atomix}
\end{figure}
Similar to \sbac, \atomix uses an atomic commit protocol to process transactions across shards. However, it uses a different, client-driven approach to achieve it.
%\atomix uses a \clientled cross-shard consensus protocol called \atomix.
The client submits the transaction $T$ to the input shards. Each shard runs a BFT protocol locally to decide whether to accept or reject the transaction, and communicates its response (\preacceptt or \preabortt) to the client.\footnote{For clarity and consistency, we use the terminology used in \Cref{chainspace}. In \atomix, \preacceptt is actually a \emph{proof-of-accept} and \preabortt is a \emph{proof-of-abort}~\cite{omniledger}.}  
A shard emits \preabortt if the transaction fails local checks (\eg, if any of the input objects is inactive).
Alternatively, if a shard emits \preacceptt, it inactivates the input objects it manages. This is the first phase of \atomix, and is similar to the voting phase in the two-phase atomic commit protocol, but differs in that the protocol proceeds optimistically.
The write changes made by the input shards in the first phase of \atomix are considered permanent (\ie, there is no \locked object state), unless the client requests the input shards to revert their changes in the second phase.  
After the client has collected \preacceptt from all input shards, it submits \acceptt message (containing proof of the \preacceptt messages) to the output shards which create the output objects.
Alternatively, if any of the input shards emits \preabortt, the client sends \abortt (containing proof of \preabortt) to the relevant input shards which make the input objects active again. This is the second phase of \atomix, and is similar to the commit phase in the two-phase atomic commit protocol.
\Cref{fig:atomix} shows the execution of \atomix for a valid transaction $T(x_1,x_2)\rightarrow (y_1,y_2,y_3)$, with two active inputs ($x_1$ managed by shard 1, and $x_2$ managed by shard 2) and producing three outputs $(y_1,y_2,y_3)$ managed by shard 1, shard 2 and shard 3, respectively. The client sends $T$ to the input shards, both of which reply with \preacceptt and make the input objects $x_1$ and $x_2$ inactive. The client sends \acceptt to the output shards which respectively create objects $y_1$, $y_2$, and $y_3$.

% =========
\subsection{Message Recording} \label{sec:byzcuit:omniledger-message-recording}
Before launching the replay attacks, the \attacker first records the target shard responses. The \attacker can record shard responses in the first phase of \atomix (\ie, \preacceptt or \preabortt), enabling the attacks described in \Cref{sec:byzcuit:omniledger-attack-1}. The \attacker can also record shard responses in the second phase of \atomix (\ie, \acceptt or \abortt), enabling the attacks described in \Cref{sec:byzcuit:omniledger-attack-2}. 
In the general case, the \attacker passively collects the messages to replay, for example by protocol executions on the network, or by downloading the blockchain and selecting the appropriate messages. \Cref{sec:byzcuit:prerecoring-omniledger} shows how the \attacker can act as a client to actively elicit and record the target messages to later use in the replay attacks.

% =========
\subsection{Attacks on the First Phase of \atomix} \label{sec:byzcuit:omniledger-attack-1}

\input{chapters/byzcuit/tables/atomix-1.tex}

We present replay attacks on the first phase of \atomix by taking the example of a transaction $T(x_1,x_2)\rightarrow (y_1,y_2,y_3)$ as described in \Cref{sec:byzcuit:attack-overview}. These attacks easily generalize to transactions with $k$ inputs and $k'$ outputs managed by an arbitrary number of shards. The replay attacks work in two steps: \first the \attacker observes the traffic and records \preacceptt or \preabortt messages as described in \Cref{sec:byzcuit:omniledger-message-recording}; and \second then replay those messages.
\Cref{tab:byzcuit:atomix-attack} shows the replay attacks that the \attacker can launch, for all possible combinations of responses generated by shard 1 and shard 2 in the first phase of \atomix. 
The caption includes details about how to interpret this table.
We describe \myrow{6} of \Cref{tab:byzcuit:atomix-attack}, to help readers interpret rest of the table on their own.
In the correct execution (\myrow{5}), shard 1 emits \preabortt, and shard 2 emits \preacceptt and inactivates the input objects $x_2$. 
Upon receiving these messages, the client sends \abortt to the output shards shard 1, shard 2 and shard 3, and shard 2 re-activates $x_2$; and the protocol terminates. 
In the attack illustrated in \myrow{6} of \Cref{tab:byzcuit:atomix-attack},
the \attacker races shard 1 by sending to the client the \prerecorded \preacceptt message from shard 1. As a result, the client sends \acceptt message to the output shards shard 1, shard 2 and shard 3, which respectively create the output objects $y_1$, $y_2$, and $y_3$.
As a result, the system ends up in an inconsistent state because the output objects ($y_1$, $y_2$, $y_3$) have been created, while the input object ($x_1$) was not active---this results in a double-spend of the input object $x_1$.
\Cref{tab:byzcuit:atomix-attack} shows that through careful selection of the messages to replay, the attacks can be effective against any shard. The attacks illustrated in \myrow{2}, \myrow{3}, and \myrow{4} only affect availability, while the other attacks compromise consistency (\ie, the \attacker can trick the input shards to reactivate arbitrary objects, and trick the output shards into creating new objects in violation of the protocol). The potential victims of these attacks include the client (\eg, when the \attacker replays the shard messages to it in the first phase of \atomix) and any input or output shards.

% =========
\subsection{Attacks on the Second Phase of \atomix} \label{sec:byzcuit:omniledger-attack-2}

\input{chapters/byzcuit/tables/atomix-2.tex}

We present replay attacks on the second phase of \atomix. The \attacker \prerecords \acceptt and \abortt messages as described in \Cref{sec:byzcuit:omniledger-message-recording} and \Cref{sec:byzcuit:prerecoring-omniledger}. 
\Cref{tab:byzcuit:atomix-attack:2} shows replay attacks corresponding to the messages emitted by the client in the second phase---\ie, \acceptt in \myrow{1}, or \abortt in \myrow{3}.
The caption includes details about how to interpret this table.
%The highlighted rows indicate successful executions of \atomix, and the other rows indicate incorrect executions caused by the replay attacks.
%The attacker replays \prerecorded messages marked with the symbol $\rhd$ at the end of the first phase of \atomix.
The \abortt message at (\mycolumn{1}, \myrow{2}) means that the \attacker sends a \prerecorded \abortt message to the input shards (shard 1 and shard 2) impersonating the client. Upon receiving this message, shard 1 and shard 2 (incorrectly) re-activate $x_1$ and $x_2$, respectively. 
Furthermore, all output shards create the output objects when the correct \acceptt message emitted by the client (\myrow{1}, \mycolumn{1}) reaches them.
This results in inconsistent state, because the output objects have been created, but the input objects have not been consumed and have been reactivated by the \abortt message replayed by the adversary. 
The potential victims of \abortt replay attack are the input shards.
Similarly, \acceptt at (\myrow{4}, \mycolumn{1}) means that the \attacker sends a \prerecorded \acceptt message to the output shards (shard 1, shard 2 and shard 3) impersonating the client. Upon receiving this message, the output shards (incorrectly) create $y_1$, $y_2$ and $y_3$. 
Furthermore, the input shards (shard 1 and shard 2) reactivate $x_1$ and $x_2$ upon receiving the correct \abortt message emitted by the client (\myrow{3}, \mycolumn{1}).
This creates inconsistent state: the input objects have not been consumed and have been reactivated by the \abortt message emitted by the client, but the output objects have been created due to the \acceptt message replayed by the \attacker. 
The potential victims of \acceptt replay attack are the output shards.

These attacks are possible because output shards create objects directly upon receiving \acceptt; they do not check if the objects have been previously invalidated because shards do not keep records of inactive objects (per the UTXO model) for scalability reasons.\footnote{Verifying that objects have not been previously invalided implies either keep a forever-growing list of invalidated objects, or download and check the shard's entire blockchain.}
The \attacker can double-spend the output objects repeatedly from a single \prerecorded message by replaying it multiple times, and spending the object (and effectively purging it from the output shards' UTXO) before each replay. 
Similar to the attacks against the second phase of \sbac (\Cref{sec:byzcuit:chainspace-attack-2}), these attacks do not exploit any racing condition and can be mounted by an adversary at a leisurely pace.

% =========
\subsection{Real-world Impact} \label{sec:byzcuit:omniledger-discussion}
Contrarily to \chainspace, \omniledger does not support smart contracts and only handles a cryptocurrency. The attacks described in Sections~\ref{sec:byzcuit:omniledger-attack-1}~and~\ref{sec:byzcuit:omniledger-attack-2} allow an \attacker to: \first double-spend the coins of any user, by reactivating spent coins (\eg, the \attacker may execute the attack depicted by \myrow{2} of \Cref{tab:byzcuit:atomix-attack:2} to re-activate the objects $x_1$ and $x_2$ after the transfer is complete); and \second create coins out of thin air by replaying the message to create coins (\eg, an \attacker may execute the attack depicted by \myrow{4} of \Cref{tab:byzcuit:atomix-attack:2} to create multiple times object $y_3$, by purging it from the UTXO list of shard 3 prior to each instance of the attack). 
If the \attacker colludes with the client, it can trigger the \prerecorded messages needed for the attacks as described in \Cref{sec:byzcuit:omniledger-message-recording}.
Alternatively, the \attacker can passively observe the network and collect the target messages to replay.
Similar results can be obtained using the attacks described in \Cref{tab:byzcuit:atomix-attack}.
Note that since transaction are recorded on the blockchain, these attacks can be detected retrospectively.  
This can lead to the \attacker being exposed, or the \attacker can inculpate innocent users (the \attacker can replay messages of any user).

%% file: chapters/byzcuit/tables/atomix-1.tex
\begin{table*}[t]
\centering
%\footnotesize
\resizebox{\textwidth}{!}{\begin{tabular}{c c c  c  c c c}
\toprule
\multicolumn{3}{c}{\small Phase 1 of \atomix} & & \multicolumn{3}{c}{\small Phase 2 of \atomix}\\
\noalign{\smallskip}
&\Centerstack{\textbf{Shard 1} \\ (potential victim)} & \Centerstack{\textbf{Shard 2} \\ (potential victim)} &
\Centerstack{\textbf{Client} \\ (victim)} &
\Centerstack{\textbf{Shard 1} \\ (potential victim)} & \Centerstack{\textbf{Shard 2} \\ (potential victim)} & \Centerstack{\textbf{Shard 3} \\ (potential victim)}\\
\noalign{\smallskip}
\toprule

\rowcolor{verylightgray}
1 & \Centerstack{\preacceptt \\ inactivate $x_1$} & \Centerstack{\preacceptt \\ inactivate $x_2$} & \acceptt & \Centerstack{- \\ create $y_1$} & \Centerstack{- \\ create $y_2$} & \Centerstack{- \\ create $y_3$} \\
\midrule
2 & $\rhd$ \preabortt & & \abortt & \Centerstack{- \\ re-activate $x_1$} & \Centerstack{- \\ re-activate $x_2$} & - \\
\midrule
3 & & $\rhd$ \preabortt & \abortt & \Centerstack{- \\ re-activate $x_1$} & \Centerstack{- \\ re-activate $x_2$} & - \\
\midrule
4 & $\rhd$\preabortt & $\rhd$\preabortt & \abortt & \Centerstack{- \\ re-activate $x_1$} & \Centerstack{- \\ re-activate $x_2$} & - \\
\noalign{\smallskip}
\toprule

\rowcolor{verylightgray}
5 & \Centerstack{\preabortt \\ -} & \Centerstack{\preacceptt \\ inactivate $x_2$} & \abortt & \Centerstack{- \\ -} & \Centerstack{- \\ re-activate $x_2$} & - \\
\midrule
6 & $\rhd$\preacceptt & & \acceptt & \Centerstack{- \\ create $y_1$} & \Centerstack{- \\ create $y_2$} & \Centerstack{- \\ create $y_3$} \\
\noalign{\smallskip}
\toprule

\rowcolor{verylightgray}
7 & \Centerstack{\preacceptt \\ inactivate $x_1$} & \Centerstack{\preabortt \\ -} & \abortt & \Centerstack{- \\ re-activate $x_1$} & \Centerstack{- \\ -} & - \\
\midrule
8 & & $\rhd$ \textsf{pre-accept}($T$) & \acceptt & \Centerstack{- \\ create $y_1$} & \Centerstack{- \\ create $y_2$} & \Centerstack{- \\ create $y_3$} \\
\noalign{\smallskip}
\toprule

\rowcolor{verylightgray}
9 & \Centerstack{\preabortt \\ -} & \Centerstack{\preabortt \\ -} & \abortt & \Centerstack{- \\ -} & \Centerstack{- \\ -} & - \\
\midrule
10 & $\rhd$ \textsf{pre-accept}($T$) & $\rhd$ \textsf{pre-accept}($T$) & \acceptt & \Centerstack{- \\ create $y_1$} & \Centerstack{- \\ create $y_2$} & \Centerstack{- \\ create $y_3$} \\
\bottomrule
\end{tabular}}
\caption[List of replay attacks against the first phase of \atomix.]{\footnotesize List of replay attacks against the first phase of \atomix for all possible executions of the transaction $T(x_1,x_2)\rightarrow (y_1,y_2,y_3)$ as described in \Cref{sec:byzcuit:attack-overview}. 
The highlighted rows indicate correct executions of \atomix (\ie, without the \attacker), and the other rows indicate incorrect executions due to the replay attacks. 
 In multirows, the top sub-rows show the protocol messages emitted by shards, and the bottom sub-rows indicate local shard actions as a result of emitting those messages.
  For example, (\mycolumn{1}, \myrow{1}) means that shard 1 emits \preacceptt (top sub-row), and inactivates $x_1$ (bottom sub-row). 
The first two columns indicate the messages emitted by each shard at the end of the first phase of \atomix. 
Replayed messages are marked with the symbol $\rhd$, for example $\rhd$\preabortt at (\mycolumn{1}, \myrow{2}) means that the \attacker sends to the client a \prerecorded \preabortt message impersonating shard 1 that races the original \preacceptt (\mycolumn{1}, \myrow{1}) emitted by shard 1. 
 The third column indicates the messages sent by the client to the relevant shards, and the last three columns indicate the local actions performed by shards at the end of the second phase of \atomix.
 }
\label{tab:byzcuit:atomix-attack}
\end{table*}

%% file: chapters/byzcuit/tables/atomix-2.tex
\begin{table*}[t]
\centering
%\footnotesize
\resizebox{\textwidth}{!}{\begin{tabular}{c C C C C}
\toprule
\multicolumn{5}{c}{\small Phase 2 of \atomix}\\
& \textbf{Client} & \Centerstack{\textbf{Shard 1} \\ (potential victim)} & \Centerstack{\textbf{Shard 2} \\ (potential victim)} & \Centerstack{\textbf{Shard 3} \\ (potential victim)}\\
\midrule

\rowcolor{verylightgray}
1 & \acceptt & \Centerstack{- \\ create $y_1$} & \Centerstack{- \\ create $y_2$} & \Centerstack{- \\ create $y_3$} \\
2 &  $\rhd$\abortt & \Centerstack{- \\ re-activate $x_1$} & \Centerstack{- \\ re-activate $x_2$} & - \\
\midrule

\rowcolor{verylightgray}
3 & \abortt  & \Centerstack{- \\ re-activate $x_1$} & \Centerstack{- \\ re-activate $x_2$} & - \\
4 & $\rhd$\acceptt & \Centerstack{- \\ create $y_1$} & \Centerstack{- \\ create $y_2$} & \Centerstack{- \\ create $y_3$} \\
\bottomrule
\end{tabular}}
\caption[List of replay attacks against the second phase of \atomix.]{\footnotesize  List of replay attacks against the second phase of \atomix for all possible executions of the transaction $T(x_1,x_2)\rightarrow (y_1,y_2,y_3)$ as described in \Cref{sec:byzcuit:attack-overview}. 
The highlighted rows indicate correct executions of \atomix (\ie, without the \attacker), and the other rows indicate incorrect executions due to the replay attacks.
In multirows, the top sub-rows show the protocol messages emitted by shards, and the bottom sub-rows indicate local shard actions.
Note that we use the multirow format for consistency reasons; in this table the first column indicates the messages emitted by the client  at the beginning of the second phase of \atomix, and the last two column shows the effect of these messages on the relevant shards.
Replayed messages are marked with the symbol $\rhd$.
For example, $\rhd$\abortt at (\mycolumn{1}, \myrow{2}) means that the \attacker sends a \prerecorded \abortt message to the input shards impersonating the client.}
\label{tab:byzcuit:atomix-attack:2}
\end{table*}

%% file: chapters/byzcuit/sections/prerecoring.tex
% =========
% Prerecording
% =========
\section{Eliciting Messages to Replay} \label{sec:byzcuit:prerecoring}
We show how the \attacker can act as (or collude with) a client to actively elicit and record the target messages to later use in the replay attacks. This empowers the \attacker to actively orchestrate the attacks.
We describe how the \attacker can trigger target messages in the context of an example, without loss of generality.
Lets assume that shard 1 manages objects $x_{1}$ (\activeObj) and object $\widetilde{x_1}$ (\inactiveObj or non-existent), and shard 2 manages object $x_{2}$ (\activeObj);
$\widetilde{x*}$ means any inactive object on the shard, and $y*$ means any output object (\ie, their details do not matter).

% =========
\subsection{\Shardled Cross-Shard Consensus} \label{sec:byzcuit:prerecoring-chainspace}
We show how the \attacker can act as (or collude with) a client to actively elicit and record the target messages, in the context of \shardled cross-shard consensus protocols as illustrated by \Cref{sec:byzcuit:shard-led-consensus}.
To elicit \preacceptt for a transaction $T(x_{1},x_{2})\rightarrow (y*)$ (the output $y*$ is not relevant here) from shard 1, the key consideration is to closely precede the transaction with another transaction $T'$ that: \first locks the inputs managed by at least one other shard (in this case $x_2$ on shard 2); and \second to ensure that the preceding transaction $T'$ gets ultimately aborted, and $x_2$ becomes active again. The steps look as follows:  

\begin{itemize}
\item The \attacker submits $T'(x_{2},\widetilde{x*})\rightarrow (y*)$ to shard 2. This locks $x_2$.
\item The attacker quickly follows up by submitting $T(x_{1},x_{2})\rightarrow (y*)$ to shard 1 and shard 2. Shard 1 generates \preacceptt, which is the target message that the attacker records. Shard 2 generates \preabortt because $x_2$ is locked by $T'$. Consequently, in the second phase of \sbac, both shard 1 and shard 2 end up aborting $T$.
\item $T'$ is eventually aborted, making $x_{2}$ active again.
\end{itemize}

To elicit \preabortt for a transaction $T(x_{1},x_{2})\rightarrow (y*)$ (the output $y*$ is not relevant here) from shard 1, the key consideration is to closely precede the transaction with another transaction $T'$ that locks the input managed by the shard (in this case $x_1$ on shard 1). The steps look as follows:
\begin{itemize}
\item The attacker submits  $T'(x_{1},\widetilde{x*})\rightarrow (y*)$ to shard 1. This locks $x_1$.
\item The attacker quickly follows up by submitting $T(x_{1},x_{2})\rightarrow (y*)$ to shard 1 and shard 2. Shard 1 generates \preabortt because $x_1$ is locked by $T'$, which is the target message that the attacker records. Shard 2 generates \preacceptt. Consequently, in the second phase of \sbac, both shard 1 and shard 2 end up aborting $T$.
\item $T'$ is eventually aborted, making $x_{1}$ active again.
\end{itemize}

To elicit \acceptt used by the attacks described in \Cref{sec:byzcuit:chainspace-attack-2}, the \attacker simply submits transaction $T$ and observes and records its successful execution. The \attacker has no incentive to record \abortt messages as these are ignored by shards (see \Cref{tab:byzcuit:sbac-attack:2}).

% =========
\subsection{\Clientled Cross-Shard Consensus} \label{sec:byzcuit:prerecoring-omniledger}
We show how the \attacker can act as (or collude with) a client to actively elicit and record the target messages, in the context of \clientled cross-shard consensus protocols as illustrated by \Cref{sec:byzcuit:client-led-consensus}.
To elicit \preacceptt from shard 1 for a transaction $T(x_{1},x_{2})\rightarrow (y*)$ (the output $y*$ is not relevant here) from shard 1, the key consideration is to closely precede the transaction with another transaction that: \first temporarily spends the inputs managed by at least one other shard (in this case $x_2$ on shard 2); and \second to ensure that the preceding transaction is ultimately aborted so that $x_2$ becomes active again. The steps look as follows:

\begin{itemize}
\item The \attacker submits $T'(x_2,\widetilde{x*})\rightarrow (y*)$ to shard 2, where $\widetilde{x*}$ is managed by a different shard. shard 2 emits \preaccepttt and marks $x_2$ as inactive.
\item The \attacker follows up by submitting $T(x_{1},x_{2})\rightarrow (y*)$ to shard 1 and shard 2. Shard 1 generates \preacceptt, which is the target message that the attacker records. Shard 2 generates \preabortt because $x_2$ is inactive.
\item The \attacker submits \abortt to shard 1 to reactivate $x_1$, and sends \aborttt to shard 2 to reactivate $x_2$.
\end{itemize}

For the attacks described in \Cref{sec:byzcuit:omniledger-attack-2}, the \attacker needs to elicit \abortt and \acceptt from the target shards. For the former, the \attacker can follow the steps described previously to elicit \preacceptt and \preabortt. To elicit \acceptt, the \attacker simply submits transaction $T$ and observes and records its successful execution. 

%% file: chapters/byzcuit/sections/defenses.tex
% =========
% Defenses
% =========
\section{Defenses Against Replay Attacks}\label{sec:byzcuit:defenses}
We identify two issues that lead to the replay attacks described in \Cref{sec:byzcuit:shard-led-consensus} and \Cref{sec:byzcuit:client-led-consensus}, and discuss how to fix those:

\begin{itemize}
\item First, the input shards do not have a way to know that particular protocol messages received correspond to a specific instance (or session) of the protocol. This gap in the input shards' knowledge enables an \attacker to replay, mix and match, old messages leading to attacks.
To address this limitation, we associate a session identifier with each transaction, which has to be crafted carefully to not degrade the performance of the protocols significantly---such as by requiring nodes to store state linearly in the number of past transactions.

\item Second, in some cases the output shards are only involved in the second phase of the protocol, and therefore have no knowledge of the transaction context (to determine freshness) that is available to the input shards.
This limitation can be addressed by ensuring that all shards---input and output---witnesses the entire protocol execution, rather than just one of the protocol phases. 
\end{itemize}

Note that the two mitigation techniques described above must be used together, as part of a single defense strategy against replay attacks. 

%% file: chapters/byzcuit/sections/byzcuit.tex
% =========
% Byzcuit
% =========
\section{The \byzcuit Atomic Commit Protocol} \label{sec:byzcuit:byzcuit}
We showed that both \sbac (Sections~\ref{sec:byzcuit:chainspace-attack-1}~and~\ref{sec:byzcuit:chainspace-attack-2}) and \atomix (Sections~\ref{sec:byzcuit:omniledger-attack-1}~and~\ref{sec:byzcuit:omniledger-attack-2}) are vulnerable to replay attacks that can compromise system liveness and safety.
\atomix is the simpler protocol of the two, and using the client to coordinate cross-shard communication can reduce the cost to $O(n)$ in the number of shards (by aggregating shard messages).
However, an unresponsive or malicious client can permanently lock input objects by never initiating the second phase of the protocol, requiring additional design considerations (\eg, a new entity that periodically unlocks input objects for transactions on which no progress has been made). 
On the other hand, \sbac runs the protocol among the shards, without relying on client coordination.
But this comes at the cost of increased cross-shard communication: all input shards communicate with all other input shards, which leads to communication complexity of $O(n^2)$ where $n$ is the number of input shards.

Motivated by these insights, we present \byzcuit---a cross-shard atomic commit protocol (based on \sbac) that integrates design features from \atomix---and offers better performance and security against replay attacks. 
\byzcuit allocates a Transaction Manager (TM) to coordinate cross-shard communication, reducing its cost to $O(n)$ in the happy case\footnote{The communication complexity can be reduced to $O(n)$ in the number of shards by aggregating shard messages as described by \omniledger.}; alternatively \byzcuit also has a fall-back mode in case the TM fails, similar to \atomix and traditional two phase commit protocols. 
\byzcuit achieves resilience against the replay attacks described in \Cref{sec:byzcuit:shard-led-consensus} and \Cref{sec:byzcuit:client-led-consensus}, by leveraging the defense proposed in \Cref{sec:byzcuit:defenses}.

% =========
\subsection{\byzcuit Protocol Design} 
We describe how \byzcuit integrates the defense presented in \Cref{sec:byzcuit:client-led-consensus}. To map particular protocol messages to a specific protocol instance (or session), \byzcuit associates a session identifier with each transaction. To ensure that all the relevant (input and output) shards witness all phases of the protocol execution, \byzcuit leverages the notion of \emph{dummy objects}.
Each shard creates a fixed number of dummy objects upon configuration;
if a shard only serves as an output shard for a transaction, \byzcuit forces it to be involved in the first phase of the protocol by implicitly including a dummy object managed by the output shard in the transaction inputs, which will create a new dummy object upon completion.
Thus, the output shard also becomes an input shard (because of the dummy object in the transaction inputs) and witnesses the entire protocol, rather than just the second phase. 

\para{\byzcuit protocol execution} We illustrate \byzcuit taking the example of a transaction $T(x_1,x_2)\rightarrow (y_1,y_2,y_3)$ with two input objects, $x_1$ managed by shard 1 and $x_2$ managed by shard 2; and three outputs, $y_1$ managed by shard 1, $y_2$ managed by shard 2, and $y_3$ managed by shard 3.
\Cref{fig:byzcuit:sysname} illustrates the \byzcuit protocol; the client first sends the transaction to all input and output shards.
Note that this is different than other protocols like \sbac and \atomix, where the transaction is only sent to the input shards.
As mentioned previously, to achieve resilience against replay attacks, \byzcuit forces a shard that is \emph{only} involved in creating the output objects to also become an input shard (and witness the transnational context by participating in the first phase of the protocol) by implicitly consuming one of its dummy inputs (which creates a new dummy object upon completion). \byzcuit associates a sequence number $s_{x_i}$ to each object and dummy object (when the object is created $s_{x_i}=0$). The sequence number is intrinsically linked to the object: when clients query shards to obtain an object $x_i$, they also receive the associated sequence number $s_{x_i}$. 

\begin{figure}[t]
\centering
\includegraphics[width=\textwidth]{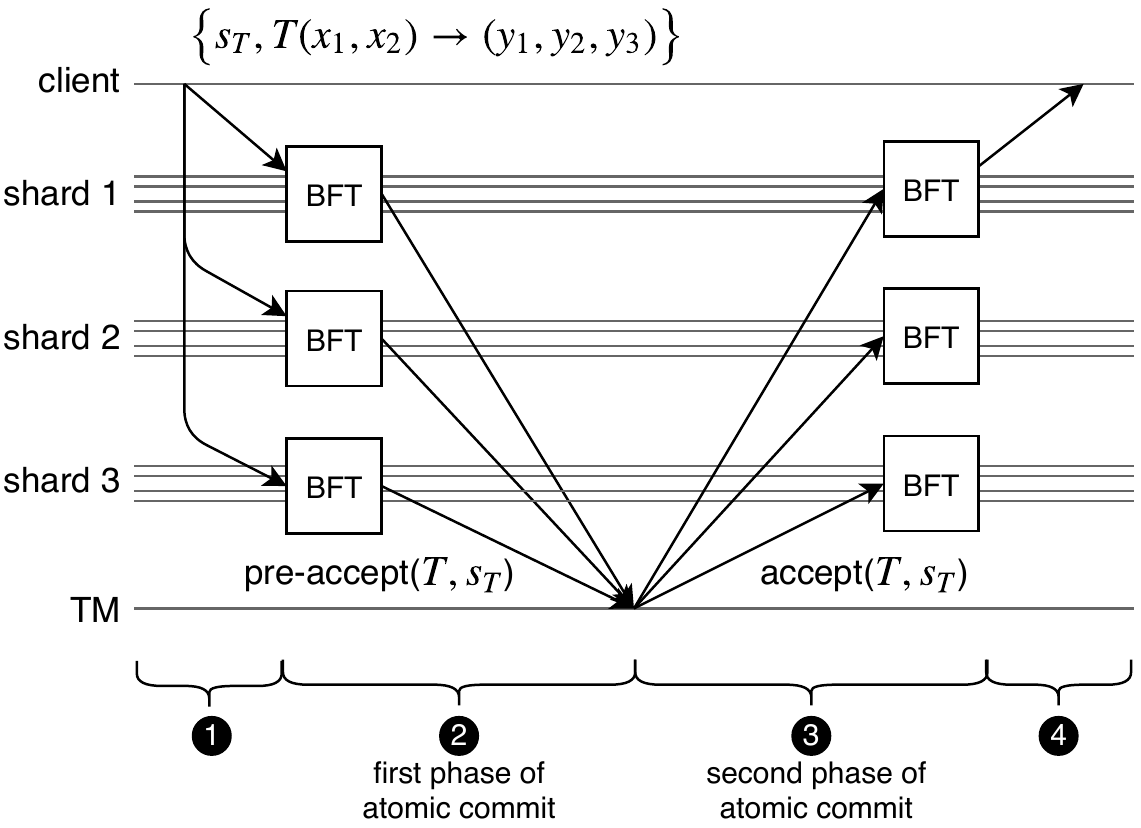}
\caption[Design overview of \byzcuit.]{\footnotesize An example execution of \byzcuit for a valid transaction $T(x_1,x_2)\rightarrow (y_1,y_2,y_3)$ with two input objects ($x_1$ and $x_2$, both are active), and three outputs $(y_1,y_2,y_3)$, where the final decision is \acceptts.}
\label{fig:byzcuit:sysname}
\end{figure}

When submitting the transaction $T$, the client also sends along a transaction sequence number $s_{T}=max\{ s_{x_1}, s_{x_2}, s_{d_3} \}$,
where the transaction sequence number $s_{T}$ is the maximum of the sequence numbers $s_{x_i}$ of each input object $x_i$ and dummy objects $d_i$~(\ding{202}). 
Upon receiving a new pair $(T,s_T)$, each shard saves $(T,s_T)$ in a local cache memory---the transaction sequence number $s_T$ acts as session identifier associated with the transaction $T$. Each shard internally verifies that the transaction passes local checks, and that $s_T$ is equal to (or bigger than) the sequence numbers of the objects they manage (\ie, shard 1 checks $s_T \geq s_{x_1}$, shard 2 checks $s_T \geq s_{x_2}$, shard 3 checks $s_T \geq s_{d_3}$). The shards send their local decision to the TM: \preacceptts for local accept (and the shard locks the objects it manages), or \preabortts for local abort.
After receiving all the messages corresponding to the first phase of \byzcuit from the concerned shards, the TM sends a suitable message to the shards (\acceptts if all the shards respond with \preacceptts, or \abortts otherwise). 
Upon receiving \acceptts or \abortts from the TM, shards first verify that they previously cached the pair $(T, s_T)$ associated with the message; otherwise they ignore it~(\ding{203}).

The \acceptts or \abortts messages sent by the TM provide enough evidence to the shards to verify whether $s_T$ is correctly computed; \ie shards verify that $s_T$ is at least the maximum of the sequence numbers of each input and dummy object by inspecting the transaction $T$ signed by each shard. If \acceptts has a correct $s_T$, the shards inactivate the input objects and create the output objects $(y_1,y_2,y_3)$, and shard 3 creates a new dummy object $\widetilde{d}_{3}$;  otherwise, they update the sequence numbers of each input object $(s_{x_1},s_{x_2})$ and dummy object $d_{3}$ to $(s_T+1)$, \ie shards locally update $s_{x_1} \leftarrow (s_T+1)$ and $s_{x_2} \leftarrow (s_T+1)$, and $s_{d_{3}}\leftarrow (s_T+1)$. Shards delete $(T,s_T)$ from their local cache~(\ding{204}).
Since we assume that shards are honest---inline with the threat model of the systems discussed---it suffices if only one shard notifies the client of the protocol outcome; we may set any arbitrary rule to decide which shard notifies the client (\eg, the shard handling the first input object)~(\ding{205}).
\Cref{fig:byzcuit:byzcuit-fsm} shows the finite state machine describing the life cycle of \byzcuit objects. 

\begin{figure}[t]
\centering
\includegraphics[width=\textwidth]{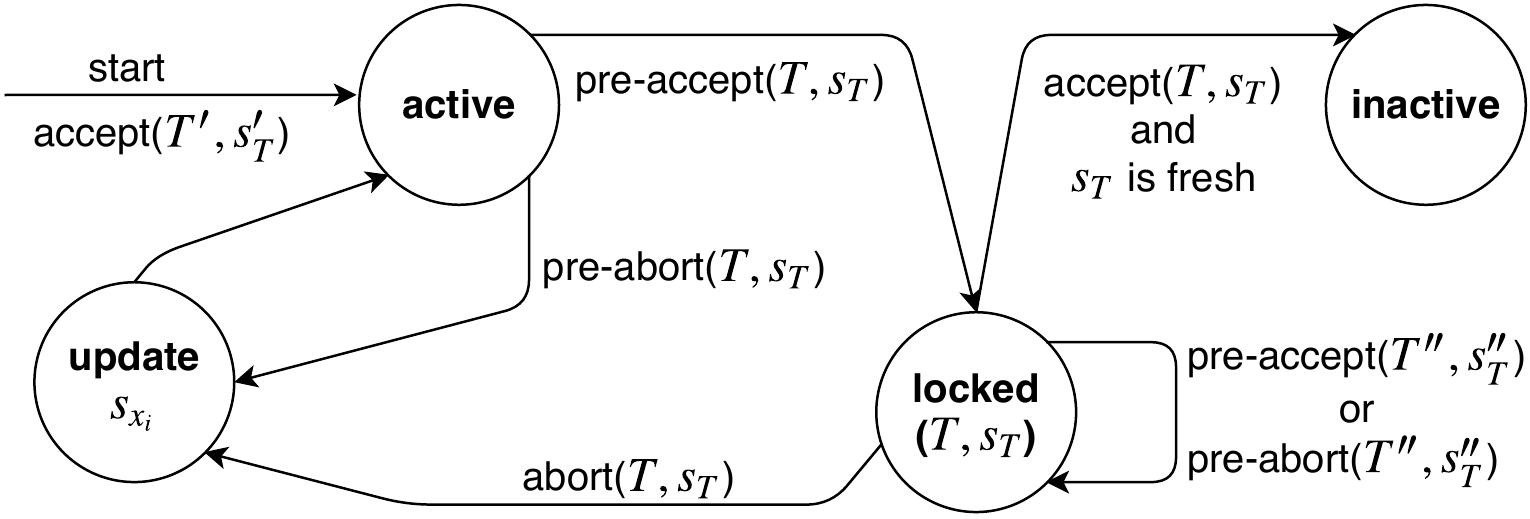}
\caption[State machine representation of objects within \byzcuit.]{\footnotesize State machine representing the life cycle of objects in \byzcuit. Objects are initially \activeObj. Upon receiving a transaction that passes local checks, a shard changes its input objects' state to \locked  (objects are locked for a given transaction $T$ and transaction sequence number $s_T$) and emits \preacceptts; otherwise it updates the sequence number of every object it manages and emits \abortts. 
Once a shard locks input objects for a given $(T,s_T)$, any \acceptts and \abortts with malformed transaction sequence numbers are ignored, and do not cause any transition (not included in the figure). Any incoming transaction $T'$ that requires processing \locked input object(s) is aborted.
Upon receiving \acceptts with a well formed $s_T$, a shard makes its input objects \inactiveObj and creates the output objects.
Alternatively, upon receiving \abortts shards unlock their inputs and update the corresponding sequence numbers.
}
\label{fig:byzcuit:byzcuit-fsm}
\end{figure}

\para{Transaction manager} 
The Transaction Manager (TM) coordinates cross-shard communication in \byzcuit. We now discuss who might play the role of the TM, and argue that \byzcuit guarantees liveness even if the TM is faulty (Byzantine) or crashes.
Keeping with the overall design goal of decentralization, we envision that a designated shard will act as the TM. If the shard is honest, the TM is live---and therefore progress is always made. The input shards contact in turn each node of the TM shard until they reach one honest node. The TM shard may have up to $f$ dishonest nodes; therefore, the client or the input shards need to send messages to at least $f+1$ nodes of the TM shard to ensure that it is received by at least one honest node\footnote{Clients may take a statistical view of availability. Given that fewer than $2/3$ of nodes in a shard are dishonest, sending the transaction to $\rho$ nodes will fail to reach an honest node with probability only $(1/3)^{\rho}$. Clients may also send messages sequentially to nodes, and only continue if they do not observe progress within some timeout to further reduce costs.}. Thus, as soon as the first honest node receives the message, the protocol is guaranteed to make progress.
If the TM is the client or any centralized party, it may act arbitrarily---but this does not stall the protocol because anyone can make the protocol progress by taking over at any time the role of the TM. This is possible because the TM does not act on the basis of any secrets, therefore anyone else can take over and complete the protocols. This `anyone' may be an honest node in a shard that wants to finally unlock an object (\eg, upon a timeout); or other clients that wish to use a locked object; or it may be an external service that has a job to periodically close open \byzcuit instances. \byzcuit ensures such parties may attempt to make progress asynchronously and concurrently safely.
Therefore, \byzcuit guarantees liveness as long as there is at least one honest entity in the system.

\para{Handling sequence number overflow}
An \attacker can try to exhaust the possible sequence numbers to make them overflow. The \attacker submits a pair $(T,s_T)$ such that the sequence number $s_T$ is just below the system overflow value; the sequence numbers associate with the inputs overflow upon the next updates, and the system would be again prone to the attacks described in \Cref{sec:byzcuit:chainspace-attack-1}\footnote{Note that this overflow vulnerability is common to every system relying on nonces chosen by the users, like Byzantine Quorum Systems~\cite{malkhi1998byzantine}.}. To mitigate this issue, shards define a \emph{clone} procedure allowing to update any of their objects to an unchanged version of themselves (\ie it creates a fresh copy of the object). This clone procedure effectively creates a new object with serial number $s_x'=0$. When shards detect that the serial number of one of their objects approaches the overflow value, they execute internally this clone procedure. The \attacker may exploit this mechanism to DoS the system, forcing shards to constantly update their objects; as a result, the target objects are not available to users. DoS countermeasures are out of scope, and are typically addressed by introducing transaction fees.

% =========
\subsection{Security against Replay Attacks}\label{sec:byzcuit:security}
We argue that \byzcuit is resilient to replay attacks. We recall the Honest Shard assumption from \chainspace (\Cref{sec:chainspace:threat-security}) under which \byzcuit operates, and assume that messages are authenticated as in traditional BFT protocols. 

\begin{SecAssumption}\label{def:byzcuit:honest-shard}(Honest Shard)
The adversary may create arbitrary smart contracts, and input arbitrary transactions into \byzcuit, however they are bound to only control up to $f$ faulty nodes in any shard. As a result, and to ensure the correctness and liveness properties of Byzantine consensus, each shard must have a size of at least $3f + 1$ nodes.
\end{SecAssumption}

Any message emitted by shards comes with at least $f+1$ signatures from nodes. Assuming honest shards, the \attacker can forge at most $f$ signatures, which is not enough to impersonate a shard. We use the Lemma below to prove the security of \byzcuit through this section.

\begin{lemma} \label{le:byzcuit:secure-phase-1}
Under Honest Shard assumption, no \attacker can obtain \prerecorded messages containing a fresh transaction sequence number $s_T$.
\end{lemma}
\begin{proof}
\small
The key idea protecting \byzcuit from these replay attacks is that the \attacker can only obtain \prerecorded messages associated with old transaction sequence numbers $s_T$. The transaction sequence number $s_T$ is fresh only if it is at least equal the maximum of the sequence number of all input and dummy objects of the transaction $T$. Shards update every input and dummy object sequence number upon aborting transactions in such a way that sequence numbers only increase. That is, after emitting \preabortts or \preacceptts , either the sequence number of all input and dummy objects of $T$ are updated to a value bigger than $s_T$ (in case of \preabortts), or the objects are inactivated which prevents any successive transaction to use them as input (in case of \preacceptts).
It is thus impossible for the adversary to hold a \prerecorded message for a fresh $s_T$; the only \prerecorded messages that the adversary can obtain contain sequence numbers smaller than $s_T$.
\end{proof}

\para{Security of the first phase of \byzcuit}
An \attacker may try to replay \preacceptts and \preabortts during the first phase of the protocol, similarly to the attacks described in Sections~\ref{sec:byzcuit:chainspace-attack-1}~and~\ref{sec:byzcuit:omniledger-attack-1}; the TM then aggregates these messages into either \acceptts or \abortts, and forwards them to the shards during the second phase of the protocol. 
Security Theorem~\ref{th:byzcuit:secure-phase-1} shows that \byzcuit detects that they originate from replayed messages and ignores them. Intuitively, the transaction sequence number $s_T$ acts as a monotonically increasing session identifier associated with the transaction $T$; the \attacker cannot obtain \prerecorded messages containing a fresh $s_T$. \byzcuit shards can then distinguish replayed messages (\ie, messages with old $s_T$) from the messages coming from the instance of the protocol that they are executing (\ie, messages with fresh $s_T$).

\begin{theorem} \label{th:byzcuit:secure-phase-1}
Under Honest Shard assumption, \byzcuit ignores \acceptts and \abortts messages issued from replayed \preacceptts and \preabortts.
\end{theorem}
\begin{proof}
\small
\Cref{fig:byzcuit:byzcuit-fsm} shows that once \byzcuit locks objects for a particular pair $(T,s_T)$, the protocol can only progress toward \acceptts or \abortts; \ie shards can either accept or abort the transaction $T$. The \attacker aims to trick one or more shards to incorrectly accept or abort $T$ by injecting \prerecorded messages during the first phase of \byzcuit; we show that the \attacker fails in every possible scenario. 

Suppose transaction $T$ should abort (the TM outputs \abortts), but the \attacker tries to trick some shards to accept the transaction. \Cref{fig:byzcuit:byzcuit-fsm} shows that the \attacker can only succeed the attack if they gather \acceptts containing a fresh transaction sequence number $s_T$. Lemma~\ref{le:byzcuit:secure-phase-1} states that no \attacker can obtain \prerecorded messages over a fresh transaction sequence number $s_T$; therefore the only messages available to the adversary at this point of the protocol are (at most) $n-1$ \preacceptts and (at most) $n$ \abortts, where $n$ is the number of concerned shards. This is not enough to form an \acceptts message with a fresh transaction sequence number $s_T$ (which is composed of $n$ \preacceptts); therefore the \attacker cannot trick any shard to accept the transaction.

Suppose transaction $T$ should be accepted (the TM outputs \acceptts with a fresh $S_T$), but the \attacker tries to trick some shards to abort the transaction. \Cref{fig:byzcuit:byzcuit-fsm} show that \byzcuit does not require a fresh transaction sequence number $s_T$ to abort transactions (the freshness of $s_T$ is only enforced upon accepting a transaction); but shards locked the input and dummy objects of the transaction for the pair $(T,s_T)$ (with fresh $s_T$), so the attacker needs to gather \abortts containing the same transaction sequence number $s_T$ locked by shards.  Lemma~\ref{le:byzcuit:secure-phase-1} shows that the \attacker cannot obtain \prerecorded messages over fresh $s_T$; therefore the only messages available to the adversary containing the (fresh) $s_T$ locked by shards at this point of the protocol are $n$ \preacceptts. This is not enough to form an \abortts message (which is composed of at least one \preabortts); therefore the \attacker cannot trick any shard to abort the transaction.
\end{proof}

\para{Security of the second phase of \byzcuit}
An \attacker may try to replay \acceptts and \abortts messages during the second phase of the protocol, similarly to the attacks described in Sections~\ref{sec:byzcuit:chainspace-attack-2}~and~\ref{sec:byzcuit:omniledger-attack-2}.
Security Theorem~\ref{th:byzcuit:secure-phase-2} shows that \byzcuit ignores those replayed messages. Intuitively, these attacks target shards acting only as output shards (and not also as input shards) and exploit the fact that they are only involved in the second phase of the protocol, and therefore have no knowledge of the transaction context (to determine freshness) that is available to the input shards. \byzcuit is resilient to these replay attacks as it is designed in such a way that there are no shards that act only as output shards; all output shards are forced to also become input shards, by introducing dummy objects if they do not manage any input objects; this prevents the attacks by removing the attack target.

\begin{theorem} \label{th:byzcuit:secure-phase-2}
Under Honest Shard assumption, \byzcuit ignores replayed \acceptts and \abortts messages.
\end{theorem}
\begin{proof}
\small
\Cref{fig:byzcuit:byzcuit-fsm} shows that shards only act upon \acceptts and \abortts messages if they have the pair $(T,s_T)$ saved in their local cache\footnote{Contrarily to \sbac and \atomix, all \byzcuit shards have the pair $(T,s_T)$ in their local cache after as they all participate to the first phase of the protocol.}. Shards save a pair $(T,s_T)$ in their local cache upon emitting \preacceptts or \preabortts, and delete it at the end of the protocol; therefore the only attack windows where the adversary can replay \acceptts and \abortts messages is while the transaction $T$ (associated with $s_T$) is being processed by the second phase of \byzcuit. This forces the attacker to operates under the same conditions as Security Theorem~\ref{th:byzcuit:secure-phase-1}, which is shown secure above.
\end{proof}
%
%A proof of \Cref{th:secure-phase-2} can be found in \Cref{sec:proof-th2}. Intuitively, these attacks target shards acting only as output shards (and not also as input shards) and exploit the fact that they are only involved in the second phase of the protocol, and therefore have no knowledge of the transaction context (to determine freshness) that is available to the input shards. \byzcuit is resilient to these replay attacks as it is designed in such a way that there are no shards that act only as output shards; all output shards are forced to also become input shards, by introducing dummy objects if they do not manage any input objects; this prevents the attacks by removing the attack target.

%Appendix~\ref{properties} shows that \byzcuit guarantees liveness, consistency and validity, similarly to \sbac. 

% =========
\subsection{\byzcuit Security \& Correctness}
We show that \byzcuit guarantees liveness, consistency, and validity.

\begin{theorem}
(Liveness) Under Honest Shards assumption, a transaction $T$ that is proposed to at least one honest concerned node, eventually results in either being committed or aborted, namely all parties deciding \acceptts or \abortts.
\end{theorem}
\begin{proof}
\small
We rely on the liveness properties of the Byzantine agreement (shards with only $f$ nodes eventually reach consensus on a sequence), and the broadcast from nodes of shards to all other nodes of shards, channelled through the Transaction Manager. Assuming $T$ has been given to an honest node, it will be sequenced withing an honest shard BFT sequence, and thus a \preacceptts or \preabortts will be sent from the $2f+1$ honest nodes of this shard, aggregated into \acceptts or \abortts, and sent to the $f+1$ nodes of the other concerned shards. Upon receiving these messages the honest nodes from other shards will process the transaction within their shards, and the BFT will eventually sequence it. Thus the user will eventually receive a decision from at least $f+1$ nodes of a shard.
\end{proof}
\begin{theorem}
(Consistency) Under Honest Shards assumption, no two conflicting transactions, namely transactions sharing the same input will be committed. Furthermore, a sequential executions for all transactions exists.
\end{theorem}
\begin{proof}
\small
A transaction is accepted only if some nodes receive \acceptts, which presupposes all shards have provided enough evidence to conclude \preacceptts for each of them. Two conflicting transaction, sharing an input, must share a shard of at least $3f+1$ concerned nodes for the common object---with at most $f$ of them being malicious. Without loss of generality upon receiving the \preacceptts message for the first transaction, this shard will sequence it, and the honest nodes will emit messages for all---and will lock this object until the two phase protocol concludes. Any subsequent attempt to \preacceptts for a conflicting $T'$ will result in a \preabortts and cannot yield a accept, if all other shards are honest majority too. After completion of the first \acceptts the shard removes the object from the active set, and thus subsequent $T'$ would also lead to \preabortts. Thus there is no path in the chain of possible interleavings of the executions of two conflicting transactions that leads to them both being committed.
\end{proof}

\begin{theorem}
(Validity) Under Honest Shards assumption, a transaction may only be accepted if it is valid according to the smart contract (or application) logic.
\end{theorem}
\begin{proof}
\small
A transaction is committed only if some nodes conclude that \acceptts, which presupposes all shards have provided enough evidence to conclude \preacceptts for each of them. The concerned nodes include at least one shard per input object for the transaction; for any contract logic represented in the transaction, at least one of those shards will be managing object from that contract. Each shard checks the validity rules for the objects they manage (ensuring they are active) and the contracts those objects are part of (ensuring the transaction is valid with respect to the contract logic) in order to \preacceptts. Thus if all shards say \preacceptts to conclude that \acceptts, all object have been checked as active, and all the contract calls within the transaction have been checked by at least one shard---whose decision is honest due to at most $f$ faulty nodes. If even a single object is inactive or locked, or a single trace for a contract fails to check, then the honest nodes in the shard will emit \preabortts, and the final decision will be \abortts.
\end{proof}

%% file: chapters/byzcuit/sections/implementation.tex
% =========
% Implementation
% =========
\section{Implementation \& Evaluation} \label{sec:byzcuit:implementation}
We implement a prototype of \byzcuit (Section~\ref{sec:byzcuit:byzcuit}) in Java and evaluate its performance and scalability. To analyze the overhead introduced by our replay attack defenses (\ie, with message sequence numbers and dummy objects), we compare \byzcuit with replay defenses (\byzcuitreplay) with the baseline of \byzcuit without any replay attack defenses (\byzcuitbaseline). 
Our implementation of \byzcuit is a fork of an early \chainspace code~\cite{chainspace}, and is released as an open-source project\footnote{\url{https://github.com/sheharbano/byzcuit}}.
For BFT consensus, we use the \bftsmart~\cite{bftsmart} Java library (based on PBFT~\cite{pbft}), which is one of the very few maintained open source BFT libraries. End users run a client to communicate with \byzcuit nodes, which sends transactions according to the \bftsmart protocol. The \byzcuit client also acts as the Transaction Manager (TM) and is responsible for driving the cross-shard consensus.
We evaluate the performance and scalability of our \byzcuit implementation through deployments on Amazon EC2 containers. 
%We also compare \byzcuit with \chainspace to measure performance improvements, by running our evaluations in a similar setup as \chainspace. 
We launch up to 96 instances for shard nodes and 96 instances for clients on \emph{t2.medium} virtual machines, each containing 8 GB of RAM on 2 virtual CPUs and running GNU/Linux Debian 8.1. We use 4 nodes per shard. 
Each measured data point corresponds to 10 runs represented by error bars. The error bars in \Cref{fig:byzcuit:tpsVSshards} and \Cref{fig:byzcuit:tpsVSdummy} show the average and standard deviation, and the error bars in \Cref{fig:byzcuit:latencyVStps} show the median and the 75th and 25th percentiles.

\begin{figure}[t]
\centering
\includegraphics[width=.7\textwidth]{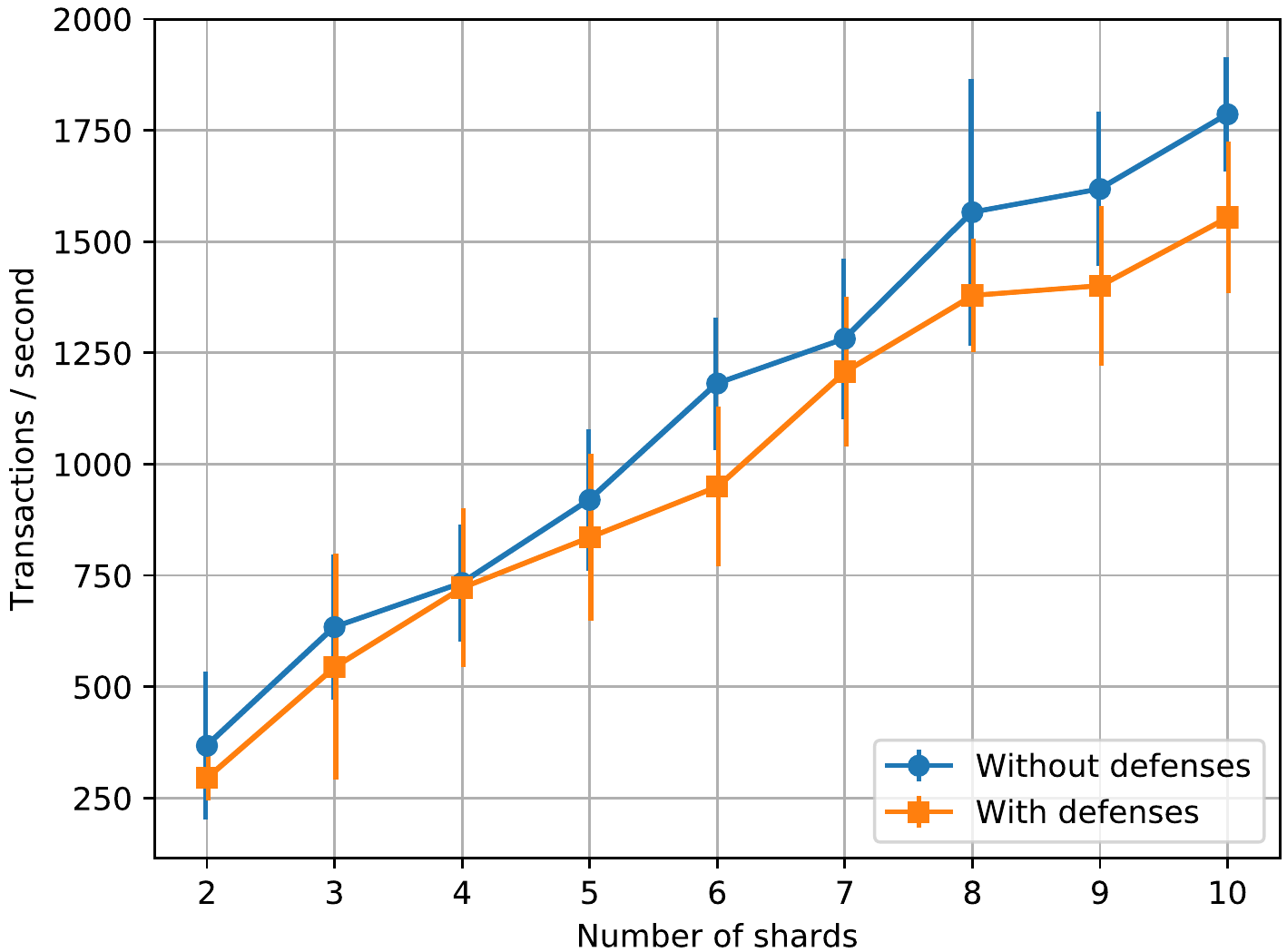}
\caption[\byzcuit throughput \vs number of shards.]{\footnotesize The effect of the number of shards on throughput. Each transaction has 2 input objects and 5 output objects, both chosen randomly from shards.}
\label{fig:byzcuit:tpsVSshards}
\end{figure}
\begin{figure}[t]
\centering
\includegraphics[width=.7\textwidth]{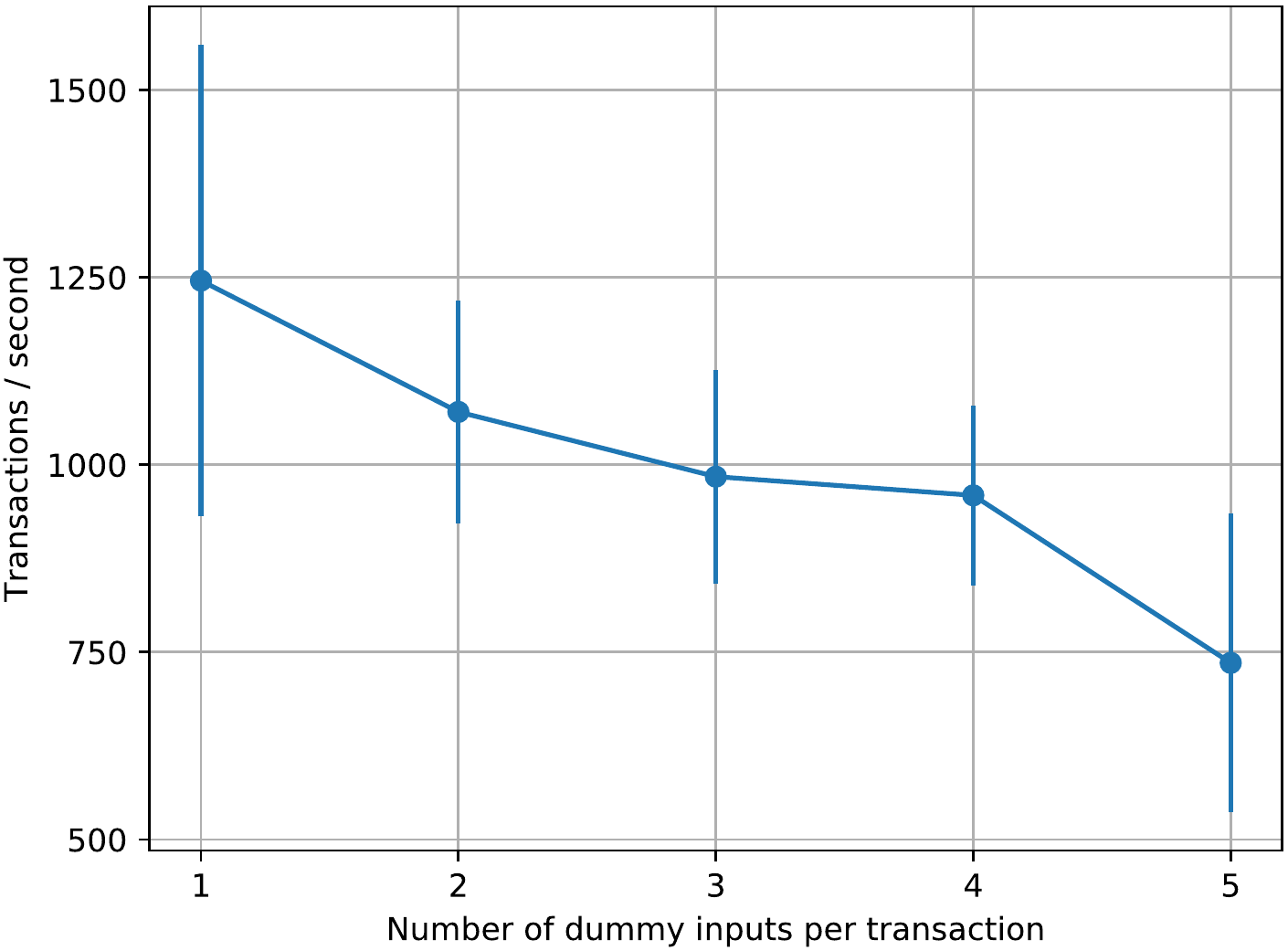}
\caption[\byzcuit throughput \vs number of dummy objects.]{\footnotesize Decrease of \byzcuit throughput with the number of dummy objects. Each transaction has 1 input object, and up to 5 dummy objects randomly selected from unique non-input shards. 6 shards are used.}
\label{fig:byzcuit:tpsVSdummy}
\end{figure}
\begin{figure}[t]
\centering
\includegraphics[width=.7\textwidth]{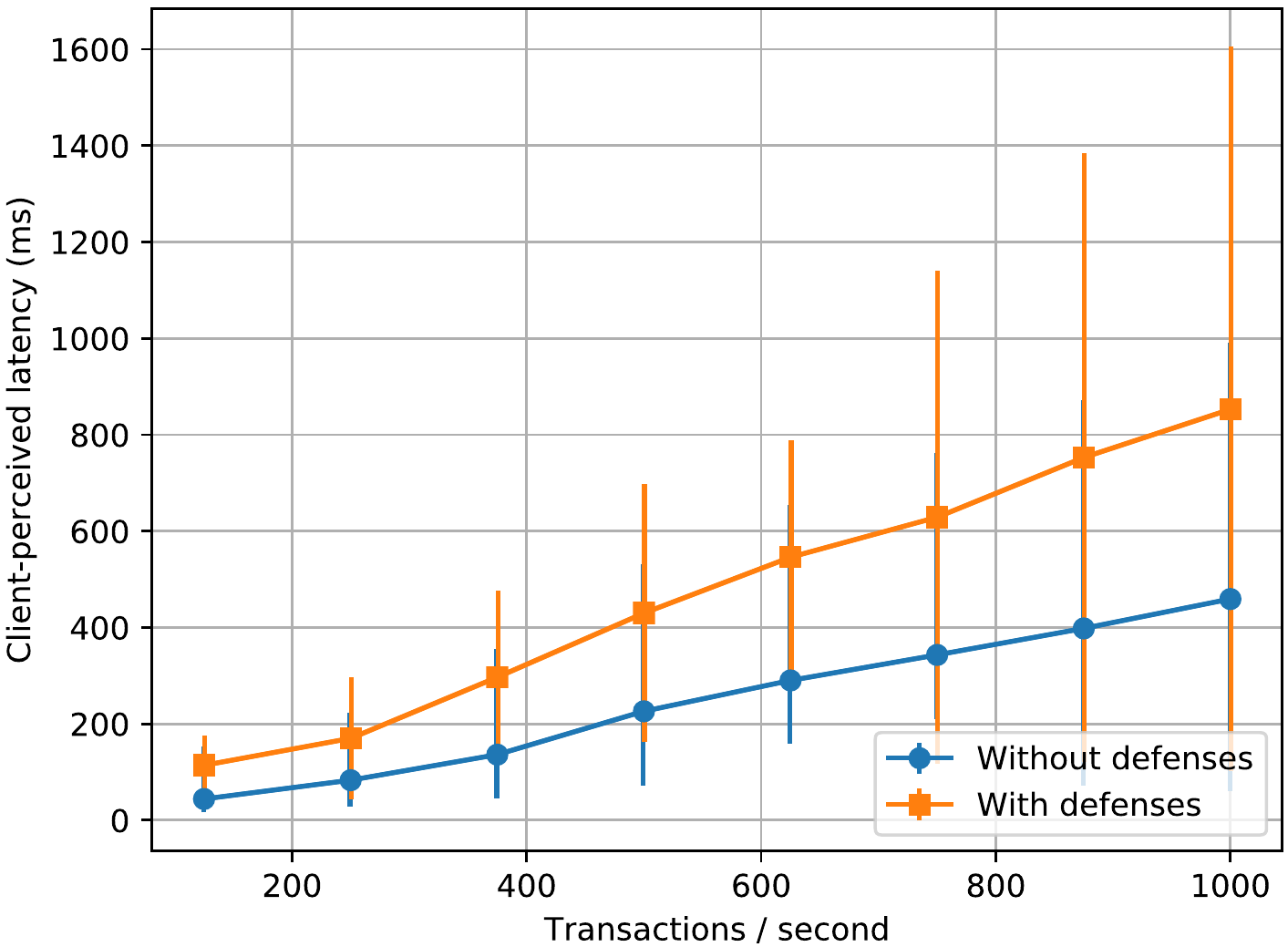}
\caption[\byzcuit client-perceived latency.]{\footnotesize Client-perceived latency vs. system load (\ie transactions received per second), for 6 shards with 2 inputs and 5 outputs per transaction (both chosen randomly from shards).}
\label{fig:byzcuit:latencyVStps}
\end{figure}

\mypara{Throughput and scalability}
\Cref{fig:byzcuit:tpsVSshards} shows the throughput of \byzcuit (the number of transactions processed per second, tps) corresponding to an increasing number of shards. 
Each transaction has 2 input objects and 5 output objects, chosen randomly from shards. We test transactions with 5 output objects for a fair evaluation of \byzcuit's replay defenses by triggering the creation of dummy objects (\ie, a large number of output objects and a small number of input objects implies a higher probability of output-only shards getting selected, triggering the creation of dummy objects). We find that \byzcuitreplay has a throughput of 260 tps for 2 shards, and linearly scales with the addition of more shards achieving up to 1,550 tps for 10 shards\footnote{Expanding the testbed to evaluate \byzcuit with more than 10 shards is unpractical for prototyping: the number of machines required for the evaluation grows fast (10 shards with 4 nodes each already requires 40 different AWS instances), and quickly becomes costly.}.
As expected, the throughput of \byzcuitreplay is lower than \byzcuitbaseline by a somewhat constant factor ranging from 20--200 tps, but still increases linearly. This is expected because the creation of dummy objects in \byzcuitreplay leads to a higher number of shards processing the same transaction compared to \byzcuitbaseline, leading to lower concurrency and lower throughput. 

%Another interesting observation is that the design and implementation optimizations in \byzcuitreplay lead to significantly higher throughput than \chainspace, even though the former has lower concurrency due to the dummy objects. For similar experimental setup and for 2--10 shards, \chainspace achieves 70--180 tps, while \byzcuitreplay achieves 260--1550 tps. This is due to the improved design of the cross-shard consensus protocol (Section~\ref{sec:byzcuit:byzcuit}), which results in communication complexity of $O(n)$ in contrast to \chainspace's $O(n^2)$ (where $n$ is the number of input shards). Another reason for \byzcuitreplay's significant throughput improvement is that unlike \chainspace, all interactions between the Transaction Manager and the shards are asynchronous. This eliminates the blocking condition in \chainspace where a shard cannot commit a transaction in the second phase of the cross-shard consensus protocol, until it receives messages from all shards corresponding to the first phase.

\mypara{The effect of dummy objects on throughput}
We previously observed that dummy objects reduce the throughput of \byzcuitreplay with respect to \byzcuitbaseline. \Cref{fig:byzcuit:tpsVSdummy} shows the extent of throughput degradation due to dummy objects. We submit specially crafted transactions to 6 shards, such that each transaction has 1 input object, and we vary the number of dummy objects from 1--5 selected from unique shards, resulting in a corresponding decrease in concurrency because as many shards end up processing the transaction. For example, 2 dummy objects means that 3 shards process the transaction (1 input shard, and 2 more shards corresponding to the dummy objects).
As expected, the throughput decreases by 20--250 tps with the addition of each dummy object, and reaches 750 tps when all 6 shards handle all transactions (which is the worst-case scenario).

\mypara{Client-perceived latency}
\Cref{fig:byzcuit:latencyVStps} shows the client-perceived latency---the time from when a client submits a transaction, until it receives a decision from \byzcuit about whether the transaction has been committed---under varying system loads (expressed as transactions submitted to \byzcuit per second). We submit a total of 1,200 transactions at 200--1,000 transactions per second to \byzcuit with 6 shards. Each transaction has 2 inputs objects and 5 output objects, both chosen randomly from shards. When the system is experiencing a load of up to 1,000 tps, clients hear back about their transactions in less than a second on average, even with our replay defenses in place.

%% file: chapters/byzcuit/sections/conclusion.tex
% =========
% Conclusion
% =========
\section{Chapter Summary} \label{sec:byzcuit:conclusion}
This chapter presented the first replay attacks against cross-shard consensus protocols in sharded distributed ledgers. These attacks affect both shard-driven and client-driven consensus protocols, and allow attackers to double-spend objects with minimal effort. The attacker can act independently without colluding with any nodes, and succeed even if all nodes are honest. %most of the attacks work without making any assumptions on the underlying network. 
While addressing these attacks seems like an implementation detail, their many variants illustrate that a fundamental re-think of cross-shard commit protocols is required to protect against them.
We developed \byzcuit, a new cross-shard consensus protocol merging features from \shardled and \clientled consensus protocols, and withstanding replay attacks. \byzcuit can be seen as unifying \atomix and \sbac, into a protocol that is efficient and secure. We implemented and evaluated it on a real cloud-based testbed, showing that it can process over 1,550 tps for 10 shards, and that its capacity scales linearly with the number of shards.
%is more efficient than \chainspace, and on par with \omniledger performance. The resulting protocol is a drop-in replacement for either, and can be adopted to immunize systems based on those designs.

%% file: chapters/fastpay/fastpay.tex
\chapter{\fastpay: High-Performance Byzantine Fault Tolerant Settlement} 
\label{fastpay}

%\input{chapters/fastpay/sections/abstract.tex}
\input{chapters/fastpay/sections/introduction.tex}
\input{chapters/fastpay/sections/background.tex}
\input{chapters/fastpay/sections/overview.tex}

\input{chapters/fastpay/sections/design.tex}
\input{chapters/fastpay/sections/security.tex}
\input{chapters/fastpay/sections/implementation.tex}
\input{chapters/fastpay/sections/evaluation.tex}

\input{chapters/fastpay/sections/discussion.tex}
\input{chapters/fastpay/sections/related.tex}

\input{chapters/fastpay/sections/conclusion.tex}

%% file: chapters/fastpay/sections/introduction.tex
% ======
% Introduction
% ======

%
\begin{figure}[t]
\centering
\includegraphics[width=.7\textwidth]{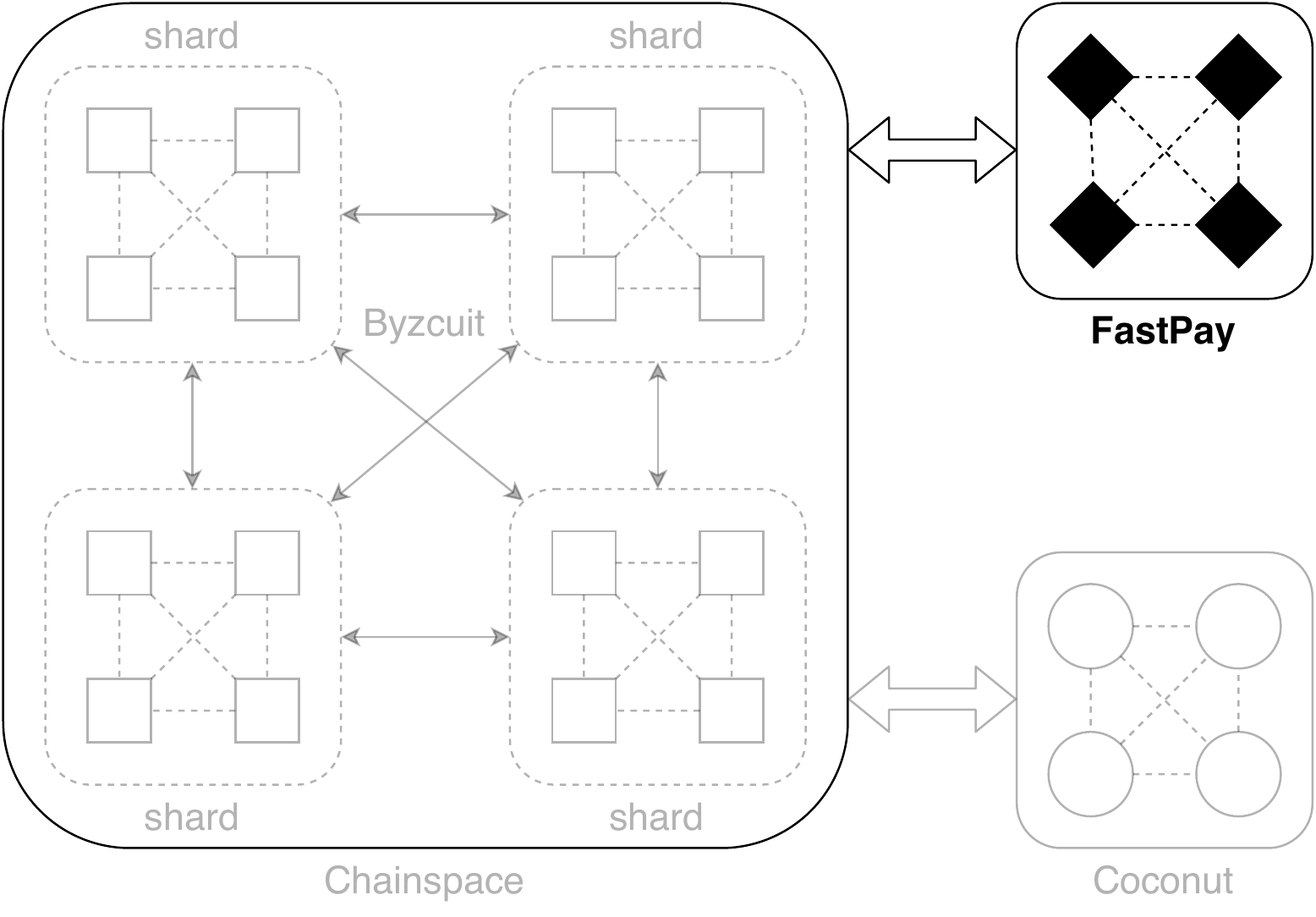}
\caption[Global overview: \fastpay.]{Example of instantiation of \fastpay with four authorities. Authorities are represented by black diamonds and the dashed lines connecting them to each other indicate that they are part of the same committee.}
\label{fig:fastpay:global-overview}
\end{figure}

Chapters~\ref{chainspace} and~\ref{byzcuit} demonstrate that sharded blockchains can be solid backbone systems that can scale to accommodate high throughput. However, their latency bottleneck remains their underlying consensus protocol which is too high for practical retail payments at physical points of sale. \Cref{byzcuit} shows that \byzcuit can process transactions in a few seconds but retail payments often require sub-second latency.
This chapter overcomes this limitation through a side-infrastructure called \fastpay; \Cref{fig:fastpay:global-overview} places \fastpay in the big picture of this thesis. \fastpay is a distributed settlement system for pre-funded payments that can be used as a financial side-infrastructure to support low-latency retail payments of a primary system such as \chainspace; it achieves extremely low latency by foregoing the expenses of consensus. The black diamonds in \Cref{fig:fastpay:global-overview} represent \fastpay nodes and the white bidirectional arrow illustrates the communication between \fastpay and \chainspace allowing users to move funds across those systems.

Settlement systems and Real-Time Gross Settlement Systems (RTGS)~\cite{rtgs-2} constitute the most common approach to financial payments in closed banking networks, that is, between reputable institutions. In contrast, blockchain platforms have proposed a radically different paradigm, allowing account holders to interact directly with an online, yet highly secure, distributed ledger. Blockchain approaches aim to enable new use cases such as personal e-wallets or private transactions, and generally provide ecosystems more favorable to users. However, until now, such open, distributed settlement solutions have come at a high performance cost and questionable scalability compared to traditional, closed RTGS systems.

\fastpay is a Byzantine Fault Tolerant (BFT) real-time gross settlement (RTGS) system. It enables authorities to jointly maintain account balances and settle pre-funded retail payments between accounts. It supports extremely low-latency confirmation (sub-second) of eventual transaction finality, appropriate for physical point-of-sale payments. It also provides extremely high capacity, comparable with peak retail card network volumes, while ensuring gross settlement in real-time. \fastpay eliminates counter-party and credit risks of net settlement and removes the need for intermediate banks, and complex financial contracts between them, to absorb these risks. \fastpay can accommodate arbitrary capacities through efficient sharding architectures at each authority. Unlike any traditional RTGS, and more like permissioned blockchains, \fastpay can tolerate up to $f$ Byzantine failures out of a total of $3f+1$ authorities, and retain both safety, liveness, and high-performance.
\fastpay can be deployed in a number of settings. First, it may be used as a settlement layer for a native token and cryptocurrency, in a standalone fashion. Second, it may be deployed as a side-chain of another cryptocurrency, or as a high performance settlement layer on the side of an established RTGS to settle fiat retail payments. In this chapter we present this second functionality in detail, since it exercises all features of the system, both payments between \fastpay accounts, as well as payments into and out of the system. 

\section*{Contributions} This chapter makes the following contributions:
\begin{itemize}
\setlength\itemsep{0em}
    \item The \fastpay design is novel in that if forgoes full consensus; it leverages the semantics of payments to minimize shared state between accounts and to increase the concurrency of asynchronous operations; and it supports sharded authorities.
    \item We provide proofs of safety and liveness in a Byzantine and fully asynchronous network setting. We show that \fastpay keeps all its properties despite the lack of total ordering, or asynchrony of updates to recipient accounts. 
    \item We experimentally demonstrate comparatively very high throughput and low latency, as well as the robustness of the system under conditions of extremely high concurrency and load. We show that performance is maintained even when some (Byzantine) authorities fail.
\end{itemize}

\section*{Outline} This chapter is organized as follows:
\Cref{sec:fastpay:overview} introduces the entities within \fastpay, their interactions, and the security properties and threat model.
\Cref{sec:fastpay:design} details the design of \fastpay both as a standalone system, and operated in conjunction with a \sysmain. 
\Cref{sec:fastpay:security} discusses safety and liveness.
\Cref{sec:fastpay:implementation} briefly describes the implementation of the \fastpay prototype.
\Cref{sec:fastpay:evaluation} provides a full performance evaluation of \fastpay as we modulate its security parameters and load. 
\Cref{sec:fastpay:discussion} discusses key open issues such as privacy, governance mechanisms and economics of the platform.
\Cref{sec:fastpay:related} covers the related work, both in terms of traditional financial systems and crypto-currencies.
\Cref{sec:fastpay:conclusion} concludes.

%% file: chapters/fastpay/sections/background.tex
% ======
% Background
% ======
\section{Background} \label{sec:fastpay:background}
Real-time gross settlement systems (RTGS)~\cite{rtgs-2} are the backbone of modern financial systems. Commercial banks use them to maintain an account with central banks and settle large value payments. 
RTGS systems are limited in their capacity\footnote{For example, the relatively recent European Central Bank TARGET2 system has a maximum capacity of 500 transactions per second~\cite{target2}}, making them unsuitable for settling low-value high-volume retail payments directly. Such retail payments are deferred: banks exchange information bilaterally about pending payments (often through SWIFT~\cite{swift, swift-iso}), they aggregate and net payments, and only settle net balances through an RTGS, often daily. The often quoted volume figure of around 80,000 transactions per second for retail card networks~\cite{mastercard-performance,visa-performance} represents the rate at which `promises' for payments are exchanged between banks, and not settled payments.
Traditional RTGS systems are implemented as monolithic centralized services operated by a single authority, and must employ a number of technical and organizational internal controls to ensure they are reliable (through a primary-replica architecture with manual switch over) and correct---namely ensuring availability and integrity. Traditionally only regulated entities have accounts in those systems. This result in a Balkanized global financial system where financial institutions connect to multiple RTGS, directly or indirectly through corresponding banks, to execute international payments.

As mentioned in \Cref{literature-review}, blockchain-based technologies, starting with Bitcoin~\cite{bitcoin} in 2009, provide more open settlement systems often combined with their own digital tokens to represent value. Permissionless blockchains have been criticized for their low performance~\cite{croman2016scaling} in terms of capacity and finality times. However, a comparison with established settlement systems leads to a more nuanced assessment. Currently, Ethereum~\cite{ethereum} can process 15 transactions per second. The actual average daily load on the EU ECB TARGET2 system is about 10 transactions per second~\cite{target2} (in 2018) which is a comparable figure (and lower than the peak advertised capacity of 500 transaction per second). However, it falls very short of the advertised transaction rate of 80,000 transaction per second peak for retail payment networks---even though this figure does not represents settled transactions. The stated ambitions of permissionless projects is to be able to settle transactions at this rate on an open and permissionless network, which remains an open research challenge~\cite{vukolic2015quest}. 
Permissioned blockchains~\cite{permissioned-blockchains, sok-consensus} provide a degree of decentralization---allowing multiple authorities to jointly operate a ledger---at the cost of some off-chain governance to control who maintains the blockchain. The most prominent of such proposals is the Libra network~\cite{libra}, developed by the Libra association. Other technical efforts include \hyperledger~\cite{hyperledger}, Corda~\cite{corda} and \tendermint~\cite{tendermint}.  Those systems are based on traditional notions of Byzantine Fault Tolerant state machine replication (or sometimes consensus with crash-failures only), which presupposes an atomic commit channel (often referred as `consensus') that sequences all transactions. Such architectures allow for higher capacities than Bitcoin and \ethereum. LibraBFT~\cite{libraBFT}, for example, aims for 1,000 transactions per second at peak capacity~\cite{libraTPS-1, libraTPS-2}; this exceeds many RTGS systems, but is still below the peak volumes for retail payment systems. A latency of multiple seconds when it comes to transaction finality is competitive with RTGS, but is unusable for retail payment at physical points of sale. 

%% file: chapters/fastpay/sections/overview.tex
% ======
% Overview
% ======
\section{Overview} \label{sec:fastpay:overview}
To illustrate its full capabilities, we describe \fastpay as a side chain of a primary system holding the primary records of accounts. We call such a primary ledger \emph{the \sysmain} for short, and its accounts \emph{the \sysmain accounts}. The \sysmain can be instantiated as a programmable blockchain, through smart contracts, like \chainspace.

%The \sysmain can be instantiated in two ways: \first as a programmable blockchain, through smart contracts, like Ethereum~\cite{ethereum} or \libra~\cite{libra}. The \sysmain can also be instantiated \second as a traditional monolithic RTGS operated by a central bank. In this case the components interfacing with \fastpay are implemented as database transactions within the \sysmain. In both cases \fastpay acts as a side infrastructure to enable pre-funded retail payments. \fastpay can also operate with a native asset, without a primary ledger. In this case sub-protocols involving the \sysmain are superfluous, since all value is held within \fastpay accounts and never transferred out or into the system. 
%In this configuration a central bank could directly issue some of the monetary supply on \fastpay to enable retail transactions in central bank money.

% ======
\subsection{Participants} \label{sec:fastpay:participants}
\fastpay involves two types of participants: \first authorities, and \second account owners (\emph{users}, for short). All participants generate a key pair consisting of a private signature key and the corresponding public verification key.
As a side-chain, \fastpay requires a smart contract on the main blockchain, or a software component on an RTGS system that can authorize payments based on the signatures of a threshold of authorities from a \emph{committee} with fixed membership.

By definition, an \emph{honest} authority always follows the \fastpay protocol, while a \emph{faulty} (or \emph{Byzantine}) one may deviate arbitrarily.
We present the \fastpay protocol for $3f+1$ equally-trusted authorities, assuming a fixed (but unknown) subset of at most $f$ Byzantine authorities. In this setting, a \emph{quorum} is defined as any subset of $2f+1$ authorities. (As for many BFT protocols, our proofs only use the classical properties of quorums thus apply to all Byzantine quorum systems~\cite{malkhi1998byzantine}.)
When a protocol message is signed by a quorum of authorities, it is said to be \emph{certified}: we call such a jointly signed message a \emph{certificate}.

% ======
\subsection{Accounts \& Actions} \label{sec:fastpay:account-actions}

A \fastpay \emph{account} is identified by its \emph{address}, which we instantiate as the cryptographic hash of its public verification key. The state of a \fastpay account is affected by four main high-level actions:
\begin{enumerate}
    \item Receiving funds from a \sysmain account.
    \item Transferring funds to a \sysmain account.
    \item Receiving funds from a \fastpay account.
    \item Transferring funds to a \fastpay account.
\end{enumerate}
Actions (3) and (4) are necessary to execute payments within \fastpay.
Actions (1) and (2) are required only when interfacing with a primary system.
\fastpay also supports two read-only actions that are necessary to ensure liveness despite faults: reading the state of an account at a \fastpay authority, and obtaining a certificate for any action executed by an authority.

% ======
\subsection{Protocol Messages} \label{sec:fastpay:protocol-messages}
The \fastpay protocol consists of \emph{transactions} on the \sysmain, denoted with letter $T$, and network requests that users send to \fastpay authorities, which we call \emph{orders} and denote with letter $O$. Users are responsible for broadcasting their orders to authorities and for processing the corresponding responses. The authorities are passive and \emph{do not communicate directly with each other}.

\para{Transfer orders} All transfers initiated by a \fastpay account start with a \emph{transfer order} $O$ including the following fields:
\begin{itemize}
    \item The sender's \fastpay address, written $\sender(O)$.
    \item The recipient, either a \fastpay or a \sysmain address, written $\recipient(O)$.
    \item A non-negative amount to transfer, written $\amount(O)$.
    \item A sequence number $\sequencenumber(O)$.
    \item Optional user-provided data.
    \item A signature by the sender over the above data.
\end{itemize}

Authorities respond to valid transfer orders by counter-signing them (see \Cref{sec:fastpay:design} for validity checks). A quorum of such signatures is meant to be aggregated into a \emph{transfer certificate}, noted $C$.

\para{Notations} We write $O = \val(C)$ for the original transfer order $O$ certified by $C$. For simplicity, we omit the operator $\val$ when the meaning is clear, \eg $\sender(C) = \sender(\val(C))$.
\fastpay addresses are denoted with letters $a$ and $b$. We use $\alpha$ for authorities and by extension for the shards of authorities.

% ======
\subsection{Security Properties \& Threat Model} \label{sec:fastpay:properties}
\fastpay guarantees the following security properties:
\begin{itemize}
    \item \textbf{Safety:} No units of value are ever created or destroyed; they are only transferred between accounts.
    \item \textbf{Authenticity:} Only the owner of an account may transfer value out of the account.
    \item \textbf{Availability:} Correct users can always transfer funds from their account.
    \item \textbf{Redeemability:} A transfer to \fastpay or \sysmain is guaranteed to eventually succeed whenever a valid transfer certificate has already been produced.
    \item \textbf{Public Auditability:} There is sufficient public cryptographic evidence for the state of \fastpay to be audited for correctness by any party. 
    \item \textbf{Worst-case Efficiency:} Byzantine authorities (or users) cannot significantly delay operations from correct users.
\end{itemize}

The above properties are maintained under a number of security assumptions: \first~there are at most $f$ Byzantine authorities out of $3f+1$ total authorities.
\second~The network is fully asynchronous, and the adversary may arbitrarily delay and reorder messages~\cite{dwork1988consensus}. However, messages are eventually delivered. \third~Users may behave arbitrarily but availability only holds for \emph{correct users} (defined in \Cref{sec:fastpay:clients}). \fourth~The \sysmain provides safety and liveness (when \fastpay is used in conjunction with it).
We further discuss the security properties of \fastpay in \Cref{sec:fastpay:security}.

%% file: chapters/fastpay/sections/design.tex
% ======
% Design
% ======
\section{The \fastpay Protocol} \label{sec:fastpay:design}
\fastpay authorities hold and persist the following information.

\para{Authorities}
The state of an authority $\alpha$ consists of the following information:
\begin{itemize}
    \item The authority name, signature and verification keys.
    \item The committee, represented as a set of authorities and their verification keys.
    \item A map $\accounts(\alpha)$ tracking the current account state of each \fastpay address $a$ in use (see below).
    \item An integer value, noted $\lasttransactionindex(\alpha)$, referring to the last transaction that paid funds into the \sysmain. This is used by authorities to synchronize \fastpay accounts with funds from the \sysmain (see Section~\ref{sec:fastpay:libra-to-fastpay}).
\end{itemize}

\para{\fastpay accounts}
The state of a \fastpay account $a$ within the authority $\alpha$ consists of the following:
\begin{itemize}
    \item The public verification key of $a$, used to authenticate spending actions.
    \item An integer value representing the balance of payment, written $\balance^a(\alpha)$.
    \item An integer value tracking the expected sequence number for the next spending action to be created, written $\nextsequencenumber^a(\alpha)$. This value starts at $0$ thus can be seen as the number of spending actions ever confirmed for this account.
    \item A field $\pendingconfirmation^a(\alpha)$ tracking the last transfer order $O$ signed by $a$ such that the authority $\alpha$ considers $O$ as \emph{pending confirmation}, if any; and absent otherwise.
    \item A list of certificates, written $\confirmedlog^a(\alpha)$, tracking all the transfer certificates $C$ that have been \emph{confirmed} by $\alpha$ and such that $\sender(C) = a$. One such certificate is available for each sequence number $n$ ($0 \leq n < \nextsequencenumber^a(\alpha)$).
    \item A list of \emph{synchronization} orders, $\synchronizationlog^a(\alpha)$, having transferred funds from the \sysmain to account $a$. (See Section~\ref{sec:fastpay:libra-to-fastpay}.)
\end{itemize}
We also define $\received^a(\alpha)$ as the list of confirmed certificates for transfers received by $a$; formally: 
\[
\received^a(\alpha) = \{C \text{ s.t. } \exists b.\; C \in \confirmedlog^b(\alpha) \text{ and } \recipient(C) = a \}
\]
 
We assume arbitrary size integers. Although \fastpay does not let users overspend, (temporary) negative balances for account states are allowed for technical reasons discussed in \Cref{sec:fastpay:security}. When present, a pending (signed) transfer order $O = \pendingconfirmation^a(\alpha)$ effectively locks the sequence number of the account $a$ and prevents $\alpha$ from accepting new transfer orders from $a$ until \emph{confirmation}, that is, until a valid transfer certificate $C$ such that $\val(C) = O$ is received. This mechanism can be seen as the `Signed Echo Broadcast' implementation of a Byzantine consistent broadcast on the label (account, next sequence number)~\cite{cachinBook}.

\para{Storage considerations} 
The information contained in the lists of certificates $\confirmedlog^a(\alpha)$ and $\received^a(\alpha)$ and in the synchronization orders $\synchronizationlog^a(\alpha)$ is self-authenticated---being signed by a quorum of authorities and by the \sysmain, respectively. Remarkably, this means that authorities may safely outsource these lists to an external high-availability data store. Therefore, \fastpay authorities only require a constant amount of local storage per account, rather than a linear amount in the number of transactions.

% ======
\subsection{Transferring Funds within \fastpay} \label{sec:fastpay:fastpay-to-*}

\fastpay operates by implementing a Byzantine consistent broadcast channel per account, specifically using a `Signed Echo Broadcast' variant (Algorithm 3.17 in~\cite{cachinBook}). It operates in two phases and all messages are relayed by the initiating user. Consistent Broadcast ensures \emph{Validity}, \emph{No duplication}, \emph{Integrity}, and \emph{Consistency}. It always terminates when initiated by a correct user. However, if a \fastpay user equivocates, current operations may fail, and the funds present on the users account may become inaccessible.

\begin{figure}[t]
\centering
\includegraphics[width=.7\textwidth]{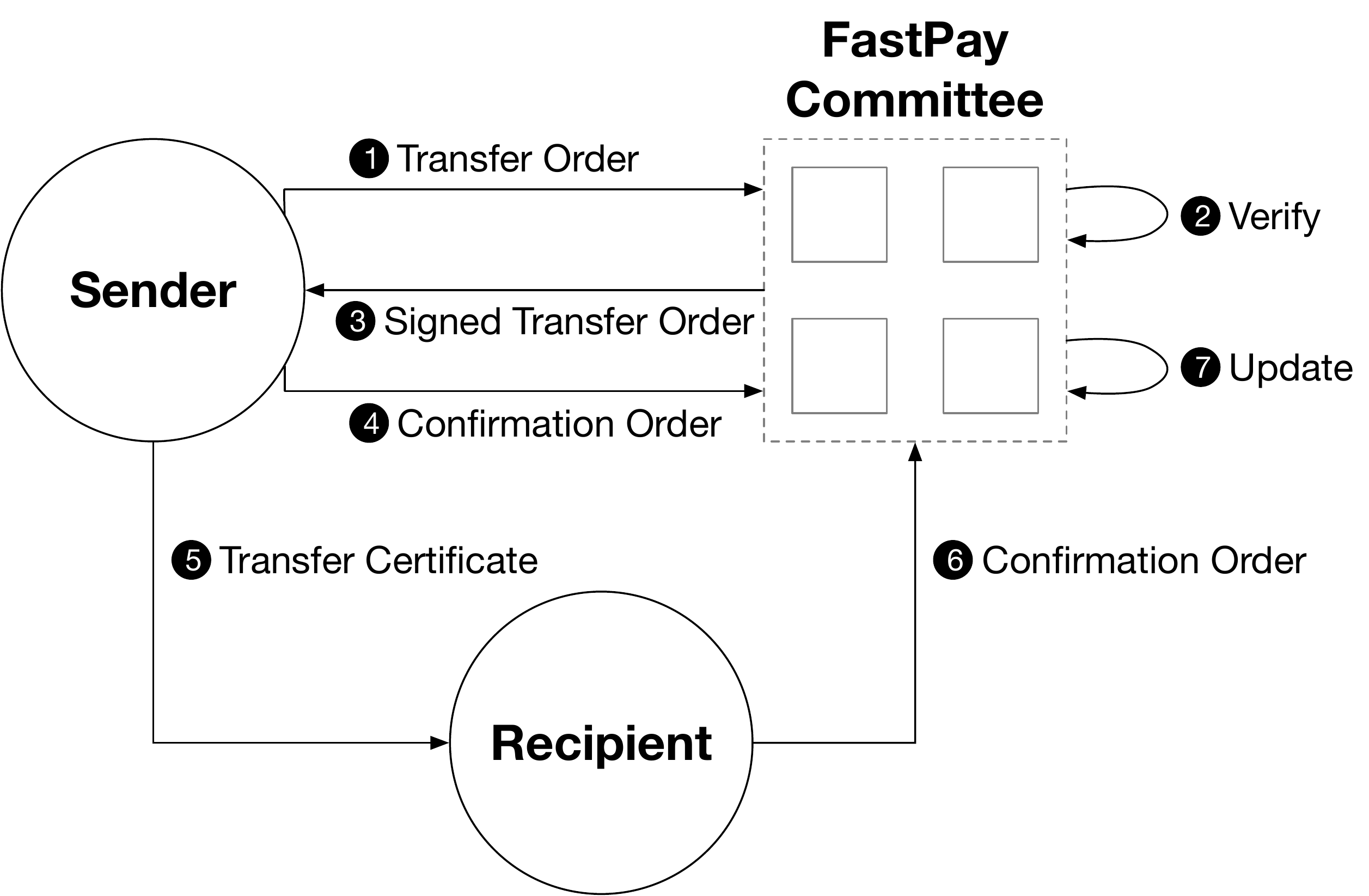}
\caption[Transfer of funds from \fastpay to \fastpay.]{Transfer of funds from \fastpay to \fastpay. The sender first generates a \emph{transfer order} that is counter-signed by each authority; a \emph{confirmation order} then finalizes the transaction.}
\label{fig:fastpay:fastpay-to-fastpay}
\end{figure}

\para{Transferring funds}
\Cref{fig:fastpay:fastpay-to-fastpay} illustrates a transfer of funds within \fastpay.
To transfers funds to another \fastpay account, the sender creates a \fastpay transfer order ($\transfer$) with the next sequence number in their account, and signs it. They then send the \fastpay transfer order to all authorities. Each authority checks (\ding{202}) \first that the signature is valid, \second that no previous transfer is pending (or its for the same transfer), \third that the amount is positive, \fourth that the sequence number matches the expected next one ($\sequencenumber(O) = \nextsequencenumber^a(\authority)$), and \fourth that the balance ($\balance^a(\alpha)$) is sufficient (\ding{203}). Then, it records the new transfer as pending and sends back a signature on the transfer order (\ding{204}) which is also stored. The authority algorithm to handle transfer orders, corresponding to step \ding{203}, is presented in \Cref{fig:fastpay:code}.
The user collects the signatures from a quorum of authorities, and use them along the \fastpay transfer order to form a transfer certificate.
The sender provides this transfer certificate to the recipient as a proof that the payment will proceed (\ding{206}).
To conclude the transaction, the sender (\ding{205}) or the recipient (\ding{207}) must broadcast the transfer certificate ($\cert$) to the authorities (we call this a \emph{confirmation order})\footnote{Aggregating signed transfer orders into a transfer certificate does not requires knowledge of any secret; therefore, anyone (and not only the sender or the recipient) can broadcast the transfer certificate to the authorities to conclude the transaction.}.
Upon reception of a confirmation order for the current sequence number, each authority $\alpha$ (\ding{208}) \first checks that a quorum of signatures was reached, \second decreases the balance of the sender, \third increments the sequence number ($\nextsequencenumber^a(\alpha) + 1$) to ensure `deliver once' semantics, and \fourth sets the pending order to None ($\pendingconfirmation^a(\alpha)=\text{None}$). Each authority $\alpha$ also \fifth adds the certificate to the list $\confirmedlog^a(\alpha)$, and \sixth increases the balance of the recipient account asynchronously (\ie without sequencing this write in relation to any specific payments from this account across authorities). The authority algorithm to handle confirmation orders (as in step \ding{208}) is presented in \Cref{fig:fastpay:code}.

In \Cref{sec:fastpay:security}, we show that the \fastpay protocol is safe thanks to the semantics of payments into an account and their commutative properties. \fastpay is a significant simplification and deviation from an orthodox application of Guerraroui~\etal~\cite{consensus-number}, where accounts are single-writer objects and all write actions are mediated by the account owner. \fastpay allows payments to be executed after a single consistent broadcast, rather than requiring recipients to sequence payments into their accounts separately. This reduces both latency and the state necessary to prevent replays.

\begin{figure}
\begin{lstlisting}[mathescape=true]
fn handle_transfer_order($\authority$, $\transfer$) -> Result {
    /// Check shard and signature.
    ensure!($\authority$.in_shard($\sender(\transfer)$));
    ensure!($\transfer$.has_valid_signature());

    /// Obtain sender account.
    match $\accounts(\authority)$.get($\sender(\transfer)$) {
        None => bail!(),
        Some(account) => {
            /// Check if the same order is already pending.
            if let Some(pending) = account.pending {
                ensure!(pending.transfer == $\transfer$);
                return Ok();
            }
            ensure!(account.next_sequence == $\sequencenumber(\transfer)$);
            ensure!(account.balance >= $\amount(\transfer)$);
            /// Sign and store new transfer.
            account.pending = Some($\authority$.sign($\transfer$));
            return Ok();
}   }   }

fn handle_confirmation_order($\authority$, $\cert$)
    -> Result<Option<CrossShardUpdate>> {
    /// Check shard and certificate.
    ensure!($\authority$.in_shard($\sender(\cert)$));
    ensure!($\cert$.is_valid($\authority$.committee));
    let $\transfer$ = $\val(\cert)$;

    /// Obtain sender account.
    let sender_account =
        $\accounts(\authority)$.get($\sender(\transfer)$)
        .or_insert(AccountState::new());

    /// Ignore old certificates.
    if sender_account.next_sequence > $\sequencenumber(\transfer)$ {
        return Ok(None);
    }

    /// Check sequence number and balance.
    ensure!(sender_account.next_sequence == $\sequencenumber(\transfer)$);
    ensure!(sender_account.balance >= $\amount(\transfer)$);

    /// Update sender account.
    sender_account.balance -= $\amount(\transfer)$;
    sender_account.next_sequence += 1;
    sender_account.pending = None;
    sender_account.confirmed.push($\cert$);

    /// Update recipient locally or cross-shard.
    let recipient = match $\recipient(\transfer)$ {
        Address::FastPay(recipient) => recipient,
        Address::Primary(_) => { return Ok(None) }
    };

    /// Same shard: read and update the recipient.
    if $\authority$.in_shard(recipient) {
        let recipient_account = $\accounts(\authority)$.get(recipient)
            .or_insert(AccountState::new());
        recipient_account.balance += $\amount(\transfer)$;
        return Ok(None);
    }

    /// Other shard: request a cross-shard update.
    let update = CrossShardUpdate {
        shard_id: $\authority$.which_shard(recipient),
        transfer_certificate: $\cert$,
    };
    Ok(Some(update))
}
\end{lstlisting}
\vspace{-1em}
\caption[\fastpay algorithms for handling transfer and confirmation orders.]{Authority algorithms for handling transfer and confirmation orders. (The cross-shard update logic is in presented in \Cref{fig:fastpay:code-core}.)}
\label{fig:fastpay:code}
\vspace{-3em} % Prevents page overflow
\end{figure}

\begin{figure}
\begin{lstlisting}[mathescape=true]
fn handle_cross_shard_commit($\authority$, $\cert$) -> Result {
    let $\transfer$ = $\val(\cert)$;
    let recipient = match $\recipient(\transfer)$ {
        Address::FastPay(recipient) => recipient,
        Address::Primary(_) => { bail!(); }
    };
    ensure!($\authority$.in_shard(recipient));
    let recipient_account = $\accounts(\authority)$.get(recipient)
        .or_insert(AccountState::new());
    recipient_account.balance += $\amount(\transfer)$;
    Ok()
}

fn handle_primary_synchronization_order($\authority$, $\sync$) -> Result {
    /// Update $\recipient(S)$ assuming that $\sync$ comes from
    /// a trusted source (e.g. Primary client).
    let recipient = $\recipient(\sync)$;
    ensure!($\authority$.in_shard(recipient));

    if $\transactionindex(\sync)$ <= $\lasttransactionindex(\authority)$ {
        /// Ignore old synchronization orders.
        return Ok();
    }
    ensure!($\transactionindex(\sync)$ == $\lasttransactionindex(\authority)$ + 1);

    $\lasttransactionindex(\authority)$ += 1;
    let recipient_account = $\accounts(\authority)$.get(recipient)
        .or_insert(AccountState::new());
    recipient_account.balance += $\amount(\sync)$;
    recipient_account.synchronized.push($\sync$);
    Ok()
}
\end{lstlisting}
\vspace{-1em}
\caption[\fastpay core algorithms.]{Authority algorithms for cross-shard updates and (\sysmain) synchronization orders.}
\label{fig:fastpay:code-core}
\end{figure}

\para{Payment finality} 
Once a transfer certificate \emph{could} be formed, namely $2f+1$ authorities signed a transfer order, no other order can be processed for an account until the corresponding confirmation order is submitted. Technically, the payment is final: it cannot be canceled, and will proceed eventually. As a result, showing a transfer certificate to a recipient convinces them that the payment will proceed. We call showing a transfer certificate to a recipient a \emph{confirmation}, and then subsequently submitting the confirmation order, to move funds, \emph{settlement}. \emph{Confirmation} requires only a single round trip to a quorum of authorities resulting in very low-latency (see \Cref{sec:fastpay:evaluation}), and giving the system its name.

\para{Proxies, gateways and crash recovery} 
The protocols as presented involve the sender being on-line and mediating all communications.
However, the only action that the sender \emph{must} perform personally is forming a transfer order, requiring their signature. All subsequent operations, including sending the transfer order to the authorities, forming a certificate, and submitting a confirmation order can be securely off-loaded to a proxy trusted only for liveness. Alternatively, a transfer order may be given to a merchant (or payment gateway) that drives the protocol to conclusion. In fact, any party in possession of a signed transfer order may attempt to make a payment progress concurrently. And as long as the sender is correct the protocol will conclude (and if not may only lock the account of the faulty sender).

This provides significant deployment and implementation flexibility. A sender client may be implemented in hardware (in a NFC smart card) that only signs transfer orders. These are then provided to a gateway that drive the rest of the protocol. Once the transfer order is signed and handed over to the gateway the sender may go off-line or crash. Authorities may also attempt to complete the protocol upon receiving a valid transfer order.
Finally, the protocol recovers from user crash failures: anyone may request a transfer order that is partially confirmed from any authority, proceed to form a certificate, and submit a confirmation order to complete the protocol.

% ======
\subsection{Sharding authorities} \label{sec:fastpay:sharding}
\fastpay requires minimal state sharing between accounts, and allows for a very efficient sharding at each authority by account. The consistent broadcast channel is executed on a per-account basis. Therefore, the protocol does not require any state sharing between accounts (and shards) up to the point where a valid confirmation order has to be settled to transfer funds between \fastpay accounts.
On settlement, the sender account is decremented and the funds are deposited into the account of the recipient, requiring interaction between at most two shards (second algorithm of \Cref{fig:fastpay:code}).
Paying into an account can be performed asynchronously, and is an operation that cannot fail (if the account does not exist it is created on the spot).  Therefore, the shard managing the recipient account only needs to be notified of the confirmed payment through a reliable, deliver once, authenticated, point to point channel (that can be implemented using a message authentication code, inter-shard sequence number, re-transmission, and acknowledgments) from the sender shard. This is a greatly simplified variant of a two-phase commit protocol coordinated by the sender shard (for details see the Presume Nothing and Last Agent Commit optimizations~\cite{DBLP:journals/dpd/SamarasBCM95,DBLP:conf/vldb/LampsonL93}). Modifying the validity condition of the consistent broadcast to ensure the recipient account exists (or any other precondition on the recipient account) would require a full two-phase commit before an authority signs a transfer order, and can be implemented while still allowing for (slightly less) efficient sharding.
The algorithms in \cref{fig:fastpay:code} implement shading. An authority shard checks that the transfer order ($\transfer$) or certificate ($\cert$) is to be handled by a specific shard and otherwise rejects it without mutating its state. Handling confirmation orders depends on whether a recipient account is on the same shard. If so, the recipient account is updated locally. Otherwise, a \emph{cross shard message} is created for the recipient shard to update the account (see code in the Appendix for this operation).

The ability to shard each authority has profound implications: increasing the number of shards at each authority, increases the theoretical throughput linearly, while latency remains constant. Our experimental evaluation confirms this experimentally (see \Cref{sec:fastpay:evaluation}).

% ======
\subsection{Interfacing with the \sysmain} \label{sec:fastpay:onchain-state}
We describe the protocols required to couple \fastpay with the \sysmain, namely transferring funds from the \sysmain to a \fastpay account, and conversely from a \fastpay to a \sysmain account.
% Those protocols allow \fastpay to be used as a high-performance settlement layer for retail transactions on the side of any traditional RTGS or cryptocurrency.
We refer throughout to the logic on the \sysmain as a \emph{smart contract}, and the primary store of information as the \emph{blockchain}. 
A traditional RTGS would record this state and manage it in conventional ways using databases and stored procedures, rather than a blockchain and smart contracts. We write $\sigma$ for the state of the `blockchain' at a given time, and $\transactions(\sigma)$ for the set of \fastpay transactions $T$ already processed by the blockchain.

\para{Smart contract}
The smart contract mediating interactions with the \sysmain requires the following data to be persisted in the blockchain:
\begin{itemize}
\setlength\itemsep{0em}
    \item The \fastpay committee composition: a set of authority names and their verification keys.
    \item A map of accounts where each \fastpay address is mapped to its current \sysmain state (see below).
    \item The total balance of funds in the smart contract, written $\totalbalance(\sigma)$.
    \item The transaction index of the last transaction that added funds to the smart contract, written $\lasttransactionindex(\sigma)$.
\end{itemize}

\para{Accounts}
The \sysmain state of a \fastpay account $a$ consists of the set of sequence numbers of transfers already executed from this account to the \sysmain. This set is called the \emph{redeem log} of $a$ and written $\redeemlog^a(\sigma)$.

\para{Adding funds from the \sysmain to \fastpay} \label{sec:fastpay:libra-to-fastpay}
\Cref{fig:fastpay:libra-to-fastpay} shows a transfer of funds from the \sysmain to \fastpay. 
The owner of the \fastpay account (or anyone else) starts by sending a payment to the \fastpay smart contract using a \sysmain transaction~(\ding{202}). This transaction is called a \emph{funding transaction}, and includes the recipient \fastpay address for the funds and the amount of value to transfer.
When the \sysmain transaction is executed, the \fastpay smart contract generates a \emph{\sysmain event} that instructs authorities of a change in the state of the \fastpay smart contract. We assume that each authority runs a full \sysmain client to authenticate such events. For simplicity, we  model such an event as a \emph{(\sysmain) synchronization order}~(\ding{203}). The smart contract ensures this event and the synchronization order contain a unique, always increasing, sequential transaction index.
When receiving a synchronization order, each authority \first checks the transaction index follows the previously recorded one, \second increments the last transaction index in their global state, \third creates a new \fastpay account if needed, and \fourth increases the account balance of the target account by the amount of value specified~(\ding{204}). \Cref{fig:fastpay:code-core} presents the authority algorithm for handling funding transactions.

\begin{figure}[t]
\centering
\includegraphics[width=.7\textwidth]{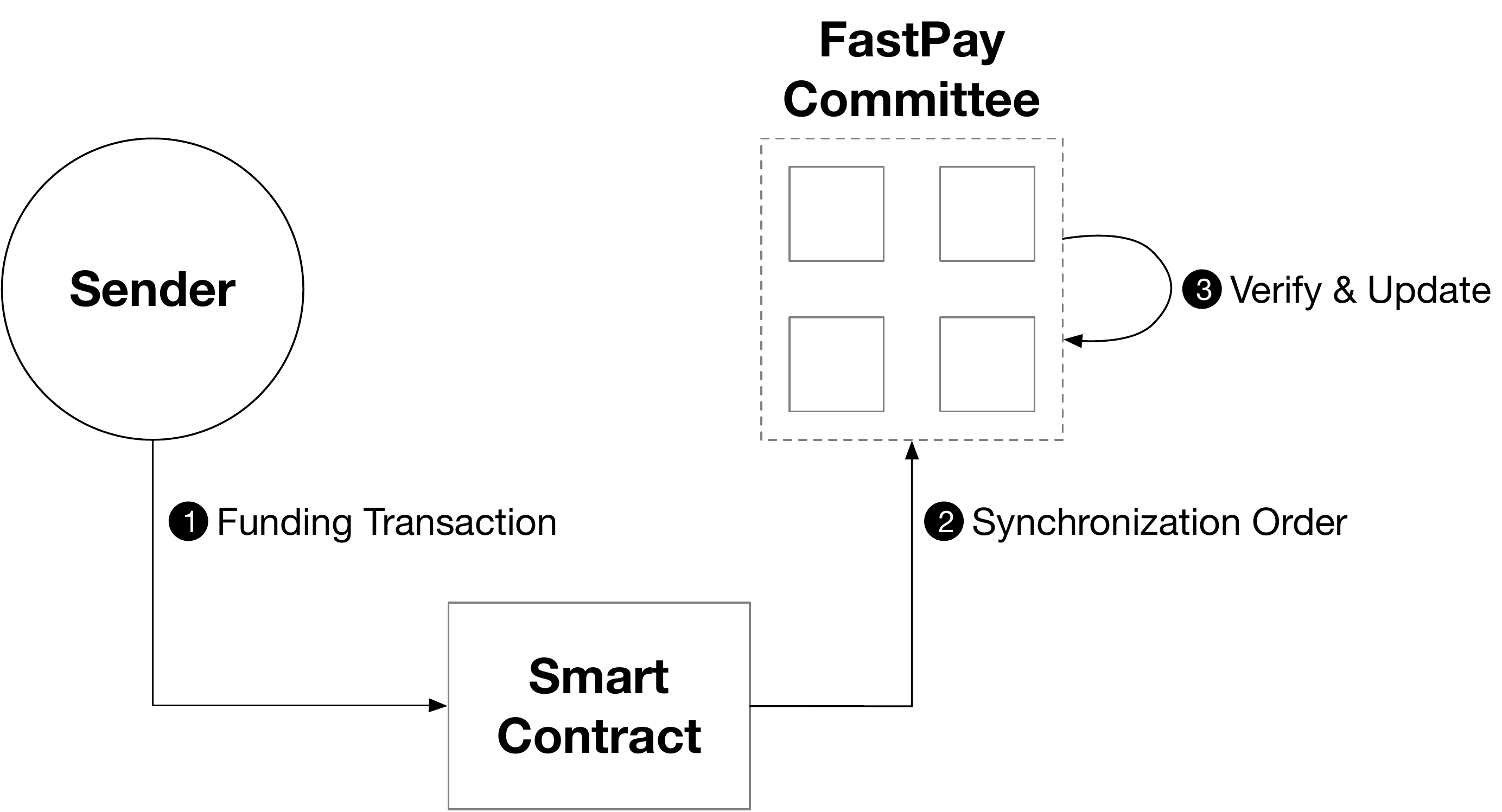}
\caption[Transfer of funds from the \sysmain to \fastpay.]{Transfer of funds from the \sysmain to \fastpay. The sender first deposits funds into a dedicated smart contract; the authorities periodically synchronize their internal state with the smart contract to reflect the new sender's balance.}
\label{fig:fastpay:libra-to-fastpay}
\end{figure}

\para{Transferring funds from \fastpay to the \sysmain}
\Cref{fig:fastpay:fastpay-to-libra} shows a transfer of funds from \fastpay to \sysmain. The \fastpay sender signs a \emph{\sysmain transfer order} using their account key and broadcasts it to the authorities (\ding{202}). This is simply a transfer order with a \sysmain address as the recipient.
Once a quorum of signatures is reached (\ding{203} and \ding{204}), the sender creates a certified (\sysmain) transfer order, also called a \emph{transfer certificate} for short.
The sender broadcasts this certificate to authorities to confirm the transaction (\ding{205}) and unlock future spending from this account. When an authority receives a confirmation order containing a certificate of transfer (\ding{206}), it must \first check that a quorum of signatures was reached, and \second that the account sequence number matches the expected one; \third they then set the pending order to None, \fourth increment the sequence number, and \fifth decrease the account balance.
Finally, the recipient of the transfer should send a redeem transaction to the \fastpay smart contract on the \sysmain blockchain (\ding{207}). When the \fastpay smart contract receives a valid redeem transaction (\ding{208}), it must \first check that the sequence number is not in the \sysmain redeem log of the sender, to prevent reuse; \second update this redeem log; \fourth transfer the amount of value specified from the smart contract into the recipient's \sysmain account.

\begin{figure}[t]
\centering
\includegraphics[width=.7\textwidth]{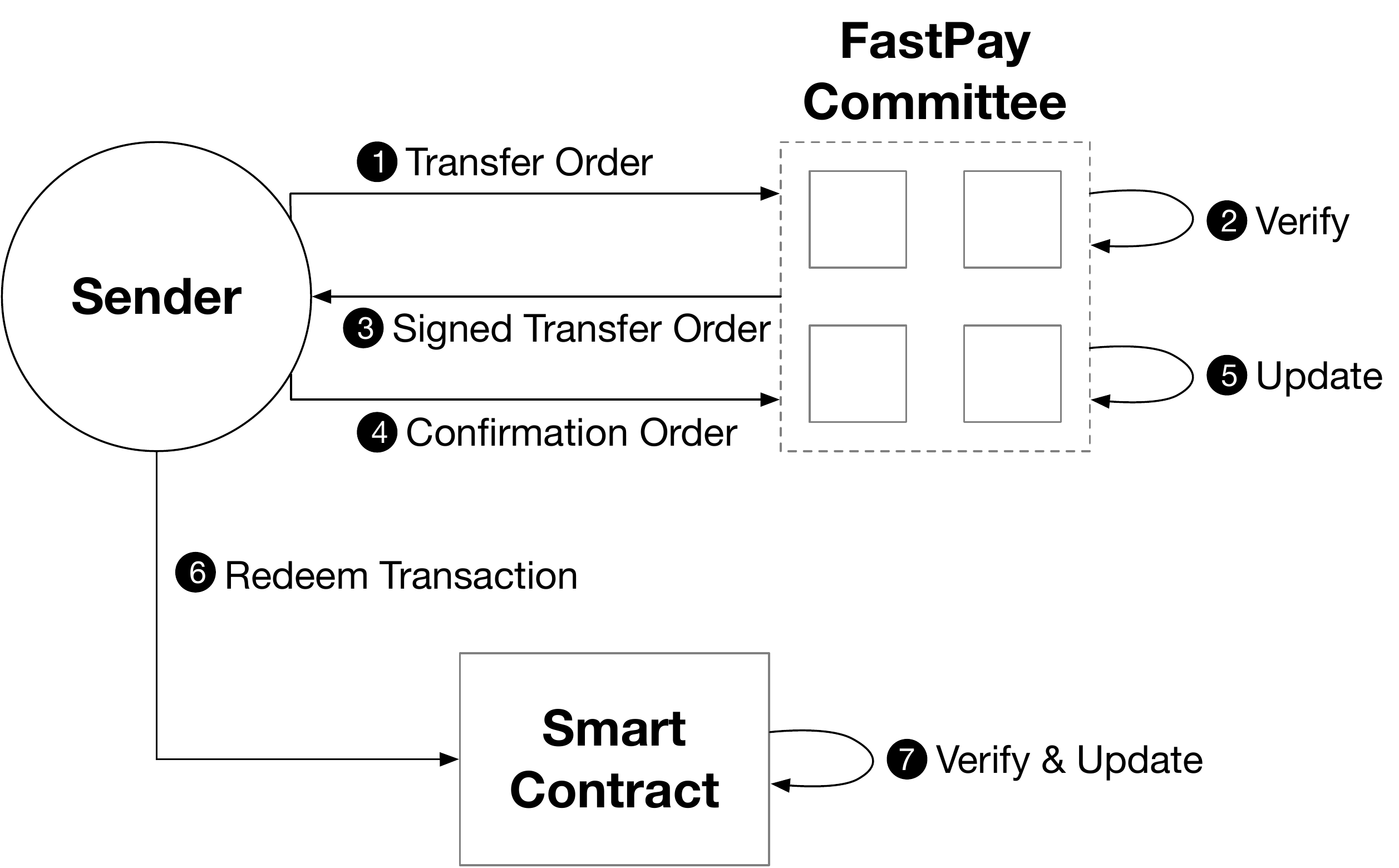}
\caption[Transfer of funds from \fastpay to the \sysmain.]{Transfer of funds from \fastpay to the \sysmain. The sender first obtains a \emph{confirmation order} from the authorities; it then submits this confirmation order to the smart contract through a \emph{redeem transaction} to unlock funds on the primary.}
\label{fig:fastpay:fastpay-to-libra}
\end{figure}
%

% ======
\subsection{State Recovery \& Auditing} \label{sec:fastpay:audits}
For every account $a$, each authority $\alpha$ must make available the pending order $\pendingconfirmation^a(\alpha)$, the sequence number $\nextsequencenumber^a(\alpha)$, the synchronization orders $\synchronizationlog^a(\alpha)$, and the certificates confirmed so far, indexed by senders (\ie $\confirmedlog^a(\alpha)$) and receivers ($\received^a(\alpha)$).
Sharing these data fulfills two important roles: \first this lets anyone read the state of any incomplete transfer and drive the protocol all the way to settlement; \second it enables auditing authority states and detect Byzantine faults (\eg incorrect balance checks).

% ======
\subsection{Correct Users \& Client Implementation} \label{sec:fastpay:clients}
A \emph{correct user} owning a \fastpay account $a$ follows the correctness rules below:
\begin{enumerate}
\item The user sets the sequence number of a new transfer order $O$ to be the next expected integer after the previous transfer (starting with $0$); \ie they sign exactly one transfer order per sequence number;
\item They broadcast the new transfer order $O$ to enough authorities until they (eventually) obtain a certificate $C$;
\item They successfully broadcast the certificate $C$ to a quorum of authorities.
\end{enumerate}

\para{\fastpay Client} To address the correctness rules above, our reference implementation of a \fastpay client holds and persists the following minimal state:
\begin{itemize}
\item The address $a$ and the secret key of the account;
\item The \fastpay committee;
\item The sequence number to be used in the next transfer;
\item The transfer order that it signed last, in case it is still pending.
\end{itemize}
In this setting, the available balance of a user account is not tracked explicitly but rather evaluated (conservatively) from the \sysmain transactions and the available logs for incoming transfers and outgoing transfers (\Cref{sec:fastpay:audits}). Evaluating the balance before starting a transfer is recommended, as signing a transfer order with an excessive amount will block (correct) client implementations from initiating further transfers until the desired amount is available.

%% file: chapters/fastpay/sections/security.tex
% ======
% Security
% ======
\section{Security Analysis} \label{sec:fastpay:security}
This section discusses safety and liveness of \fastpay.
Let $\sigma$ denote the current state of the \sysmain. We define $\funding^a(\sigma)$ as the sum of all the amounts transferred to a \fastpay address $a$ from the \sysmain:
\[
\funding^a(\sigma) \,=\, \sum_{\left\{\!\!\!\!\scriptsize\begin{array}{l} T \in \transactions(\sigma)\\ \recipient(T)=a \end{array}\right.} \amount(T)
\]
For simplicity, we write $\sum_C \amount(C)$ when we mean to sum over certified transfer orders: $\sum_{O \text { s.t. } \exists C. O = \val(C)} \amount(O)$.
We define $\funding^a(\alpha)$ as the sum of all the amounts received from the \sysmain by a \fastpay address $a$, as seen at a given time by an authority\ $\alpha$:
\[\funding^a(\alpha) = \sum_{\sync \in \synchronizationlog^a(\alpha)} \amount(\sync)\]

% ======
\subsection{Safety}
\begin{lemma}[Transfer certificate uniqueness] \label{le:fastpay:lemma-s0}
If
\[
\left\{
\begin{array}{l}
\sender(C) = \sender(C'), \; \text{and} \\
\sequencenumber(C) = \sequencenumber(C')
\end{array}
\right.
\]
then $C$ and $C'$ certify the same transfer order:
\[\val(C) = \val(C') \]
\end{lemma}
\begin{proof}
\small
Both certificates $C$ and $C'$ are signed by a quorum of authorities. By construction, any two quorums intersect on at least one honest authority. Let $\alpha$ be an honest authority in both quorums. $\alpha$ signs at most one transfer order per sequence number, thus $C$ and $C'$ certify the same transfer order.
\end{proof}

\begin{lemma}[\fastpay invariant] \label{le:fastpay:lemma-s1}
For every honest authority $\alpha$, for every account $a$, it holds that
\[
\begin{array}{l}
\bigg( \balance^a(\alpha) \;+\; \sum_{C \in \confirmedlog^a(\alpha)} \amount(C) \bigg)\\[0.6em]
\quad \leq \bigg( \funding^a(\alpha) \;+\; \sum_{C \in \received^a(\alpha)} \amount(C) \bigg)
\end{array}
\]
Besides, if $n = \nextsequencenumber^a(\alpha)$, we have that $\confirmedlog^a(\alpha) = \{C_0 \ldots C_{n-1}\}$ for some certificates $C_k$ such that $\sequencenumber(C_k) = k$ and $\sender(C_k)=a$.
\end{lemma}
\begin{proof}
\small
By construction of the \fastpay authorities (\Cref{fig:fastpay:code} and \Cref{fig:fastpay:code-core}):
Whenever a confirmed certificate $C$ is added to $\confirmedlog^a(\alpha)$, $\balance^a(\alpha)$ is decreased by $\amount(C)$, and $\nextsequencenumber^a(\alpha)$ is incremented into $\sequencenumber(C) + 1$.
Any new synchronization order equally increases $\balance^a(\alpha)$ and $\funding^a(\alpha)$.
Whenever a confirmed certificate $C$ is added to $\received^a(\alpha)$, possibly later (due to cross-shard updates), $\balance^a(\alpha)$ is increased once by $\amount(C)$.
\end{proof}

\begin{lemma}[\sysmain invariant] \label{le:fastpay:lemma-s2}
The total balance of all \fastpay accounts on \sysmain is such that
\[
\begin{array}{l}
\totalbalance(\sigma) = \bigg( \sum_a \funding^a(\sigma) - \\
\qquad \sum_{C \in \redeemlog(\sigma)}\amount(C) \bigg)
\end{array}
\]
\end{lemma}
\begin{proof}
\small
By construction of the smart contract handling funding and redeeming transactions (\Cref{sec:fastpay:libra-to-fastpay}): whenever a funding transaction $T$ is executed by the smart contract, both $\funding^a(\sigma)$ and $\totalbalance(\sigma)$ increase by $\amount(T)$. Conversely, $\totalbalance(\sigma)$ decreases by $\amount(C)$ whenever a \sysmain transfer certificate $C$ is redeemed on-chain and added to $\redeemlog(\sigma)$.
\end{proof}

\begin{lemma}[Funding log synchronization] \label{le:fastpay:lemma-s3}
For every honest authority $\alpha$ and every account $a$, it holds that
\[ \funding^a(\alpha) \leq \funding^a(\sigma) \]
\end{lemma}
\begin{proof}
\small
By definition of the synchronization with the \sysmain (\Cref{sec:fastpay:libra-to-fastpay}), and by security of the \sysmain and its client,  $\funding^a(\alpha)$ only increases after a funding transaction has already increased $\funding^a(\sigma)$ by the same amount.
\end{proof}

\begin{lemma}[Balance check] \label{le:fastpay:lemma-s4}
For every honest authority $\alpha$, whenever $O = \pendingconfirmation^a(\alpha)$ holds, then we also have \[\amount(O) \leq \balance^a(\alpha)\]
\end{lemma}
\begin{proof}
\small
By construction of the \fastpay authorities (\Cref{fig:fastpay:code}), if $O = \pendingconfirmation^a(\alpha)$, then $O$ was successfully processed by~$\alpha$ as a new transfer order from account $a$. At the time of the request, $\amount(O)$ did not exceed the current balance $B$. Since $O$ is still pending, in the meantime, no other transfer certificates from account $a$ have been confirmed by $\alpha$. (A confirmation would reset the field $\pendingconfirmation$ and prevent $O$ from being pending again due to the increased sequence number.) Therefore, the balance did not decrease, and $\balance^a(\alpha) \geq B \geq \amount(O)$.
\end{proof}

\begin{proposition}[Account safety] \label{pr:fastpay:proposition-s5}
For every account $a$, at any given time, we have that
\[
\begin{array}{l}
\sum_{\sender(C)=a} \amount(C) \;\leq\; \\
\\
\qquad \funding^a(\sigma) \;+\; \sum_{\recipient(C)=a} \amount(C)
\end{array}
\]
\end{proposition}
\begin{proof}
\small
Let $n$ be the highest sequence number of a transfer certificate $C_n$ from $a$.
Let $\alpha$ an honest authority whose signature is included in the certificate. At the time of signature,  we had $\val(C_n) = \pendingconfirmation^a(\alpha)$, therefore by \Cref{le:fastpay:lemma-s1} and \Cref{le:fastpay:lemma-s4}:
\[
\begin{array}{l}
\bigg( \amount(C_n) \;+\; \sum_{C\in\confirmedlog^a(\alpha)} \amount(C) \bigg) \\[0.6em]
\quad \leq \bigg( \funding^a(\alpha) \;+\; \sum_{C \in \received^a(\alpha)} \amount(C) \bigg)
\end{array}
\]
Given that $n$ is the highest sequence number, by \Cref{le:fastpay:lemma-s0} and \Cref{le:fastpay:lemma-s1}, the left-hand term exactly covers the certified transfer orders from $a$ and is equal to $\sum_{\sender(C)=a} \amount(C)$.

Given that amounts are non-negative, for every honest node $\alpha$, we have \[\sum_{C \in \received^a(\alpha)} \amount(C) \;\leq\; \sum_{\recipient(C)=a}\amount(C)\]

Finally, $\funding^a(\alpha) \leq \funding^a(\sigma)$ by \Cref{le:fastpay:lemma-s3}.
\end{proof}

\begin{theorem}[Solvency of \fastpay] \label{th:fastpay:theorem-s6} At any time, the sum of the amounts of all existing certified transfers from \fastpay to the \sysmain cannot exceed the funds collected by all transactions on the \sysmain smart contract:
\[
\sum_{\recipient(C) \in \sysmainaddresses} \!\!\amount(C)
\;\leq\; \sum_a \funding^a(\sigma)
\]
\end{theorem}
\begin{proof} 
\small
By applying \Cref{pr:fastpay:proposition-s5} on every account and summing, we obtain:
\[
\begin{array}{ll}
\sum_a \funding^a(\sigma) \geq \\[0.4em]
\qquad \bigg( \sum_{C} \amount(C) - \sum_{\recipient(C)\in \fastpayaddresses} \amount(C) \bigg) \\[0.6em]
\qquad = \; \sum_{\recipient(C) \in \sysmainaddresses} \amount(C)
\end{array}
\]
\end{proof}

% ======
\subsection{Liveness}
Next, we describe how receivers of valid transfer certificates can finalize transactions and make funds available on their own accounts (\sysmain and \fastpay).

\begin{proposition}[Redeemability of valid transfer certificates to \sysmain] \label{pr:fastpay:proposition-l1B}
A new valid \sysmain transfer certificate $C$ can always be redeemed by sending a new redeem transaction $T$ to the smart contract.
\end{proposition}
\begin{proof}
\small
\Cref{th:fastpay:theorem-s6} shows that the smart contact always has enough funding for all certified \sysmain transfer orders. The definition of $\redeemlog(\sigma)$ (\Cref{sec:fastpay:libra-to-fastpay}) thus ensures that any new certified \sysmain transfer order can be redeemed exactly once.
\end{proof}

\begin{proposition}[Redeemability of valid transfer certificates to \fastpay] \label{pr:fastpay:proposition-l1A}
Any user can eventually have a valid \fastpay transfer certificate $C$ confirmed by any honest authority. \end{proposition}
\begin{proof}
\small
If a certificate $C$ exists for account $a$ and sequence number $n$, this means at least $f+1$ honest authorities contributed signatures to the transfer order $O = \val(C)$.
By construction of \fastpay, these authorities have received (\Cref{fig:fastpay:code}) and will keep available (\Cref{sec:fastpay:audits}) all the previous confirmation orders $C_0 \ldots C_{n-1}$ with $\sender(C_k) = a$, $\sequencenumber(C_k) = k$. Therefore, any client can retrieve them and eventually bring any other honest authority up to date with $C$.
\end{proof}

Specifically, in \Cref{pr:fastpay:proposition-l1A}, the confirmation order for $C$ is guaranteed to succeed for every honest authority $\alpha$, provided that the user first recovers and transfers to $\alpha$ all the \emph{missing certificates required by $\alpha$}, defined as the sequence $C_k \ldots C_{n-1}$ such that $k = \nextsequencenumber^a(\alpha)$, $a = \sender(C)$, $n = \sequencenumber(C)$, $\sender(C_i) = a$ $(k\leq i \leq n-1)$. The fact that no other certificates need to be confirmed (\eg to credit the balance of $\sender(C)$ itself) is closely related to the possibility of (temporary) negative balances for authorities.

Note that having a \fastpay certificate confirmed by an authority $\alpha$ only affects $\alpha$'s recipient and the sender's balances (\ie \emph{redeems the certificate}) the first time it is confirmed.
Finally, we state that \fastpay funds credited on an account can always be spent.
We write $\received(a)$ for the set of \emph{incoming} transfer certificates $C$ such that $\recipient(C)=a$ and $C$ is known to the owner of the account\ $a$.

\begin{proposition}[Availability of transfer certificates] \label{pr:fastpay:proposition-l2}
Let $a$ be an account owned by a correct user, $n$ be the next available sequence number after the last signed transfer order (if any, otherwise $n=0$), and $O$ be a new transfer order signed by~$a$ with $\sequencenumber(O) = n$ and $\sender(O) = a$.

Assume that the owner of $a$ has secured enough funds for a new order $O$ based on their knowledge of the chain $\sigma$, the history of outgoing transfers, and the set $\received(a)$. That is, formally:
\[
\begin{array}{l}
\bigg( \amount(O) \; + \; \sum_{\left\{\!\!\!\!\scriptsize\begin{array}{l} \sender(C)=a\\ \sequencenumber(C) < n\end{array}\right.}\!\!\!\!\amount(C) \bigg) \\[0.4em]
\quad \leq \; \bigg( \funding^a(\sigma) \; + \; \sum_{C \in \received(a)} \amount(C) \bigg)
\end{array}
\]
Then, for any honest authority $\alpha$, the user will always eventually obtain a valid signature of $O$ from $\alpha$ after sending the following orders to $\alpha$:
\begin{enumerate}
\setlength\itemsep{0em}
\item A synchronization order from the \sysmain based on the known state $\sigma$;
\item A confirmation order for every $C \in \received(a)$, preceded by all the missing certificates required by $\alpha$ (if any) for the sender of $C$;
\item Then, the transfer order $O$.
\end{enumerate}
\end{proposition}
\begin{proof}
\small
Let $B \geq \amount(O)$ be the following value evaluated at the time of the creation of the new transfer order $O$:

\[
\begin{array}{l}
B \; = \; \bigg( \funding^a(\sigma) \;-\; \sum_{\left\{\!\!\!\!\scriptsize\begin{array}{l} \sender(C)=a\\ \sequencenumber(C) < n\end{array}\right.}\!\!\!\!\amount(C) \\[0.4em]
\quad \; + \; \sum_{C \in \received(a)} \amount(C) \bigg)
\end{array}
\]
By a case analysis similar to the proof of \Cref{le:fastpay:lemma-s1}, provided that the owner of $a$ is communicating the information described in \Cref{pr:fastpay:proposition-l2} to the authority $\alpha$, it will hold eventually that $\balance^a(\alpha) \geq B \geq \amount(O)$ and $\nextsequencenumber^a(\alpha) = n$. We deduce that eventually $\alpha$ will accept the transfer order~$O$ and make the value of its signed (pending) order available.
\end{proof}

% ======
\subsection{Performance under Byzantine Failures}
The \fastpay protocol does not rely on any designated leader (like PBFT~\cite{pbft}) to make progress or create proposals; \fastpay authorities do not directly communicate with each other, and their actions are symmetric. Clients create certificates by gathering the first $2f+1$ responses to a valid transfer order, and no action of a Byzantine authority may delay the creation of a certificate. A Byzantine authority may not even present a signature on a different order as a response to confuse a correct client, since it would have to be signed by the correct payer. Subsequently, a correct client submits the confirmation order to all authorities. Again, Byzantine authorities cannot in any way delay honest authorities from processing the payment locally in their databases, and enabling a subsequent payment for the sending account.

Byzantine clients may attempt denial of service attacks by over-using the system, and for example creating a very large number of receiving accounts (this could be disincentivized by charging some fee for an account creation). However, an attempt to equivocate by sending two transfer orders for a single sequence number could either result in their own account being locked (no single transfer order can achieve $2f+1$ signatures to form a certificate and move to the next sequence number), or one of them succeeding---neither of which degrade performance. Transfer orders with insufficient funds or incorrect sequence numbers are simply rejected, which does not significantly affect performance (if anything they do not result in confirmation orders that are more costly to process than transfer orders, see \Cref{sec:fastpay:evaluation}).

% ======
\subsection{Worst-Case Efficiency of \fastpay Clients} 
To initiate a transfer (\Cref{pr:fastpay:proposition-l2}) or receive funds (\Cref{pr:fastpay:proposition-l1A}) from a sender account $a$, a \fastpay client must address a quorum of authorities. During the exchange, each authority $\alpha$ may require missing certificates $C_k \ldots C_{n-1}$, where $k = \nextsequencenumber^a(\alpha)$ is provided by $\alpha$. In an attempt to slow down the client, a Byzantine authority could return $k=0$ and/or fail to respond at some point. To address this, a client should query each authority $\alpha$ in parallel. After retrieving the sequence number~$k$, the required missing certificates should be downloaded sequentially, in reverse order, then forwarded to~$\alpha$. Given that \fastpay client operations succeed as soon as a quorum of authorities complete their exchanges, this strategy ensures client efficiency despite Byzantine authorities.

%% file: chapters/fastpay/sections/implementation.tex
% ======
% Implementation
% ======
\section{Implementation} \label{sec:fastpay:implementation}
We implemented both a \fastpay client and a networked multi-core multi-shard \fastpay authority in Rust, using Tokio\footnote{\url{https://tokio.rs}} for networking and ed25519-dalek\footnote{\url{https://github.com/dalek-cryptography/ed25519-dalek}} for signatures. For the verification of the multiple signatures composing a certificate we use ed25519 batch verification. To reduce latency we use UDP for \fastpay requests and replies, and make the core of \fastpay idempotent to tolerate retries in case of packet loss; we also provide an experimental \fastpay implementation using exclusively TCP. Currently, data-structures are held in memory rather than persistent storage. 
We implement an authority shard as a separate operating system process with its own networking and Tokio reactor core, to validate the low overhead of intra-shard coordination (through message passing rather than shared memory). We experimented with manually pinning processes to physical cores without a noticeable increases in performance through the Linux \emph{taskset} feature. It seems the Linux OS does a good job of distributing processes and keeping them on inactive cores. We also experimented with a single process multi-threaded implementation of \fastpay, using a single Tokio reactor for all shards on multi-core machines. However, this led to significantly lower performance, and therefore we opted for using separate processes even on a single machine for each shard. The exact bottleneck justifying this lower performance---whether at the level of Tokio multi-threading or OS resource management---still eludes us.

The implementation for both server and client is less than 4,000 LOC (of which half are for the networking), and a further 1,375 LOC of unit tests. It required about 2.5 months of work for 3 engineers, and a bit over 1,500 git commits. Keeping the core small required constant re-factoring and its simplicity is a significant advantage of the proposed \fastpay design. We are open sourcing the Rust implementation, Amazon web services orchestration scripts, benchmarking scripts, and measurements data to enable reproducible results\footnote{
\url{https://github.com/calibra/fastpay}
}.

%% file: chapters/fastpay/sections/evaluation.tex
% ======
% Evaluation
% ======
\section{Evaluation} \label{sec:fastpay:evaluation}
We evaluate the throughput and latency of our implementation of \fastpay through experiments on AWS. We  particularly aim to demonstrate that \first sharding is effective, in that it increases throughput linearly as expected; \second latency is not overly affected by the number of authorities or shards, and remains near-constant, even when some authorities fail; and \third that the system is robust under extremely high concurrency and transaction loads.

% ======
\subsection{Microbenchmarks}
We report on microbenchmarks of the single-CPU core time required to process transfer orders, authority signed partial certificates, and certificates. \Cref{tab:fastpay:microbenchark} displays the cost of each operation in micro seconds ($\mu s$) assuming 10 authorities (recall $1\mu s = 10^{-6}s$); each measurement is the result of 500 runs on an Apple laptop (MacBook Pro) with a 2.9 GHz Intel Core i9 (6 physical and 12 logical cores), and 32 GB 2400 MHz DDR4 RAM. The first 3 rows respectively indicate the time to create and serialize \first a transfer order, \second a partial certificate signed by a single authority, and \third a transfer certificate as part of a confirmation order. The last 3 rows indicate the time to deserialize them and check their validity. 
The dominant CPU cost involves the deserialization and signature check on certificates ($236\mu s$), which includes the batch verification of the 8 signatures (7 from authorities and 1 from sender). However, deserializing orders ($58\mu s$) and votes ($60\mu s$) is also expensive: it involves 1 signature verification (no batching) and creating 1 signature. Those results indicate that a single core shard implementation may only settle just over 4,000 transactions per second---highlighting the importance of sharding to achieve high-throughput. 
In terms of networking costs, a transfer order is 146 bytes, and the signed response is 293 bytes. This could be reduced by only responding with a signature (64 bytes) rather than the full signed order, but we chose to echo back the order to simplify client implementations. A full certificate for an order is 819 bytes, and the response---consisting of an update on the state of the \fastpay account---is 51 bytes. For deployments using many authorities we can compress certificates by using an aggregate signature scheme (such as BLS~\cite{bls}). However, verification CPU costs of BLS only make this competitive for committees larger than 50-100 authorities. We note that all \fastpay message types fit within the common maximum transmission unit of commodity IP networks, allowing requests and replies to be executed using a single UDP packet (assuming no packets loss and 10 authorities).

\begin{table}[t]
\centering
\begin{tabular}{lrr} \toprule
\textbf{Measure} & \textbf{Mean (\bm{$\mu$}s)} & \textbf{Std. (\bm{$\mu$}s)} \\
\midrule
Create \& Serialize Order & 27 & 1 \\
Create \& Serialize Partial Cert. & 27 & 2 \\
Create \& Serialize Certificate & 4 & 0 \\
Deserialize \& Check Order & 58 & 1 \\
Deserialize \& Check Partial Cert. & 60 & 1 \\
Deserialize \& Check Certificate & 236 & 10 \\
\bottomrule
\vspace{0.5mm}
\end{tabular}
\caption[\fastpay microbenchmark.]{Microbenchmark of single core CPU costs of \fastpay operations; average and standard deviation of 500 measurements for 10 authorities.}
\label{tab:fastpay:microbenchark}
\end{table}
%

% ======
\subsection{Throughput} \label{sec:fastpay:throughput-eval}
We deploy a \fastpay multi-shard authority on Amazon Web Services (Stockholm, eu-north-1 zone), on a m5d.metal instance. This class of instance guarantees 96 virtual CPUs (48 physical cores), on a 2.5 GHz, Intel Xeon Platinum 8175, and 384 GB memory. The operating system is Linux Ubuntu server 18.04, where we increase the network buffer to about 96MB. 
In all graphs, each measurement is the average of 9 runs, and the error bars represent one standard deviation; all experiments use our UDP implementation.
We measure the variation of throughput with the number of shards. Our baseline experiment parameters are: 4 authorities (for confirmation orders), a load of 1M transactions, and applying back-pressure to allow a maximum of 1000 concurrent transactions at the time into the system (\ie the \emph{in-flight} parameter). We then vary those baseline parameters through our experiments to illustrate their impact on performance.

\para{Robustness and performance under high concurrency}
Figures~\ref{fig:fastpay:x-1000000-z-4-transfer} and~\ref{fig:fastpay:x-1000000-z-4-confirmation} respectively show the variation of the throughput of processing transfer and confirmation orders as we increase the number of shards per authority, from 15 to 85. We measure those by processing 1M transactions, across 4 authorities. \Cref{fig:fastpay:x-1000000-z-4-transfer} shows that the throughout of transfer orders slowly increases with the number of shards. The in-flight parameter---the maximum number of transactions that is allowed into the system at any time---influences the throughput by about 10\%, and setting it to 1,000 seems optimal for performance. The degree of concurrency in a system depends on the number of concurrent client requests, and we observe that \fastpay is stable and performant even under extremely high concurrency peaks of 50,000 concurrent requests. Afterwards, the Operating System UDP network buffers fill up, and the authority network stacks simply drop the requests.

\Cref{fig:fastpay:x-1000000-z-4-confirmation} shows that the throughput of confirmation orders initially increases linearly with the number of shards, and then reaches a plateau at around 48 shards. This happens because our experiments are run on machines with 48 physical cores, running at full speed, and 48 logical cores.
The in-flight parameter of concurrent requests does not influence the throughput much, but setting it too low (\eg at 100) does not saturate our CPUs. These figures show that \fastpay can support up to 160,000 transactions per second on 48 shards (about 7x the peak transaction rate of the Visa payments network~\cite{visa-performance}) while running on commodity computers that cost less than 4,000 USD/month per authority\footnote{AWS reports a price of 5.424 USD/hour for their \texttt{m5d.metal} instances. \url{https://aws.amazon.com/ec2/pricing/on-demand} (January 2020)}.

\begin{figure}[t]
\centering
\includegraphics[width=.7\textwidth]{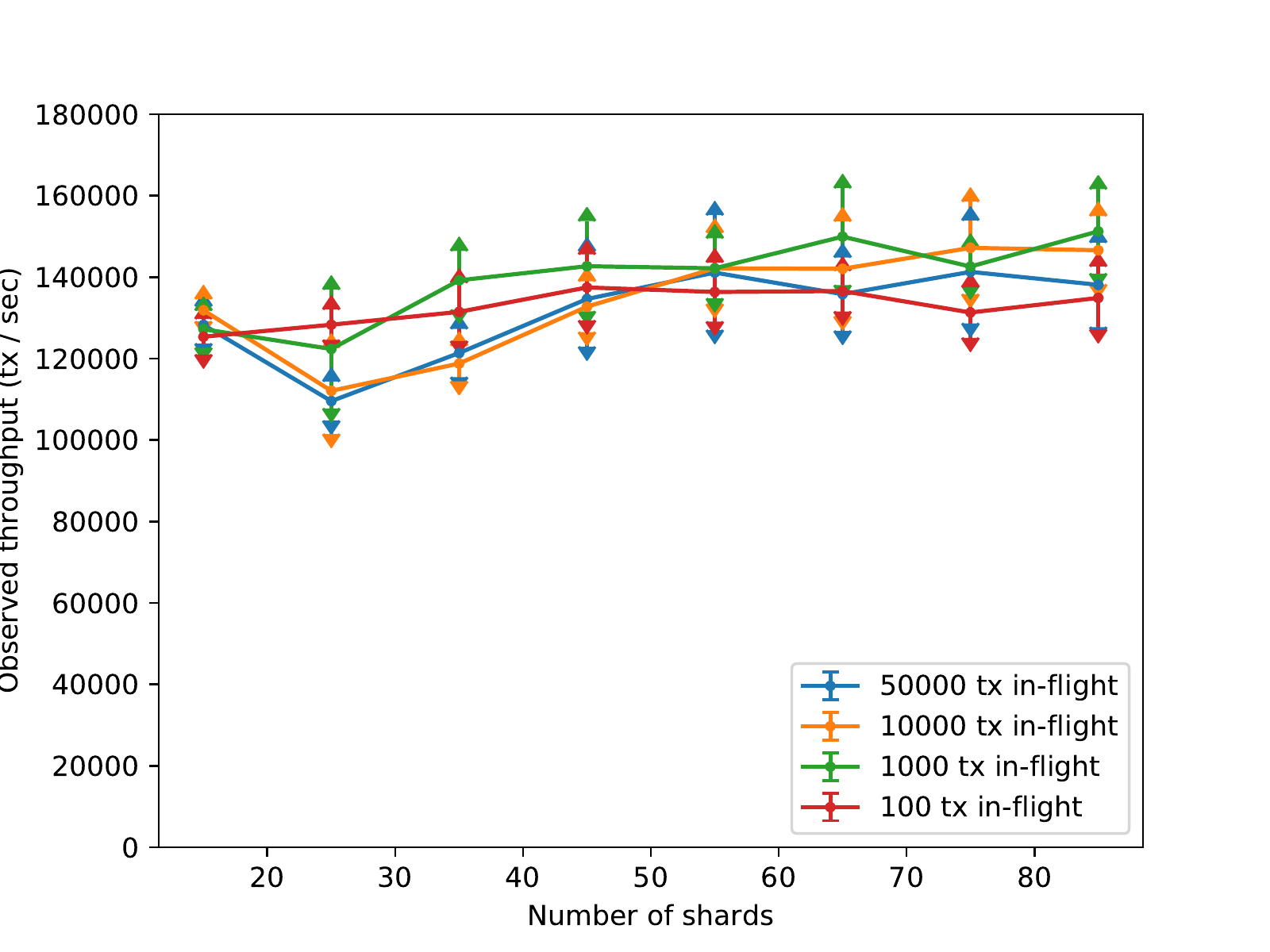}
\caption[\fastpay transfer orders throughput under high concurrency.]{Variation of the throughput of transfer orders with the number of shards, for various levels of concurrency (in-flight parameter). The measurements are run under a total load of 1M transactions.}
\label{fig:fastpay:x-1000000-z-4-transfer}
\end{figure}
\begin{figure}[t]
\centering
\includegraphics[width=.7\textwidth]{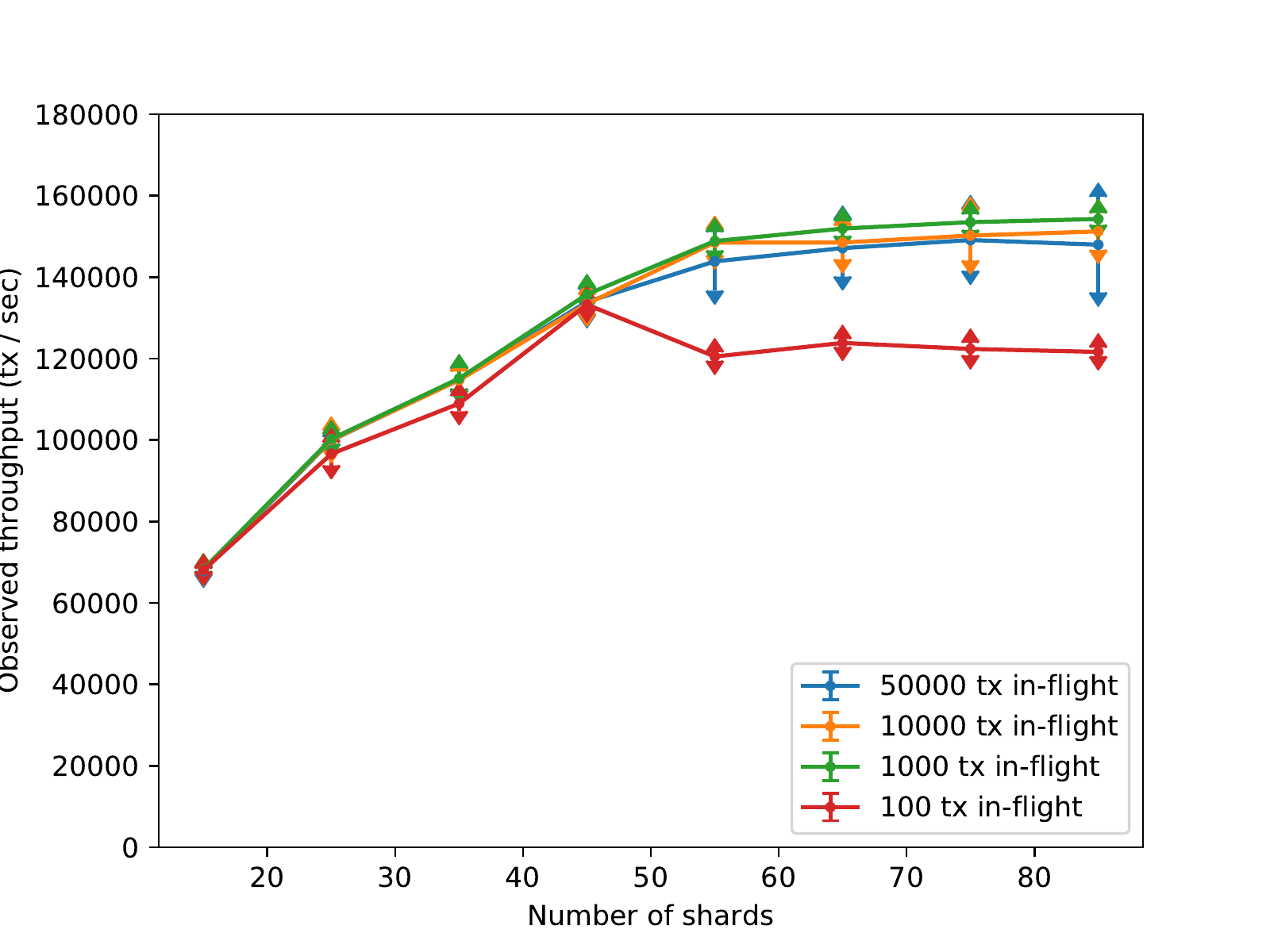}
\caption[\fastpay confirmation orders throughput under high concurrency.]{Variation of the throughput of confirmation orders with the number of shards, for various levels of concurrency (in-flight parameter). The certificates are issued by 4 authorities, and the measurements are run under a total load of 1M transactions.}
\label{fig:fastpay:x-1000000-z-4-confirmation}
\end{figure}

\para{Robustness and performance under total system load}
Figures \ref{fig:fastpay:x-z-1000-4-transfer} and \ref{fig:fastpay:x-z-1000-4-confirmation} show the variation of the throughput of transfer and confirmation orders with the number of shards, for various total system loads---namely the total number of transactions in the test; they show that the throughput is not affected by the system load. The tests were performed with 4 authorities, and the client concurrency in-flight parameter set to 1,000. 
These figures illustrate that \fastpay can process about 160,000 transactions per second even under a total load of 1.5M transactions, and that the total load does not significantly affect performance. These supplement Figures \ref{fig:fastpay:x-1000000-z-4-transfer} and~\ref{fig:fastpay:x-1000000-z-4-confirmation} that illustrate that the concurrent transaction rate also does not influence performance significantly (except when it is too low by under-utilizing the system).

Readers may be surprised those measurements are important. The key measurement work by Han~\etal~\cite{han2019performance} compares a number of permissioned systems under a high load, and shows that for all of Hyperledger Fabric (v0.6 with PBFT)~\cite{hyperledger-v06}, Hyperledger Fabric (v1.0 with BFT-Smart)~\cite{hyperledger-v10}, Ripple~\cite{ripple-documentation} and R3 Corda v3.2~\cite{corda-v32} the successful requests per second \emph{drops to zero} as the transaction rate increases to more than a few thousands transactions per second (notably for Corda only a few hundred). An exception is Tendermint~\cite{tendermint}, that maintains a processed transaction rate of about 4,000 to 6,000 transactions per second at a high concurrency rate. Those findings were confirmed for Hyperledger Fabric that reportedly starts saturating at a rate of 10,000 transactions per second~\cite{nasir2018performance}. Our results demonstrate that \fastpay continues to be very performant even under the influence of extremely high rates of concurrent transactions (in-flight parameter) and overall work load (total number of transactions processed), as expected. This is apparently not the norm.

\begin{figure}[t]
\centering
\includegraphics[width=.7\textwidth]{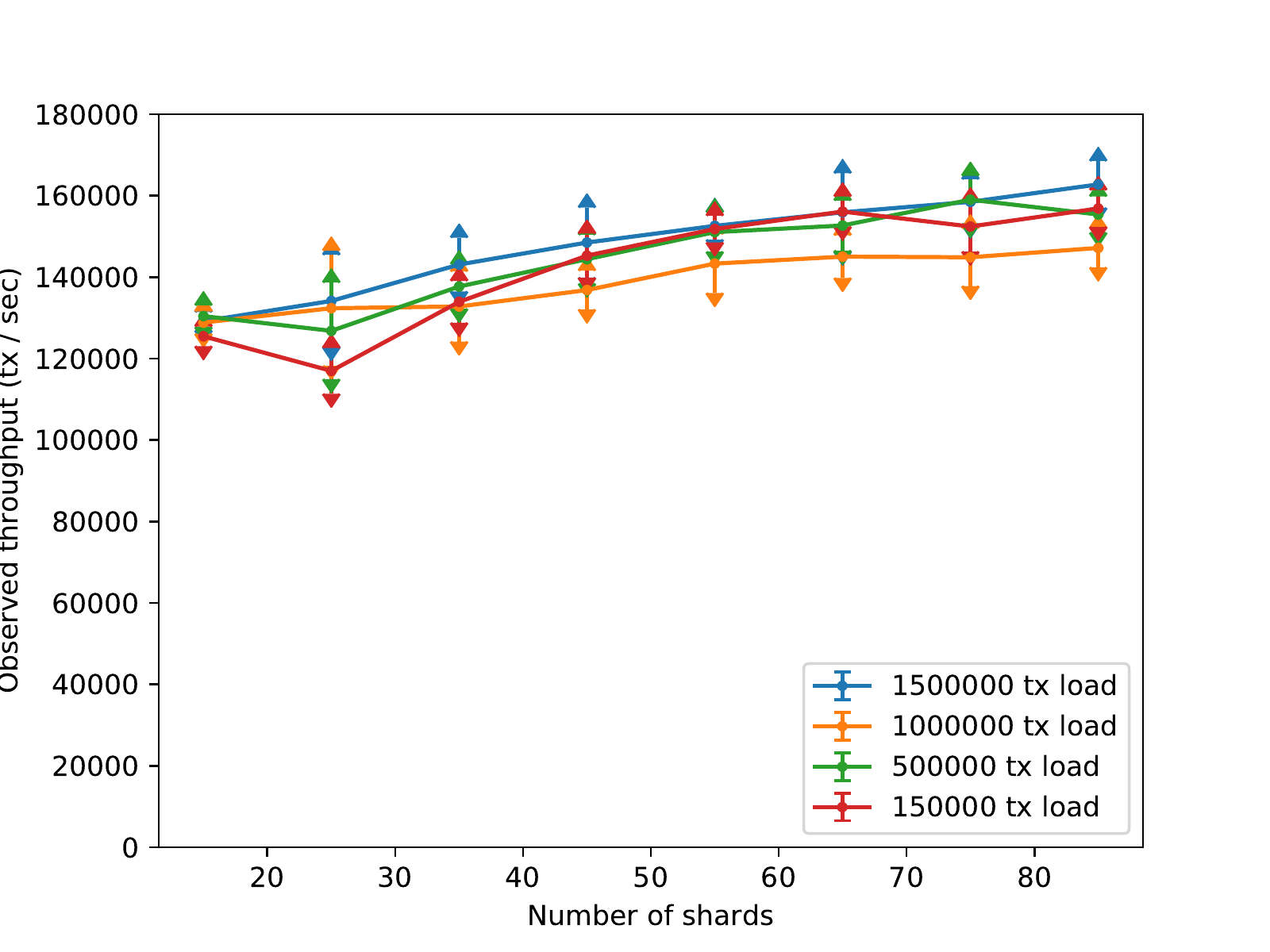}
\caption[\fastpay transfer orders throughput under high load.]{Variation of the throughput of transfer orders with the number of shards, for various loads. The in-flight parameter is set to 1,000.}
\label{fig:fastpay:x-z-1000-4-transfer}
\end{figure}
\begin{figure}[t]
\centering
\includegraphics[width=.7\textwidth]{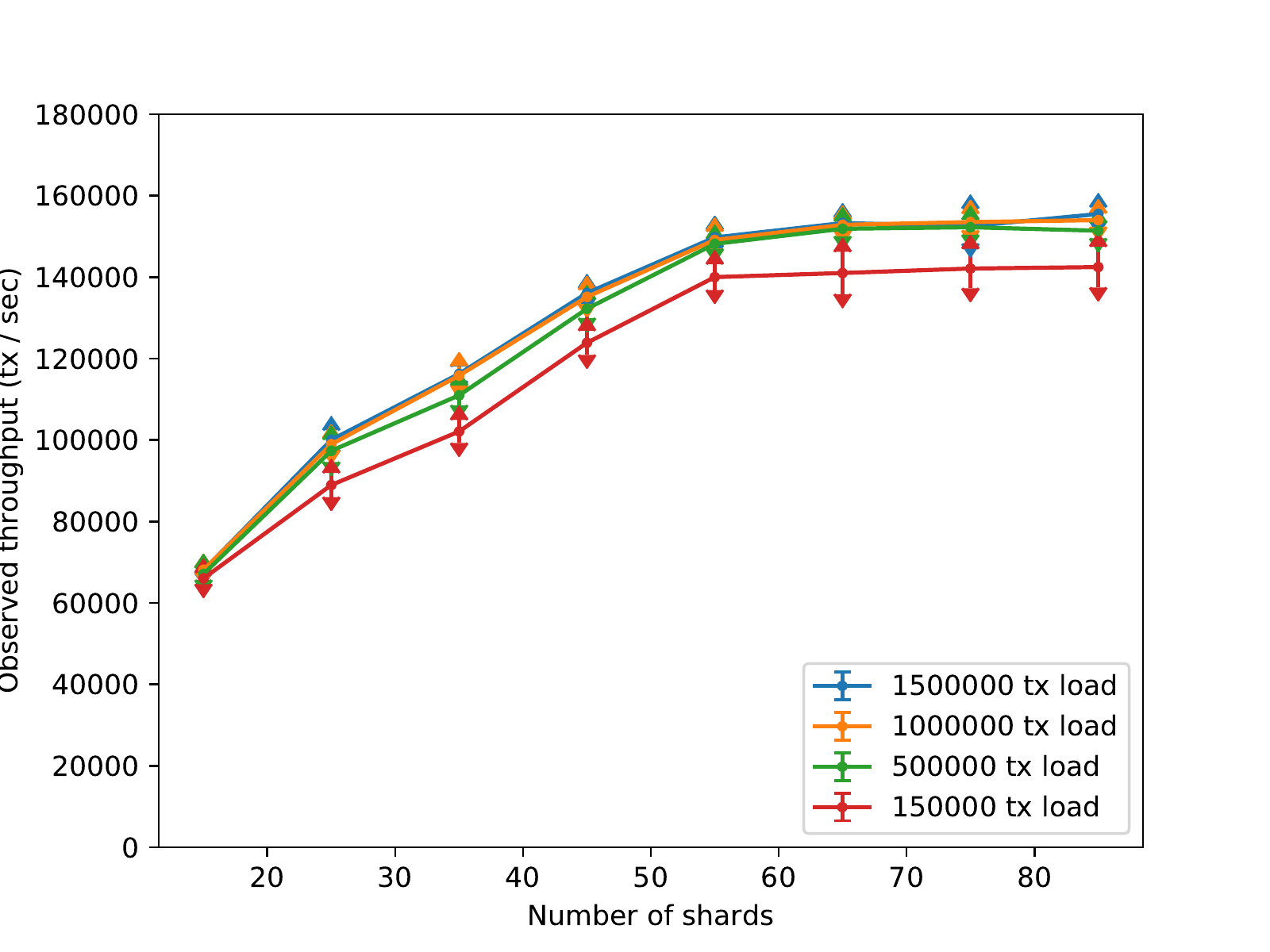}
\caption[\fastpay confirmation orders throughput under high load.]{Variation of the throughput of confirmation orders with the number of shards, for various loads. The certificates are issued by 4 authorities, and the in-flight parameter is set to 1,000.}
\label{fig:fastpay:x-z-1000-4-confirmation}
\end{figure}

\para{Influence of the number of authorities}
As discussed in \Cref{sec:fastpay:design}, we expect that increasing the number of authorities only impacts the throughput of confirmation orders (that need to transfer and check transfer certificates signed by $2f+1$ authorities), and not the throughput of transfer orders. \Cref{fig:fastpay:z-1000000-1000-x-confirmation} confirms that the throughput of confirmation orders decreases as the number of authorities increases. \fastpay can still process about 80,000 transactions per second with 20 authorities (for 75 shards). The measurements are taken with an in-flight concurrency parameter set to 1,000, and under a load of 1M total transactions. We note that for higher number of authorities, using an aggregate signature scheme (\eg BLS~\cite{bls}) would be preferable since it would result in constant time verification and near-constant size certificates. However, due to the use of batch verification of signatures, the break even point may be after 100 authorities in terms of verification time.

\begin{figure}[t]
\centering
\includegraphics[width=.7\textwidth]{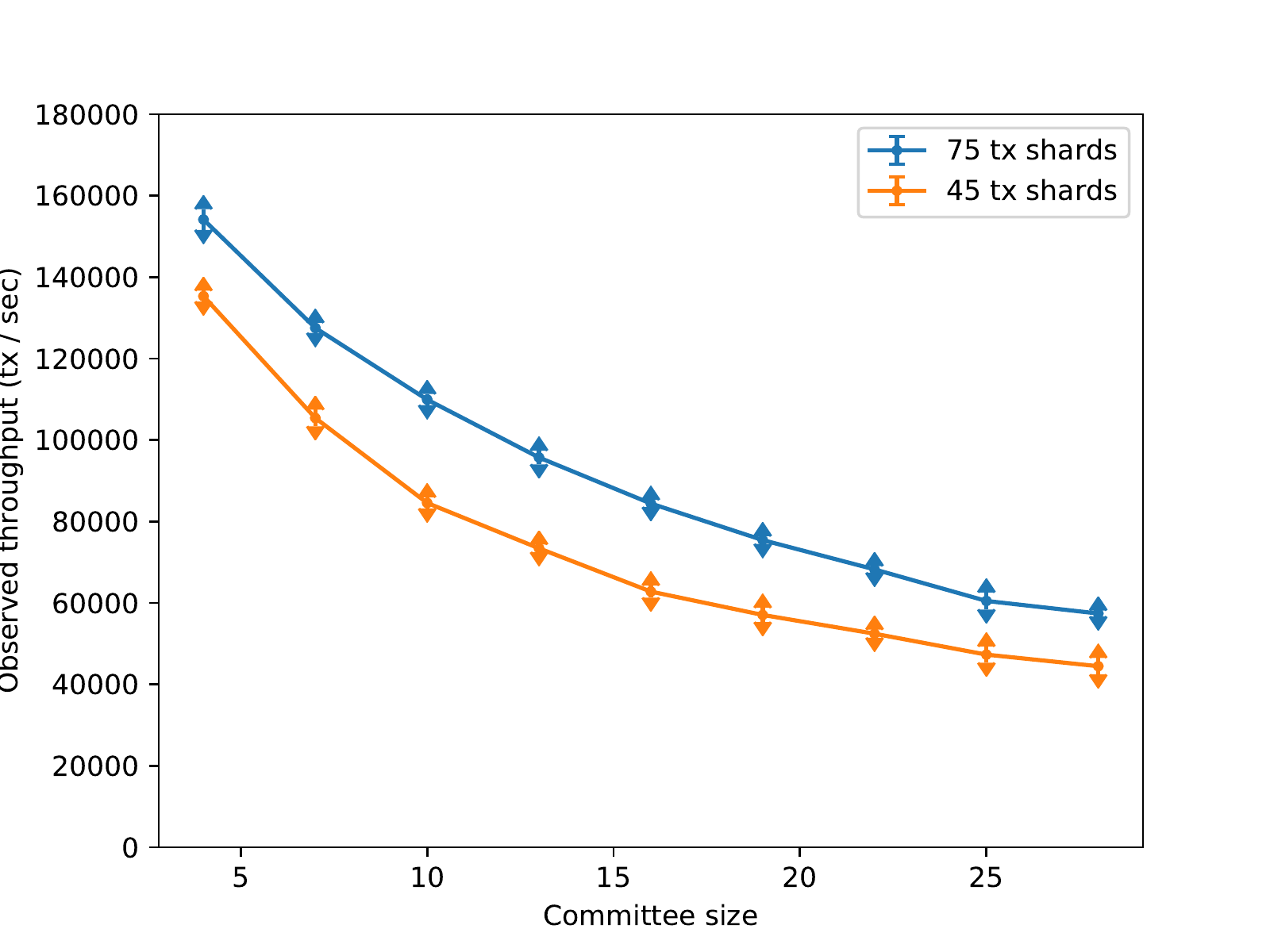}
\caption[\fastpay confirmation orders throughput for multiple authorities.]{Variation of the throughput of confirmation orders with the number of authorities, for various number of shards. The in-flight parameter is set to 1,000 and the system load is of 1M transactions.}
\label{fig:fastpay:z-1000000-1000-x-confirmation}
\end{figure}
%

% ======
\subsection{Latency}
We measure the variation of the client-perceived latency with the number of authorities. We deploy several \fastpay multi-shard authorities on Amazon Web Services (all in Stockholm, eu-north-1 zone), each on a m5d.8xlarge instance. This class of instance guarantees 10Gbit network capacity, on a 3.1 GHz, Intel Xeon Platinum 8175 with 32 cores, and 128 GB memory. The operating system is Linux Ubuntu server 16.04. Each instance is configured to run 15 shards.
The client is run on an Apple laptop (MacBook Pro) with a 2.9 GHz Intel Core i9 (6 physical and 12 logical cores), and 32 GB 2400 MHz DDR4 RAM; and connected to a reliable WIFI network. We run experiments with the client in two different locations; \first in the U.K. (geographically close to the authorities, same continent), and \second in the U.S. West Coast (geographically far from the authorities, different continent).
Each measurement is the average of 300 runs, and the error bars represent one standard deviation; all experiments use our UDP implementation.

\begin{figure}[t]
\centering
\includegraphics[width=.7\textwidth]{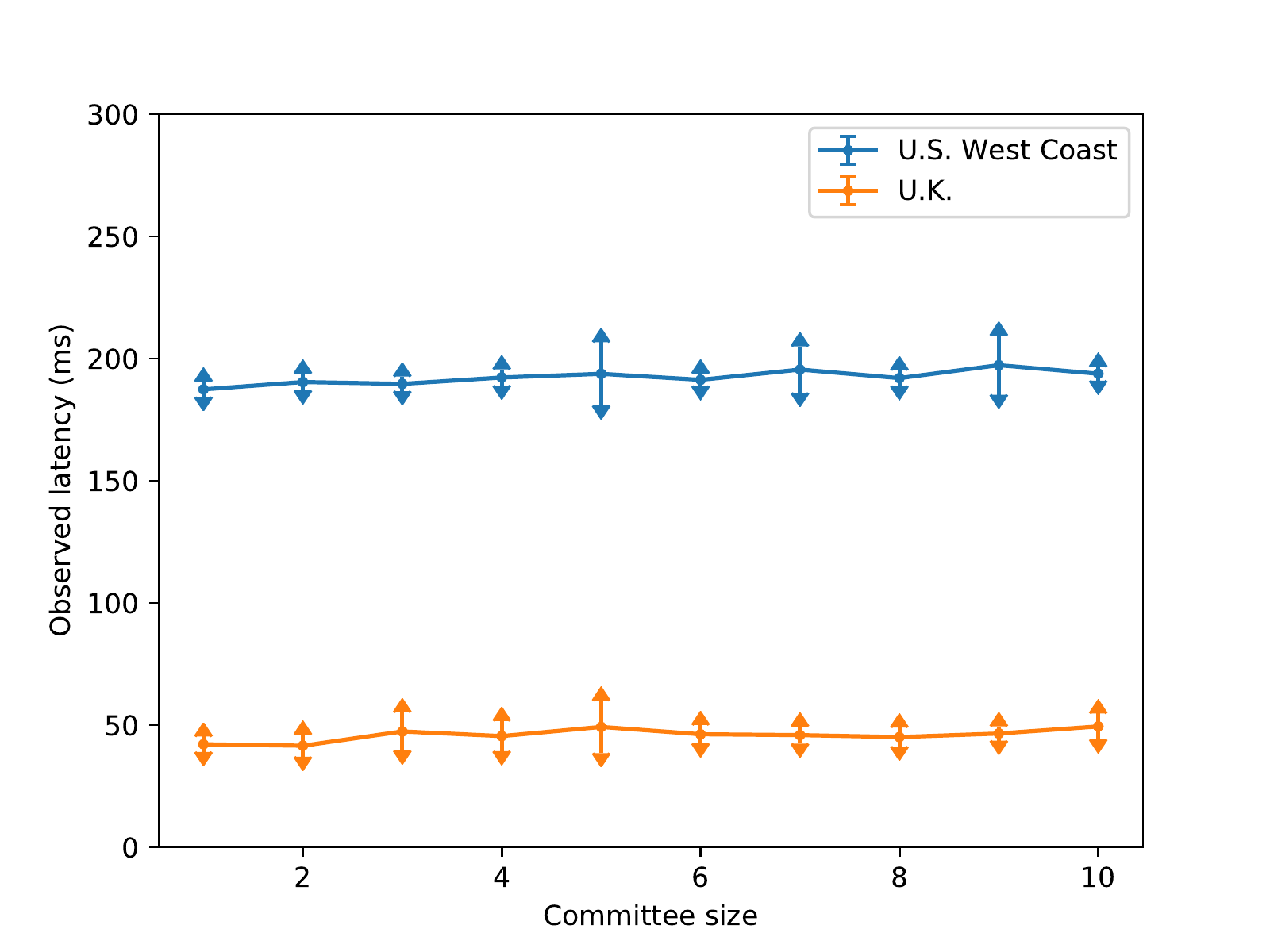}
\caption[\fastpay transfer orders latency.]{Variation of the latency of transfer orders with the number of authorities, for various locations of the client.}
\label{fig:fastpay:latency-transfer}
\end{figure}
\begin{figure}[t]
\centering
\includegraphics[width=.7\textwidth]{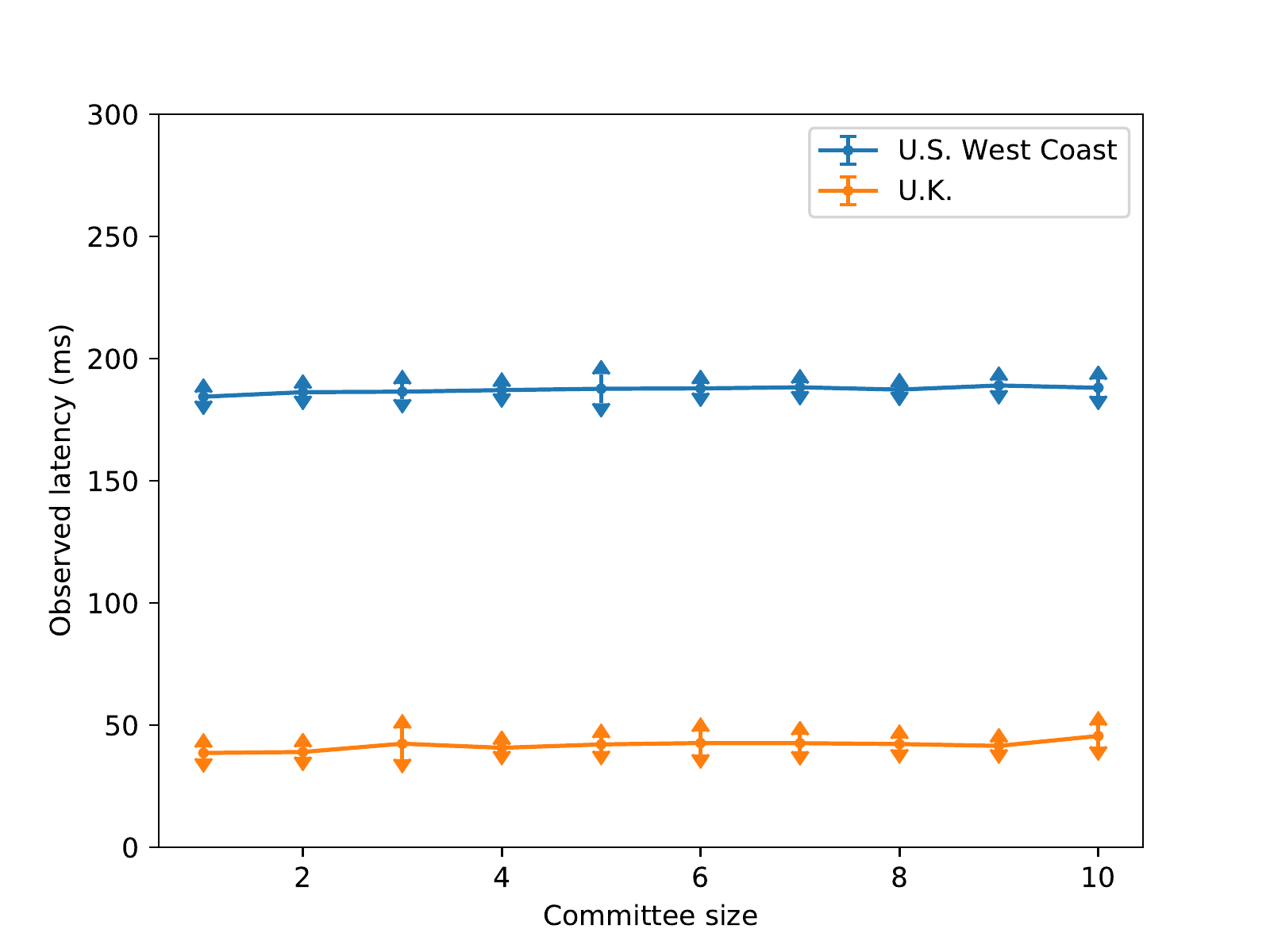}
\caption[\fastpay confirmation orders latency.]{Variation of the latency of confirmation orders with the number of authorities, for various locations of the client.}
\label{fig:fastpay:latency-confirmation}
\end{figure}

We observe that the client-authority WAN latency is low for both transfer and confirmation orders; the latency is under 200ms when the client is in the U.S. West Coast, and about 50ms when the client is in the U.K.
\Cref{fig:fastpay:latency-transfer} illustrates the latency between a client creating and sending a transfer order to all authorities, and receiving sufficient signatures to form a transfer certificate (in our experiment we wait for all authorities to reply to measure the worse case where $f$ authorities are Byzantine). The latency is virtually constant as we increase the number of authorities, due to the client emitting orders asynchronously to all authorities and waiting for responses in parallel. 
\Cref{fig:fastpay:latency-confirmation} illustrates the latency to submit a confirmation order, and wait for all authorities to respond with a success message. It shows latency is virtually constant when increasing the number of authorities. This indicates that the latency is largely dominated by the network (and not by the verification of certificates). However, since even for 10 authorities a \fastpay message fits within a network MTU, the variation is very small. Due to our choice of using UDP as a transport there is no connection initiation delay (as for TCP), but we may observe packet loss under very high congestion conditions. Authority commands are idempotent to allow clients to re-transmit to overcome loss without sacrificing safety.

\para{Performance under failures} 
Research literature suggests permissioned blockchains based on (often leader-based) consensus suffer an enormous performance drop when some authorities fail~\cite{DBLP:conf/icdcs/LeeSHKN14}. We measure the effect of authority failure in \fastpay and show that latency is not affected when $f$ or fewer authorities are unavailable.
We run our baseline experimental setup (10 authorities distributed over 10 different AWS instances), when a different number of authorities are not available for $f = 0\ldots3$. 
\begin{wraptable}{r}{0.40\columnwidth}
\centering
\footnotesize
\begin{tabular}{lcc}  
\toprule
\bm{$f$}  & \textbf{Mean (ms)} & \textbf{Std. (ms)} \\
\midrule
0 & 43 & 2\\
1 & 41 & 3\\
2 & 44 & 4\\
3 & 47 & 2\\
\bottomrule
\end{tabular}
\caption[\fastpay crash-failure Latency.]{Crash-failure Latency.}
\label{tab:fastpay:fail}
\end{wraptable}
We measure the latency experienced by a client  on the same continent (Europe), sending a transfer order until it forms a valid transfer certificate. \Cref{tab:fastpay:fail} summarizes the mean latency and standard deviation for different $f$. There is no statistically significant difference in latency, no matter how many tolerable failures \fastpay experiences (up to $f \leq 3$ for 10 authorities). We also experimented with killing authorities one by one with similar results, up to $f > 3$ when the system did observably lose liveness as expected. The underlying reason for the steady performance under failures is \fastpay's lack of reliance on a leader to drive the protocol.

%% file: chapters/fastpay/sections/discussion.tex
% ======
% Discussion
% ======
\section{Limitations \& Future Work} \label{sec:fastpay:discussion}

% ======
\para{Threats to validity of experiments} 
Our experiments represent the best case performance, for a set number of authorities and shards, as they are performed in laboratory conditions. In particular, real-wold transactions may have the same sender account, which would prevent them from being executed in parallel. Further, the throughput evaluation places transaction load on an authority through the local network interface, and therefore does not take fully into account the operating system networking costs of a full WAN stack. Further, our WAN latency experiments were performed against authorities with very low-load. Finally, the costs of persisting databases to storage are not taken into account when measuring latency and throughput (we leave the implementation of low-latency persistent storage to future work).

%\para{Integrating privacy} 
%As presented, \fastpay exposes information about all transactions, namely the sender-recipient accounts and the amounts transferred, as well as the timings of those transfers. Fully integrating stronger privacy protections is a separate research project. However, we want to highlight that the architecture of \fastpay is highly compatible with threshold issuance selective disclosure credential designs, such as Coconut~\cite{coconut}. In those scheme a threshold of authorities can jointly sign a credential that the user can subsequently randomize and present to execute a payment. 
%Implementing hidden balances, like MinbleWinble~\cite{mimblewimble} and combining them with credentials should be possible---but beyond the scope of the present work.

\para{Checkpointing, authority, and key rotation}
The important enabler for the good performance of \fastpay, but also an important limitation, is the fact that authorities do not need to reach consensus on the state of their databases. We demonstrate that payments are secure in this context, but various system maintenance operations are harder to implement. For example, checkpointing the state of all accounts in the systems, to compress the list of stored certificates would be beneficial, but cannot be straightforwardly implemented without consensus. Similarly, it would be beneficial for authorities to be able to rotate in and out of the committee, as well as to update their cryptographic signature keys. Due to the lack of tight synchronization between authorities there is no natural point that guarantees they all update their committees at the same logical time. Further, our proofs of liveness under asynchrony presume that transfer orders and certificates that were once valid, will always be valid. Integrating such governance features into \fastpay will require careful design to safely leverage either some timing (synchrony) assumptions or use a more capable (but maybe lower performance) consensus layer, such as one facilitated by the \sysmain.

\para{Economics and fees} Some cost to insert transactions into a system (like fees in Bitcoin), allows for sound accounting and prevents Denial of Service attacks by clients over-using an open system. The horizontal scalability of \fastpay alleviates somehow the need to integrate such a scheme, since issues of capacity can be resolved by increasing its capacity through more shards (as well as deploying network level defenses). However, if there was a need to implement fees for using \fastpay we would not recommend using micro-payments associated with each payment like in Bitcoin. Rather, we would recommend allowing a client to deposit some payment into a service account with all authorities, and then allow them to deduct locally some of this fee for any services rendered (namely any signed transfer order or confirmation order processed). In practical terms, the variable costs of processing transactions in \fastpay is low. There is no artificial shortage due to lack of scalability, and a flat periodic fee on either senders or recipients might be sufficient to support operations (rather than a charge per transaction).

%% file: chapters/fastpay/sections/related.tex
% ======
% Related
% ======
\section{Comparison with Related Work} \label{sec:fastpay:related}
We compare \fastpay with traditional payment systems and some relevant crypto-currencies based on permissioned blockchains.

\para{Traditional payment systems} In the context of traditional payment systems \fastpay is a real-time gross settlement system (RTGS)~\cite{rtgs-1,rtgs-2}---payments are executed in close to real-time, there is no netting between participants, and  the transfer of funds is final upon the full payment protocol terminating. All payments are pre-funded so there is no need to keep track of credit or liquidity, which makes the design vastly simpler.
A well known issue with RTGS systems is the need for higher liquidity, as compared with settlement systems based on daily settlement after netting---since more money moves around exposing accounts to higher volatility. In that respect \fastpay is state of the art, in that it allows immediate liquidity recycling~\cite{uk-rtgs}, namely as soon as a payment is processed the value payed into an account may be used to pay other parties in the system. 

\fastpay, from an assurance and performance perspective is significantly superior to deployed RTGS systems: it \first implements a fully Byzantine fault tolerant architecture (established systems rely on master-slave configurations to only recover from few crash failures), \second has higher throughput (as compared, for example with the TARGET2~\cite{target2} European Central Bank RTGS systems that has a target throughput of 500 tx/sec), and \third faster finality (as compared to TARGET2 providing finality of a few seconds). Since \fastpay allows for fast gross settlement, participants are not exposed to credit risk, as is the case for retail payment systems such as VISA and Mastercard (that use daily netting, and have complex financial arrangements to mitigate credit risk in case of bank default). Furthermore, it does achieve both throughput and latency, comparable to those systems combined---about 80,000 tx/sec at peak times, when adding up the throughput of Visa and Mastercard together~\cite{visa-performance,mastercard-performance}.
On the downside, \fastpay lacks certain features of mature RTGS systems: in particular it does not support Delivery-on-Payment transactions that atomically swap securities when payment is provided, or Payment-versus-Payment, that atomically swap amounts in different currencies to minimize the risk of foreign exchange transactions. These require atomic operations across accounts controlled by different users, and would therefore require extending \fastpay to support them (namely operations with consensus number of 2 per Herlihy~\cite{herlihy1991wait}).

\para{Crypto-currencies} \fastpay provides high assurance in the context of Byzantine failures within its infrastructure. So in that respect it is comparable with systems encountered in the space of permissioned blockchains and crypto-currencies, as well as their eco-system of payment channels. \fastpay is permissioned in that the set of authorities managing the system is closed---in fact we do not even propose a way to rotate those authorities and leave this to future work. Qualitatively, \fastpay differs from other permissioned (or permissionless) crypto-currencies in a number of important ways: it is secure under full network asynchrony (since it does not require or rely on atomic broadcast channels or consensus, but only consistent broadcast)---leading to higher performance. This direction was explored in the past in relation to central bank cryptocurrency systems~\cite{rscoin} and high performance permissionless systems~\cite{avalanche}. It was recently put on a formal footing by Guerraroui~\etal~\cite{consensus-number}. Our work extends this theory to allow increased concurrency, correctness under sharding, and rigorous interfacing with external settlement mechanisms. \fastpay achieves auditability through a set of certificates signed by authorities rather than a sequential log of actions (blockchain), which would require authorities to reach agreement on a common sequence. 
Quantitatively, compared with other permissioned systems \fastpay is extremely performant. HyperLedger Fabric~\cite{hyperledger} running with 10 nodes achieves about 1,000 transactions per second and a latency of about 10 seconds~\cite{nasir2018performance}; and Libra~\cite{libra} and Corda~\cite{corda, corda-performance} achieve similar performance.
JP Morgan developed a digital coin built from the Ethereum codebase, which can achieve about 1,500 transactions per second with four nodes, and imposing a block time of 1 second~\cite{baliga2018performance}.
Tendermint~\cite{tendermint} reportedly achieves 10,000 transactions per second with 4 nodes, with a few seconds latency~\cite{tendermint-performance}. However, as we discussed in \Cref{sec:fastpay:evaluation}, many of those systems see their performance degrading dramatically under heavy load---whereas \fastpay performs as expected.

\fastpay can be used as a side chain of any cryptocurrency with reasonable finality guarantees, and sufficient programmability. As compared to bilateral payment channels it is superior in that it allows users to pay anyone in the system without locking liquidity into the bilateral channel, and is fully asynchronous. However, \fastpay does rely on an assumption of threshold non-Byzantine authorities for safely and liveness, whereas payment channel designs only rely on network synchrony for safety and liveness (under conditions of asynchrony safety may be lost). As compared to payment channel networks (such as the lighting network~\cite{lightning}) \fastpay is simpler and does not require complex path finding algorithms~\cite{lightning, flare, grunspan2018ant, sivaraman2018routing}. %As a result its performance \george{Compare with lightnening perf, and complexity.}

%% file: chapters/fastpay/sections/conclusion.tex
% ======
% Conclusion
% ======
\section{Chapter Summary} \label{sec:fastpay:conclusion}
\fastpay is a settlement layer based on consistent broadcast channels, rather than full consensus. The \fastpay design leverages the nature of payments to allow for asynchronous payments into accounts, and optional interactions with an external \sysmain to build a practical system, while providing proofs of both safely and liveness; it also proposes and evaluates a design for sharded implementation of authorities to horizontally scale and match any throughput need. 
The performance and robustness of \fastpay is beyond and above the state of the art, and validates that moving away from both centralized solutions and full consensus to manage pre-funded retail payments has significant advantages. Authorities can jointly process tens of thousands of transactions per second (we observed a peak of 160,000 tx/sec) using merely commodity hardware and lean software. A payment confirmation latency of less than 200ms between continents make \fastpay practical for point of sale payments---where goods and services need to be delivered fast and in person. Pretty much instant settlement enables retail payments to be freed from intermediaries, such as banks payment networks, since they eliminate any credit risk inherent in deferred netted end-of-day payments, that underpin today most national Fast Payment systems~\cite{bolt2014fast}. Further, \fastpay can tolerate up to one-third of authorities crashing or even becoming Byzantine without losing either safety or liveness (or performance). This is in sharp contrast with existing centralized settlement layers operating on specialized mainframes with a primary / backup crash fail strategy (and no documented technical strategy to handle Byzantine operators). Surprisingly, it is also in contrast with permissioned blockchains, which have not achieved similar levels of performance and robustness yet, due to the complexity of engineering and scaling full \bftlong consensus protocols.

%% file: chapters/coconut/coconut.tex
\chapter{\coconut: Threshold Issuance Selective Disclosure Credentials} \label{coconut}

\input{chapters/coconut/sections/introduction.tex}
\input{chapters/coconut/sections/architecture.tex}

\input{chapters/coconut/sections/construction.tex}
\input{chapters/coconut/sections/security.tex}

\input{chapters/coconut/sections/implementation.tex}
\input{chapters/coconut/sections/applications.tex}
\input{chapters/coconut/sections/evaluation.tex}

\input{chapters/coconut/sections/related.tex}

\input{chapters/coconut/sections/limitations.tex}
\input{chapters/coconut/sections/conclusion.tex}

%\input{chapters/coconut/appendices/eth-tumbler.tex}
%\input{chapters/coconut/appendices/multi_attribute_credentials.tex}
%\input{chapters/coconut/appendices/security.tex}

%% file: chapters/coconut/sections/introduction.tex
% =========
% Introduction 
% =========

%
\begin{figure}[t]
\centering
\includegraphics[width=.7\textwidth]{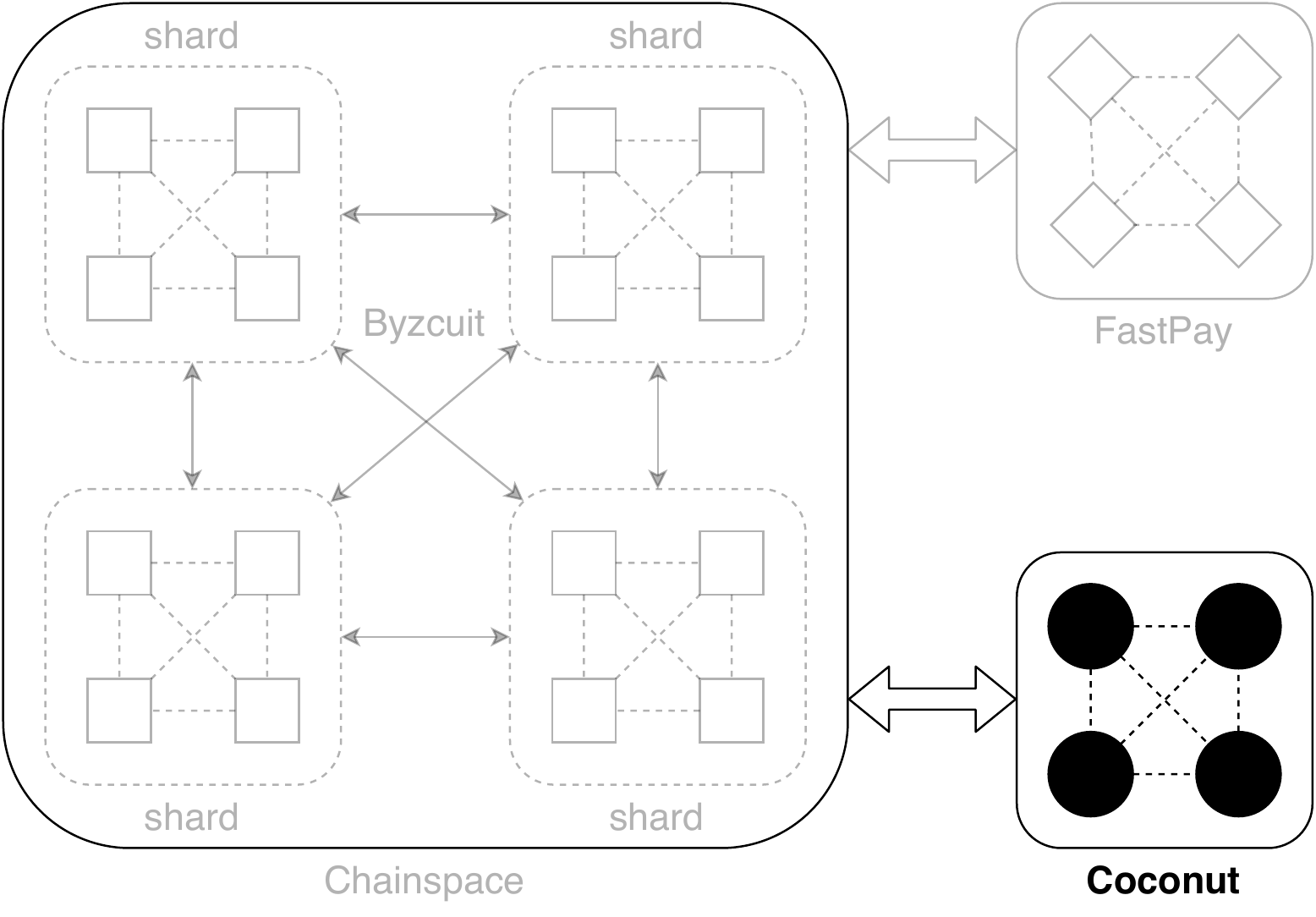}
\caption[Global overview: \coconut.]{Example of instantiation of \coconut with four authorities. Authorities are represented by black circles and the dashed lines connecting them to each other indicate that they are part of the same committee.}
\label{fig:coconut:global-overview}
\end{figure}

\chainspace (\Cref{chainspace}) is the backbone of this thesis. It provides a sharded blockchain that can scale to accommodate high throughput and can process transactions in a few seconds by running \byzcuit at its core (\Cref{byzcuit}). Further, it revisits the execution of smart contracts on blockchains by introducing a model where transactions are executed at the client side to support privacy-preserving applications by design.
%Then, \fastpay (\Cref{fastpay}) augments this backbone by providing a high-capacity side-infrastructure capable of processing payment transactions with sub-second latency, thus making blockchain technologies practical for retail payments at physical points of sale.
%
This last chapter presents \coconut, the final component of this thesis that reveals the full potential of \chainspace's execution model by presenting a number of decentralized and scalable privacy-preserving applications implemented as \chainspace smart contracts. \coconut is selective disclosure credential scheme that natively integrates with blockchains; \Cref{fig:coconut:global-overview} places \coconut in the big picture of this thesis by instantiating it as a side-infrastructure of \chainspace. The black circles in \Cref{fig:coconut:global-overview} represent \coconut nodes and the white bidirectional arrow illustrates the communication between \coconut and \chainspace allowing users to obtain and use \coconut credentials within \chainspace smart contracts.

Selective disclosure credentials~\cite{cl,amac} allow the issuance of a credential to a user, and the subsequent unlinkable revelation (or `showing') of some of the attributes it encodes to a verifier for the purposes of authentication, authorization or to implement electronic cash. However, as explained in \Cref{sec:literature-review:coconut-related}, established schemes have shortcomings. Some entrust a single issuer with the credential signature key, allowing a malicious issuer to forge any credential or electronic coin. Other schemes do not provide the necessary efficiency, re-randomization, or blind issuance properties necessary to implement practical selective disclosure credentials. No existing scheme provides all of efficiency, threshold distributed issuance, private attributes, re-randomization, and unlinkable multi-show selective disclosure. This is especially troublesome in applications related to e-cash or token schemes, since the issuer is effectively given a license to generate coins, and the unlinkable nature of the showing protocols prevents auditors from detecting such behavior.
Moreover, the security models of these systems generally assume that integrity should hold in the presence of a threshold number of dishonest or faulty nodes (Byzantine fault tolerance); it is desirable for similar assumptions to hold for multiple credential issuers (threshold issuance). Integrity should hold when a threshold of infrastructure nodes are honest; however, not all authorities are expected to be online or honest.
The lack of efficient general purpose selective disclosure credentials impacts platforms that support `smart contracts',
such as \ethereum~\cite{ethereum}, \hyperledger~\cite{hyperledger} and \chainspace~\cite{chainspace}. 
They all share the limitation that verifiable smart contracts may only perform operations recorded on a public \blockchain.
Thus, smart contracts themselves cannot execute operations requiring secrets, such as issuing signatures or credentials. \chainspace overcomes this limitation by providing an infrastructure where smart contracts are partially executed client-side, and thus allowing a framework where they can operate on the users' secret inputs. However, the nodes backing the ledger can only operate on public data, which prevents smart contracts from directly accessing any node-side secret input. 

Issuing credentials through smart contracts would be very desirable: a smart contract could conditionally issue user credentials depending on the state of the \blockchain, or attest some claim about a user operating through the contract---such as their identity, attributes, or even the balance of their wallet. 
This is not possible, as current selective credential schemes would either entrust a single party as an issuer, or would not provide appropriate efficiency, re-randomization, blind issuance and selective disclosure capabilities (as in the case of threshold signatures~\cite{back2014enabling}). For example, the \hyperledger system supports CL credentials~\cite{cl} through a trusted third party issuer, illustrating their usefulness, but also their fragility against the issuer becoming malicious. 
%Garman~\etal~\cite{dac} present a decentralized anonymous credentials system integrated into distributed ledgers; they provide the ability to issue publicly verifiable claims without central issuers, but do not focus on threshold issuance or on general purpose credentials, and showing credentials requires expensive double discrete-logarithm proofs.

\coconut addresses these challenges, and allows a subset of decentralized mutually distrusting authorities to jointly issue credentials, on public or private attributes. Those credentials cannot be forged by users, or any small subset of potentially corrupt authorities. Credentials can be re-randomized before selected attributes are shown to a verifier, protecting privacy even in the case in which all authorities and verifiers collude. The \coconut scheme is based on a threshold issuance signature scheme that allows partial claims to be aggregated into a single credential. Mapped to the context of permissioned and semi-permissioned \blockchains, \coconut allows collections of authorities in charge of maintaining a \blockchain, or a side chain~\cite{back2014enabling} based on a federated peg, to jointly issue selective disclosure credentials.
\coconut uses short and computationally efficient credentials, and efficient revelation of selected attributes and verification protocols.
Each partial credential and the consolidated credential is composed of exactly two group elements. The size of the credential remains constant regardless of the number of attributes or authorities/issuers. Furthermore, after a one-time setup phase where the users collect and aggregate a threshold number of verification keys from the authorities, the attribute showing and verification are $O(1)$ in terms of both cryptographic computations and communication of cryptographic material---irrespective of the number of authorities. Our evaluation of the \coconut primitives shows very promising results. Verification takes about 10ms, while signing a private attribute is about 3 times faster.
The latency is about 600 ms when the client aggregates partial credentials from 10 authorities distributed across the world.

% =========
\subsection*{Contributions} 
This chapter makes the following key contributions: 
\begin{itemize}
\item It describes the signature schemes underlying \coconut, including how key generation, distributed issuance, aggregation and verification of signatures operate. The scheme is an extension and hybrid of the Waters signature scheme~\cite{waters}, the BGLS signature~\cite{bgls}, and the signature scheme of Pointcheval and Sanders~\cite{pointcheval}. This is the first general purpose, fully distributed threshold issuance, multi-show credential scheme of which we are aware.
    
\item It uses \coconut to implement a generic smart contract library for \chainspace~\cite{chainspace} and one for \ethereum~\cite{ethereum}, performing public and private attribute issuance, aggregation, randomization and selective disclosure. We evaluate their performance and cost within those platforms.
    
\item It presents the design of three applications using the \coconut contract library: a coin tumbler providing payment anonymity; a privacy preserving electronic petitions; and a proxy distribution system for a censorship resistance system. We implement and evaluate the first two applications on the \chainspace platform, and provide a security and performance evaluation.
\end{itemize}

% =========
\subsection*{Outline} 
\Cref{sec:coconut:architecture} presents an overview of the \coconut system; \Cref{sec:coconut:construction} presents the cryptographic primitives underlying \coconut; \Cref{sec:coconut:implementation} and \Cref{sec:coconut:applications} respectively present an implementation and some example of applications backed by the \coconut system. \Cref{sec:coconut:evaluation} provides the evaluation of the core primitives and of the previously discussed applications; \Cref{sec:coconut:related} presents a comparison with related work; and \Cref{sec:coconut:conclusion} concludes the chapter.

%% file: chapters/coconut/sections/architecture.tex
% =========
% Overview 
% =========
\section{Overview} 
\label{sec:coconut:architecture}

\begin{figure}[t]
\centering
\includegraphics[width=.7\textwidth]{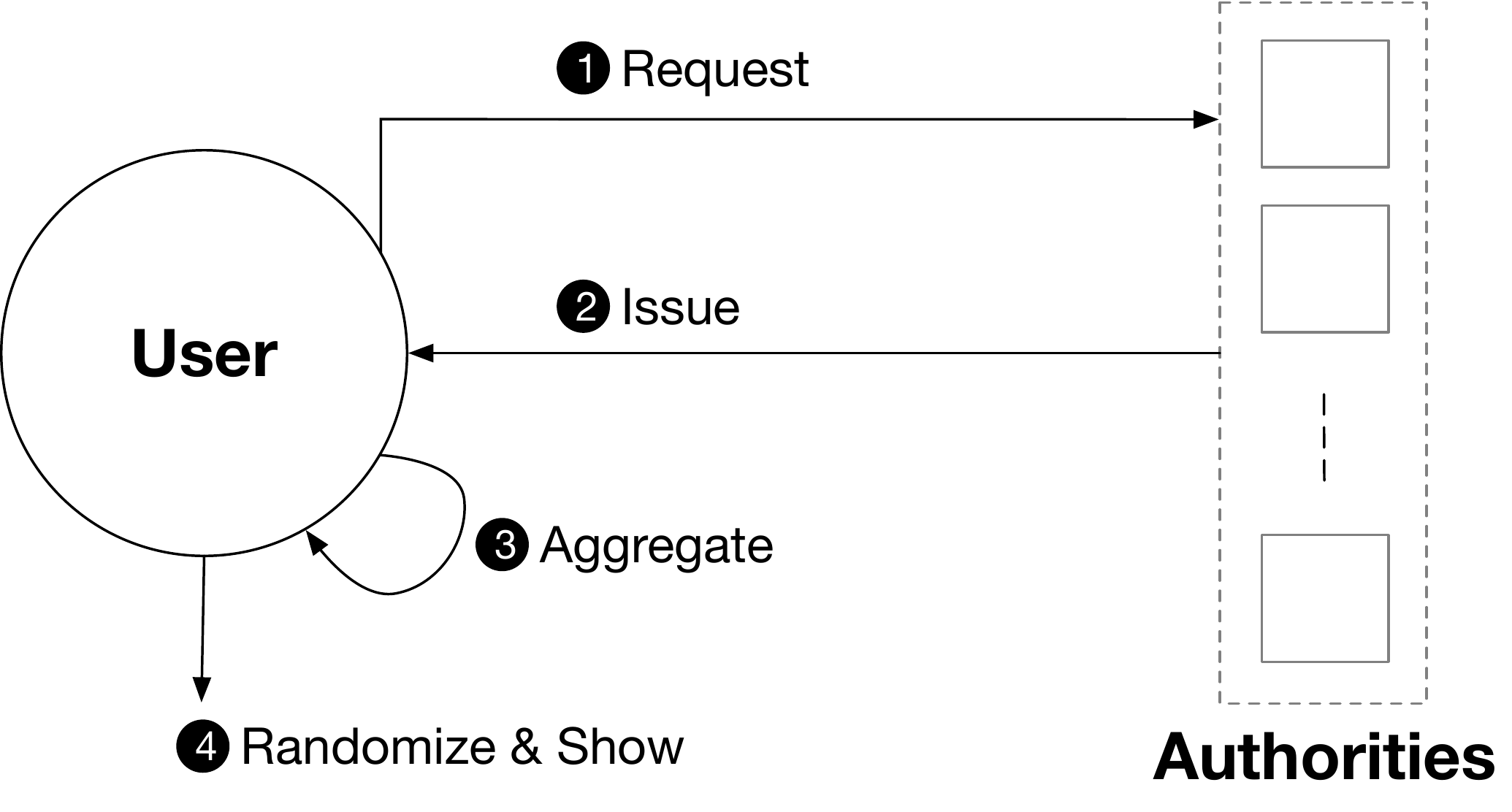}
\caption[Overview of \coconut.]{\footnotesize A high-level overview of \coconut architecture. The user first request a partial credential to each authority and locally aggregates them into a single credential. The user can then randomize and show multiple times its credential to the verifier.}
\label{fig:coconut:gen_arch}
\end{figure}

\coconut is a selective disclosure credential system, supporting threshold credential issuance of public and private attributes, re-randomization of credentials to support multiple unlinkable revelations, and the ability to selectively disclose a subset of attributes. It is embedded into a smart contract library that can be called from other contracts to issue credentials.
The \coconut architecture is illustrated in \Cref{fig:coconut:gen_arch}. Any \coconut user may send a \coconut \emph{request} command to a set of \coconut signing authorities; this command specifies a set of public or encrypted private attributes to be certified into the credential~(\ding{202}). Then, each authority answers with an \emph{issue} command delivering a partial credential~(\ding{203}). Any user can collect a threshold number of shares, aggregate them to form a single consolidated credential, and re-randomize it (\ding{204}). The use of the credential for authentication is however restricted to a user who knows the private attributes embedded in the credential---such as a private key.  The user who owns the credentials can then execute the \emph{show} protocol to selectively disclose attributes or statements about them~(\ding{205}). The showing protocol is publicly verifiable, and may be publicly recorded. \coconut has the following design goals:

\begin{itemize}
\item \textbf{Threshold authorities:} Only a subset of the authorities is required to issue partial credentials in order to allow the users to generate a consolidated credential~\cite{boldyreva2002efficient}. The communication complexity of the \emph{request} and \emph{issue} protocol is thus $O(t)$, where $t$ is the size of the subset of authorities. It is impossible to generate a consolidated credential from fewer than $t$ partial credentials.
\item \textbf{Blind issuance \& unlinkability:} The authorities issue the credential without learning any additional information about the private attributes embedded in the credential. Furthermore, it is impossible to link multiple showings of the credentials with each other, or the issuing transcript, even if all the authorities collude (see \Cref{sec:coconut:definitions}).
\item \textbf{Non-interactivity:} The authorities may operate independently of each other, following a simple key distribution and setup phase to agree on public security and cryptographic parameters---they do not need to synchronize or further coordinate their activities.

\item \textbf{Liveness:} \coconut guarantees liveness as long as a threshold number of authorities remains honest and weak synchrony assumptions holds for the key distribution~\cite{cryptoeprint:2012:377}.

\item \textbf{Efficiency:} The credentials and all zero-knowledge proofs involved in the protocols are short and computationally efficient. After aggregation and re-randomization, the attribute showing and verification involve only a single consolidated credential, and are therefore $O(1)$ in terms of both cryptographic computations and communication of cryptographic material---no matter the number of authorities.
\item \textbf{Short credentials:} Each partial credential---as well as the consolidated credential---is composed of exactly two group elements, no matter the number of authorities or the number of attributes embedded in the credentials.
\end{itemize}

As a result, a large number of authorities may be used to issue credentials, without significantly affecting efficiency.

%% file: chapters/coconut/sections/construction.tex
% =========
% Construction 
% =========
\section{The \coconut Construction} \label{sec:coconut:construction}
We introduce the cryptographic primitives supporting the \coconut architecture, step by step from the design of Pointcheval and Sanders~\cite{pointcheval} and Boneh~\etal~\cite{bls, bgls} to the full \coconut scheme.

\begin{itemize}
\item \textbf{Step 1:} We first recall (\Cref{sec:coconut:pointcheval_recall}) the scheme of Pointcheval~\etal~\cite{pointcheval} for single-attribute credentials. We present its limitations preventing it from meeting our design goals presented in \Cref{sec:coconut:architecture}, and we show how to incorporate principles from Boneh~\etal~\cite{bls} to overcome them. 

\item \textbf{Step 2:} We introduce (\Cref{sec:coconut:threshold_credentials_scheme}) the \emph{\coconut threshold credentials scheme}, which has all the properties of Pointcheval and Sanders~\cite{pointcheval} and Boneh~\etal~\cite{bls}, and allows us to achieve all our design goals.

\item \textbf{Step 3:} Finally, we extend (\Cref{sec:coconut:multi_message_scheme}) our schemes to support credentials embedding $q$ distinct attributes $(m_1,\dots,m_{q})$ simultaneously.
\end{itemize}

% =========
\subsection{Notations \& Assumptions} \label{sec:coconut:background_and_assumptions}
We present the notation used in the rest of the chapter, as well as the security assumptions on which our primitives rely.

\para{Zero-knowledge proofs} 
Our credential scheme uses non-interactive zero-knowledge proofs to assert knowledge and relations over discrete logarithm values. We represent these non-interactive zero-knowledge proofs with the notation introduced by Camenisch~\etal~\cite{camenisch1997proof}:
\begin{equation}\nonumber
\textrm{NIZK}\{(x,y,\dots): \textrm{statements about } x, y, \dots\}
\end{equation}
which denotes proving in zero-knowledge that the secret values $(x,y,\dots)$ (all other values are public) satisfy the statements after the colon. 

\para{Cryptographic assumptions}
\coconut requires groups $(\mathbb{G}_1,\mathbb{G}_2,\mathbb{G}_T)$ of prime order $p$ with a bilinear map $e:\mathbb{G}_1 \times \mathbb{G}_2 \rightarrow \mathbb{G}_T$ and satisfying \first\emph{Bilinearity}, \second\emph{Non-degeneracy}, and \third\emph{Efficiency} as described in \Cref{sec:literature-review:crypto}.
\coconut also relies on a cryptographically secure hash function $\hashtopoint$, hashing an element $\mathbb{G}_1$ into an other element of $\mathbb{G}_1$, namely $\hashtopoint: \mathbb{G}_1\rightarrow\mathbb{G}_1 $. We implement this function by serializing the $(x,y)$ coordinates of the input point and applying a full-domain hash function to hash this string into an element of $\mathbb{G}_1$ (as Boneh~\etal~\cite{bls}).

\para{Threshold and communication assumptions}
\coconut assumes honest majority ($n/2 < t$) to prevent malicious authorities from issuing credentials arbitrarily. \coconut authorities do not need to communicate with each other; users wait for $t$-out-of-$n$ replies (in any order of arrival) and aggregate them into a consolidated credential; thus \coconut implicitly assumes an asynchronous setting. However, our current implementations rely on the distributed key generation protocol of Kate~\etal~\cite{cryptoeprint:2012:377}, which requires \first weak synchrony for liveness (but not for safety), and \second at most one third of dishonest authorities.

% =========
\subsection{Scheme Definitions and Security Properties}\label{sec:coconut:definitions}
We present the protocols that comprise a threshold credentials scheme:
\begin{description}
\item[\definition{Setup($1^\lambda$)}{$params$}] defines the system parameters $params$ with respect to the security parameter $\lambda$. These parameters are publicly available.

\item[\definition{KeyGen($params$)}{$sk,vk$}] is run by the authorities to generate their secret key $sk$ and verification key $vk$ from the public $params$.

\item[\definition{AggKey($vk_1, \dots, vk_t$)}{$vk$}] is run by whoever wants to verify a credential to aggregate any subset of $t$ verification keys $vk_i$ into a single consolidated verification key $vk$. \algorithm{AggKey} needs to be run only once.

\item[\definition{IssueCred($m,\phi$)}{$\sigma$}] is an interactive protocol between a user and each authority, by which the user obtains a credential $\sigma$ embedding the private attribute $m$ satisfying the statement $\phi$.

\item[\definition{AggCred($\sigma_1, \dots, \sigma_t$)}{$\sigma$}] is run by the user to aggregate any subset of $t$ partial credentials $\sigma_i$ into a single consolidated credential.

\item[\definition{ProveCred($vk, m, \phi'$)}{$\Theta,\phi'$}] is run by the user to compute a proof $\Theta$ of possession of a credential certifying that the private attribute $m$ satisfies the statement $\phi'$ (under the corresponding verification key $vk$).

\item[\definition{VerifyCred($vk, \Theta, \phi'$)}{$true/false$}] is run by whoever wants to verify a credential embedding a private attribute satisfying the statement $\phi'$, using the verification key $vk$ and cryptographic material $\Theta$ generated by \textsf{ProveCred}. 
\end{description}

A threshold credential scheme must satisfy the following security properties:
\begin{description}
\item [Unforgeability: ] It must be unfeasible for an adversarial user to convince an honest verifier that they are in possession of a credential if they are in fact not (\ie, if they have not received valid partial credentials from at least $t$ authorities).

\item [Blindness: ] It must be unfeasible for an adversarial authority to learn any information about the attribute $m$ during the execution of the \algorithm{IssueCred} protocol, except for the fact that $m$ satisfies $\phi$.

\item [Unlinkability / Zero-knowledge: ] It must be unfeasible for an adversarial verifier (potentially working with an adversarial authority) to learn anything about the attribute $m$, except that it satisfies $\phi'$, or to link the execution of \algorithm{ProveCred} with either another execution of \algorithm{ProveCred} or with the execution of \algorithm{IssueCred} (for a given attribute $m$).
\end{description}

% =========
\subsection{Foundations of Coconut}\label{sec:coconut:pointcheval_recall}
Before giving the full \coconut construction, we first recall the credentials scheme proposed by Pointcheval and Sanders~\cite{pointcheval}; their construction has the same properties as CL-signatures~\cite{cl} but is more efficient. The scheme works in a bilinear group $(\mathbb{G}_1,\mathbb{G}_2,\mathbb{G}_T)$ of type 3, with a bilinear map $e:\mathbb{G}_1 \times \mathbb{G}_2 \rightarrow \mathbb{G}_T$ as described in \Cref{sec:coconut:background_and_assumptions}. 

\begin{description}
\item[\definition{P.Setup($1^\lambda$)}{$params$}] Choose a bilinear group $(\mathbb{G}_1,\mathbb{G}_2,\mathbb{G}_T)$ with order $p$, where $p$ is a $\lambda$-bit prime number. Let $g_1$ be a generator of $\mathbb{G}_1$, and $g_2$ a generator of $\mathbb{G}_2$. The system parameters are $params=(\mathbb{G}_1, \mathbb{G}_2, \mathbb{G}_T, p, g_1, g_2)$. 

\item[\definition{P.KeyGen($params$)}{$sk,vk$}] Choose a random secret key $sk = (x,y) \in \mathbb{F}_p^2$. Parse $params=(\mathbb{G}_1, \mathbb{G}_2, \mathbb{G}_T, p, g_1, g_2)$, and publish the verification key $vk = (g_2, \alpha,\beta) = (g_2, g_2^x,g_2^y)$.

\item[\definition{P.Sign($params, sk, m$)}{$\sigma$}] Parse $sk = (x, y)$. Pick a random $r \in \mathbb{F}_p$ and set $h=g_1^{r}$. Output $\sigma = (h, s) = (h, h^{x+y\cdot m})$.

\item[\definition{P.Verify($params, vk, m, \sigma$)}{$true/false$}] Parse $vk = (g_2, \alpha,\beta)$ and $\sigma = (h, s)$. Output $true$ if $h\neq1$ and $e(h,\alpha\beta^m)=e(s,g_2)$; otherwise output $false$.
\end{description}

The signature $\sigma=(h,s)$ is randomizable by choosing a random $r' \in \mathbb{F}_p$ and computing $\sigma'=(h^{r'},s^{r'})$. The above scheme can be modified to obtain credentials on a private attribute: to run \algorithm{IssueCred} the user first picks a random $t \in \mathbb{F}_p$, computes the commitment $c_p=g_1^tY^m$ to the message $m$, where $Y=g_1^y$; and sends it to a single authority along with a zero-knowledge proof of the opening of the commitment. The authority verifies the proof, picks a random $u \in \mathbb{F}_p$, and returns $\widetilde{\sigma}=(h,\widetilde{s})=(g^u,(Xc_p)^u)$ where $X=g_1^x$. The user unblinds the signature by computing $\sigma=(h,\widetilde{s}(h)^{-t})$, and this value acts as the credential.

This scheme provides blindness, unlinkability, efficiency and short credentials; but it does not support threshold issuance and therefore does not achieve our design goals. This limitation comes from the \textsf{P.Sign} algorithm---the issuing authority computes the credentials using a private and self-generated random number $r$ which prevents the scheme from being efficiently distributed to a multi-authority setting\footnote{The original paper of Pointcheval and Sanders~\cite{pointcheval} proposes a sequential aggregate signature protocol that is  unsuitable for threshold credentials issuance (see \Cref{sec:coconut:related}).}. To overcome that limitation, we take advantage of a concept introduced by BLS signatures~\cite{bls}; exploiting a hash function $\hashtopoint: \mathbb{F}_p\rightarrow\mathbb{G}_1$ to compute the group  element $h=\hashtopoint(m)$. The next section describes how \coconut incorporates these concepts to achieve all our design goals.

% =========
\subsection{The \coconut Threshold Credential Scheme} \label{sec:coconut:threshold_credentials_scheme}
We introduce the \emph{\coconut} threshold credential scheme, allowing users to obtain a partial credential $\sigma_i$ on a private or public attribute $m$. In a system with $n$ authorities, a $t$-out-of-$n$ threshold credentials scheme offers great flexibility as the users need to collect only $n/2< t \leq n$ of these partial credentials in order to recompute the consolidated credential (both $t$ and $n$ are scheme parameters). 

\para{Cryptographic primitives} For the sake of simplicity, we describe below a key generation algorithm \algorithm{TTPKeyGen} as executed by a trusted third party; this protocol can however be executed in a distributed way as illustrated by Gennaro~\etal~\cite{gennaro1999secure} under a synchrony assumption, and as illustrated by Kate~\etal~\cite{cryptoeprint:2012:377} under a weak synchrony assumption. Adding and removing authorities implies a re-run of the key generation algorithm---this limitation is inherited from the underlying Shamir's secret sharing protocol~\cite{shamir1979share} and can be mitigated using techniques introduced by Herzberg~\etal~\cite{herzberg1995proactive}.

\begin{description}
% Setup
\item[\definition{Setup($1^\lambda$)}{$params$}] Choose a bilinear group $(\mathbb{G}_1,\mathbb{G}_2,\mathbb{G}_T)$ with order $p$, where $p$ is a $\lambda$-bit prime number. Let $g_1, h_1$ be generators of $\mathbb{G}_1$, and $g_2$ a generator of $\mathbb{G}_2$. The system parameters are $params=(\mathbb{G}_1, \mathbb{G}_2, \mathbb{G}_T, p, g_1, g_2, h_1)$. 

% TTPKeyGen
\item[\definition{TTPKeyGen($params, t, n$)}{$sk,vk$}] Pick\footnote{This algorithm can be turned into the \textsf{KeyGen} and \textsf{AggKey} algorithms described in \Cref{sec:coconut:definitions} using techniques illustrated by Gennaro~\etal~\cite{gennaro1999secure} or Kate~\etal~\cite{cryptoeprint:2012:377}.} two polynomials $v,w$ of degree $t-1$ with coefficients in $\mathbb{F}_p$, and set $(x,y) = (v(0), w(0))$. Issue to each authority $i \in [1, \dots, n]$ a secret key $sk_i = (x_i,y_i) = (v(i), w(i))$, and publish their verification key $vk_i$ = $(g_2,\alpha_i,\beta_i) = (g_2,g_2^{x_i},g_2^{y_i})$.

% IssueCred
\item[\definition{IssueCred($m, \phi$)}{$\sigma$}] Credentials issuance is composed of three algorithms:
\begin{description}
\item \definition{PrepareBlindSign($m, \phi$)}{$d,\Lambda,\phi$} The users generate an \elgamal key-pair $(d, \gamma=g_1^{d})$; pick a random $o\in\mathbb{F}_p$,  compute the commitment $c_m$ and the group element $h\in\mathbb{G}_1$ as follows:
\begin{equation}\nonumber
c_m = g_1^m h_1^o \qquad\textrm{and}\qquad h = \hashtopoint(c_m)
\end{equation} 
Pick a random $k \in \mathbb{F}_p$ and compute an \elgamal encryption of $m$ as below:
\begin{equation}\nonumber
c = Enc(h^m)=(g_1^k,\gamma^k h^m)
\end{equation}
Output $(d, \Lambda=(\gamma, c_m, c, \pi_{s}), \phi)$, where $\phi$ is an application-specific predicate satisfied by $m$, and $\pi_{s}$ is defined by:
\begin{align}\nonumber
& \pi_{s} = \textrm{NIZK}\{(d, m, o, k): \gamma = g_1^d \;\land\; c_m=g_1^mh_1^o\\ \nonumber
& \qquad\land\; c = (g_1^k,\gamma^k h^m) \;\land\;  \phi(m)=1\}
\end{align}

\item \definition{BlindSign($sk_i, \Lambda, \phi$)}{$\tilde{\sigma}_i$} The authority $i$ parses $\Lambda=(\gamma, c_m, c, \pi_{s})$, $sk_i=(x_i,y_i)$, and $c=(a,b)$. Recompute $h = \hashtopoint(c_m)$. Verify the proof  $\pi_{s}$ using $\gamma$, $c_m$ and $\phi$; if the proof is valid, build $\tilde{c}_i=(a^y,h^{x_i}b^{y_i})$ and output $\tilde{\sigma}_i = (h, \tilde{c}_i)$; otherwise output $\perp$ and stop the protocol.

\item \definition{Unblind($\tilde{\sigma}_i, d$)}{$\sigma_i$} The users parse $\tilde{\sigma}_i=(h, \tilde{c})$ and $\tilde{c}=(\tilde{a},\tilde{b})$; compute $\sigma_i = (h,\tilde{b}(\tilde{a})^{-d})$. Output $\sigma_i$.
 \end{description}
 
% AggCred
\item[\definition{AggCred($\sigma_1, \dots, \sigma_t$)}{$\sigma$}] Parse each $\sigma_i$ as $(h,s_i)$ for $i \in [1, \dots, t]$. Output $(h,\prod^t_{i=1} s_i^{l_i})$, where $l$ is the Lagrange coefficient:
\begin{equation}\nonumber
l_i = \left[\prod^t_{j=1, j\neq i} (0-j)\right] \left[\prod^t_{j=1, j\neq i} (i-j)\right]^{-1} \;\textrm{mod}\; p
\end{equation}

% ProveCred
\item[\definition{ProveCred($vk, m, \sigma, \phi'$)}{$\Theta,\phi'$}] Parse $\sigma=(h,s)$ and $vk=(g_2,\alpha,\beta)$. Pick at random $r',r \in \mathbb{F}_p^2$; set $\sigma'=(h',s')=(h^{r'},s^{r'})$; build $\kappa = \alpha\beta^m g_2^r$ and $\nu=\left(h'\right)^r$. Output $(\Theta=(\kappa, \nu, \sigma',\pi_v),\phi')$, where $\phi'$ is an application-specific predicate satisfied by $m$, and $\pi_v$ is:
\begin{equation}\nonumber
    \pi_v=\textrm {NIZK}\{(m,r): \kappa=\alpha\beta^m g_2^r \ \land \ \nu=\left(h'\right)^r \ \land \  \phi'(m)=1\} 
\end{equation}

% VerifyCred
\item[\definition{VerifyCred($vk, \Theta, \phi'$)}{$true/false$}] Parse $\Theta = (\kappa, \nu, \sigma',\pi_v)$ and $\sigma'=(h',s')$; verify $\pi_v$ using $vk$ and $\phi'$. Output $true$ if the proof verifies, $h'\neq1$ and $e(h',\kappa)=e(s'\nu,g_2)$; otherwise output $false$.
\end{description}

\para{Correctness and explanation} The \algorithm{Setup} algorithm generates the public parameters. Credentials are elements of $\mathbb{G}_1$, while verification keys are elements of $\mathbb{G}_2$. \Cref{fig:coconut:protocol_priv} illustrates the protocol exchanges.
To keep an attribute $m \in \mathbb{F}_p$ hidden from the authorities, the users run \algorithm{PrepareBlindSign} to produce $\Lambda=(\gamma, c_m, c, \pi_{s})$. They create an \elgamal keypair $(d, \gamma=g_1^{d})$, pick a random  $o \in \mathbb{F}_p$, and compute a commitment $c_m=g_1^mh_1^o$. Then, the users compute $h=\hashtopoint(c_m)$ and  the encryption of $h^m$ as below:
\begin{equation}\nonumber
c = Enc(h^m) = (a, b) = (g_1^k,\gamma^kh^m),
\end{equation} 
where $k \in \mathbb{F}_p$. Finally, the users send $(\Lambda, \phi)$ to the signer, where $\pi_{s}$ is a zero-knowledge proof ensuring that $m$ satisfies the application-specific predicate $\phi$, and correctness of $\gamma, c_m,c$~(\ding{202}). All the zero-knowledge proofs required by \coconut are based on standard sigma protocols to show knowledge of representation of discrete logarithms; they are based on the DH assumption~\cite{camenisch1997proof} and do not require any trusted setup.
To blindly sign the attribute, each authority $i$ verifies the proof $\pi_{s}$, and uses the homomorphic properties of \elgamal to generate an encryption $\tilde{c}$ of $h^{x_i+y_i\cdot m}$ as below:
\begin{equation}\nonumber
\tilde{c} = (a^y, h^{x_i} b^{y_i}) = (g_1^{ky_i}, \gamma^{ky_i}h^{x_i+y_i\cdot m})
\end{equation}

Note that every authority must operate on the same element $h$. Intuitively, generating $h$ from $h=\hashtopoint(c_m)$ is equivalent to computing $h=g_1^{\tilde{r}}$  where $\tilde{r} \in \mathbb{F}_p$ is unknown by the users (as in Pointcheval and Sanders~\cite{pointcheval}). However, since $h$ is deterministic, every authority can uniquely derive it in isolation and forgeries are prevented since different $m_0$ and $m_1$ cannot lead to the same value of $h$.\footnote{If an adversary $\mathcal{A}$ can obtain two credentials $\sigma_0$ and $\sigma_1$ on respectively $m_0=0$ and $m_1=1$ with the same value $h$ as follows: $\sigma_0 = h^{x} \quad \textrm{and} \quad \sigma_1=h^{x+y}$; then $\mathcal{A}$ could forge a new credential $\sigma_2$ on $m_2=2$:
$\sigma_2 = (\sigma_0)^{-1} \sigma_1 \sigma_1 = h^{x+2y}$.}
As described in \Cref{sec:coconut:pointcheval_recall}, the blind signature scheme of Pointcheval and Sanders builds the credentials directly from a commitment of the attribute and a blinding factor secretly chosen by the authority; this is unsuitable for issuance of threshold credentials. We circumvent that problem by introducing the \elgamal ciphertext $c$ in our scheme and exploiting its homomorphism, as described above.

\begin{figure}[t]
\centering
\includegraphics[width=.7\textwidth]{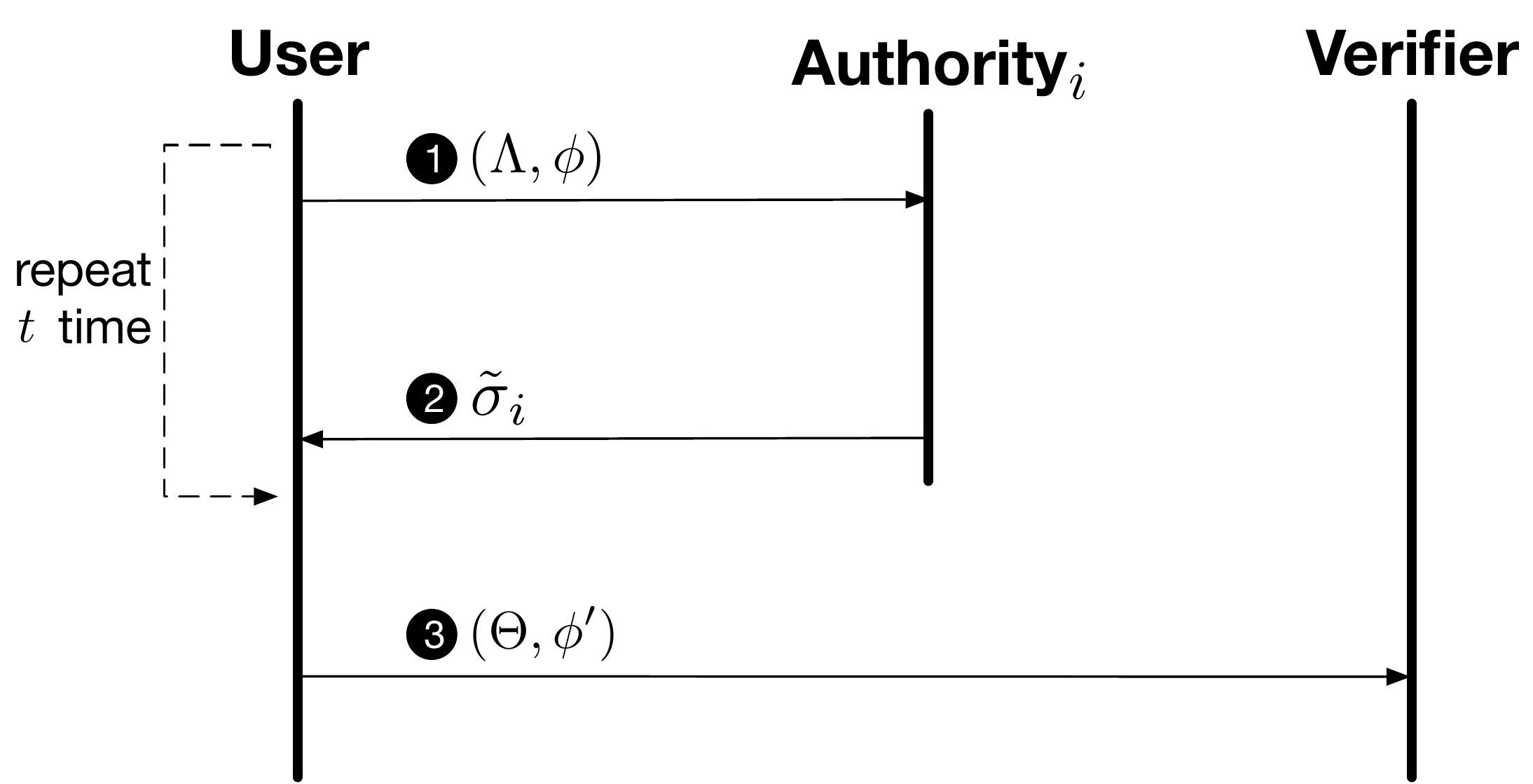}
\caption[\coconut threshold credentials protocol exchanges.]{\coconut threshold credentials protocol exchanges. The user first provides the authorities with the cryptographic material $\Lambda$ embedding its private attributes, and each authority replies with a blinded partial credential $\tilde{\sigma}_i$. The user aggregates all $\tilde{\sigma}_i$ into a single credential and unblinds it. Finally, the user can show its credentials by providing the verifier with the cryptographic material $\Theta$.}
\label{fig:coconut:protocol_priv}
\end{figure}

Upon reception of $\tilde{c}$, the users decrypt it using their \elgamal private key $d$ to recover the partial credentials $\sigma_i = (h, h^{x_i+y_i\cdot m})$; this is performed by the \algorithm{Unblind} algorithm (\ding{203}). Then, the users can call the \algorithm{AggCred} algorithm to aggregate any subset of $t$ partial credentials. This algorithm uses the Lagrange basis polynomial $l$ which allows to reconstruct the original $v(0)$ and $w(0)$ through polynomial interpolation;
\begin{equation}\nonumber
v(0) = \sum^t_{i=1} v(i)l_i \quad \textrm{and} \quad w(0) = \sum^t_{i=1} w(i)l_i
\end{equation}
However, this computation happens in the exponent---neither the authorities nor the users should know the values $v(0)$ and $w(0)$. One can easily verify the correctness of \algorithm{AggCred} of $t$ partial credentials $\sigma_i=(h_i,s_i)$ as below.
\begin{align} \nonumber
	s & = \prod^t_{i=1} \left(s_i\right)^{l_i} = \prod^t_{i=1} \left(h^{x_i+y_i\cdot m}\right)^{l_i} \\ \nonumber
	&= \prod^t_{i=1} \left(h^{x_i}\right)^{l_i} \prod^t_{i=1} \left(h^{y_i\cdot m}\right)^{l_i} = \prod^t_{i=1} h^{(x_i l_i)} \prod^t_{i=1} h^{(y_i l_i)\cdot m} \\ \nonumber
	&= h^{v(0)+w(0)\cdot m} = h^{x+y\cdot m}
\end{align}
Before verification, the verifier collects and aggregates the verifications keys of the authorities---this process  happens only once and ahead of time. The algorithms \algorithm{ProveCred} and \algorithm{VerifyCred} implement verification. First, the users randomize the credentials by picking a random $r' \in \mathbb{F}_p$ and computing $\sigma'=(h',s')=(h^{r'},s^{r'})$; then, they compute $\kappa$ and $\nu$ from the attribute $m$, a blinding factor $r\in\mathbb{F}_p$ and the aggregated verification key:
\begin{equation*}
\kappa=\alpha\beta^m g_2^r \qquad\textrm{and}\qquad \nu=(h')^r
\end{equation*}
Finally, they send $\Theta=(\kappa, \nu, \sigma', \pi_v)$ and $\phi'$ to the verifier where $\pi_v$ is a zero-knowledge proof asserting the correctness of $\kappa$  and $\nu$; and that the private attribute $m$ embedded into $\sigma$ satisfies the application-specific predicate $\phi'$~(\ding{204}). The proof $\pi_v$ also ensures that the users actually know $m$ and that $\kappa$ has been built using the correct verification keys and blinding factors. The pairing verification is similar to Pointcheval and Sanders~\cite{pointcheval} and Boneh~\etal~\cite{bls}; expressing $h'=g_1^{\tilde{r}} \; | \; \tilde{r} \in \mathbb{F}_p$, the left-hand side of the pairing verification can be expanded as:
\begin{equation*}
    e(h',\kappa) = e(h',g_2^{(x+my+r)}) = e(g_1,g_2)^{(x+my+r) \tilde{r}}
\end{equation*}
and the right-hand side:
\begin{equation*}
    e(s' \nu,g_2) = e(h'^{(x+my+r)},g_2) = e(g_1,g_2)^{(x+my+r)  \tilde{r}}
\end{equation*}
From where the correctness of \algorithm{VerifyCred} follows.

\para{Security}
The proof system we require is based on standard sigma protocols to show knowledge of representation of discrete logarithms, and can be rendered non-interactive using the Fiat-Shamir heuristic~\cite{fiat1986prove} in the random oracle model.  As our signature scheme is derived from the ones due to Pointcheval and Sanders~\cite{pointcheval} and BLS~\cite{bls}, we inherit their assumptions; namely, LRSW~\cite{lysyanskaya1999pseudonym} and XDH~\cite{bls}.

\begin{theorem} \label{th:coconut:coconut_theorem}
Assuming LRSW, XDH, and the existence of random oracles, \coconut is a secure threshold credentials scheme, meaning it satisfies unforgeability (as long as fewer than $t$ authorities collude), blindness, and unlinkability.
\end{theorem}
\noindent A sketch of this proof, based on the security of the underlying components of \coconut, can be found in \Cref{sec:coconut:security_proofs}. \coconut guarantees unforgeability as long as less than $t$ authorities collude ($t>n/2$), and guarantees blindness and unlinkability no matter how many authorities collude (and even if the verifier colludes with the authorities).

% =========
\subsection{Multi-Attribute Credentials}\label{sec:coconut:multi_message_scheme}
We expand our scheme to embed multiple attributes into a single credential without increasing its size; this generalization follows directly from the Waters signature scheme~\cite{waters} and Pointcheval and Sanders~\cite{pointcheval}. The authorities' key pairs become:
\begin{equation}\nonumber
sk= (x,y_1,\dots,y_{q}) \quad\textrm{and}\quad vk = (g_2,g_2^x,g_2^{y_1}, \dots, g_2^{y_{q}}) 
\end{equation}
where $q$ is the number of attributes. The multi-attribute credential is derived from the commitment $c_m$ and the group element $h$ as below:
\begin{equation*}
c_m = g_1^o \prod_{j=1}^{q} h_j^{m_j} \qquad\textrm{and}\qquad h = \hashtopoint(c_m)
\end{equation*}
and the credential generalizes as follows:
\begin{equation*}
\sigma = (h,h^{x+\sum_{j=1}^{q} m_j y_j})
\end{equation*}
The credential's size does not increase with the number of attributes or authorities---it is always composed of two group elements. The security proof of the multi-attribute scheme relies on a reduction against the single-attribute scheme and is analogous to Pointcheval and Sanders~\cite{pointcheval}. Moreover, it is also possible to combine public and private attributes to keep only a subset of the attributes hidden from the authorities, while revealing some others; the \algorithm{BlindSign} algorithm only verifies the proof $\pi_{s}$ on the private attributes (similar to Chase~\etal~\cite{amac}).
If the credentials include only non-random attributes, the verifier could guess its value by brute-forcing the verification algorithm\footnote{Let assume for example that some credentials include a single attribute $m$ representing the age of the user; the verifier can run the verification algorithm $e(h,\kappa(\alpha\cdot\beta^m)^{-1})=e(\nu,g_2)$ for every $m \in [1,100]$ and guess the value of $m$.}. This issue is prevented by always embedding a private random attribute into the credentials, that can also act as the authorization key for the credential.

\para{Cryptographic primitives} 
As in \Cref{sec:coconut:threshold_credentials_scheme}, we describe below a key generation algorithm \algorithm{TTPKeyGen} as executed by a trusted third party; this protocol can however be execute in a distributed way as illustrated by Kate~\etal~\cite{cryptoeprint:2012:377}.
\begin{description}
\item[\definition{Setup($1^\lambda,q$)}{$params$}] Choose a bilinear group $(\mathbb{G}_1,\mathbb{G}_2,\mathbb{G}_T)$ with order $p$, where $p$ is an $\lambda$-bit prime number. Let $g_1, h_1, \dots, h_{q}$ be generators of $\mathbb{G}_1$, and $g_2$ a generator of $\mathbb{G}_2$. The system parameters are $params=(\mathbb{G}_1, \mathbb{G}_2, \mathbb{G}_T, p, g_1, g_2, h_1, \dots, h_{q} )$. 

% TTPKeyGen
\item[\definition{TTPKeyGen($params, t, n, q$)}{$sk,vk$}] Choose $(q+1)$ polynomials of degree $(t-1)$ with coefficients in $\mathbb{F}_p$, noted $(v, w_1, \dots, w_{q})$,  and set:
\begin{equation*}
(x,y_1, \dots y_{q}) = (v(0),w_1(0), \dots, w_{q}(0))
\end{equation*}
Issue a secret key $sk_i$ to each authority $i \in [1, \dots, n]$ as below:
\begin{equation*}
sk_i = (x_i,y_{i,1}, \dots, y_{i,q}) = (v(i), w_{1}(i),\dots, w_{q}(i))
\end{equation*}
and publish their verification key $vk_i$ computed as follows:
\begin{equation*}
vk_i = (g_2,\alpha_i,\beta_{i,1}, \dots, \beta_{i,q)}) = (g_2,g_2^{x_i},g_2^{y_{i,1}}, \dots, g_2^{y_{i,1q}})
\end{equation*}

% IssueCred
\item[\definition{IssueCred($m_1,\dots,m_{q}, \phi$)}{$\sigma$}] Credentials issuance is composed of three algorithms:
\begin{description}
\item \definition{PrepareBlindSign($m_1, \dots, m_q,\phi$)}{$d,\Lambda,\phi$} The users generate an \elgamal key-pair $(d, \gamma=g_1^{d})$; pick a random $o\in\mathbb{F}_p$ compute the commitment $c_m$ and the group element $h\in\mathbb{G}_1$ as follows:
\begin{equation}\nonumber
c_m = g_1^o \prod_{j=1}^{q} h_j^{m_j} \qquad\textrm{and}\qquad h = \hashtopoint(c_m)
\end{equation} 
Pick at random $(k_1,\dots,k_{q}) \in \mathbb{F}_p^{q}$ and compute an \elgamal encryption of each $m_j$ for $\forall j \in [1, \dots, q]$  as below:
\begin{equation}\nonumber
c_j = Enc(h^{m_j})=(g_1^{k_j},\gamma^{k_j} h^{m_j})
\end{equation}
Output ($d,\Lambda=(\gamma, c_m, c_j, \pi_{s}),\phi$) $\forall j \in [1, \dots, q]$, where $\pi_{s}$ is defined by:
\begin{align}\nonumber
\pi_{s} = & \textrm{NIZK}\{(d, m_1,\dots,m_q, o, k_1,\dots,k_q): \gamma = g_1^d \\ \nonumber
 &\land\; c_m = g_1^o \prod_{j=1}^{q} h_j^{m_j}  \;\land\; c_j = (g_1^{k_j,}\gamma^{k_j} h^{m_j}) \\ \nonumber
 & \land\;  \phi(m_1,\dots,m_{q})=1\} \quad \forall j \in [1, \dots, q] 
 \end{align}
 
% BlindSign
\item \definition{BlindSign($sk, \Lambda, \phi$)}{$\tilde{\sigma}_i$} The authority $i$ parses $\Lambda=(\gamma, c_m, c_j, \pi_{s})$ and $c_j=(a_j,b_j)$ $\forall j \in [1, \dots, q]$, and $sk_i=(x,y_1, \dots, y_{q})$. Recompute $h = \hashtopoint(c_m)$. Verify the proof  $\pi_{s}$ using $\gamma,c_m$ and $\phi$. If the proof is invalid, output $\perp$ and stop the protocol; otherwise output $\tilde{\sigma}_i=(h,\tilde{c})$, where $\tilde{c}$ is defined as below:
\begin{equation}\nonumber
\tilde{c}=\left(\prod_{j=1}^{q}a_j^{y_j}, h^x \prod_{j=1}^{q}b_j^{y_j}\right)
\end{equation}

\item \definition{Unblind($\tilde{\sigma}_i, d$)}{$\sigma_i$} The users parse $\tilde{\sigma}_i=(h, \tilde{c})$ and $\tilde{c}=(\tilde{a},\tilde{b})$; compute $\sigma_i = (h,\tilde{b}(\tilde{a})^{-d})$. Output $\sigma_i$.
 \end{description}

% AggCred
\item[\definition{AggCred($\sigma_1, \dots, \sigma_t$)}{$\sigma$}] Parse each $\sigma_i$ as $(h,s_i)$ for $i \in [1, \dots, t]$. Output $(h,\prod^t_{i=1} s_i^{l_i})$, where:
\begin{equation}\nonumber
l_i = \left[\prod^t_{i=1, j\neq i} (0-j)\right] \left[\prod^t_{i=1,s j\neq i} (i-j)\right]^{-1} \;\textrm{mod}\; p
\end{equation}

% ProveCred
\item[\definition{ProveCred($vk, m_1, \dots, m_{q}, \sigma, \phi'$)}{$\sigma',\Theta,\phi'$}] Parse $\sigma=(h,s)$ and $vk=(g_2,\alpha,\beta_{1}, \dots, \beta_{q})$. Pick at random $r',r\in\mathbb{F}_q^2$; set $\sigma'=(h',s')=(h^{r'},s^{r'})$, and build $\kappa$ and $\nu$ as below:
\begin{equation}\nonumber
\kappa =  \alpha\prod_{j=1}^{q}{\beta_j^{m_j}} g_2^r \qquad\textrm{and}\qquad \nu= \left(h'\right)^r
\end{equation}
Output $(\Theta=(\kappa, \nu, \sigma', \pi_v),\phi')$, where $\pi_v$ is:
\begin{align}\nonumber
& \pi_{v} = \textrm{NIZK}\{(m_1,\dots,m_q,r): \kappa= \alpha\prod_{j=1}^{q}{\beta_j^{m_j}} g_2^r \\ \nonumber
& \qquad \land \ \nu=\left(h'\right)^r \ \land \  \phi(m_1, \dots, m_{q})=1\} 
\end{align}

% VerifyCred
\item[\definition{VerifyCred($vk, \Theta, \phi'$)}{$true/false$}] Parse $\Theta=(\kappa, \nu, \sigma',  \pi_v)$ and $\sigma'=(h',s')$; verify $\pi_v$ using $vk$ and $\phi'$; Output $true$ if the proof verifies, $h'\neq1$ and $e(h',\kappa)=e(s'\nu,g_2)$; otherwise output $false$.
\end{description}

%% file: chapters/coconut/sections/security.tex
% =========
% Security 
% =========
\section{Sketch of Security Proofs} \label{sec:coconut:security_proofs}
This section sketches the security proofs of the cryptographic construction described in \Cref{sec:coconut:construction}.

\para{Unforgeability} There are two possible ways for an adversary to forge a proof of a credential: \first an adversary without a valid credential nevertheless manages to form a proof such that \algorithm{VerifyCred} passes; and \second, an adversary that has successfully interacted with fewer than $t$ authorities generates a valid consolidated credential (of which they then honestly prove possession using \algorithm{ProveCred}).

Unforgeability in scenario \first is ensured by the soundness property of the zero-knowledge proof.  For scenario \second, running \algorithm{AggCred} involves performing Lagrange interpolation.  If an adversary has fewer than $t$ partial credentials, then they have fewer than $t$ points, which makes the resulting polynomial (of degree $t-1$) undetermined and information-theoretically impossible to compute.  The only option available to the adversary is thus to forge the remaining credentials directly.  This violates the unforgeability of the underlying blind signature scheme, which was proved secure by Pointcheval and Sanders~\cite{pointcheval} under the LRSW assumption~\cite{lysyanskaya1999pseudonym}.

\para{Blindness} Blindness follows directly from the blindness of the signature scheme used during \algorithm{IssueCred}, which was largely proved secure by Pointcheval and Sanders~\cite{pointcheval} under the XDH assumption~\cite{bls}.  There are only two differences between their protocol and ours.

First, the \coconut authorities generate the credentials from a group element $h=\hashtopoint(c_m)$ instead of from $g_1^{\tilde{r}}$ for random $\tilde{r} \in \mathbb{F}_p$.  The hiding property of the commitment $c_m$, however, ensures that $\hashtopoint(c_m)$ does not reveal any information about $m$. Second, Pointcheval and Sanders use a commitment to the attributes as input to \algorithm{BlindSign} (see \Cref{sec:coconut:pointcheval_recall}), whereas \coconut uses an encryption instead.  The IND-CPA property, however, of the encryption scheme ensures that the ciphertext also reveals no information about $m$.
Concretely, \coconut uses Pedersen Commitments~\cite{pedersen} for the commitment scheme, which is secure under the discrete logarithm assumption.  It uses \elgamal for the encryption scheme in $\mathbb{G}_1$, which is secure assuming DDH.  Finally, it relies on the blindness of the Pointcheval and Sanders signature, which is secure assuming XDH~\cite{bls}.  As XDH implies both of the previous two assumptions, our entire blindness argument is implied by XDH.

\para{Unlinkability / Zero-knowledge} Unlinkability and zero-knowledge are guaranteed under the XDH assumption~\cite{bls}. The zero-knowledge property of the underlying proof system ensures that \algorithm{ProveCred} does not on its own reveal anything more than the validity of the statement $\phi'$, which may include public attributes (see \Cref{sec:coconut:multi_message_scheme}. The fact that credentials are re-randomized at the start of \algorithm{ProveCred} in turn ensures unlinkability, both between different executions of \algorithm{ProveCred} and between an execution of \algorithm{ProveCred} and of \algorithm{IssueCred}.

%% file: chapters/coconut/sections/implementation.tex
% =========
% Implementation
% =========
\section{Implementation}\label{sec:coconut:implementation}
We implement a Python library for \coconut as described in \Cref{sec:coconut:construction} and publish the code on GitHub as an open-source project\footnote{\url{https://github.com/asonnino/coconut}}.
We also implement a smart contract library in Chainspace~\cite{chainspace} to enable other application-specific smart contracts (see \Cref{sec:coconut:applications}) to conveniently use our cryptographic primitives.
We present the design and implementation of the \coconut smart contract library in Section~\ref{sec:coconut:smart_contract_library}.
In addition, we implement and evaluate some of the functionality of the \coconut smart contract library in \ethereum~\cite{ethereum} (\Cref{sec:coconut:ethereum_smart_contract_library}).
 Finally, we show how to integrate \coconut into existing semi-permissioned \blockchains (\Cref{sec:coconut:integrations_into_ledgers}). 

% =========
\subsection{The \coconut Smart Contract Library}  \label{sec:coconut:smart_contract_library}
We implement the \coconut smart contract in \chainspace\footnote{\url{https://github.com/asonnino/coconut-chainspace}} (which can be used by other application-specific smart contracts) as a library to issue and verify randomizable threshold credentials through cross-contract calls. 
Running this library as an independent smart contract enables other application-specific smart contracts to securely delegate the credentials issuance and verification processes to the library through a cross-contract call. 
The contract has four functions, \algorithm{(Create, Request, Issue, Verify)}, as illustrated in \Cref{fig:coconut:library}. First, a set of authorities call the \algorithm{Create} function to initialize a \coconut instance defining the \emph{contract info}; \ie their verification key, the number of authorities and the threshold parameter~(\ding{202}). The initiator smart contract can specify a callback contract that needs to be executed by the user in order to request credentials; \eg this callback can be used for authentication. The instance is public and can be read by the user~(\ding{203}); any user can request a credential through the \algorithm{Request} function by executing the specified callback contract, and providing the public and private \emph{attributes} to include in the credentials~(\ding{204}). The public attributes are simply a list of clear text strings, while the private attributes are encrypted as described in \Cref{sec:coconut:threshold_credentials_scheme}. Each signing authority monitors the \blockchain at all times, looking for credential requests. If the request appears on the \blockchain (\ie a transaction is executed), it means that the callback has been correctly executed~(\ding{205}); each authority issues a partial \emph{credential} on the specified attributes by calling the \algorithm{Issue} procedure~(\ding{206}). In our implementation, all partial credentials are in the \blockchain; however, these can also be provided to the user off-chain. Users collect a threshold number of partial credentials, and aggregate them to form a full credential~(\ding{207}). Then, the users locally randomize the credential. The last function of the \coconut library contract is \algorithm{Verify} that allows the \blockchain---and anyone else---to check the validity of a given credential~(\ding{208}). 

\begin{figure}[t]
\centering
\includegraphics[width=\textwidth]{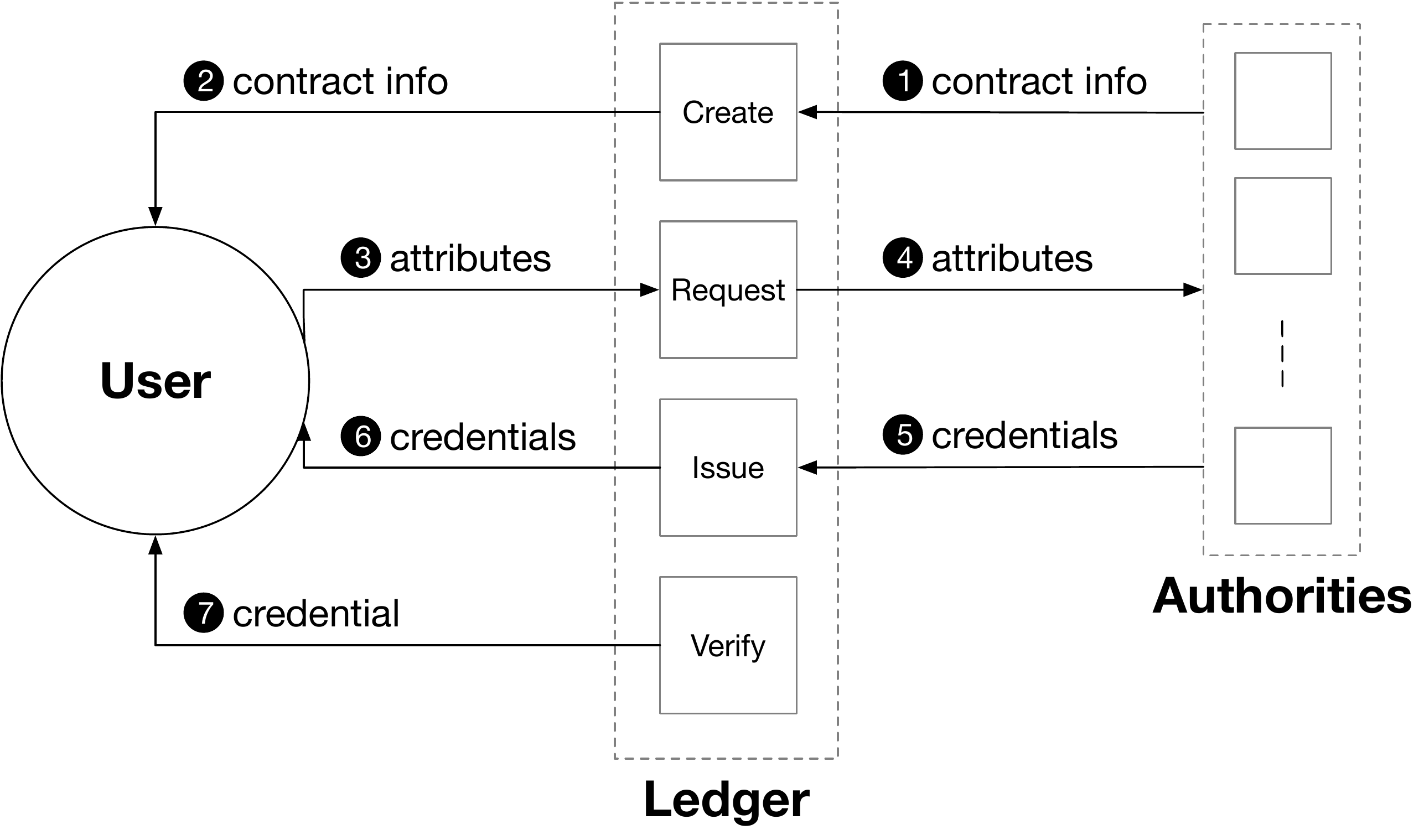}
\caption[The \coconut smart contract library.]{The \coconut smart contract library is composed of four functions: \algorithm{(Create, Request, Issue, Verify)}. \algorithm{Create} initializes a \coconut instance defining a number of public information such that the identities of the authorities. \algorithm{Request} and \algorithm{Issue} respectively allow the users to request credentials and the authorities to issue them on-chain. \algorithm{Verify} allows smart contracts to verify \coconut credentials.}
\label{fig:coconut:library}
\end{figure}

A limitation of this architecture is that it is not efficient for the authorities to continuously monitor the \blockchain. \Cref{sec:coconut:integrations_into_ledgers} explains how to overcome this limitation by embedding the authorities into the nodes running the \blockchain.

% =========
\subsection{\ethereum Smart Contract Library}\label{sec:coconut:ethereum_smart_contract_library}
To make \coconut more widely available, we also implement it in \ethereum---a popular permissionless smart contract \blockchain~\cite{ethereum}.
We release the \coconut \ethereum smart contract as an open source library\footnote{\url{https://github.com/musalbas/coconut-ethereum}}.
The library is written in Solidity, a high-level JavaScript-like language that compiles down to \ethereum Virtual Machine (EVM) assembly code.
\ethereum recently hardcoded a pre-compiled smart contract in the EVM for performing pairing checks and elliptic curve operations on the alt\_bn128 curve~\cite{eip197,eip196}, for efficient verification of zkSNARKs. The execution of an \ethereum smart contract has an associated `gas cost', a fee that is paid to miners for executing a transaction. Gas cost is calculated based on the operations executed by the contract; \ie the more operations, the higher the gas cost. The pre-compiled contracts have lower gas costs than equivalent native \ethereum smart contracts.
We use the pre-compiled contract for performing a pairing check, in order to implement \coconut verification within a smart contract. The \ethereum code only implements elliptic curve addition and scalar multiplication on $\mathbb{G}_1$, whereas \coconut requires operations on $\mathbb{G}_2$ to verify credentials. 
Therefore, we implement elliptic curve addition and scalar multiplication on $\mathbb{G}_2$ as an \ethereum smart contract library written in Solidity that we also release open source\footnote{\url{https://github.com/musalbas/solidity-BN256G2}}. This is a practical solution for many \coconut applications, as verifying credentials with one revealed attribute only requires one addition and one scalar multiplication. It would not be practical however to verify credentials with attributes that will not be revealed---this requires three $\mathbb{G}_2$ multiplications using our elliptic curve implementation, which would exceed the current \ethereum block gas limit (10M as of September 2019).

We can however use the \ethereum contract to design a federated peg for side chains, or a coin tumbler as an \ethereum smart contract, based on credentials that reveal one attribute. We go on to describe and implement this tumbler using the \coconut \chainspace library in \Cref{sec:coconut:tumbler}, however the design for the \ethereum version differs slightly to avoid the use of attributes that will not be revealed. The library shares the same functions as the \chainspace library described in \Cref{sec:coconut:smart_contract_library}, except for \algorithm{Request} and \algorithm{Issue} which are computed off the blockchain to save gas costs.
As \algorithm{Request} and \algorithm{Issue} functions simply act as a communication channel between users and authorities, users can directly communicate with authorities off the \blockchain to request tokens.
 This saves significant gas costs that would be incurred by storing these functions on the \blockchain. The \algorithm{Verify} function simply verifies tokens against \coconut instances created by the \algorithm{Create} function. 

% =========
\subsection{Deeper Blockchain Integration} \label{sec:coconut:integrations_into_ledgers}
The designs described in \Cref{sec:coconut:smart_contract_library} and \Cref{sec:coconut:ethereum_smart_contract_library} rely on authorities on-the-side for issuing credentials. In this section, we present designs that incorporate \coconut authorities within the infrastructure of a number of semi-permissioned \blockchains. This enables the issuance of credentials as a side effect of the normal system operations, taking no additional dependency on extra authorities. It remains an open problem how to embed \coconut into permissionless systems. These systems have a highly dynamic set of nodes maintaining the state of their blockchains, which cannot readily be mapped into \coconut issuing authorities.

Integration of \coconut into \hyperledger Fabric~\cite{hyperledger}---a permissioned \blockchain platform---is straightforward. Fabric contracts run on private sets of computation nodes---and use the Fabric protocols for cross-contract calls. In this setting, \coconut issuing authorities can coincide with the Fabric smart contract authorities. Upon a contract setup, they perform the setup and key distribution, and then issue partial credentials when authorized by the contract. For issuing \coconut credentials, the only secrets maintained are the private issuing keys; all other operations of the contract can be logged and publicly verified. \coconut has obvious advantages over using traditional CL credentials relying on a single authority---as currently present in the \hyperledger roadmap\footnote{\url{http://nick-fabric.readthedocs.io/en/latest/idemix.html}}. The threshold trust assumption---namely that integrity and availability is guaranteed under the corruption of a subset of authorities is preserved, and prevents forgeries by a single corrupted node.
We can also naturally embed \coconut into sharded scalable \blockchains, as exemplified by \chainspace~\cite{chainspace} (which supports general smart contracts), and \omniledger~\cite{omniledger} (which supports digital tokens). 
In these systems, transactions are distributed and executed on `shards' of authorities, whose membership and public keys are known. 
\coconut authorities can naturally coincide with the nodes within a shard---a special transaction type in \omniledger, or a special object in \chainspace, can signal to them that issuing a credential is authorized. 
The authorities,  then issue the partial signature necessary to reconstruct the \coconut credential, and attach it to the transaction they are processing anyway. 
Users can aggregate, re-randomize and show the credential. 

%% file: chapters/coconut/sections/applications.tex
% =========
% Applications 
% =========
\section{Applications} \label{sec:coconut:applications}
In this section, we present three applications that leverage \coconut to offer improved security and privacy properties---a coin tumbler (\Cref{sec:coconut:tumbler}), a privacy-preserving petition system (\Cref{sec:coconut:petition}), and a system for censorship-resistant distribution of proxies (\Cref{sec:coconut:proxy}).  
For generality, the applications assume authorities external to the \blockchain, but these can also be embedded into the \blockchain as described in \Cref{sec:coconut:integrations_into_ledgers}.

% =========
\subsection{Coin Tumbler} \label{sec:coconut:tumbler}
We implement a coin tumbler (or mixer) on \chainspace as depicted in \Cref{fig:coconut:tumbler}. Coin tumbling is a method to mix cryptocurrency associated with an address visible in a public ledger with other addresses, to ``clean'' the coins and obscure the trail back to the coins' original source address. A limitation of previous similar schemes~\cite{mixcoin,blindcoin,tumblebit,coinjoin,coinshuffle,xim,mobius} is that they are either centralized (\ie there is a central authority that operates the tumbler, which may go offline), or require users to coordinate with each other.
The \coconut tumbler addresses these issues \via a distributed design (\ie security relies on a set of multiple authorities that are collectively trusted to contain at least $t$ honest ones), and does not require users to coordinate with each other.
Zcash~\cite{zcash} achieves a similar goal: it theoretically hides the totality of the transaction but at a large computational cost, and offers the option to cheaply send transactions in clear. In practice, the computational overhead of sending hidden transactions makes it impractical, and only a few users take advantage of the optional privacy provided by Zcash; as a result, transactions are easy to de-anonymize~\cite{Kappos}, and recent works aim to reduce the computational overhead of Zcash hidden transactions~\cite{fast-zcash}. \coconut provides efficient proofs taking only a few milliseconds (see \Cref{sec:coconut:evaluation}), and makes hidden transactions practical. Trust assumptions in Zcash are different from \coconut. However, instead of assuming a threshold number of honest authorities, Zcash relies on zk-SNARKs which assumes a setup algorithm executed by a trusted authority\footnote{Recent proposals aim to distribute this trusted setup~\cite{mpc-zcash-1}.}.
M{\"o}bius~\cite{mobius}---which was developed concurrently---is a coin tumbler based on \ethereum smart contracts that achieves strong notions of anonymity and low off-chain communication complexity. M{\"o}bius relies on ring signatures to allow parties to prove group membership without revealing exactly which public key belongs to them.

\begin{figure}[t]
    \centering
    \includegraphics[width=.7\textwidth]{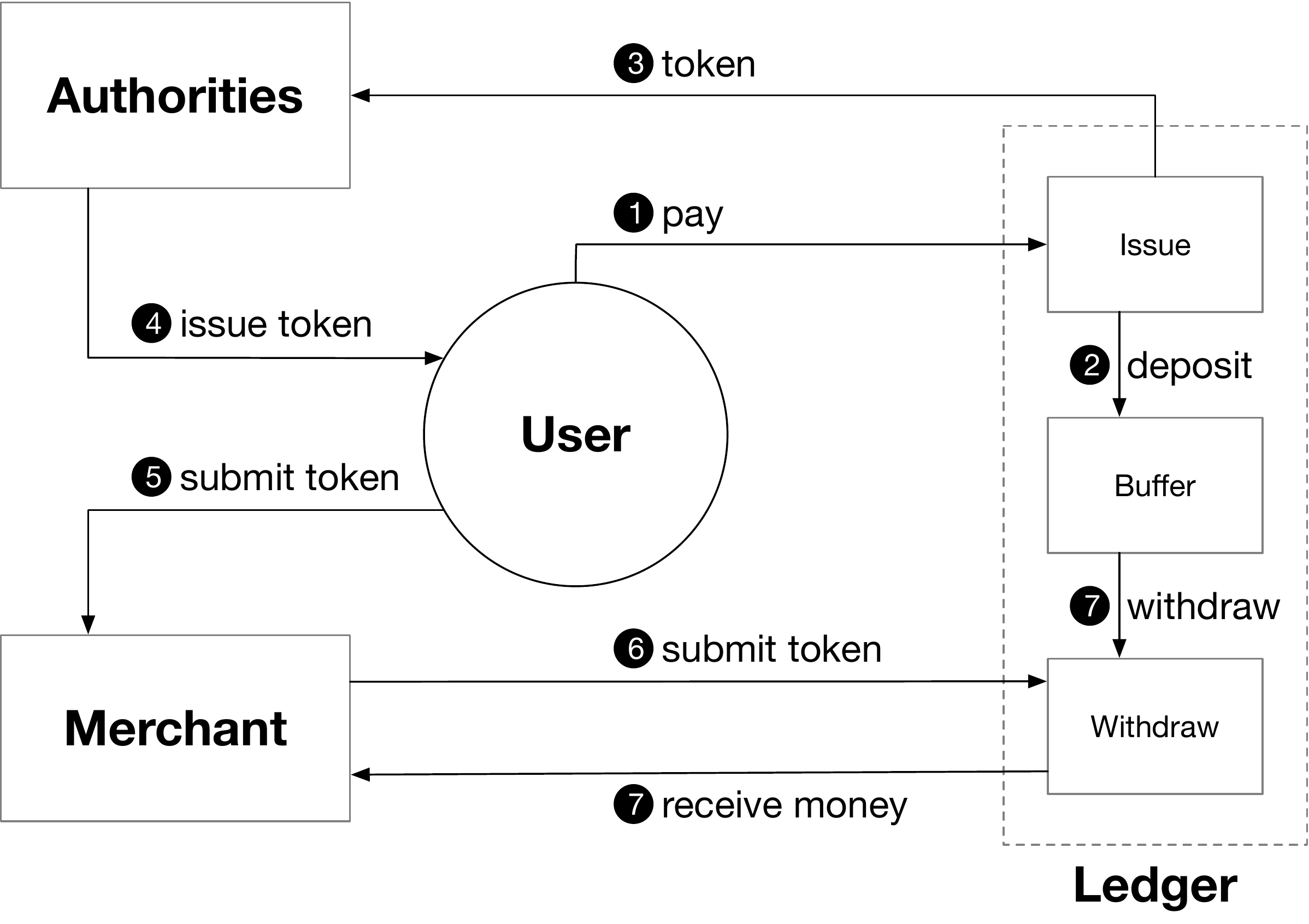}
    \caption[\coconut coin tumbler application.]{\coconut coin tumbler application. The user first deposits funds into a dedicated smart contract, and the authorities issue a credential to the user. The user can later use this credential as a coin token at a merchant, who can redeem the token and withdraw funds from the smart contract.}
    \label{fig:coconut:tumbler}
\end{figure}

\para{\chainspace coin tumbler}
Our tumbler uses \coconut to instantiate a pegged side-chain~\cite{back2014enabling}, providing stronger value transfer anonymity than the original cryptocurrency platform, through unlinkability between issuing a credential representing an e-coin~\cite{chaum1988untraceable}, and spending it. The tumbler application is based on the \coconut contract library and an application specific smart contract called `tumbler'. 
A set of authorities jointly create an instance of the \coconut smart contract as described in \Cref{sec:coconut:smart_contract_library} and specify the smart contract handling the coins of the underlying \blockchain as callback. Specifically, the callback requires a coin transfer to a buffer account. Then users execute the callback and \emph{pay} $v$ coins to the buffer to ask a credential on the public attribute $v$, and on two private attributes: the user's private key $k$ and a randomly generated sequence number $s$~(\ding{202}). Note that to prevent tracing traffic analysis, $v$ should be limited to a specific set of possible values (similar to cash denominations). The request is accepted by the \blockchain only if the user \emph{deposited} $v$ coins to the buffer account~(\ding{203}). 
Each authority monitors the \blockchain and detects the \emph{request}~(\ding{204}); and issues a partial \emph{credential} to the user (either on chain or off-chain)~(\ding{205}). The user aggregates all partial credentials into a consolidated credential, re-randomizes it, and \emph{submits} it as coin token to a merchant. First, the user produces a zk-proof of knowledge of its private key by binding the proof to the merchant's address $addr$; then, the user provides the merchant with the proof along with the sequence number $s$ and the consolidated credential~(\ding{206}). The coins can only be spent with knowledge of the associated sequence number and by the owner of $addr$.
To accept the above as payment, the merchant \emph{submits} the token by showing the credential and a group element $\zeta=g_1^s\in\mathbb{G}_1$ to the tumbler contract along with a zero-knowledge proof ensuring that $\zeta$ is well-formed~(\ding{207}). To prevent double spending, the tumbler contract keeps a record of all elements $\zeta$ that have already been shown. Upon showing a $\zeta$ embedding a fresh (unspent) sequence number $s$, the contract verifies that the credential and zero-knowledge proofs check, and that $\zeta$ doesn't already appear in the spent list. Then it \emph{withdraws} $v$ coins from the buffer~(\ding{208}), sends them to be \emph{received} by the merchant account determined by $addr$, and adds $\zeta$ to the spent list~(\ding{209}). For the sake of simplicity, we keep the transfer value $v$ in clear-text (treated as a public attribute), but this could be easily hidden by integrating a range proof; this can be efficiently implemented using the technique developed by B{\"u}nz~\etal~\cite{bunz2018bulletproofs}.

\para{Security consideration} \coconut provides blind issuance which allows the user to obtain a credential on the sequence number $s$ without the authorities learning its value. Without blindness, any authority seeing the user key $k$ could potentially race the user and the merchant, and spend it---blindness prevents authorities from stealing the token. Furthermore, \coconut provides unlinkability between the \emph{pay} phase~(\ding{202}) and the \emph{submit} phase~(\ding{206}) (see \Cref{fig:coconut:tumbler}), and prevents any authority or third parties from keeping track of the user's transactions. As a result, a merchant can receive payments for good or services offered, yet not identify the purchasers. Keeping a spent list of all elements $\zeta$ prevents double-spending attacks~\cite{karame2012double} without revealing the sequence number $s$; this prevents an attacker from exploiting a race condition in the \emph{submit token} phase~(\ding{207}) and lock user's funds\footnote{An attacker observing a sequence number $s$ during a \emph{submit token} phase~(\ding{207}) could exploit a race condition to lock users fund by quickly buying a token using the same $s$, and spending it before the original \emph{submit token} phase is over.}.
Finally, this application prevents a single authority from creating coins to steal all the money in the buffer. The threshold property of \coconut implies that the adversary needs to corrupt at least $t$ authorities for this attack to be possible. A small subset of authorities cannot block the issuance of a token---the service is guaranteed to be available as long as at least $t$ authorities are running. 

\para{Adapting the coin tumbler to \ethereum}
We extend the example of the tumbler application described above to the \ethereum version of the \coconut library, with a few modifications to reduce the gas costs.
Instead of having $v$ (the number of coins) as an attribute, which would increase the number of elliptic curve multiplications required to verify the credentials, we allow for a fixed number of instances of \coconut to be setup for different denominations for $v$. The Tumbler has a \algorithm{Deposit} method, where users deposit Ether into the contract, and then send an issuance request to authorities on one private attribute: $addr || s$, where $addr$ is the destination address of the merchant, and $s$ is a randomly generated sequence number~(1). It is necessary for $addr$ to be a part of the attribute because once the attribute is revealed, the credential can be spent by anyone with knowledge of the attribute (including any peers monitoring the blockchain for transactions), therefore the credential must be bounded to a specific recipient address before it is revealed. This issuance request is signed by the \ethereum address that deposited the Ether into the smart contract, as proof that the request is associated with a valid deposit, and sent to the authorities~(2). As $addr$ and $s$ will be both revealed at the same time when withdrawing the token, we concatenate these in one attribute to save on elliptic curve operations. Users aggregate the credentials issued by the authorities~(3). The resulting token can then be passed to the \algorithm{Withdraw} function, where the withdrawer reveals $addr$ and $s$~(4). As usual, the contract maintains a map of $s$ values associated with tokens that have already been withdrawn to prevent double-spending. After checking that the token's credentials verifies and that it has not already been spent, the contract sends $v$ to the \ethereum destination address $addr$~(5).

% =========
\subsection{Privacy-Preserving Petition} \label{sec:coconut:petition}
We consider the scenario where several authorities managing the country \emph{C} wish to issue some long-term credentials to its citizens to enable any third party to organize a privacy-preserving petition.  All citizens of \emph{C} are allowed to participate, but should remain anonymous and unlinkable across petitions. This application extends the work of Diaz~\etal~\cite{diaz2008privacy} which does not consider threshold issuance of credentials.

\begin{figure}[t]
\centering
\includegraphics[width=.7\textwidth]{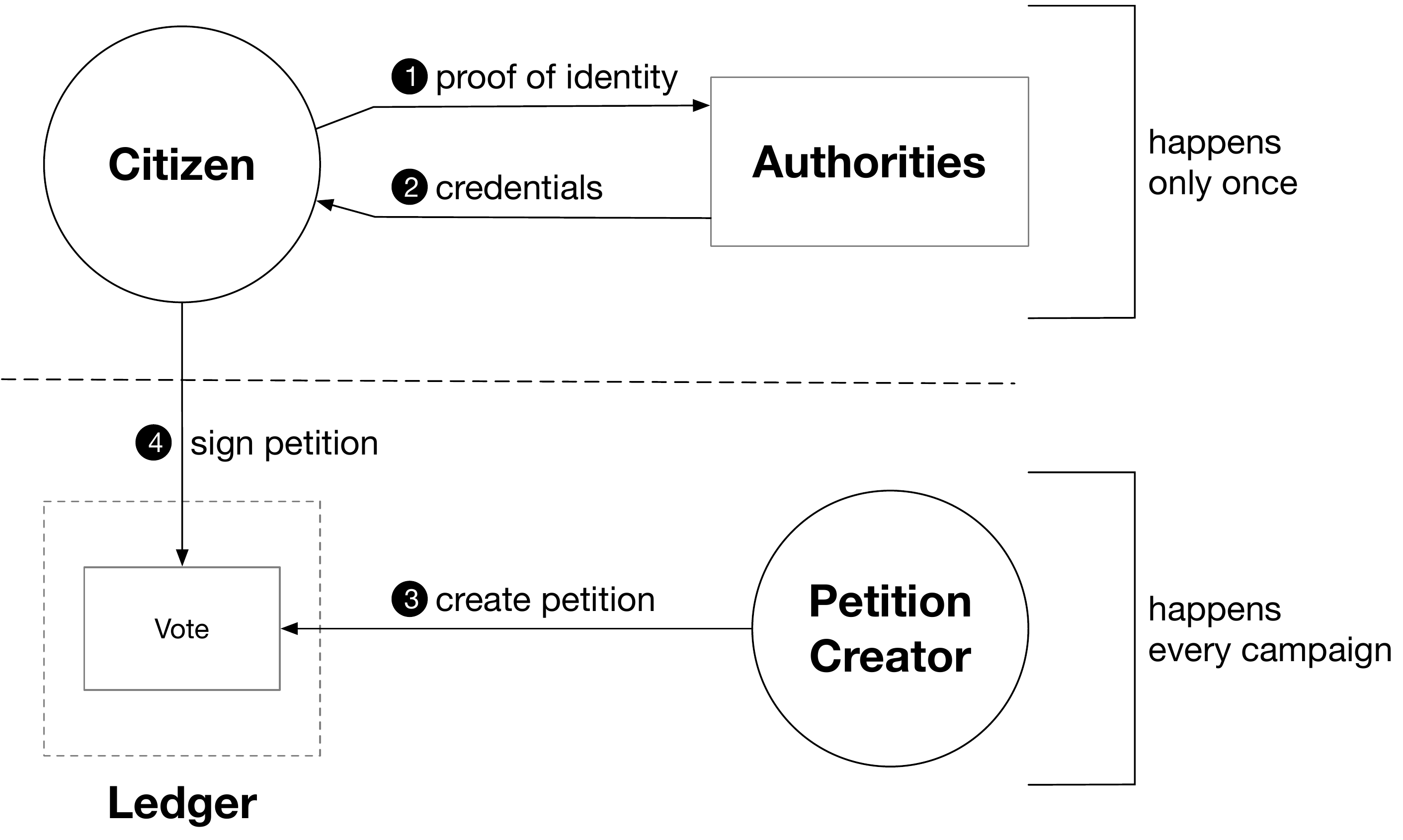}
\caption[\coconut petition application.]{\coconut petition application. The citizen first receives a credential from the authorities, it can then use it to vote on multiple petitions unlinkably.}
\label{fig:coconut:petition}
\end{figure}

\para{\chainspace petition system}
Our petition system is based on the \coconut library contract for \chainspace and a simple smart contract called `petition'.
There are three types of parties: a set of signing authorities representing \emph{C}, a petition initiator, and the citizens of \emph{C}.  
The signing authorities create an instance of the \coconut smart contract as described in \Cref{sec:coconut:smart_contract_library}. As shown in \Cref{fig:coconut:petition}, the citizen provides a \emph{proof of identity} to the authorities~(\ding{202}). 
The authorities check the citizen's identity, and issue a blind and long-term signature on her private key $k$. This signature, which the citizen needs to obtain only once, acts as her long term \emph{credential} to sign any petition~(\ding{203}).
Any third party can \emph{create a petition} by creating a new instance of the petition contract and become the `owner' of the petition. The petition instance specifies an identifier $g_s \in \mathbb{G}_1$ unique to the petition where its representation is unlinkable to the other points of the scheme\footnote{This identifier can be generated through a hash function $\mathbb{F}_p \rightarrow \mathbb{G}_1:\widetilde{H}(s)=g_s \;|\; s\in\mathbb{F}_p$.}, as well as the verification key of the authorities issuing the credentials and any application specific parameters (\eg the options and current votes)~(\ding{204}). In order to \emph{sign} a petition, the citizens compute a value $\zeta=g_s^{k}$. They then adapt the zero-knowledge proof of the \algorithm{ProveCred} algorithm of \Cref{sec:coconut:threshold_credentials_scheme} to show that $\zeta$ is built from the same attribute $k$ in the credential; the petition contract checks the proofs  and the credentials, and checks that the signature is fresh by verifying that $\zeta$ is not part of a spent list. If all the checks pass, it adds the citizens' signatures to a list of records and adds $\zeta$ to the spent list to prevents a citizen from signing the same petition multiple times (prevent double spending)~(\ding{205}). The zero-knowledge proof ensures that $\zeta$ is built from a signed private key $k$, meaning that the users correctly executed the callback to prove that they are citizens of \emph{C}. 

\para{Security consideration} \coconut's blindness property prevents the authorities from learning the citizen's secret key, and misusing it to sign petitions on behalf of the citizen. 
Another benefit is that it lets citizens sign petitions anonymously; citizens only have to go through the issuance phase once, and can then re-use credentials multiple times while staying anonymous and unlinkable across petitions. \coconut allows for distributed credentials issuance, removing a central authority and preventing a single entity from creating arbitrary credentials to sign petitions multiple times. 

% =========
\subsection{Censorship-Resistant Distribution of Proxies} \label{sec:coconut:proxy}
Proxies can be used to bypass censorship, but often become the target of censorship themselves. We present a system based on \coconut for censorship-resistant distribution of proxies (CRS). In our CRS, the volunteer \emph{V} runs proxies, and is known to the \coconut authorities through its long-term public key. The authorities establish reputability of volunteers (identified by their public keys) through an out of band mechanism. The user \emph{U} wants to find proxy IP addresses belonging to reputable volunteers, but volunteers want to hide their identity. As shown in~\Cref{fig:coconut:proxy}, \emph{V} gets an ephemeral public key $pk'$ from the proxy~(\ding{202}), provides \emph{proof of identity} to the authorities~(\ding{203}), and gets a \emph{credential} on two private attributes: the proxy IP address, $pk'$, and one public attribute: the time period $\delta$ for which the credential is valid~(\ding{204}).
\emph{V} shares the credential with the concerned proxy~(\ding{205}), which creates the \emph{proxy info} including $pk'$, $\delta$, and the credential; the proxy `registers' itself by appending this information to the \blockchain along with a zero-knowledge proof and the material necessary to verify the validity of the credential~(\ding{206}).
The users \emph{U} monitor the \blockchain for proxy registrations. When a registration is found, \emph{U} indicates the intent to use a proxy by publishing to the \blockchain a \emph{request info} message which looks as follows: user IP address encrypted under $pk'$ which is embedded in the registration \blockchain entry~(\ding{207}).  
The proxy continuously monitors the \blockchain, and upon finding a user request
addressed to itself, \emph{connects} to \emph{U} and presents proof of knowledge of the private key associated with $pk'$~(\ding{208}). 
\emph{U} verifies the proof, the proxy IP address and its validity period, and then starts relaying its traffic through the proxy.  
The proposed CRS assumes multiple volunteers disposing of numerous IP addresses and that it is unfeasible for a censor to act as user to learn and block all addresses of the volunteers.

\begin{figure}[t]
\centering
\includegraphics[width=.7\textwidth]{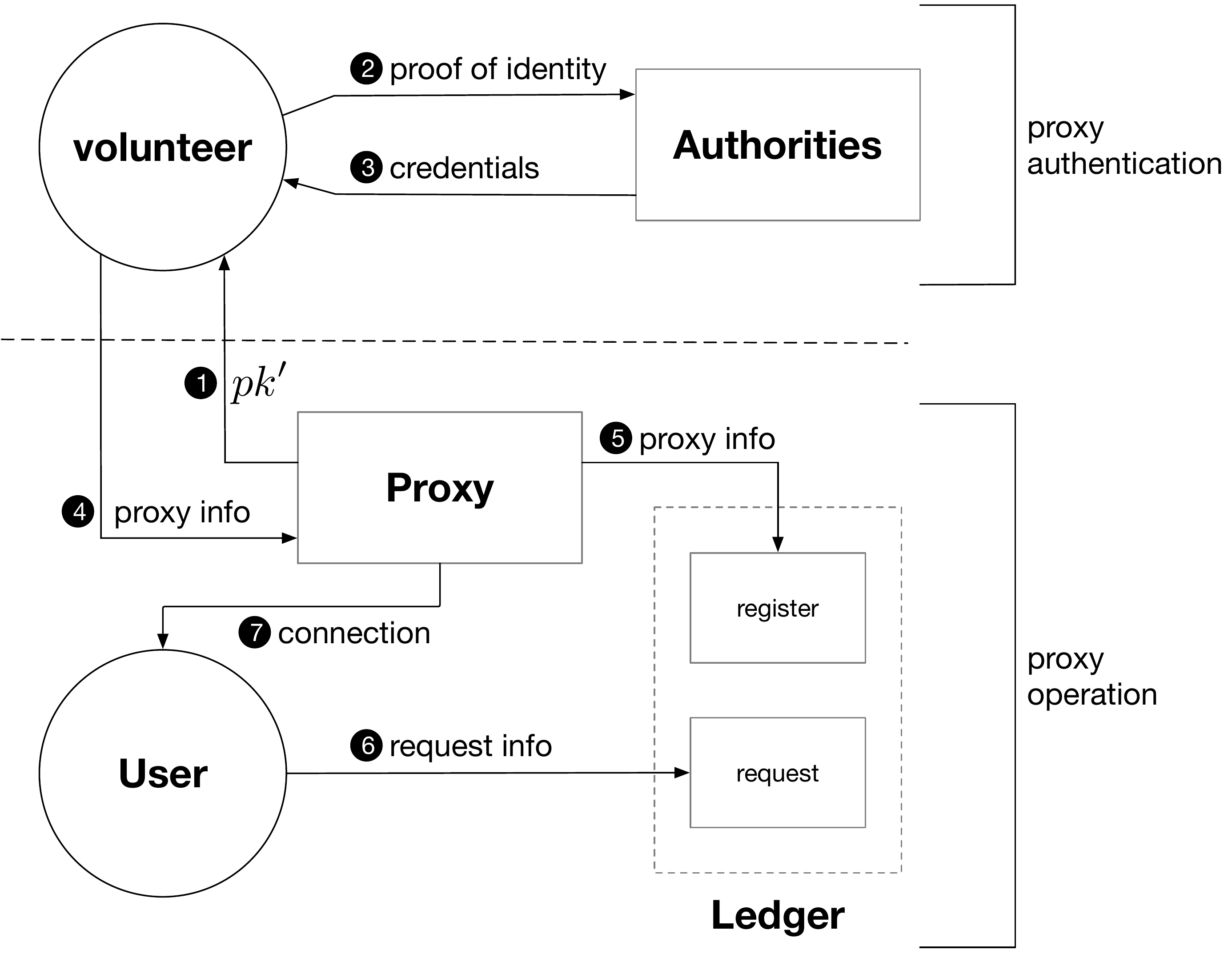}
\caption[\coconut censorship-resistant proxy distribution system.]{\coconut censorship-resistant proxy distribution system. Volunteers running proxies first receive credentials from the \coconut authorities. They then use these credentials to advertise their services on chain and connect to users.}
\label{fig:coconut:proxy}
\end{figure}

\para{Security consideration}
A number of CRSes have been previously proposed that employ techniques such as tunnelling traffic through `unblockable' protocols, and covert proxies and channels (\eg user-generated content on social media platforms)~\cite{censorship:sok}. %
A common limitation of censorship resistance schemes is relying on volunteers that are \emph{assumed} to be resistant to coercion: either \first the volunteer is a large, commercial organization (\eg Amazon or Google) over which the censor cannot exert its influence; and/or \second the volunteer is located outside the country of censorship. 
However, both these assumptions were proven wrong~\cite{meek-suspended,five-eyes}. 
The proposed CRS overcomes this limitation by offering coercion-resistance to volunteers from censor-controlled users and authorities. Due to \coconut's blindness property, a volunteer can get a credential on its IP address and ephemeral public key without revealing those to the authorities. The users get proxy IP addresses run by the volunteer, while being unable to link it to the volunteer's long-term public key.
It might be hard for a censor to take down large, commercial parties---but these can be potentially forced to stop supporting the CRS~\cite{meek-suspended}. Similarly, the emergence of global surveillance coalitions invalidates prevailing CRS assumptions based on the censor's geographic reach~\cite{five-eyes}.
The proposed CRS relies on multiple organizations validated by a distributed set of authorities, that can run proxies with complete deniability. 
Moreover, the authorities operate independently and can be controlled by different entities, and are resilient against a threshold number of authorities being dishonest or taken down. 
%This is a significant improvement over existing systems like Tor where proxies are authenticated by a dedicated group of authorities called directory servers controlled by a single party (Tor), and where the directory servers are vulnerable to take down~\cite{tor-ds}. \george{Bano: I do not think this applies to router authorities, but it does, I think, to the bridge authority.}

%% file: chapters/coconut/sections/evaluation.tex
% =========
% Evaluation 
% =========
\section{Evaluation} \label{sec:coconut:evaluation}
We present the evaluation of the \coconut threshold credentials scheme; first we present a benchmark of the cryptographic primitives described in \Cref{sec:coconut:construction} and then we evaluate the smart contracts described in \Cref{sec:coconut:applications}.

% =========
\subsection{Cryptographic Primitives}
We implement the primitives described in Section~\ref{sec:coconut:construction} in Python using petlib~\cite{petlib} and bplib~\cite{bplib}. The bilinear pairing is defined over the Barreto-Naehrig~\cite{Barreto-Naehrig} curve, using OpenSSL as arithmetic backend.

\begin{table}[t]
\centering
\begin{tabular}{lcc}
\toprule
\small\textbf{Operation}\hspace{2cm} & \small\textbf{Mean (ms)} & \small\textbf{Std. (ms)}\\
\midrule
\textsf{PrepareBlindSign} & 2.633 & $\pm$ 0.003 \\ 
\textsf{BlindSign} & 3.356 & $\pm$ 0.002 \\
\textsf{Unblind} & 0.445 & $\pm$ 0.002 \\ 
\textsf{AggCred} & 0.454 & $\pm$ 0.000 \\
\textsf{ProveCred} & 1.544 & $\pm$ 0.001 \\
\textsf{VerifyCred} & 10.497 & $\pm$ 0.002 \\
\bottomrule
\end{tabular}
\caption[Execution time of \coconut primitives.]{\footnotesize Execution times for the cryptographic primitives described in \Cref{sec:coconut:construction}, measured for one private attribute over 10,000 runs. \textsf{AggCred} is computed assuming two authorities; the other primitives are independent of the number of authorities.}
\label{tab:coconut:timing}
\end{table}
\begin{table}[t]
\centering
\begin{tabular}{lcc}
\toprule
\multicolumn{3}{l}{Number of authorities: $n$, Signature size: 132 bytes}\\
\small\textbf{Transaction} & \small\textbf{Complexity} & \small\textbf{Size [B]}\\
\midrule
\multicolumn{1}{l}{Signature on one public attribute:}\\
\ding{202} request credential & $O(n)$ & 32 \\ 
\ding{203} issue credential &  $O(n)$ & 132 \\
\ding{204} verify credential &  $O(1)$ &  162 \\ 
&&\\
\multicolumn{1}{l}{Signature on one private attribute:}\\
\ding{202} request credential &  $O(n)$ & 516 \\ 
\ding{203} issue credential &  $O(n)$ & 132 \\
\ding{204} verify credential &  $O(1)$ &  355 \\
\bottomrule
\end{tabular}
\caption[Communication complexity of \coconut credentials.]{\footnotesize Communication complexity and transaction size for the \coconut credentials scheme when signing one public and one private attribute (see \Cref{fig:coconut:protocol_priv} of \Cref{sec:coconut:construction}).}
\label{tab:coconut:transactions}
\end{table}

\para{Timing benchmark} \Cref{tab:coconut:timing} shows the mean and standard deviation of the execution of each procedure described in section \Cref{sec:coconut:construction}. Each entry is the result of 10,000 runs measured on an desktop computer, 3.6GHz Intel Xeon. Signing is much faster than verifying credentials---due to the pairing operation in the latter; verification takes about 10ms; signing a private attribute is about 3 times faster.

\para{Communication complexity and packets size} \Cref{tab:coconut:transactions} shows the communication complexity and the size of each exchange involved in the \coconut credentials scheme, as presented in \Cref{fig:coconut:protocol_priv}. The communication complexity is expressed as a function of the number of signing authorities ($n$), and the size of each attribute is limited to 32 bytes as the output of the SHA-2 hash function. The size of a credential is 132 bytes. The highest transaction sizes are to request and verify credentials embedding a private attribute; this is due to the proofs $\pi_s$ and $\pi_v$ (see \Cref{sec:coconut:construction}). The proof $\pi_s$ is approximately 318 bytes and $\pi_v$ is 157 bytes.

\para{Client-perceived latency} We evaluate the client-perceived latency for the \coconut threshold credentials scheme for authorities deployed on Amazon AWS~\cite{aws} when issuing partial credentials on one public and one private attribute. The client requests a partial credential from 10 authorities, and latency is defined as the time it waits to receive $t$-out-of-10 partial signatures. \Cref{fig:coconut:latency} presents measured latency for a threshold parameter t ranging from 1--10. 
The dots correspond to the average latency and the error-bars represent the normalized standard deviation, computed over 100 runs. The client is located in London while the 10 authorities are geographically distributed across the world; US East (Ohio), US West (N. California), Asia Pacific (Mumbai), Asia Pacific (Singapore), Asia Pacific (Sydney), Asia Pacific (Tokyo), Canada (Central), EU (Frankf{\"u}rt), EU (London), and South America (S{\~a}o Paulo). All machines are running a fresh 64-bit Ubuntu distribution, the client runs on a \emph{large} AWS instance and the authorities run on \emph{nano} instances.
As expected, we observe that the further the authorities are from the client, the higher the latency due to higher response times; the first authorities to respond are always those situated in Europe, while Sidney and Tokyo are the latest. 
Latency grows linearly, with the exception of a large jump (of about 150 ms) when $t$ increases from 2 to 3---this is due to the 7 remaining authorities being located outside Europe. 
The latency overhead between credential requests on public and private attributes remains constant.

\begin{figure}[t]
\centering
\includegraphics[width=.7\textwidth]{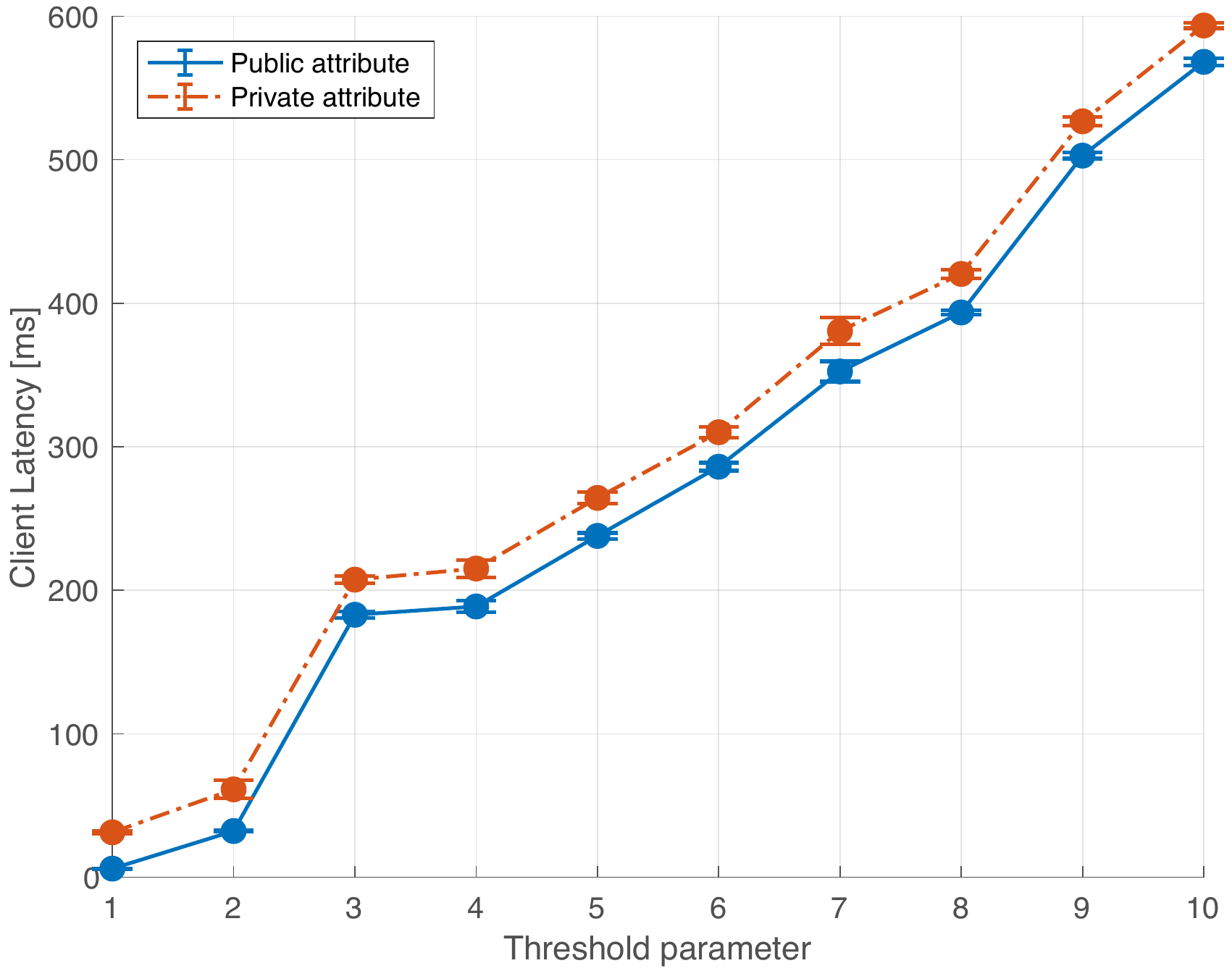}
\caption[\coconut client-perceived latency.]{\footnotesize Client-perceived latency for \coconut threshold credentials scheme with geographically distributed authorities, measured for one attribute over 100 runs. }
\label{fig:coconut:latency}
\end{figure}
\begin{table}[t]
\centering
\begin{tabular}{lccc}
\toprule
\multicolumn{4}{l}{\coconut smart contract library}\\
\small\textbf{Operation} & \small\textbf{Mean (ms)} & \small\textbf{Std. (ms)} & \textbf{Size (kB)}\\
\midrule
\textsf{Create} [g] & 0.195 & $\pm$ 0.065 & $\sim1.38$\\ 
\textsf{Create} [c] & 12.099 & $\pm$ 0.471 & -\\ 
\textsf{Request} [g] & 7.094 & $\pm$ 0.641 & $\sim3.77$\\ 
\textsf{Request} [c] & 6.605 & $\pm$ 0.559 & -\\ 
\textsf{Issue} [g] & 4.382 & $\pm$ 0.654 & $\sim3.08$\\ 
\textsf{Issue} [c] & 0.024 & $\pm$ 0.001 & -\\  
\textsf{Verify} [g] & 5.545 & $\pm$ 0.859 & $\sim1.76$\\ 
\textsf{Verify} [c] & 10.814 & $\pm$ 1.160 & - \\
\bottomrule
\end{tabular}
\caption[Performance of the \coconut smart contract library for \chainspace]{\footnotesize Timing and transaction size of the \chainspace implementation of the \coconut smart contract library described in \Cref{sec:coconut:smart_contract_library}, measured for two authorities and one private attribute over 10,000 runs. The notation [g] denotes the execution the procedure and [c] denotes the execution of the checker.}
\label{tab:coconut:library}
\end{table}
\begin{table}[t]
\centering
\begin{tabular}{lccc}
\toprule
\multicolumn{4}{l}{Coin tumbler}\\
\small\textbf{Operation} & \small\textbf{Mean (ms)} & \small\textbf{Std. (ms)} & \textbf{Size (kB)}\\
\midrule
\textsf{InitTumbler} [g] & 0.235 & $\pm$ 0.065 & $\sim1.38$\\ 
\textsf{InitTumbler} [c] & 19.359 & $\pm$ 0.773 & -\\ 
\textsf{Pay} [g] & 11.939 & $\pm$ 0.792 & $\sim4.28$\\ 
\textsf{Pay} [c] & 6.625 & $\pm$ 0.559 & -\\ 
\textsf{Redeem} [g] & 0.132 & $\pm$ 0.012 & $\sim3.08$\\ 
\textsf{Redeem} [c] & 11.742 & $\pm$ 0.757 & -\\  
\bottomrule
\end{tabular}
\caption[Performance of the \coconut coin tumbler application.]{\footnotesize Timing and transaction size of the \chainspace implementation of the coin tumbler smart contract (described in \Cref{sec:coconut:applications}), measured over 10,000 runs. The transactions are independent of the number of authorities.  The notation [g] denotes the execution the procedure and [c] denotes the execution of the checker.}
\label{tab:coconut:tumbler}
\end{table}
\begin{table}[t]
\centering
\begin{tabular}{lccc}
\toprule
\multicolumn{4}{l}{Privacy-preserving e-petition}\\
\small\textbf{Operation} & \small\textbf{Mean (ms)} & \small\textbf{Std. (ms)} & \textbf{Size (kB)}\\
\midrule
\textsf{InitPetition} [g] & 3.260 & $\pm$ 0.209 & $\sim1.50$\\ 
\textsf{InitPetition} [c] & 3.677 & $\pm$ 0.126 & -\\ 
\textsf{SignPetition} [g] & 7.999 & $\pm$ 0.467 & $\sim3.16$\\ 
\textsf{SignPetition} [c] & 15.801 & $\pm$ 0.537 & -\\ 
\bottomrule
\end{tabular}
\caption[Performance of the \coconut petition application.]{\footnotesize Timing and transaction size of the \chainspace implementation of the privacy-preserving e-petition smart contract (described in \Cref{sec:coconut:applications}), measured over 10,000 runs. The transactions are independent of the number of authorities. The notation [g] denotes the execution the procedure and [c] denotes the execution of the checker.}
\label{tab:coconut:petition}
\end{table}
%

% =========
\subsection{\chainspace Implementation}
We evaluate the \coconut smart contract library implemented in \chainspace, as well as the coin tumbler (\Cref{sec:coconut:tumbler}) and the privacy-preserving e-petition (\Cref{sec:coconut:petition}) applications that use this library.
As expected, \Cref{tab:coconut:library} shows that the most time consuming procedures are the checker of \textsf{Create} and the checker of \textsf{Verify}; i.e., they call the \textsf{VerifyCred} primitives which takes about 10 ms (see \Cref{tab:coconut:timing}). \Cref{tab:coconut:library} is computed assuming two authorities; the transaction size of \textsf{Issue} increases by about 132 bytes (\ie the size of the credentials) for each extra authority\footnote{The \textsf{Request} and \textsf{Issue} procedures are only needed in the case of on-chain issuance (see \Cref{sec:coconut:smart_contract_library}).} while the other transactions are independent of the number of authorities.
Similarly, the most time consuming procedure of the coin tumbler (\Cref{tab:coconut:tumbler}) application and of the privacy-preserving e-petition (\Cref{tab:coconut:petition}) are the checker of \textsf{InitTumbler}  and the checker of \textsf{SignPetition}, respectively; these two checkers call the \textsf{BlindVerify} primitive involving pairing checks. The \textsf{Pay} procedure of the coin tumbler presents the highest transaction size as it is composed of two distinct transactions: a coin transfer transaction and a \textsf{Request} transaction from the \coconut contract library. However, they are all practical, and they all run in a few milliseconds. These transactions are independent of the number of authorities as issuance is either handled off-chain or by the \coconut smart contract library.

% =========
\subsection{\ethereum Implementation}
We evaluate the \coconut Ethereum smart contract library described in \Cref{sec:coconut:ethereum_smart_contract_library} using the Go implementation of Ethereum on an Intel Core i5 laptop with 12GB of RAM running Ubuntu 17.10. \Cref{tab:coconut:ethlibrary} shows the execution times and gas costs for different procedures in the smart contract. The execution times for \algorithm{Create} and \algorithm{Verify} are higher than the execution times for the \chainspace version (\Cref{tab:coconut:library}) of the library, due to the different implementations. The arithmetic underlying \coconut in \chainspace is performed through Python naively binding to C libraries, while in \ethereum arithmetic is defined in solidity and executed by the EVM.

\begin{table}[t]
\centering
\begin{tabular}{lccc}
\toprule
\multicolumn{4}{l}{\coconut \ethereum smart contract library}\\
\small\textbf{Operation} & \small\textbf{Mean (ms)} & \small\textbf{Std. (ms)} & \textbf{Gas}\\
\midrule
\textsf{Create} & 27.45 & $\pm$ 3.054 & $\sim23,000$\\ 
\textsf{Verify} & 120.17 & $\pm$ 25.133 & $\sim2,150,000$\\ 
\bottomrule
\end{tabular}
\caption[Performance of the \coconut smart contract library for \ethereum.]{\footnotesize Timing and gas cost of the \ethereum implementation of the \coconut smart contract library described in \Cref{sec:coconut:ethereum_smart_contract_library}. Measured over 100 runs, for one public attribute. The transactions are independent of the number of authorities.}
\label{tab:coconut:ethlibrary}
\end{table}

We also observe that the \algorithm{Verify} function has a significantly higher gas cost than \algorithm{Create}. This is mostly due to the implementation of elliptic curve multiplication as a native \ethereum smart contract---the elliptic curve multiplication alone costs around $1,700,000$ gas, accounting for the vast majority of the gas cost, whereas the pairing operation using the pre-compiled contract costs only 260,000 gas. The actual fiat USD costs corresponding to those gas costs, fluctuate wildly depending on the price of Ether---Ethereum's internal value token---the load on the network, and how long the user wants to wait for the transaction to be mined into a block. As of February 7th 2018, for a transaction to be confirmed within 6 minutes, the transaction fee for \algorithm{Verify} is \$1.74, whereas within 45 seconds, the transaction fee is \$43.5.\footnote{\url{https://ethgasstation.info/}}
The bottleneck of our \ethereum implementation is the high-level arithmetic in $\mathbb{G}_2$. However, \ethereum provides a pre-compiled contract for arithmetic operations in $\mathbb{G}_1$. We could re-write our cryptographic primitives by swapping all the operations in $\mathbb{G}_1$ and $\mathbb{G}_2$, at the cost of relying on the SXDH assumption~\cite{ramanna2016efficient} (which is stronger than the standard XDH assumption that we are currently using).

%% file: chapters/coconut/sections/related.tex
% =========
% Related Works 
% =========
\section{Comparison with Related Works} \label{sec:coconut:related}

\definecolor{verylightgray}{gray}{0.9}
\begin{table*}[t]
\centering
\resizebox{\textwidth}{!}{\begin{tabular}{lccccc}
\toprule
\small\textbf{Scheme} & \small\textbf{Blindness} & \small\textbf{Unlinkable} & \small\textbf{Aggregable}  & \small\textbf{Threshold}  & \small\textbf{Size} \\
\midrule

\textbf{\cite{waters}} Waters Signature & \no & \no & \L & \no & 2 Elements\\
\textbf{\cite{lossw}} LOSSW Signature & \no & \no & \M & \no & 2 Elements\\
\textbf{\cite{bgls}} BGLS Signature & \yes & \no & \H & \yes & 1 Element\\
\textbf{\cite{cl}} CL Signature & \yes & \yes & \M & \no & $O(q)$ Elements\\ %(3+2q)
\textbf{\cite{idemix}} Idemix & \yes & \yes & \L & \no & $O(q)$ Elements \\
\textbf{\cite{uprove}} U-Prove & \yes & \yes & \L & \no & $O(v)$ Elements\\ 
\textbf{\cite{acl}} ACL & \yes & \yes & \L & \no & $O(v)$ Elements\\ 
\textbf{\cite{pointcheval}} Pointcheval and Sanders & \yes & \yes & \M & \no & 2 Elements\\ 
\textbf{\cite{dac}} Garman~\etal & \yes & \yes & - & \no & 2 Elements\\ 
\rowcolor{verylightgray} \textbf{[\Cref{sec:coconut:construction}]} \coconut & \yes & \yes & \H & \yes & 2 Elements\\ 

\bottomrule
\end{tabular}}
\caption[Comparison of \coconut with related work.]{\footnotesize Comparison of \coconut with other relevant cryptographic constructions. The aggregability of the signature scheme reads as follow;  \L: not aggregable,  \M: sequentially aggregable,  \H: aggregable. The signature size is measured asymptotically or in terms of the number of group elements it is made of (for constant-size credentials); $q$ indicates the number of attributes embedded in the credentials and $v$ the number of times the credential may be shown unlinkably.}
\label{tab:coconut:related_works}
\end{table*}

We compare the \coconut cryptographic constructions and system with related work in \Cref{tab:coconut:related_works}, along the dimensions of key properties offered by \coconut---blindness, unlinkability, aggregability (i.e., whether multiple authorities are involved in issuing the credential), threshold aggregation (i.e., whether a credential can be aggregated using signatures issued by a subset of authorities), and signature size (see Sections~\ref{sec:coconut:architecture}~and~\ref{sec:coconut:construction}). \Cref{sec:literature-review:coconut-related} provides a detailed description of all the schemes mentioned in \Cref{tab:coconut:related_works}.

%% file: chapters/coconut/sections/limitations.tex
% =========
% Limitations
% =========
\section{Limitations} \label{sec:coconut:limitations}
\coconut has a number of limitations that are beyond the scope of this work, and deferred to future work. 
Adding and removing authorities implies to re-run the key generation algorithm---this limitation is inherited from the underlying Shamir's secret sharing protocol~\cite{shamir1979share} and can be mitigated using techniques coming from proactive secret sharing introduced by Herzberg~\etal~\cite{herzberg1995proactive}. However, credentials issued by authorities with different key sets are distinguishable, and therefore frequent key rotation reduces the privacy provided. 
As any threshold system, \coconut is vulnerable if more than the threshold number of authorities are malicious; colluding authorities could create coins to steal all the coins in the buffer of the coin tumbler application (\Cref{sec:coconut:tumbler}), create fake identities or censor legitimate users of the electronic petition application (\Cref{sec:coconut:petition}), and defeat the censorship resistance of our proxy distribution application (\Cref{sec:coconut:proxy}). Note that users' privacy is still guaranteed under colluding authorities, or an eventual compromise of their keys.
Implementing the \coconut smart contract library on \ethereum is expensive (\Cref{tab:coconut:ethlibrary}) as \ethereum does not provide pre-compiled contracts for elliptic curve arithmetic in $\mathbb{G}_2$; re-writing our cryptographic primitives by swapping all the operations in $\mathbb{G}_1$ and $\mathbb{G}_2$ would dramatically reduce the gas cost (and be secure under SXDH~\cite{ramanna2016efficient}).

%% file: chapters/coconut/sections/conclusion.tex
% =========
% Conclusion 
% =========
\section{Chapter Summary} \label{sec:coconut:conclusion}
Existing selective credential disclosure schemes do not provide the full set of desired properties, particularly when it comes to efficiency and issuing general purpose selective disclosure credentials without sacrificing desirable distributed trust assumptions. This limits their applicability in distributed settings such as distributed ledgers.%, and prevents security engineers from implementing separation of duty policies that are effective in preserving integrity. 
In this chapter, we present \coconut---a novel scheme that supports distributed threshold issuance, public and private attributes, re-randomization, and multiple unlinkable selective attribute revelations. 
We provide an overview of the \coconut system, and the cryptographic primitives underlying \coconut;
an implementation and evaluation of \coconut as a smart contract library in \chainspace and \ethereum, a sharded and a permissionless blockchain respectively; and three diverse and important application to anonymous payments, petitions and censorship resistance. \coconut fills an important gap in the literature and enables general purpose selective disclosure credentials---an important privacy enhancing technology---to be efficiently used in settings with no natural single trusted third party to issue them, and to interoperate with modern transparent computation platforms.

%% file: chapters/conclusion/conclusion.tex
% =========
% Discussion
% =========
\chapter{Conclusion}
This thesis proposed technologies to overcome the following limitations of blockchain technologies and smart contract platforms: \first poor scalability, \second high latency, and \third difficulty to operate on secret values (privacy).

\Cref{chainspace} presented \chainspace---an open, distributed ledger platform for high-integrity and transparent processing of transactions. \chainspace offers extensibility though privacy-friendly smart contracts. We presented an instantiation of \chainspace by parameterizing it with a number of `system' and `application' contracts, along with their evaluation. However, unlike existing smart-contract based systems such as \ethereum, it offers high scalability through sharding across nodes, while offering high auditability. As such it offers a competitive alternative to both centralized and permissioned systems, as well as fully peer-to-peer, but unscalable systems.

\Cref{byzcuit} presented the first replay attacks against cross-shard consensus protocols in sharded distributed ledgers. These attacks affect both shard-driven and client-driven consensus protocols, and allow attackers to double-spend with minimal efforts. The attacker can act independently without colluding with any nodes, and succeed even if all nodes are honest; most of the attacks work without making any assumptions on the underlying network. While addressing these attacks seems like an implementation detail, their many variants illustrate that a fundamental re-think of cross-shard commit protocols is required to protect against them.
We developed \byzcuit, a new cross-shard consensus protocol merging features from \shardled and \clientled consensus protocols, and withstanding replay attacks. \byzcuit can be seen as unifying \atomix and \sbac, into a protocol that is efficient and secure. We implemented and evaluated it on a real cloud-based testbed, showing that it can process over 1,550 tps for 10 shards while keeping latency under 1 second (using \bftsmart as intra-shard consensus implementation), and that its capacity scales linearly with the number of shards.

\Cref{fastpay} presented \fastpay---a settlement layer based on consistent broadcast channels, rather than full consensus. The \fastpay design leverages the nature of payments to allow for asynchronous payments into accounts, and optional interactions with an external \fastpay to build a practical system, while providing proofs of both safely and liveness; it also proposes and evaluates a design for sharded implementation of authorities to horizontally scale and match any throughput need. 
The performance and robustness of \fastpay is beyond and above the state of the art, and validates that moving away from both centralized solutions and full consensus to manage pre-funded retail payments has significant advantages. Authorities can jointly process tens of thousands of transactions per second (we observed a peak of 160,000 tx/sec) using merely commodity hardware and lean software. A payment confirmation latency of less than 200ms across continents make \fastpay practical for point of sale payments---where goods and services need to be delivered fast and in person. Pretty much instant settlement enables retail payments to be freed from intermediaries, such as banks payment networks, since they eliminate any credit risk inherent in deferred netted end-of-day payments, that underpin today most national Fast Payment systems~\cite{bolt2014fast}. Further, \fastpay can tolerate up to one-third of authorities crashing or even becoming Byzantine without losing either safety or liveness (or performance). This is in sharp contrast with existing centralized settlement layers operating on specialized mainframes with a primary / backup crash fail strategy (and no documented technical strategy to handle Byzantine operators). Surprisingly, it is also in contrast with permissioned blockchains, which have not achieved similar levels of performance and robustness yet, due to the complexity of engineering and scaling full \bftlong consensus protocols.

\Cref{coconut} presented \coconut---a novel scheme that supports distributed threshold issuance, public and private attributes, re-randomization, and multiple unlinkable selective attribute revelations. Previous selective credential disclosure schemes do not provide the full set of desired properties, particularly when it comes to efficiency and issuing general purpose selective disclosure credentials without sacrificing desirable distributed trust assumptions. This limits their applicability in distributed settings such as distributed ledgers, and prevents security engineers from implementing separation of duty policies that are effective in preserving integrity. 
We provide an overview of the \coconut system, and the cryptographic primitives underlying \coconut; an implementation and evaluation of \coconut as a smart contract library in \chainspace and \ethereum, a sharded and a permissionless blockchain respectively; and three diverse and important application to anonymous payments, petitions and censorship resistance. \coconut fills an important gap in the literature and enables general purpose selective disclosure credentials---an important privacy enhancing technology---to be efficiently used in settings with no natural single trusted third party to issue them, and to interoperate with modern transparent computation platforms.

\section{Future Directions}
A number of future directions are left open to explore. 

\para{Permissionless sharded systems}
Sharded distributed systems such as \chainspace raise the concern of how to map nodes to shards. In permissioned systems, this process is usually done according to the policy of a designated party, but it remains an open challenge for open systems.
Similarly, the reconfiguration of sharded ledgers is also an open question; if nodes are dynamically reconfigured across shards, there needs to be a mechanism  to transfer the blockchain’s state from one node to another without losing liveness during this process.
The design of a mechanism to avoid the creation of malicious shards is also an open question. \chainspace allows contract creators to designate which shards are responsible for handling the state associated with their smart contract. It is however unclear how the system could recover if smart contract creators select malicious shards.

\para{Scaling intra-shard consensus}
\chainspace scales by adding new shards to the system and effectively running multiple instances of intra-shard consensus in a coordinated fashion. However, it uses intra-shard consensus as a black box; the prototype of \byzcuit presented in \Cref{sec:byzcuit:implementation} is implemented using \bftsmart but any BFT consensus protocol would work as a drop-in replacement, and it works in a setting where each node is run by a distinct authority. Another research direction is to explore how a single authority could scale out by implementing its node across multiple machines (possibly a whole data center). Such a system is not trivial to design as it requires an intra-shard consensus protocol that can efficiently share its load and take advantage of multiple machines.

\para{Scaling execution}
Scaling smart contract execution is another open research direction. Smart contracts may contain heavy operations and execution could quickly become the bottleneck of the system, eventually slowing down consensus and harming scalability.  Despite \chainspace does not execute smart contracts, it still needs to verify the correctness of transactions which could sometimes be even more expensive than re-executing the transaction (depending on the contract). A possible direction to investigate to solve this problem is the design of a general purpose zk-SNARK verifier that could be used to verify any transaction; zk-SNARKs are particularly suited for blockchains as they are typically cheap to verify (and heavy to compute). An another direction could be to scale execution across multiple machines; this could be simpler for smart contract with special semantics (\eg payments, as shown by \fastpay in \Cref{fastpay}) but is a challenge for general purpose smart contracts. The traditional database literature explores this path but blockchains have the option to rely on totally ordered sequence of transactions before attempting execution.

\para{Ensuring resilience under attack}
The key advantage of blockchain technologies is their resilience to Byzantine behavior. Blockchains operate under the assumption that a subset of the nodes can behave arbitrarily maliciously, but how to test blockchains under attack is an open question. The core of the challenge comes from allowing Byzantine nodes to perform any kind of actions (in contrast with crash-tolerant protocols where the adversary is constrained to a specific set of actions). Further, it is not enough to ensure that safety and liveness are preserved under attack, we must also ensure that the adversary cannot significantly slow down the system and make it unpractical. This is particularly true for leader-based protocols that need to handle the possibility of malicious leaders.

\para{Incentive mechanism}
Another direction for future works is the integration of incentive mechanisms into open systems; sharded distributed systems raise the additional concern that any transaction fee may be diluted as it is split amongst the nodes of multiple shards, and thus may not be sufficient to incentivize honest behavior.
Closely related to that issue, preventing denial of service attacks against sharded ledgers is a challenge; sharded consensus usually involve a high number of communications (coming from both the intra-shard and inter-shard consensus protocols) which are potential vectors for denial-of-service attacks.

\para{Reconfiguration without consensus}
Efficiently replacing a set of authorities in the absence of consensus remains an open challenge; this applies to both \fastpay and \coconut, as well as to any distributed system that does not implement consensus.

% =========
\section{Closing Thoughts}
We showed that it is now possible to build secure and scalable distributed ledgers that can accommodate high throughput. We demonstrated that blockchain technologies can be brought to retail payment system through the use of extremely low-latency distributed side-infrastructures.
Further, we showed how to design distributed ledgers with native support for privacy-preserving applications, and how to use them to issue credentials in a blockchain setting; this enables a number of novel distributed applications.
The technologies described in this thesis can be implemented to increase the fairness, robustness, efficiency, and privacy of current payment systems, while decreasing their costs.